\pdfoutput=1
\documentclass[11pt]{article}

\newif\ifshowproofs
\newif\ifmainbuild
\newif\ifsupplement
\showproofstrue     %
\mainbuildfalse     %
\supplementfalse

 \usepackage{amssymb,amsmath,amsthm}
\usepackage[mathscr]{eucal}
\usepackage{tikz}
\usepackage{amscd}
\usepackage{mathtools}
\usepackage{verbatim}
\usepackage{bbm}
\usepackage{longtable}
\usepackage[round]{natbib}
\usepackage{setspace}

\usepackage{graphicx}
\usepackage{placeins}

\newtheorem{theorem}{Theorem}[section]
\newtheorem{lemma}[theorem]{Lemma}
\newtheorem{cor}{Corollary}[section]
\newtheorem{prop}{Proposition}[section]
\newtheorem{addendum}{Addendum}[section]

\theoremstyle{definition}

\newtheorem{exm}{Example}
\newtheorem{defn}{Definition}[section]
\newtheorem{rem}{Remark}[section]
\newtheorem{asm}{Assumption}[section]
\usepackage{chngcntr}
\counterwithin{table}{section}

\newcommand{\R}{\mathbb{R}}

\newcommand{\N}{\mathbb{N}}

\newcommand{\ve}{\varepsilon}

\DeclarePairedDelimiterX{\infdivx}[2]{(}{)}{%
  #1\;\delimsize\|\;#2%
}
\newcommand{\inner}[1]{\left\langle #1 \right\rangle}

\DeclarePairedDelimiter{\norm}{\lVert}{\rVert}
\newcommand{\E}[1]{\mathrm{E}\left[ #1 \right]}

\renewcommand{\hat}{\widehat}
\renewcommand{\P}[1]{\mathrm{P}\left( #1 \right)}

\newcommand{\one}{\mathbf{1}}
\newcommand{\cp}{\overset{p}{\rightarrow}}

\newcommand{\EE}[2]{\mathrm{E}_{#1}\left[ #2 \right]}

\newcommand{\PP}[2]{\mathrm{P}_{#1}\left( #2 \right) }

\renewcommand{\mod}[1]{{\ifmmode\text{\rm\ (mod~$#1$)}\else\discretionary{}{}{\hbox{ }}\rm(mod~$#1$)\fi}}
\newcommand{\ra}{\rightarrow}

\newcommand{\mP}{\mathrm{P}}

\renewcommand{\bar}{\overline}
\numberwithin{equation}{section}
\numberwithin{figure}{section}

\numberwithin{table}{section}

\newcommand{\rs}{\rightsquigarrow}

\usepackage{enumitem}
\setlist{
listparindent=\parindent,
parsep=0pt,
}
\usepackage[margin=1in]{geometry}
\newcommand{\suppref}{}
\newcommand{\Secapp}{\ifmainbuild Section\else\ifsupplement Section\else Appendix\fi\fi}
\newcommand{\suppmatin}{\unskip}
\makeatletter
\newcommand{\xrimport}[1]{%
  \let\xr@saved@bibcite\bibcite
  \let\bibcite\@gobbletwo
  \externaldocument{#1}%
  \let\bibcite\xr@saved@bibcite
}
\makeatother

\ifsupplement
  \usepackage{xr-hyper}
  \usepackage{hyperref}
  \xrimport{main}
\else
  \ifmainbuild
    \usepackage{xr-hyper}
    \usepackage{hyperref}
    \xrimport{supplement}
  \else
    \usepackage{hyperref}
  \fi
\fi

\ifsupplement
  \hypersetup{pdftitle={Supplementary Material: Inference under partial identification with minimax test statistics}}
\else
  \hypersetup{pdftitle={Inference under partial identification with minimax test statistics}}
\fi
\hypersetup{
  pdfauthor={Isaac Loh},
  pdfsubject={Inference under partial identification},
  pdfkeywords={partial identification; minimax test statistics; composite hypotheses; directional differentiability; bootstrap}
}
\hypersetup{colorlinks=true, linkcolor=blue!60!black, citecolor=blue!60!black, urlcolor=blue!60!black}

\begin{document}

\title{Inference under partial identification with minimax test statistics}

\author{Isaac Loh%
\thanks{Department of Economics and Finance, University of North Carolina Wilmington,
Wilmington, NC 28403.
Email: \href{mailto:lohi@uncw.edu}{lohi@uncw.edu}.}}

\date{\ifmainbuild\else This version: November 2026\fi}

\maketitle

\begin{abstract}
\noindent In partially identified models, specification and subvector tests often minimize a
measure of sample violation over the parameter space.  Such statistics have the minimax form
$T_n=r_n\inf_{\theta\in\Theta}\sup_{t\in\mathcal{T}}v_n(\theta,t)$, which includes
$r_n\inf_{\theta\in\Theta}\xi(m_n(\theta))$ for a convex loss $\xi$ of Banach-space-valued
moments $m_n$.  We characterize their limiting distribution through familiar geometric objects.
The limit is an upper bound in general and exact under convexity in $\theta$, a standard
polynomial minorant condition, or continuity of the contact set.  It retains the loss directions that bind
locally and cannot be absorbed by moving $\theta$, which under convex loss are the subgradients
of $\xi$ at $m(\theta)$ that are nonnegative along every first-order movement of $\theta$.  A
re-optimization bootstrap common to \citet{HongLi2018} and \citet{BCS2017} re-solves the
minimization after a random perturbation and needs no estimate of the identified set, the contact set, or any
derivative.  We show that it consistently estimates the limit under any of the three conditions for exactness.  Under a standard uniform minorant condition, it is also uniformly valid for minimax criteria and for general convex losses on Banach-valued moments.  Without a uniform minorant, the two properties can be incompatible.  In simulations, the re-optimization bootstrap has close to nominal size and closely tracks the sampling distribution of the statistic.
\end{abstract}

\noindent\textbf{Keywords:} partial identification; minimax test statistics; composite hypotheses; directional differentiability; bootstrap

\medskip
\noindent\textbf{JEL classification:} C12; C14; C15

\bigskip

\section{Introduction}

Many econometric models imply feasibility restrictions rather than a single
identified parameter value.  In a moment inequality model, for example, a
specification test asks whether any $\theta\in\Theta$ satisfies all
population inequalities.  In an overidentified GMM model with a partially
identified structural parameter, the specification question is whether the
moment equalities hold at some $\theta$.  Similar
nonemptiness questions arise in conditional moment models, random-set models,
and linear-programming formulations of sharp bounds.  They also arise in
subvector inference, which tests whether a candidate value $c$ of a functional
$g(\theta)$ is compatible with the model.

This paper studies inference on such composite feasibility hypotheses.  Let
$m:\Theta\to\mathfrak{X}$ be a moment map taking values in a Banach space, and
let $\xi:\mathfrak{X}\to\R\cup\{\infty\}$ be a convex loss (nonnegative in the leading examples) measuring violation of the model
restrictions.  The null hypothesis is
\[
  H_0:\Theta_{\mathrm{I}}\neq\emptyset,\qquad
  \Theta_{\mathrm{I}}=\{\theta\in\Theta:\xi(m(\theta))=0\}.
\]
Given an estimator $m_n$ of $m$, the natural statistic is the minimum sample
violation
\[
  r_n\inf_{\theta\in\Theta}\xi(m_n(\theta)),
\]
where $r_n$ is a normalizing rate (e.g., $\sqrt{n}$).  Representing the convex loss by its supporting affine
functions puts the statistic in the minimax form
\begin{align*}
T_n = r_n \inf_{\theta \in \Theta} \sup_{t \in \mathcal{T}} v_n(\theta, t),
\end{align*}
where the outer infimum selects the parameter value most consistent with the
data and the inner supremum scores the remaining violation.  Composite
moment inequality tests, conditional moment tests with a fixed instrument class,
GMM-type overidentification tests, support-function
specification tests, and LP-valued inference are all examples (see Section~\ref{S:examples}).

The main result characterizes the pointwise limiting
distribution of $T_n$ under the composite null by an upper bound that is exact under convexity in $\theta$, a standard polynomial minorant condition, or continuity of the contact set.  For the general
minimax statistic, Theorem~\ref{T:one} shows that under these conditions the limit is
\[
  \inf_{\theta\in\Theta_{\mathrm{I}}}\sup_{t\in\mathcal{K}(\theta)}
  \mathbb{G}(\theta,t),
\]
where $\mathbb{G}$ is the limit of the centered process $\mathbb{V}_n=r_n(v_n-v)$, $v$ is the population analogue of $v_n$, and $\mathcal{K}(\theta)$ is a
contact set.  The contact set keeps precisely those indices $t$ that attain the population
inner supremum at $\theta$ and cannot be made irrelevant by a local movement of
$\theta$.

For convex-loss statistics, the contact set becomes explicit.  Writing
$v(\theta,t)=t(m(\theta))$ for the supporting affine functions $t$ of $\xi$, convex duality makes
$\mathcal{K}(\theta)$ the intersection of two parts.  Its \emph{envelope part}, the subdifferential
$\partial\xi(m(\theta))$, keeps the supporting hyperplanes of the loss at $m(\theta)$.  Its
\emph{stationarity part}, the polar cone $(-m'(\theta)(T_\theta\Theta))^\circ$, keeps the directions
that local movement of $\theta$ cannot absorb, where $T_\theta\Theta$ is the tangent cone of $\Theta$
at $\theta$ and $m'(\theta)(T_\theta\Theta)$ is the linearized set of locally reachable directions.
The limit therefore takes the form
\begin{align} \label{E:limit_intro}
\inf_{\theta \in \Theta_{\mathrm{I}}} \sup_{y \,\in\, (-m'(\theta)(T_\theta\Theta))^\circ \cap \partial\xi(m(\theta))} y(\mathbb{X}(\theta)),
\end{align}
where $\mathbb{X}(\theta)$ is the Gaussian limit of $r_n(m_n-m)$.  Selection thus operates on dual
directions of the loss, rather than on moments as in generalized moment selection \citep{as2010},
so equality and inequality restrictions enter through the same first-order geometry.

This geometry recovers familiar limits in special cases.  In a regular GMM
overidentification problem, suppose there are $L$ moments, a $k$-dimensional parameter, and a
singleton identified set, with full-rank derivative $m'(\theta_0)$ and efficient
weighting. The tangent-cone polar then projects the Gaussian moment process off the $k$ locally reachable
directions, and the $J$ statistic has its classical $\chi^2(L-k)$ limit.  In moment inequality models, the subdifferential selects the
locally binding inequalities while the tangent-cone polar removes binding
directions that can be relaxed by moving $\theta$ within $\Theta$.  The same
geometry carries over to Banach-space moment maps.  When the loss is the distance
to a convex set $C$, which includes both examples, the limit is the infimum over $\Theta_{\mathrm{I}}$ of the distance from
the Gaussian process to the cone $T_{m(\theta)}C-m'(\theta)(T_\theta\Theta)$.

The main technical difficulty is that the statistic combines a nonsmooth loss
with an outer minimization over a possibly non-singleton identified set.  At a
fixed parameter value, directional delta methods handle inference on
$\xi(m_n(\theta_0))$.  For the composite statistic
$\inf_\theta\xi(m_n(\theta))$, however, the usual Danskin argument \citep{Danskin1966,Shapiro1991} is not immediately available.  It treats $(\theta,m)\mapsto\xi(m(\theta))$ as
differentiable in the moment coordinate, but under the null $\xi$ is typically
nondifferentiable at the points where it is evaluated.  For instance, a norm loss
has a kink at the origin, which is exactly where $m(\theta)$ lies under the null it defines. The missing
gradient is replaced by the subdifferential, and the first-order effect of perturbing
$m(\theta)$ is the support function
\[
  \sup_{y\in\partial\xi(m(\theta))} y(\mathbb{X}(\theta)).
\]
The minimizing $\theta$ can also move locally, and ignoring that movement loses the
degrees-of-freedom correction in the GMM example above.  Working directly with the minimax form and dualizing $\xi$ makes the
subgradient the inner index, so that the expansion can record which subgradient
directions survive local movement of $\theta$.  The resulting distribution theory for minimax statistics is the
paper's first contribution, which accommodates infinite-dimensional parameter spaces and Banach-valued moments.

As a second contribution, the paper extends the theory of a bootstrap already in use to cover pointwise exact and
uniformly valid inference in this general setting.  Re-solving the
minimum-violation problem after a perturbation of size $s_n^{-1}$ is the numerical delta method of
\citet{HongLi2018} and, as \citet[Supplement, Section~9]{HongLi2020} note, for finitely many moments the penalize-resampling statistic of \citet{BCS2017}.  We show that this re-optimization bootstrap, which estimates none of the geometry above, is pointwise exact
under the same conditions as the limit, for any step $s_n$ that diverges more slowly than $r_n$ (Theorem~\ref{T:two} and
Proposition~\ref{P:ex2}).  Under a standard uniform minorant condition it is also uniformly valid for minimax criteria and for general convex losses on Banach-valued moments (Theorem~\ref{P:unif1}\suppref{} and Proposition~\ref{P:unifbound}).  Consequently, critical values formed from its quantile plus an arbitrarily small slack control size uniformly.  Where the limit is exact and its distribution function is continuous and strictly increasing at the quantile, they are also, up to the slack, the least conservative valid critical values that converge to a constant.  In particular, for finitely many moments the penalize-resampling draw is pointwise exact and
uniformly valid on its own, and the discard-resampling component of Test~MR \citep{BCS2017} can be dropped without
asymptotic loss.  In our simulations, the draw alone with an intermediate step performs about as well as Test~MR.

When the identification rate is slower, a localized and penalized variant remains uniformly valid
(Propositions~\ref{Add:2}\suppref{} and~\ref{Add:unifxi}) and pointwise consistent for the limiting bound of Theorem~\ref{T:one} (Proposition~\ref{ad:one}\suppref{}).  If no uniform identification rate can be assumed, \Secapp~\ref{S:uniform_separation}\suppref{}
shows that pointwise exactness and uniform validity can be incompatible.
For norm losses, a stationarity-only variant needs continuity of the linearized moment image alone
(\Secapp~\ref{S:stationarity}\suppref{}).  We also propose a multiscale hybrid that takes the smallest critical value over a
fixed grid of steps.

\paragraph{Related literature.}  The paper relates first to the literature on inference under partial
identification.  \citet{CHT2007} develop a general level-set approach to
estimation and confidence regions for parameter sets.  In moment inequality
models, procedures based on moment selection or recentering, including \citet{as2010} and its refinements and alternatives
\citep{AndrewsBarwick2012,RomanoShaikhWolf2014} and extensions to conditional moment inequalities \citep{andrews2013,andrews2017}, deliver uniformly valid inference for fixed parameter
hypotheses.  \citet[hereafter BCS]{BCS2015} study the related specification problem
of testing whether a finite system of moment equalities and inequalities has a
solution.  They derive the limit of the minimum-violation statistic along sequences of null distributions as an infimum over the Hausdorff limit of the scaled moment graph \citep[Theorem~3.1]{BCS2015}, and they propose tests with uniform asymptotic size control.  \citet[Theorem~S.2.4]{BCS2017} extend this limit to subvector hypotheses.  Formula~\eqref{E:limit_intro} describes the geometry of this limit through the subdifferential of the loss and the tangent cone of the parameter space, for general convex losses and Banach-valued moments.  Under the conditions that make it exact, the bootstrap of Section~\ref{S:boot} estimates this limit consistently.
\citet{marcoux2024} develop a related specification test that is valid under
weaker identification assumptions.
Under possible misspecification, \citet{AK2024} study the population analogue
of the minimum-violation statistic, the infimum over $\theta$ of the maximum
standardized moment violation.  They construct confidence intervals for the resulting misspecification
index and tests derived from it. For incomplete models, \citet{KaidoMolinari2024} develop a related
information-based method whose Rao score statistic for a candidate $\theta$ is
asymptotically chi-squared when the score covariance is nonsingular, and conservative otherwise.
Subvector inference is a further point of contact, developed for moment inequality
models by \citet{BCS2017}, through local linearization by \citet{Bei2024}, and
through regularized support functions in affine models by \citet{Gafarov2025}.  \citet{BCS2017} prove uniform size control for Test MR and observe exactness of its penalize-resampling component in an example in their Section~3, and \citet{BBC2019} extend these tests to many moment inequalities.
\citet{CCT2018} construct Monte Carlo confidence sets for the identified set and
its subvectors as level sets of quasi-likelihood-ratio and profiled
quasi-likelihood-ratio criteria.  The profiling perspective also informs \citet{HHW2025}, who
study varying-coefficient models with partially identified functional
coefficients.  The present paper complements these approaches by characterizing the local distribution of the minimum-violation statistic for nonlinear moment maps,
general convex losses, and Banach-space-valued moments.  The resulting tests may be inverted to form confidence sets for a component or other functional of the
parameter.  This connects the paper to the calibrated-projection method of
\citet{KMS2019}.

Second, the paper connects to recent work obtaining exact or asymptotically
exact tests in structured moment inequality problems.  \citet{CS2023} derive
adaptive, size-exact quasi-likelihood ratio tests for full-vector and subvector
inference in moment inequality models.  \citet{CSS2025} extend this approach to
inequalities that are linear in possibly partially identified nuisance
parameters with an estimated coefficient matrix, using a tuning-free chi-squared
critical value.  For conditional moment inequalities that are linear in the
nuisance parameters, \citet{ARP2023} profile a studentized maximum statistic by
linear programming and construct conditional and hybrid tests.  These papers exploit important finite-dimensional
linear structures.  The contribution here is to show that an analogous
geometric principle, which intersects active loss directions with directions not
absorbed by nuisance movements, continues to characterize the limit for
general Fréchet differentiable $m$ and general convex $\xi$ over Banach spaces.

Third, the argument builds on the literature on inference for directionally
differentiable functionals.  \citet{fangsantos2018}, building on \citet{Dumbgen1993}, develop resampling
methods for inference on directionally differentiable functions, for which the
delta method holds only in directional form and the standard bootstrap is inconsistent.  \citet{HongLi2018} propose a numerical delta method for
approximating directional derivatives, and \citet{FangSeo2021} use projection
methods for testing convex-cone restrictions.  The present statistic is
directionally differentiable in a more composite sense, because its derivative involves both the subgradient of $\xi$ and the outer infimum over
$\theta$.  Formula \eqref{E:limit_intro} identifies exactly how these two
sources of nonstandard behavior interact.  Corollary~\ref{C:dirderiv} supplies this composite derivative, which is needed for Hong and Li's consistency result, and
Section~\ref{S:convexloss} sharpens the first-order formula of their online Appendix~A.4.  Section~\ref{S:uniform} augments their results with uniform validity for the minimax functional, which that appendix notes is outside the scope of their analysis.

Finally, the paper relates to sieve-based inference in nonparametric and semiparametric conditional
moment models.  For nonparametric and semiparametric conditional moment
equalities, \citet{S2012} and \citet{hong2017} derive exact pointwise limits of
sieve minimum-distance statistics for the same kind of composite null, that some
element of the identified set satisfies a restriction.  \citet{CNS2023} provide
uniformly valid inference for constrained conditional moment restriction models
via a growing approximation to the local parameter space.  A related
shape-restricted treatment is \citet{Zhu2020}.  Our results complement these sieve-based analyses by working with the full parameter space and index class rather than sieve approximations to them, and by allowing general convex losses $\xi$ for mixed systems of moment equalities and inequalities.  When the parameter space or index class is too large to treat directly, sieves remain available (Remark~\ref{R:sieve}).  For norm losses, our limits coincide with the
sieve-based ones when the sieve directions are dense and are otherwise no more conservative
(display~\eqref{E:CNStype}).

\paragraph{Outline.}  Section~\ref{S:examples} introduces five motivating examples.
Section~\ref{S:model} presents the general model and asymptotic distribution
theory (Theorem~\ref{T:one} and Proposition~\ref{P:ex1}).
Section~\ref{S:boot} establishes pointwise bootstrap consistency and its
convex-loss implementation (Theorem~\ref{T:two} and
Proposition~\ref{P:ex2}).  Section~\ref{S:uniform} derives uniform asymptotic
validity of the same bootstrap over classes of data-generating
processes (Theorem~\ref{P:unif1}\suppref{} and Proposition~\ref{P:unifbound}).
Section~\ref{S:mc} reports Monte Carlo evidence comparing the re-optimization
bootstrap of Proposition~\ref{P:ex2} and its multiscale hybrid with the tests of \citet{BCS2017}.  \ifmainbuild\else Proofs are collected in Appendix~\ref{S:proofs}.  Appendices~\ref{S:stationarity} and~\ref{S:studentized_detail} treat the stationarity-only bootstraps and the studentized statistics.  Appendix~\ref{S:validity} states a supporting critical-value lemma, Appendix~\ref{S:supp_tables} reports supplementary simulations,
and Appendix~\ref{S:uniform_separation} shows that uniformity can fail without uniform separation.\fi

\subsection{Notation}

\subsubsection{General notation}

We write $\R_+$ for the set of nonnegative reals. In a real seminormed space $\mathfrak{B}$, we write $d_{\mathfrak{B}}$ for the pseudometric induced by the seminorm $\norm{\cdot}_{\mathfrak{B}}$. For a subset $S$ of $\mathfrak{B}$, let $B_S(x,r)=\{s\in S:d_{\mathfrak{B}}(s,x)\le r\}$ denote the closed radius-$r$ ball around $x$ relative to $S$.  We write $B_{\mathfrak{X}}(x,r)$ and $B_{\mathfrak{X}^*}(x,r)$ for closed balls in $\mathfrak{X}$ and its dual. We let $d(x,S) = \inf_{y \in S} d(x,y)$ denote the distance from a point $x$ to a set $S$, for any metric or pseudometric $d$. When both arguments are sets, $d(A,B)$ denotes the symmetric Hausdorff pseudodistance, with $d(\emptyset,\emptyset)=0$ and $d(A,\emptyset)=\infty$ if exactly one of the sets is empty. We write $a\vee b=\max\{a,b\}$, $a\wedge b=\min\{a,b\}$, and $v^+=v\vee0$, and set $\sup\emptyset=-\infty$ and $\inf\emptyset=\infty$.

If $m: \mathfrak{B} \ra \mathfrak{X}$ is a map between seminormed spaces and $\Theta$ is a subset of $\mathfrak{B}$, we say that $m$ is uniformly Fr\'{e}chet differentiable over $\Theta$ (with order $\alpha_m$) if there are continuous linear operators $m'(\theta): \operatorname{span}(\Theta) \ra \mathfrak{X}$, $\theta\in\Theta$, and a constant $\alpha_m > 1$ such that
\begin{align} \label{E:frechet}
    \sup_{\substack{\theta \in \Theta,\, \theta' \in B_\Theta(\theta,\delta)}} \norm{m(\theta') - m(\theta) - m'(\theta)(\theta' - \theta)}_{\mathfrak{X}} = O(\delta^{\alpha_m}).
\end{align}
We write $\alpha_m=\infty$ when the Taylor remainder in~\eqref{E:frechet} is identically zero. We use the same convention for the smoothness exponent $\alpha_v$ of Assumption~\ref{A:two}: when $c_n\to0$, expressions such as $c_n^{\alpha_m}$ and $c_n^{\alpha_v}$ are read as zero if the corresponding exponent is infinite.

Let $\mathcal{F}$ be an index set, usually $\Theta \times \mathcal{T}$ or a subset of it. The space $\ell^\infty(\mathcal{F})$ is the set of all uniformly bounded, real functions on $\mathcal{F}$ equipped with the uniform norm. Weak convergence of random processes over an index set $\mathcal{F}$ is handled by the bounded-Lipschitz metric, which for processes $\mathbb{V}$ (not necessarily measurable) and $\mathbb{G}$ indexed by $\mathcal{F}$ is defined by
\begin{align} \label{E:dBL}
d_{\mathrm{BL}}(\mathbb{V}, \mathbb{G}) & = \sup_{\ell \in \mathrm{BL}_1} |\mathrm{E}^*[\ell(\mathbb{V})] - \E{\ell(\mathbb{G})}|,
\end{align}
where $\mathrm{BL}_1$ is the set of all real functions on $\ell^\infty(\mathcal{F})$ that are bounded by $1$ in absolute value and have Lipschitz constant at most $1$, and $\mathrm{E}^*$ is an outer expectation \citep[see][\S\S1.2 and~1.12]{VW1996}. If nets $\mathbb{V}_\alpha$ and $\mathbb{G}_\alpha$ of processes satisfy $d_{\mathrm{BL}}(\mathbb{V}_\alpha, \mathbb{G}_\alpha) \ra 0$, we write, with some abuse of notation, $\mathbb{V}_\alpha \rs \mathbb{G}_\alpha$. If $\mathbb{V}_\alpha^*$ is an asymptotically measurable bootstrap analogue of $\mathbb{V}_\alpha$ whose distribution is itself random conditional upon the observation of a sample of data, we write that $\mathbb{V}_\alpha^* \overset{\mP}{\rs} \mathbb{G}_\alpha$ if $d_{\mathrm{BL}}( \mathbb{V}_\alpha^* ,\mathbb{G}_\alpha)$ converges to $0$ in outer probability \citep[cf.][\S3.6]{VW1996}. For processes with values in a normed space $\mathfrak{X}$, $\ell^\infty(\mathcal{F},\mathfrak{X})$ and $\mathrm{BL}_1$ are defined analogously, with the uniform norm. A real-valued bootstrap statistic $X_n^*$ satisfies $X_n^* = o_{p^*}(1)$ if $\PP{n}{|X_n^*|>\ve} = o_p(1)$ for all $\ve > 0$, where $\mathrm{P}_n$ is the probability conditional on the sample.  We use $O_{p^*}(1)$ analogously for conditional boundedness in probability.

A collection of processes $\mathbb{G}_\alpha, \alpha \in \mathrm{A}$ over $\mathcal{F}$ is \textit{uniformly pre-tight} if, for all $\ve > 0$, there exists a measurable and totally bounded subset $S$ of $\ell^\infty(\mathcal{F})$ (or $\ell^\infty(\mathcal{F}, \mathfrak{X})$) such that $\inf_{\alpha \in \mathrm{A}} \P{\mathbb{G}_\alpha \in S} > 1 - \ve$. This is weaker than the notion of uniform tightness \citep[\S\S1.3 and~1.12]{VW1996}. A collection is \textit{uniformly $\rho_\alpha$-equicontinuous in probability} if, for all $\ve> 0$, there exists some $\delta > 0$ such that
\begin{align*}
\sup_{\alpha \in \mathrm{A}} \P{\sup_{\rho_\alpha(f_1, f_2) < \delta} | \mathbb{G}_\alpha(f_1) - \mathbb{G}_\alpha (f_2) | > \ve } < \ve.
\end{align*}
Because the bounded-Lipschitz metric \eqref{E:dBL} pairs an outer expectation with an ordinary one, the limit processes are taken to be Borel measurable and the processes $\mathbb{V}_n$ asymptotically measurable throughout.  The latter is automatic when the limit is fixed in $n$ and separable \citep[Lemma~1.3.8 and \S1.12]{VW1996}.  Outer expectations and probabilities \citep[\S1.2]{VW1996} stand in for their measurable counterparts where needed.

\subsubsection{Convex analysis}\label{S:convexanalysis}
A convex cone $C$ in a real vector space $\mathfrak{X}$ is a convex set that is closed under multiplication by positive scalars \citep[\S2]{Rock1970}. We let $\mathrm{cone}(C)$ denote the smallest cone containing an arbitrary set $C\subset \mathfrak{X}$. If $\mathfrak{X}$ is a Banach space, the polar $C^\circ \subset \mathfrak{X}^*$ is the set of functionals $y$ in the dual of $\mathfrak{X}$ that are nonpositive on $C$. If $\xi: \mathfrak{X} \ra \R$ is a convex function, we let $\partial\xi(x)=\{y\in\mathfrak{X}^*:\xi(x')\ge\xi(x)+y(x'-x)\text{ for all }x'\in\mathfrak{X}\}$ denote the subdifferential of $\xi$ at $x\in\mathfrak{X}$.

Given a subset $\Theta$ of a seminormed space $\mathfrak{B}$ that is convex in a neighborhood of a point $\theta \in \Theta$, we define the \textit{tangent cone} $T_\theta \Theta$ of $\Theta$ at $\theta$ to be the set $\bigcap_{\delta > 0} \mathrm{cone}( B_\Theta(\theta, \delta) - \theta)$. When $x\notin S$ for a subset $S$ of a seminormed space, we write $T_x S$ for the radial cone $\mathrm{cone}(S-x)$. This convention is used only where the base point may lie outside the set, as in Assumption~\ref{A:ex22}(1) and Corollary~\ref{C:geometry}.\footnote{It is compatible with the local definition: if $S$ is convex, then $T_xS=\mathrm{cone}(S-x)$ for every $x$, and if also $x\in\overline S$, then $\overline{T_xS}=\overline{T_x\overline S}$. Hence polars and distances are unaffected by replacing a convex set by its closure, because $(T_xS)^\circ=(T_x\overline S)^\circ$ and $d(z,T_xS)=d(z,T_x\overline S)$.} When $\Theta$ is convex and $\theta\in\Theta$, this definition gives the radial cone $\mathrm{cone}(\Theta-\theta)$, whose closure is the usual tangent cone of a convex set \citep[cf.][\S6.3]{Convex2005}.

\section{Examples}
\label{S:examples}

We present five examples motivating the framework of Section~\ref{S:model}, each isolating a different piece of the limit theory. In each case the test statistic takes the minimax form
\begin{equation} \label{E:teststatistic}
  T_n = r_n \inf_{\theta \in \Theta} \sup_{t \in \mathcal{T}} v_n(\theta, t),
\end{equation}
with $v_n$ and $r_n$ as in Section~\ref{S:model}. Section~\ref{S:examples_revisited} completes each example by deriving its limiting distribution from our main results. Although the examples take $\Theta\subset\R^k$ for concreteness, they extend to any compact $\Theta$ in a seminormed space that indexes a Donsker class of moment functions. By Remark~\ref{R:sieve}, $\Theta$ and $\mathcal{T}$ may be replaced by sufficiently dense sieves without changing these limits. As in BCS, \citet{CS2023}, and \citet{CNS2023}, the moments in each example may also be studentized by an estimated weight (Section~\ref{S:studentized}).

\begin{exm}[GMM overidentification test under partial identification]\label{E:jtest}
Let $g:\mathcal{X}\times\Theta\to\R^L$ be a vector of $L>k$ known moment functions and $\Theta\subset\R^k$ compact. The composite null is
\[
  H_0: \exists\,\theta\in\Theta \text{ such that } E[g(X,\theta)]=0.
\]
Set $m(\theta)=E[g(X,\theta)]\in\R^L$ and $\xi(x)=\norm{x}_W\equiv(x^\top Wx)^{1/2}$ for a positive definite weight matrix $W$. By duality,
\[
  \norm{x}_W=\sup_{u^\top W^{-1}u\le 1}u^\top x,
\]
so the statistic $T_n=\sqrt{n}\inf_{\theta\in\Theta}\norm{m_n(\theta)}_W$ is of the form~\eqref{E:teststatistic} with $v_n(\theta,u)=u^\top m_n(\theta)$. Under regular point identification and efficient weighting, $T_n^2$ is the GMM $J$ test of \citet{H1982}. At any $\theta\in\Theta_{\mathrm{I}}$, $m(\theta)=0$, so every dual direction $u$ attains the same population value zero. Thus, the envelope condition defining the contact set is vacuous, and $\mathcal{K}(\theta)$ keeps only the dual directions that a local movement of $\theta$ cannot make negative. Any convex Lipschitz loss may replace $\norm{\cdot}_W$, including $\xi(x)=d(x,C)$ for a convex cone $C$ as in \citet{FangSeo2021}.
\end{exm}

\begin{exm}[Composite moment inequality test]\label{E:momentineq}
Let $g:\mathcal{X}\times\Theta\to\R^J$ be a vector of $J$ known moment functions and $\Theta\subset\R^k$ compact. The composite null is
\[
  H_0: \exists\,\theta\in\Theta \text{ such that } E[g_j(X,\theta)]\le 0 \text{ for all } j=1,\ldots,J.
\]
Set $m_j(\theta)=E[g_j(X,\theta)]$ and $\xi(x)=\max_{1\le j\le J}x_j^+$ (where $x^+=\max(x,0)$). The resulting statistic is
\[
  T_n=\sqrt{n}\inf_{\theta\in\Theta}\max_{1\le j\le J}(m_{n,j}(\theta))^+,
\]
which is of the form~\eqref{E:teststatistic} because $\max_j x_j^+=\sup_{y\in\mathrm{conv}(\{0,e_1,\ldots,e_J\})}y^\top x$. Equivalently, $\xi(x)=\max_j x_j^+$ is the supremum-norm distance from $x$ to the negative orthant, a special case of Corollary~\ref{C:geometry}. This example displays both parts of the contact set, which consists of those $y$ in the weight set $\mathrm{conv}(\{0,e_1,\ldots,e_J\})$ with
\[
  y_j = 0 \  \text{ whenever } m_j(\theta)<0, \qquad y^\top m'(\theta)h \ge 0 \  \text{ for all } h \in T_\theta \Theta.
\]
The first condition drops the slack inequalities. The second restricts the weights $y$ jointly, rather than inequality by inequality. Proposition~\ref{P:ex1} describes the limiting distribution by combining these two selections. BCS provide valid, possibly conservative critical values for this test, and \citet{CHT2007} construct confidence regions for the identified set. Equality restrictions fit naturally into the same formulation by replacing $E[g]=0$ with the pair $E[g]\le 0$ and $E[-g]\le 0$. The finite-$J$ setting extends to a continuum of inequalities by replacing $\{1,\ldots,J\}$ with an arbitrary Donsker index set $\mathcal{J}$ and $\max_j x_j^+$ with $\sup_{s\in\mathcal{J}}f(s)^+$.
\end{exm}

\begin{exm}[Conditional moment composite null]\label{E:condmoment}
Let $q(X,\theta)$ be a scalar residual and consider the composite conditional moment null $H_0:\exists\,\theta\in\Theta$ such that $E[q(X,\theta)\mid Z]=0$. Let $\mathcal{H}$ be a symmetric class of instruments rich enough that its unconditional moments characterize this conditional restriction. Define the function-valued moment map $m(\theta)(h)=E[h(Z)q(X,\theta)]$, $h\in\mathcal{H}$, and the loss $\xi(f)=\sup_{h\in\mathcal{H}}|f(h)|$. A finite collection $\mathcal{H}=\{h_1,\ldots,h_L\}$ reduces the problem to a finite system of unconditional moment equalities with $g_j(X,\theta)=h_j(Z)q(X,\theta)$. When the product class $\{h(Z)q(X,\theta):h\in\mathcal H,\theta\in\Theta\}$ is Donsker, the statistic
\[
  T_n=\sqrt{n}\inf_{\theta\in\Theta}\sup_{h\in\mathcal{H}}\left|\frac{1}{n}\sum_{i=1}^n h(Z_i)q(X_i,\theta)\right|
\]
is still of the form~\eqref{E:teststatistic}, now with moments taking values in $\ell^\infty(\mathcal{H})$ and $\xi$ equal to the sup norm. As in the GMM example, $m(\theta)=0$ on $\Theta_{\mathrm{I}}$, and the contact set is again governed by stationarity. Proposition~\ref{P:ex1} describes the pointwise distributional limit for any fixed Donsker class. Conditional moment inequalities can be handled similarly with a sufficiently rich class of nonnegative instruments.

\end{exm}

\begin{exm}[Specification testing with convex moment predictions]\label{E:bmm}
Following \citet{molinari2011}, let $\mu=E[w(Z)]\in\R^d$ be a vector of observed moments and let $Q_\theta$ be an integrably bounded random compact set containing the moment predictions of the model at $\theta$. Suppose, as under their conditions, that its Aumann expectation $\mathcal{C}(\theta)=E[Q_\theta]$ is closed and convex. The model is correctly specified if its sharp identified set
\[
  \Theta_{\mathrm I}=\{\theta\in\Theta:\mu\in\mathcal{C}(\theta)\}
\]
is nonempty. For example, in a finite game with multiple equilibria, $\mu$ collects observed outcome probabilities and $\mathcal{C}(\theta)$ contains the probabilities permitted by the model and admissible equilibrium selection.

Let $h(A,u)=\sup_{a\in A}u^\top a$ denote the support function of $A$. The support-function characterization of convex-set membership gives
\[
  \theta\in\Theta_{\mathrm I}
  \text{ if and only if }
  m(\theta)(u)\equiv u^\top\mu-E[h(Q_\theta,u)]\leq0
  \quad\text{for every }u\in B_{\R^d}(0,1).
\]
Thus, for a sample analogue $m_n$, the specification statistic is
\[
  T_n=\sqrt n\inf_{\theta\in\Theta}\sup_{u\in B_{\R^d}(0,1)}
  \bigl(m_n(\theta)(u)\bigr)^+.
\]
This is of the form~\eqref{E:teststatistic} with the convex loss $\xi(f)=\sup_{u\in B_{\R^d}(0,1)}f(u)^+$. Any finite collection of directions generally characterizes only an outer approximation to $\Theta_{\mathrm I}$, whereas the function-valued formulation retains the sharp continuum of restrictions. Proposition~\ref{P:ex1} supplies the corresponding limit without first discretizing the unit ball. The required Donsker condition is standard when the coordinate functions in $w$ and, if estimated, the support-function class are Donsker, because $u$ ranges over a compact finite-dimensional ball and forms only bounded linear combinations. Conditional convex moment predictions can be treated by pairing each support direction $u$ with an instrument from a sufficiently rich nonnegative Donsker class, as in the preceding example. Similar restrictions appear in the random-set models of \citet{beresteanu2008}, \citet{galichon2011}, and \citet{chesher2017}, and in the linear-in-measures models of \citet{BotosaruLohMuris2026}.
\end{exm}

\begin{exm}[Inference on an optimization value with linear and nonlinear components]\label{E:lp}
Consider an optimization problem in which a potentially nonlinear parameter $\theta\in\Theta$ indexes the coefficients of an LP in $x$. Let $K\subset\R^p$ be a known compact convex set, and let $A(\theta)=E[g_A(X,\theta)]$, $b(\theta)=E[g_b(X,\theta)]$, and $c(\theta)=E[g_c(X,\theta)]$ be estimable coefficients. Define
\[
  V(\theta)=\min_{x\in K:\,A(\theta)x=b(\theta)}c(\theta)^\top x.
\]
The parameter of interest is the value $v_0=\inf_{\theta\in\Theta}V(\theta)$. Assume standard conditions for strong LP duality and suppose there is a fixed compact convex multiplier set $\Lambda$ such that
\[
  V(\theta) = \inf_{x\in K}\sup_{\lambda\in\Lambda}\bigl[c(\theta)^\top x + \lambda^\top(b(\theta) - A(\theta)x)\bigr]
\]
for every $\theta$. By LP duality, the representation holds whenever $\Lambda$ contains an optimal dual multiplier for each $\theta$. Also suppose that, whenever $V(\theta)=v_0$, every minimizing $x$ on the right is feasible for the original LP, as is true when an optimal dual multiplier lies in the interior of $\Lambda$. A sufficiently large ball satisfies both conditions under a uniform Slater condition. Replacing the population coefficients in this representation by their sample analogues defines the plug-in estimator
\[
  V_n(\theta)=\inf_{x\in K}\sup_{\lambda\in\Lambda}
  \bigl[c_n(\theta)^\top x+\lambda^\top(b_n(\theta)-A_n(\theta)x)\bigr].
\]
Treating $(\theta,x)\in\Theta\times K$ as the composite parameter gives
\[
  T_n = \sqrt{n}\bigl(\inf_{\theta\in\Theta}V_n(\theta)-v_0\bigr) = \sqrt{n}\inf_{(\theta,x)\in\Theta\times K}\sup_{\lambda\in\Lambda}\bigl[c_n(\theta)^\top x - v_0 + \lambda^\top(b_n(\theta)-A_n(\theta)x)\bigr],
\]
which is of the form~\eqref{E:teststatistic}. The inner index is now the dual variable $\lambda$. At an infimum-attaining feasible pair, the population Lagrangian is zero for every $\lambda\in\Lambda$, so the envelope condition is automatic. The contact set instead selects the multipliers that satisfy first-order optimality after both $x$ and $\theta$ are allowed to move. This extends ordinary LP sensitivity analysis to a composite problem in which the coefficients may depend nonlinearly on $\theta$ and the outer minimization is not itself an LP. A related literature develops inference on the value or solvability of a linear program with estimated coefficients, holding $\theta$ fixed \citep{HsiehShiShum2022,FSST2023,ChoRussell2024,MbakopGoff2025,Voronin2025,BaiEtal2026}.
\end{exm}

\section{Model and asymptotic distribution}
\label{S:model}

Let $\Theta$ be a parameter space.  Our basic model has an identified set $\Theta_{\mathrm{I}} \subset \Theta$ consisting of the $\theta \in \Theta$ that satisfy
\begin{align} \label{E:theta0}
\ell(\theta) = 0,
\end{align}
where $\ell: \Theta \ra \R_+$ is nonnegative. We further require that $\ell$ take the form
\begin{align} \label{E:ell}
\ell(\theta) = \sup_{t \in \mathcal{T}} v(\theta, t),
\end{align}
where $\mathcal{T}$ indexes functions $\{v(\cdot, t): t \in \mathcal{T}\}$ that generate the loss $\ell$ over $\Theta$.  We study the minimum value of $\ell$ over $\Theta$, which has the minimax representation
\begin{align} \label{E:centralobject}
    \inf_{\theta \in \Theta} \ell(\theta) = \inf_{\theta \in \Theta} \sup_{t \in \mathcal{T}} v(\theta, t)
\end{align}
Denote this minimal loss by $T_0 = \inf_{\theta \in \Theta} \ell(\theta)$.
The loss $\ell$ is not observed directly.  Instead, the researcher has a noisy estimate $\ell_n$, built from estimates $v_n$ on a sample of size $n$,
\begin{align*}
    \ell_n(\theta) = \sup_{t \in \mathcal{T}} v_n(\theta, t).
\end{align*}
Our approach to inference is to estimate the asymptotic distribution of the test statistic $T_n=r_n\inf_{\theta\in\Theta}\ell_n(\theta)$ in~\eqref{E:teststatistic}, for an appropriate normalizing sequence $r_n$ (most often $\sqrt n$), under general conditions.

\begin{asm}\label{A:zero} $\Theta$ is a subset of a seminormed space $\mathfrak{B}$ with seminorm $\norm{\cdot}_{\mathfrak{B}}$ satisfying:
\begin{enumerate}
\item  $\Theta$ is convex in a neighborhood of $\Theta_{\mathrm{I}}$.
\item The function $\ell$ maps $\Theta$ to $\R_+$, and for all $\delta > 0$,
\begin{align} \label{E:elldeltaeqn}
\inf_{\substack{\theta \in \Theta: d_{\mathfrak{B}} (\theta, \Theta_{\mathrm{I}}) \ge \delta}} \ell(\theta) > 0.
\end{align}
\item $\Theta_{\mathrm{I}}$ is compact.
\end{enumerate}
\end{asm}

Assumption \ref{A:zero}(1) holds for a broad class of parameter spaces. It is
trivial when $\Theta_{\mathrm{I}}$ lies in the interior of $\Theta$, and it
extends to non-convex $\Theta$ by reparameterization. Suppose that $g:\tilde{\Theta}\ra\Theta$
is surjective and $C^2$ with nondegenerate derivative, that $g^{-1}(\Theta_{\mathrm{I}})$ is compact, and that
$\tilde{\Theta}$ is convex near $g^{-1}(\Theta_{\mathrm{I}})$. Then, all results apply to
$(\tilde{\Theta}, v \circ g, m \circ g)$ when $v \circ g$ satisfies Assumptions~\ref{A:two} and~\ref{A:three}, with $\alpha_v$ and $\alpha_m$ replaced by $\min\{\alpha_v,2\}$ and $\min\{\alpha_m,2\}$. This holds because $g$ is
locally bi-Lipschitz and preserves the infimum, the identification rate, and the limiting
formula \eqref{E:thm}. This covers $C^2$ manifolds with boundary or corners.
Assumptions~\ref{AU:1}\suppref{} and~\ref{A:uniformbound}(1) give the uniform analogues.

\begin{asm} \label{A:one}
\begin{enumerate}
\item
For some sequence $(r_n)$ with $r_n \ra \infty$, the processes $\mathbb{V}_n(\theta, t) \equiv r_n(v_n(\theta, t) - v(\theta, t)) \in \ell^\infty(\Theta \times \mathcal{T})$ converge in the bounded-Lipschitz metric \eqref{E:dBL} to processes $\mathbb{G}_n(\theta , t)$ for which $\norm{\mathbb{G}_n}_{\ell^\infty}$ is $O_p(1)$ uniformly in $n$.

\item The limit processes $\mathbb{G}_n(\theta, t)$ are uniformly $\rho$-equicontinuous in probability for a pseudometric $\rho: (\Theta \times \mathcal{T})^2 \ra \R_+$.
\end{enumerate}
\end{asm}

\begin{asm}\label{A:two}
 \begin{enumerate}
\item  For all $(\theta, t) \in \Theta \times \mathcal{T}$ there is a lower semicontinuous, convex, and positively homogeneous functional $D_{\theta; t}v: \operatorname{span}(\Theta) \ra \R $ and constant $\alpha_v > 1$ such that
\begin{align}
\sup_{\substack{\theta \in \Theta,\, \theta' \in B_\Theta(\theta,\delta)\\ t \in \mathcal{T}}} |v(\theta',t) - v(\theta, t) - D_{\theta; t}v (\theta' - \theta)| = O(\delta^{\alpha_v}). \label{E:fdeltaeqn}
\end{align}
\item For all $n \in \N$ and $\theta, \theta' \in \Theta$, $t \mapsto v_n(\theta, t) + D_{\theta; t}v (\theta'- \theta)$ is almost surely concave.
\end{enumerate}
\end{asm}

Assumption \ref{A:one} is a standard empirical process condition. In the leading applications, $\mathbb{V}_n=r_n(v_n-v)$ converges to a fixed tight process $\mathbb{G}$, so one may take $\mathbb{G}_n=\mathbb{G}$ for every $n$. The notation also permits a moving approximation $\mathbb{G}_n$. This is useful in the triangular-array settings emphasized by \citet{FangSeo2021}, where kernel or sieve estimators, estimated normalizations, or changing finite-dimensional approximations may yield an $n$-dependent Gaussian approximation rather than a single fixed weak limit. The arguments below do not need $\mathbb{G}_n$ to converge. They need only that $\mathbb{V}_n$ be close to $\mathbb{G}_n$ in bounded-Lipschitz distance and that the approximating processes be uniformly bounded and equicontinuous.
For simplicity, the pseudometric $\rho$ does not depend on $n$ here. Section~\ref{S:uniform} lets it vary with $n$ and $(R,P)$ (Assumptions~\ref{AU:2} and~\ref{AU:4}\suppref{}).

Assumption \ref{A:two} introduces a directional derivative $D_{\theta;t}v$ for each $v(\cdot,t)$ at each $\theta\in\Theta$.  Part (1) requires the approximation error in $\theta$ to be $O(\delta^{\alpha_v})$ for some $\alpha_v > 1$.  It reduces to standard differentiability when $D_{\theta;t}v$ is linear, and Taylor's theorem gives the typical value $\alpha_v = 2$ when $v(\cdot,t)$ is twice continuously differentiable in $\theta$.  Part (2) imposes concavity of the inner objective in $t$.  It holds whenever $v_n(\theta,\cdot)$ is concave and $t \mapsto D_{\theta;t}v(\theta'-\theta)$ is affine, which covers the examples of Section~\ref{S:examples}, where both maps are affine in $t$.

For $(\theta,t)\in\Theta\times\mathcal{T}$ and $\delta>0$, define
\begin{align}
\psi_\delta(\theta, t)& = \inf_{\theta' \in B_\Theta(\theta, \delta)} \frac{ D_{\theta; t} v (\theta' - \theta)}{\delta}, \nonumber \\
 \label{E:gamma}
\gamma(\theta, t) &= \liminf_{\delta \downarrow 0} \psi_\delta(\theta, t).
\end{align}
The quantity $\psi_\delta(\theta,t)$ measures the most negative value the directional derivative $D_{\theta;t}v$ can attain over $\delta$-perturbations of $\theta$ within $\Theta$.  Its nonpositive limit $\gamma(\theta,t)$ is the corresponding quantity in the tangent cone at $\theta$. Under Assumption~\ref{A:zero}(1), $\Theta$ is convex near $\Theta_{\mathrm{I}}$, so $\psi_\delta(\theta,t)$ decreases as $\delta\downarrow0$ for $\theta$ sufficiently close to $\Theta_{\mathrm{I}}$, and the liminf is a limit.

The \textbf{contact set} $\mathcal{K}(\theta) \subseteq \mathcal{T}$ consists of those $t$ for which $\gamma$ vanishes and $v$ is nonnegative:
\begin{align} \label{E:ksetdef}
\mathcal{K}(\theta) &= \{t \in \mathcal{T}: \gamma(\theta,t) = 0 \text{ and } v(\theta, t) \ge 0\}.
\end{align}
It is useful to separate the stationarity part of this definition (see \Secapp~\ref{S:stationarity}\suppref{}). Let
\begin{align} \label{E:k0setdef}
\mathcal{K}_0(\theta) &= \{t \in \mathcal{T}: \gamma(\theta,t) = 0\}.
\end{align}
Thus, $\mathcal{K}(\theta)=\mathcal{K}_0(\theta)\cap\{t:v(\theta,t)\ge 0\}$. The condition $\gamma(\theta,t)=0$ means that no local movement within $\Theta$ can strictly decrease $v(\cdot,t)$ to first order, so a local perturbation of $\theta$ cannot remove the contribution of $t$. The condition $v(\theta,t)\ge 0$ is the classical envelope condition. At $\theta\in\Theta_{\mathrm{I}}$, it requires $t$ to attain the inner supremum. Together, the two conditions identify the indices $t$ that remain active in the limit. If $\Theta$ is a singleton, then $\gamma\equiv0$ and $\mathcal{K}(\theta)$ reduces to the argmax of $t\mapsto v(\theta,t)$, as in the standard envelope formula. In applications where $v(\theta,t)=0$ for every $t$ at each $\theta\in\Theta_{\mathrm{I}}$, the envelope condition is automatic on the identified set and $\mathcal{K}(\theta)=\mathcal{K}_0(\theta)$ for $\theta\in\Theta_{\mathrm{I}}$. This occurs, for example, when $v(\theta,t)=t(m(\theta))$ and $\mathcal{T}$ is the unit ball of the dual of a Banach space, because $\|m(\theta)\|=0$ on the identified set.

As $n\to\infty$, the distribution of $T_n$ is determined by $\mathcal{K}(\theta)$. To pass the resulting supremum over $\mathcal{K}(\theta)$ through limits, we require compactness of $\mathcal{T}$ and semicontinuity of $v(\theta,\cdot)$ and $D_{\theta;t}v$ in $t$.

\begin{asm} \label{A:three}
$\mathcal{T}$ is a convex subset of a linear topological space with topology $\mathcal{U}$, under which:
    \begin{enumerate}
    \item $\rho: (\Theta \times \mathcal{T})^2 \ra \R$ is continuous with respect to the product topology and uniformly continuous with respect to $d_{\mathfrak{B}}$:
    \begin{align} \label{E:pseudometricdelta}
        \lim_{\delta \downarrow 0} \sup_{\substack{\theta \in \Theta,\, \theta' \in B_\Theta(\theta,\delta) \\ t \in \mathcal{T}}} \rho((\theta, t), (\theta', t)) = 0.
    \end{align}
     \item $\mathcal{T}$ is compact.
    \item For all $\theta, \theta' \in \Theta$, the maps
\begin{align}
&t \mapsto v(\theta, t) \nonumber \\
&t \mapsto v_n(\theta, t) \nonumber \\
&t \mapsto D_{\theta; t}v(\theta' - \theta)
\label{E:remarkedone}
\end{align}
are $\mathcal{U}$-upper semicontinuous.
\end{enumerate}
\end{asm}

Assumption~\ref{A:three} enters through Sion's minimax theorem \citep[Corollary~3.3]{Sion1958}. Compactness of $\mathcal{T}$ in part~(2), together with the upper semicontinuity in part~(3), allows the minimax interchange without compactness of $\Theta$.  The minimax theorem of \citet[Theorem~3]{Terkelsen1972} provides an alternative that dispenses with the upper-semicontinuity conditions in part~(3) but in return requires $\Theta$ to be compact. Compactness of $\Theta$ under $d_{\mathfrak{B}}$ would also make the uniform continuity in part~(1) automatic, given the continuity of $\rho$ (Lemma~\ref{L:conv0}\suppref{}). In the examples of Section~\ref{S:examples}, $v(\theta,t)$ and $v_n(\theta,t)$ are affine in $t$, and $\mathcal{U}$ can be taken as the weak-* topology on $\mathcal{T}$. Continuity of $t\mapsto v_n(\theta,t)$ then follows from the definition of the weak-* topology, so Assumption~\ref{A:three}(3) adds no restriction in those cases.

Assumptions \ref{A:zero}--\ref{A:three} together suffice to bound $T_n$ from above by the limit functional in Theorem~\ref{T:one}. Equality holds if, in addition, $\ell$ grows away from $\Theta_{\mathrm{I}}$ quickly enough that parameters off $\Theta_{\mathrm{I}}$ cannot improve on this bound in the limit, $v(\cdot,t)$ is convex in $\theta$ near $\Theta_{\mathrm{I}}$, or the contact set varies continuously off $\Theta_{\mathrm{I}}$. To state this precisely, define the modulus of identifiability
\begin{align} \label{E:felleqn}
    f_\ell(\delta) \equiv  \sup_{\theta: \ell(\theta) \le \delta} d_{\mathfrak{B}}(\theta, \Theta_{\mathrm{I}}),
\end{align}
which gives the largest distance from $\Theta_{\mathrm{I}}$ attained by any $\theta$ with $\ell(\theta) \le \delta$. Assumption \ref{A:zero}(2) implies $f_\ell(\delta) = o(1)$, so the level sets of $\ell$ shrink to $\Theta_{\mathrm{I}}$.  Part (1) of the next assumption strengthens this to a polynomial rate (see Remark~\ref{R:polynomialminorant}). Part (2) provides an alternative route to equality via convexity of $v(\cdot,t)$ in $\theta$, which makes the minimax bound sharp without a rate condition on $\ell$. This condition is natural when the parameter enters the moments linearly. In NPIV \citep{S2012}, for example, $v(\theta,h)=E[(Y-\theta(X))h(Z)]$ is affine in $\theta$ for every instrument $h$.  Part (3) uses a localized notion of continuity for the contact set.

\begin{defn} \label{D:hemiat}
    A correspondence $K: \theta \mapsto K(\theta)$ from $\Theta$ to $\mathcal{T}$ is \textit{lower hemicontinuous at $\Theta_{\mathrm{I}}$} if, for every $\theta \in \Theta_{\mathrm{I}}$ and every open $U \subset \mathcal{T}$ intersecting $K(\theta)$, there is some $\delta > 0$ such that $K(\theta')$ intersects $U$ whenever $\theta'\in B_\Theta(\theta,\delta)$ and $\theta' \not\in \Theta_{\mathrm{I}}$.
\end{defn}

\noindent This is weaker than lower hemicontinuity on all of $\Theta$, because only deviations off $\Theta_{\mathrm{I}}$ are needed.  For part~(3), let $\widetilde{\mathcal{K}}(\theta)=\mathcal{K}(\theta)$ if $\theta\in\Theta_{\mathrm{I}}$ and $\widetilde{\mathcal{K}}(\theta)=\{t\in\mathcal{T}:v(\theta,t)\ge0\}$ (the envelope set) otherwise.  Part~(3) asks that each contact index $t\in\mathcal{K}(\theta)$ at an identified parameter remain weakly violated, in that $v(\theta',t')\ge0$ for some $t'$ near $t$, at parameters $\theta'$ just outside $\Theta_{\mathrm{I}}$. Unlike lower hemicontinuity of $\mathcal{K}$, it does not require stationarity there. It is implied by the convexity condition in part~(2), as the proof of Theorem~\ref{T:one} shows, and by lower hemicontinuity of $\mathcal{K}$ at $\Theta_{\mathrm{I}}$, for which Assumption~\ref{A:ex22} below gives sufficient conditions under convex loss.

\begin{asm}
\label{A:four}
Any of the following holds:
\begin{enumerate}
\item
For some $\alpha_\ell > 0$,
\begin{align} \label{E:felleqn2}
    f_\ell(\delta)  = O(\delta^{\alpha_\ell}).
\end{align}
If $\alpha_v < \infty$, $\alpha_\ell\alpha_v > 1$.
\item $v(\cdot, t)$ is convex in $\theta$ for all $t \in \mathcal{T}$ in some neighborhood of $\Theta_{\mathrm{I}}$.
\item The correspondence $\widetilde{\mathcal{K}}$ is lower hemicontinuous at $\Theta_{\mathrm{I}}$.
\end{enumerate}
\end{asm}

\begin{rem}[Polynomial minorant condition] \label{R:polynomialminorant}
Assumption~\ref{A:four}(1) is equivalent to the polynomial minorant condition $\ell(\theta)\ge c\min\{\delta_0,d_{\mathfrak{B}}(\theta,\Theta_{\mathrm{I}})\}^{1/\alpha_\ell}$ for some $c,\delta_0>0$ \citep[Condition~C.2, (4.1), and~(4.5)]{CHT2007}.  It is often applied for a power of $\ell$, such as $Q=\ell^2$, whose quadratic growth corresponds to $\alpha_\ell=1$ here.  \citet[Section~3.1]{KaidoMolinariStoye2022} split such minorants into a global identification condition and a local error bound.  Similarly, Assumption~\ref{A:four}(1) adds an optional local polynomial error bound to the global identification condition of Assumption~\ref{A:zero}(2), which is maintained throughout.  For smooth finite-dimensional moment equalities and inequalities, the local error bound holds with $\alpha_\ell=1$ when the Mangasarian--Fromovitz constraint qualification holds at every point of the compact set $\Theta_{\mathrm{I}}$ \citep{Robinson1976}, as \citet[Theorem~3.1]{KaidoMolinariStoye2022} show in the case without equalities.

Assumption~\ref{A:four}(1) requires $\alpha_\ell\alpha_v>1$, where $\alpha_v$ is the smoothness of the individual functions $v(\cdot,t)$ rather than of $\ell$.  Under linear identification ($\alpha_\ell=1$), the $\alpha_v>1$ already imposed by Assumption~\ref{A:two} therefore suffices for Theorem~\ref{T:one}.  Quadratic growth of $\ell$ itself ($\alpha_\ell=1/2$) needs $\alpha_v>2$, which fails for smooth $v$ whenever some $v(\cdot,t)$ has nonzero curvature in $\theta$, leaving the convexity alternative of Assumption~\ref{A:four}(2).  Section~\ref{S:uniform} discusses the minorant condition in more detail.
\end{rem}

\begin{theorem} \label{T:one}
Let Assumptions \ref{A:zero}, \ref{A:one}, \ref{A:two}, and \ref{A:three} hold. For all $\theta \in \Theta$, let $\mathcal{K}(\theta)$ be defined as in \eqref{E:ksetdef}.
Then, if $\Theta_{\mathrm{I}}$ is nonempty, $\mathcal{K}(\theta)$ is nonempty for all $\theta \in \Theta_{\mathrm{I}}$, and
\begin{align}
T_n = r_n\inf_{\theta \in \Theta} \sup_{t \in \mathcal{T}} v_n(\theta, t) &\le \inf_{\theta \in \Theta_{\mathrm{I}}} \sup_{t \in \mathcal{K}(\theta)} \mathbb{V}_n(\theta,t) + o_p(1)  \nonumber\\
& \rs \inf_{\theta \in \Theta_{\mathrm{I}}} \sup_{t \in \mathcal{K}(\theta)} \mathbb{G}_n(\theta, t). \label{E:thm}
\end{align}
If Assumption \ref{A:four} also holds, then \eqref{E:thm} holds with equality. If $\Theta_\mathrm{I}= \emptyset$, then $T_n \ge r_n T_0 - \sup_{\theta \in \Theta, t \in \mathcal{T}} |\mathbb{V}_n(\theta, t)| = r_n T_0 - O_p(1)$.

\end{theorem}

\begin{cor}[Directional derivative] \label{C:dirderiv}
Let $v\in\ell^\infty(\Theta\times\mathcal{T})$ and Assumptions~\ref{A:zero}, \ref{A:two}, and~\ref{A:three} hold with $v$ replacing $v_n$ where it appears, and suppose that $\Theta_{\mathrm{I}}$ is nonempty. Let $\Phi(w)=\inf_{\theta}\sup_{t}w(\theta,t)$ on the set $\mathbb{D}_\Phi$ of $w\in\ell^\infty(\Theta\times\mathcal{T})$ that satisfy Assumptions~\ref{A:two}(2) and~\ref{A:three}(3) with $w$ replacing $v_n$ where it appears. Let $\mathbb{D}_0$ be the set of $\varphi\in\ell^\infty(\Theta\times\mathcal{T})$ that are uniformly continuous with respect to a pseudometric $\rho$ satisfying Assumption~\ref{A:three}(1). Then $\Phi(v)=0$, and for all $h_n\downarrow0$ and $\varphi_n\to\varphi$ in $\ell^\infty(\Theta\times\mathcal{T})$ with $v+h_n\varphi_n\in\mathbb{D}_\Phi$ and $\varphi\in\mathbb{D}_0$,
\begin{align*}
\limsup_{n\to\infty}\frac{\Phi(v+h_n\varphi_n)-\Phi(v)}{h_n}\le\inf_{\theta\in\Theta_{\mathrm{I}}}\sup_{t\in\mathcal{K}(\theta)}\varphi(\theta,t).
\end{align*}
Under Assumption~\ref{A:four}, the limit exists and equals the right side.
\end{cor}

Corollary~\ref{C:dirderiv} shows that, under Assumption~\ref{A:four}, $\Phi$ is Hadamard directionally differentiable at $v$ tangentially to $\mathbb{D}_0$, with derivative $\varphi\mapsto\inf_{\theta\in\Theta_{\mathrm{I}}}\sup_{t\in\mathcal{K}(\theta)}\varphi(\theta,t)$. This justifies the numerical delta method of \citet{HongLi2018} for this functional, and Section~\ref{S:uniform} shows that, under a uniform minorant condition, the method is also uniformly valid (Theorem~\ref{P:unif1}\suppref{} and Proposition~\ref{P:unifbound}). The contact set plays the role of the active argmax set in the classical envelope theorem, augmented to account for local re-optimization over $\theta$.

\begin{rem}[Sieve approximation]
\label{R:sieve}
Although our results are stated for fixed $\Theta$ and $\mathcal{T}$, they extend to sieve settings. When sieves $\Theta_n$ and $\mathcal{T}_n$ are used in Theorem~\ref{T:one}, the sieved statistic agrees with $T_n$ up to $o_p(1)$ provided the sieves are sufficiently dense. Asymptotic equicontinuity of $\mathbb{V}_n$ (Assumption~\ref{A:one}) gives, for $\theta'$ close to $\theta$,
\begin{equation}
\label{E:sieveapprox}
r_n v_n(\theta', t') = r_n v_n(\theta,t) + r_n\bigl(v(\theta',t') - v(\theta,t)\bigr) + o_p(1),
\end{equation}
uniformly in $t$. If
\begin{equation}
\label{E:sievecond}
\sup_{\theta\in\Theta}\inf_{\theta'\in\Theta_n}\norm{v(\theta,\cdot)-v(\theta',\cdot)}_\infty
+\sup_{t\in\mathcal{T}}\inf_{t'\in\mathcal{T}_n}\norm{v(\cdot,t)-v(\cdot,t')}_\infty
= o(r_n^{-1}),
\end{equation}
with $\rho$-close approximating points, then \eqref{E:sieveapprox} yields $r_n\inf_{\Theta}\sup_{\mathcal{T}}v_n = r_n\inf_{\Theta_n}\sup_{\mathcal{T}_n}v_n + o_p(1)$. A sufficient condition is that $v$ be uniformly Lipschitz in each argument, for pseudometrics $d_\Theta$ and $d_{\mathcal{T}}$ under which $\rho$ is uniformly continuous, with sieve meshes $\sup_\theta d_\Theta(\theta,\Theta_n)$ and $\sup_t d_{\mathcal{T}}(t,\mathcal{T}_n)$ of order $o(r_n^{-1})$. The bootstrap inherits the same approximation (see Assumption~\ref{A:five}(1), or~\ref{A:ex2}(1) under convex loss). The mesh requirement \eqref{E:sievecond} is a standard sieve approximation rate \citep{Chen2007}. Because the sieves are subsets of $\Theta$ and $\mathcal{T}$, and Assumption~\ref{A:one} controls $\mathbb{V}_n$ uniformly over all of $\Theta\times\mathcal{T}$, sieves satisfying \eqref{E:sievecond} may be arbitrarily dense without affecting the asymptotic results. The choice of sieve can therefore be guided by computational need rather than by the asymptotic properties of the test.

In the setting of Proposition~\ref{P:ex1} below, sieving $\mathcal{T}$ is unnecessary because the supremum over $\mathcal{T}$ is absorbed into $\xi$, which is computed directly for each $\theta$. Only $\Theta$ is then sieved, and \eqref{E:sievecond} is needed only with $\mathcal{T}_n=\mathcal{T}$.
\end{rem}

\subsection{Exact distribution under convex loss}\label{S:convexloss}

A natural application of Theorem~\ref{T:one} is to minimal violations measured by a convex loss, as in
\[
T_n = r_n\bigl(\inf_{\theta\in\Theta}\xi(m_n(\theta))-T_0\bigr),
\]
where $\xi:\mathfrak{X}\to\R\cup\{\infty\}$ is convex and
$m_n$ estimates a population moment function $m:\Theta\to\mathfrak{X}$, so that the
minimal loss is $T_0 = \inf_{\theta\in\Theta}\xi(m(\theta))$.  Throughout this
subsection,
\[
\Theta_{\mathrm{I}}
=
\{\theta\in\Theta:\xi(m(\theta))=T_0\}
=
\{\theta\in\Theta:\ell(\theta)=0\},
\qquad
\ell=\xi\circ m-T_0.
\]
Thus, $\Theta_{\mathrm{I}}$ continues to signify the zero set of the loss, which coincides with the set of minimizers of $\xi\circ m$.
When $\xi$ is the distance to a constraint set $C$ and the null $T_0 = 0$ holds, it is the set of parameters with $m(\theta)$ in the closure of $C$.  Section~\ref{S:examples} provides several examples of the convex-loss structure.  The following assumption specializes
Assumptions~\ref{A:zero}--\ref{A:three} to this setting.

\begin{asm} \label{A:ex1}
    \begin{enumerate}
        \item $\mathfrak{B}$ is a seminormed space, $\mathfrak{X}$ is a Banach space, and $\Theta_{\mathrm{I}} \subset \mathfrak{B}$ is compact.
        \item There are $\delta>0$ and a bounded, convex, open set $D$ containing $B_{\mathfrak{X}}(0,\delta)+m(\Theta)$ on which $\xi$ is finite, convex, and Lipschitz.
        \item $\Theta$ is convex in a neighborhood of $\Theta_{\mathrm{I}}$ and \eqref{E:elldeltaeqn} holds with $\ell = \xi \circ m - T_0$.
        \item The function $m: \Theta \ra \mathfrak{X}$ is uniformly Fr\'{e}chet differentiable over $\Theta$ with order $\alpha_m > 1$.

        \item For some sequence $(r_n)$ with $r_n \ra \infty$, the processes $\mathbb{W}_n(\theta) \equiv r_n (m_n(\theta) - m( \theta))$ converge weakly to uniformly pre-tight limits $\mathbb{X}_n(\theta)$ in $\ell^\infty(\Theta, \mathfrak{X})$.
        \item The elements $\mathbb{X}_n$ are equicontinuous in probability, in that, for all $\ve > 0$,
        \begin{align*}
            \lim_{\delta \downarrow 0}\sup_{n \in \N}\P{\sup_{\substack{\theta \in \Theta,\, \theta' \in B_\Theta(\theta,\delta)}}\norm{\mathbb{X}_n(\theta) - \mathbb{X}_n(\theta')}_{\mathfrak{X}} > \ve } = 0.
        \end{align*}
    \end{enumerate}
\end{asm}

Assumptions~\ref{A:ex1}(1)--(4) are the analogues of Assumptions~\ref{A:zero} and~\ref{A:two} for the $\xi\circ m$ structure, and are relatively standard in infinite-dimensional settings. In particular, $\xi$ may be nondifferentiable, as a norm is at the origin. The local-domain condition in part (2) is automatic when $m(\Theta)$ is bounded and $\xi$ is finite, convex, and Lipschitz on all of $\mathfrak{X}$, as with norm losses.

Assumptions~\ref{A:ex1}(5)--(6) require the empirical process $\mathbb{W}_n$ to have a uniformly pre-tight weak limit $\mathbb{X}_n$ that is equicontinuous in probability. They hold in the leading case of a fixed tight limit, standard in moment-equality and moment-inequality models, where equicontinuity follows from the Donsker property of the underlying function class. The limits $\mathbb{X}_n$ may depend on $n$, as discussed after Assumption~\ref{A:one}, provided they are uniformly pre-tight.  Section~\ref{S:uniform_convex} works with a weaker, quantitative form of this condition.

\begin{prop} \label{P:ex1}
    Let Assumption \ref{A:ex1} hold. Then, if $\Theta_{\mathrm{I}}$ is nonempty,
    \begin{align}
        r_n (\inf_{\substack{\theta \in \Theta}} \xi(m_n(\theta)) - T_0) & \le \inf_{\theta \in \Theta_{\mathrm{I}}} \sup_{y \in (-m'(\theta)(T_\theta \Theta))^\circ \cap \partial \xi(m(\theta))} y (\mathbb{W}_n(\theta)) + o_p(1)  \nonumber \\
        & \rs \inf_{\theta \in \Theta_{\mathrm{I}}}\sup_{y \in (-m'(\theta)(T_\theta \Theta))^\circ \cap \partial \xi(m(\theta))} y (\mathbb{X}_n(\theta)). \label{E:ex1}
    \end{align}
    If $\Theta$ is a singleton, $m$ is affine in a neighborhood of $\Theta_{\mathrm{I}}$, Assumption~\ref{A:ex22} below holds, or there is a coefficient $\alpha_\ell$ satisfying  $\alpha_\ell\alpha_m > 1$ such that
                \begin{align} \label{E:localident1}
                    \sup_{\theta \in \Theta} \frac{d_{\mathfrak{B}}(\theta, \Theta_{\mathrm{I}})}{ (\xi(m(\theta)) - T_0)^{\alpha_\ell} } < \infty,
                \end{align}
    the first line of \eqref{E:ex1} is an equality.
\end{prop}

\ifmainbuild
\else
\begin{rem}[Singularities and reparameterization]
\label{R:singularity}
The minorant~\eqref{E:localident1} holds with $\alpha_\ell=1$ if $\xi(m(\theta))-T_0$ grows at least linearly in $d_{\mathfrak{B}}(\theta,\Theta_{\mathrm{I}})$. When $\Theta_{\mathrm{I}}=\{\theta_0\}$, $T_0=0$, and $\xi$ is a norm, this requires $m'(\theta_0)$ to be bounded below on the tangent cone, in that $\norm{m'(\theta_0)h}_{\mathfrak{X}}\ge c\norm{h}_{\mathfrak{B}}$ for some $c>0$ and all $h\in T_{\theta_0}\Theta$. When $\theta_0$ is interior in finite dimensions, this is the full-column-rank condition familiar from the GMM $J$ test. This is fairly standard, but where it fails the bound~\eqref{E:ex1} can be strict. Take $\xi=|\cdot|$, $m(\theta)=E[X]-\theta^3$ for $X$ standard normal with sample mean $\bar X_n$, and $\Theta=[-1,1]$, so that $\Theta_{\mathrm{I}}=\{0\}$ and $m'(0)=0$. Because $\theta\mapsto\theta^3$ maps $[-1,1]$ onto itself, $\inf_\theta|m_n(\theta)|=d(\bar X_n,[-1,1])=0$ once $\bar X_n\in[-1,1]$, so $T_n\cp0$. The right-hand side of~\eqref{E:ex1}, however, is nondegenerate.  As $m'(0)=0$, $(-m'(0)(T_\theta\Theta))^\circ=\R$ at $\theta=0$, and the limit is $\sup_{|y|\le1}yZ=|Z|$ with $Z\sim N(0,1)$. The first-order formula thus understates the contribution of local movements of $\theta$.  With $m(\theta)=E[X]-\theta^2$ the bound is again strict.  Even so, reparameterization may recover the sharp result.  The statistic $\inf_\theta|m_n(\theta)|$ is invariant under $u=\theta^3$, for which $m(u)=E[X]-u$ is nonsingular and~\eqref{E:ex1} holds with equality and limit $0$. This parallels the use of reparameterization for a nonconvex $\Theta$ (Assumption~\ref{A:zero}).
\end{rem}

\fi

The next assumption gives equality in Proposition~\ref{P:ex1} without a minorant condition.

\begin{asm} \label{A:ex22}
    Either of the following holds:
        \begin{enumerate}
            \item For each fixed $\theta_0\in\Theta_{\mathrm{I}}$, the correspondence
            \[
                \theta \mapsto
                (-m'(\theta)(T_\theta \Theta))^\circ
                \cap (T_{m(\theta)} (\xi^{-1}((-\infty,T_0])))^\circ
                \cap \partial \xi(m(\theta_0))
            \]
            is lower hemicontinuous\footnote{The proof only requires lower hemicontinuity over the restricted set $\{\theta_0\} \cup (\Theta \setminus \Theta_{\mathrm{I}})$.} at $\theta_0$ with respect to the norm topology on $\mathfrak{X}^*$.
            \item The correspondence $\theta \mapsto (-m'(\theta)(T_\theta \Theta))^\circ \cap \partial \xi(m(\theta))$ is lower hemicontinuous at every point $\theta_0$ of $\Theta_{\mathrm{I}}$ with respect to the norm topology on $\mathfrak{X}^*$.
    \end{enumerate}
\end{asm}

Both cases of Assumption~\ref{A:ex22} are hemicontinuity conditions on $\partial\xi$ and $m'(\theta)(T_\theta\Theta)$, and either implies the contact-set continuity of Assumption~\ref{A:four}(3) (Lemma~\ref{L:ex22K} \suppmatin{}).  In case (1), the sublevel set $\xi^{-1}((-\infty,T_0])$ is convex, and its tangent cone at $m(\theta)$ is the radial cone $\mathrm{cone}(\xi^{-1}((-\infty,T_0])-m(\theta))$ of Section~\ref{S:convexanalysis}.  For a distance loss with $T_0=0$, Lemma~\ref{L:geometry} \suppmatin{} gives a primal sufficient condition for case (1) in terms of the following continuity notion, and Remark~\ref{R:constantrank} states and verifies it in finite dimensions.

\begin{defn}[Collar upper hemicontinuity] \label{D:collar}
A correspondence $F$ from a topological space into $\mathfrak{X}$ is \emph{collar upper hemicontinuous} at $a$ if, for every weakly open neighborhood $D$ of $0$, there is a neighborhood $U$ of $a$ with $F(a')\subseteq F(a)+D$ for all $a'\in U$.
\end{defn}

Collar upper hemicontinuity is weaker than upper hemicontinuity with respect to either the weak or the norm topology on $\mathfrak{X}$, because each collar $F(a)+D$ is a weakly open set containing $F(a)$.

\begin{rem}[Verifying Assumption~\ref{A:ex22}(1) in finite dimensions]
\label{R:constantrank}
Let $\mathfrak{X}=\R^d$, let $\theta_0\in\Theta_{\mathrm{I}}$ be interior to $\Theta\subseteq\R^k$, and let $\xi=d_{\mathfrak{X}}(\cdot,C)$ with $C$ convex and $T_0=0$.  By Lemma~\ref{L:geometry} \suppmatin{}, Assumption~\ref{A:ex22}(1) holds at $\theta_0$ if $\mathrm{cone}(m(\theta_0)-m(\theta))+T_{m(\theta_0)}C-m'(\theta)(T_\theta\Theta)$, truncated to the unit ball, is collar upper hemicontinuous in $\theta$.  The ray $\mathrm{cone}(m(\theta_0)-m(\theta))$ is the envelope direction in the right panel of Figure~\ref{F:geometry_norm}\suppref{}.  It adds nothing when
\begin{equation} \label{E:conemembership}
m(\theta_0)-m(\theta)\in T_{m(\theta_0)} C-m'(\theta)(T_\theta \Theta)
\end{equation}
for $\theta$ near $\theta_0$.  What is then needed is collar upper hemicontinuity of $\theta\mapsto(T_{m(\theta_0)}C-m'(\theta)(T_\theta\Theta))\cap B_{\mathfrak{X}}(0,1)$.  This and \eqref{E:conemembership} hold when the cone $T_{m(\theta_0)}C-m'(\theta)(T_\theta\Theta)$ is all of $\mathfrak{X}$ for $\theta$ near $\theta_0$.

Suppose that $C$ encodes moment equalities and inequalities,
\[
C=\{x\in\R^d:\ x_j\le 0\ \ (j\in\mathcal{J}),\qquad x_j=0\ \ (j\in\mathcal{E})\},
\]
and let $\mathcal{B}\subseteq\mathcal{J}$ collect the inequalities binding at $m(\theta_0)$.  The cone in \eqref{E:conemembership} is polyhedral, so it is all of $\R^d$ exactly when its polar $(T_{m(\theta_0)}C)^\circ\cap\ker m'(\theta)^\top$ is trivial.  A vector $y$ lies in $(T_{m(\theta_0)}C)^\circ$ exactly when it is nonnegative on $\mathcal{B}$, zero on the slack inequalities, and unrestricted on $\mathcal{E}$.  Such a vector lies in $\ker m'(\theta)^\top$ exactly when $\sum_{j\in\mathcal{B}\cup\mathcal{E}}y_j\nabla m_j(\theta)=0$, where $\nabla m_j(\theta)$ are the rows of $m'(\theta)$.  The polar is therefore trivial exactly when no nonzero multiplier vector satisfies both conditions.  This is the dual form of the Mangasarian--Fromovitz constraint qualification at $\theta$ with active set $\mathcal{B}$.  Being an open condition, it holds near $\theta_0$ once it holds at $\theta_0$.  The same qualification supplies a local linear minorant condition (Remark~\ref{R:polynomialminorant}).
\end{rem}

\ifmainbuild\else
\begin{figure}[htbp]
\centering
\begin{tikzpicture}[scale=1.40, >=stealth, font=\small]

  \begin{scope}[xshift=-3.1cm]

    \node[font=\footnotesize, align=center] at (0,1.75)
      {Stationarity only: $D(\theta)$ stays near $D(\theta_0)$};

    \draw[->, gray!40, thin] (-1.55,0) -- (1.55,0);
    \draw[->, gray!40, thin] (0,-1.55) -- (0,1.55);
    \draw[gray!35, thin] (0,0) circle (1.3);
    \node[gray!55, font=\scriptsize] at (0.85,1.10) {$B_{\mathfrak{X}}(0,1)$};

    \draw[gray!22, line width=11pt, line join=round, line cap=round]
      (0,0) -- (200:1.3) arc[start angle=200,end angle=262,radius=1.3] -- cycle;

    \fill[green!14] (0,0) -- (200:1.3) arc[start angle=200,end angle=262,radius=1.3] -- cycle;
    \draw[green!60!black, thick] (0,0) -- (200:1.3) arc[start angle=200,end angle=262,radius=1.3] -- cycle;
    \node[green!50!black, font=\footnotesize] at (250:1.62) {$D(\theta_0)$};

    \draw[green!35!black, thick, dashed] (0,0) -- (197:1.22) arc[start angle=197,end angle=247,radius=1.22] -- cycle;
    \node[green!35!black, font=\scriptsize] at (222:0.78) {$D(\theta)$};

    \filldraw[black] (0,0) circle (1.4pt);
    \node[above right=1pt, font=\scriptsize] at (0,0) {$0$};

  \end{scope}

  \begin{scope}[xshift=3.1cm]

    \node[font=\footnotesize, align=center] at (0,1.75)
      {Assumption~\ref{A:ex22}(1): the full-contact wedge};

    \draw[->, gray!40, thin] (-1.55,0) -- (1.55,0);
    \draw[->, gray!40, thin] (0,-1.55) -- (0,1.55);
    \draw[gray!35, thin] (0,0) circle (1.3);
    \node[gray!55, font=\scriptsize] at (0.85,1.10) {$B_{\mathfrak{X}}(0,1)$};

    \draw[gray!22, line width=11pt, line join=round, line cap=round]
      (0,0) -- (200:1.3) arc[start angle=200,end angle=262,radius=1.3] -- cycle;

    \fill[green!14] (0,0) -- (200:1.3) arc[start angle=200,end angle=262,radius=1.3] -- cycle;

    \fill[blue!20] (0,0) -- (197:1.3) arc[start angle=197,end angle=250,radius=1.3] -- cycle;

    \draw[green!60!black, thick] (0,0) -- (200:1.3) arc[start angle=200,end angle=262,radius=1.3] -- cycle;
    \node[green!50!black, font=\footnotesize] at (256:1.62) {$D(\theta_0)$};

    \draw[green!35!black, thick, dashed] (0,0) -- (212:1.22) arc[start angle=212,end angle=250,radius=1.22] -- cycle;
    \node[green!35!black, font=\scriptsize] at (233:0.80) {$D(\theta)$};

    \draw[red!70!black, thick, ->] (0,0) -- (197:1.30);
    \node[red!60!black, font=\scriptsize, anchor=east] at (197:1.36) {$-m(\theta)$};

    \node[blue!55!black, font=\scriptsize, align=center] at (0.62,-1.25) {full-contact\\wedge};

    \filldraw[black] (0,0) circle (1.4pt);
    \node[above right=1pt, font=\scriptsize] at (0,0) {$0{=}m(\theta_0)$};

  \end{scope}

\end{tikzpicture}
\caption{Primal geometric conditions implying lower hemicontinuity in the norm case, with $D(\theta)\equiv -m'(\theta)(T_\theta\Theta)\cap B_{\mathfrak{X}}(0,1)$. The grey band is a neighborhood of $D(\theta_0)$ (solid green). \textit{Left} (stationarity only): for norm loss the contact set equals the stationarity set on $\Theta_{\mathrm{I}}$, so the bootstrap of Proposition~\ref{add:xi_norm} needs only collar upper hemicontinuity, which requires the nearby cone $D(\theta)$ (dashed) to lie in such a neighborhood for $\theta$ close to $\theta_0$. It may rotate or shrink, but it may not acquire new limiting directions. \textit{Right} (Assumption~\ref{A:ex22}(1)): a sufficient condition is that the full-contact wedge $\mathrm{cone}(D(\theta)\cup\{-m(\theta)\})\cap B_{\mathfrak{X}}(0,1)$ (blue), formed by $D(\theta)$ and the envelope direction $-m(\theta)$ (red), satisfy the same requirement. The envelope direction is what makes this condition more demanding.}
\label{F:geometry_norm}
\end{figure}

\fi

When $\xi=\norm{\cdot}_{\mathfrak{X}}$ and $T_0=0$, Assumption~\ref{A:four}(3) holds if, as $\theta\to\theta_0\in\Theta_{\mathrm{I}}$ with $\theta\notin\Theta_{\mathrm{I}}$, the violation direction $m(\theta)/\norm{m(\theta)}_{\mathfrak{X}}$ approaches $\overline{m'(\theta_0)(T_{\theta_0}\Theta)}$ in norm, because each element of $\mathcal{K}(\theta_0)$ can then be tilted slightly toward a functional norming the violation direction. The first line of \eqref{E:ex1} is then an equality without Assumption~\ref{A:ex22} (see the proof of Proposition~\ref{P:ex1}).

Exact limiting distributions with nonsmooth $\xi$ have mostly been derived at a fixed parameter $\theta_0$ \citep{fangsantos2018,HongLi2018}. \citet[online Appendix~A.4]{HongLi2018} and \citet[Supplement, Theorem~9.1]{HongLi2020} use a fixed-parameter directional expansion for this composite null. That expansion need not capture local movements of the minimizing parameter, so its expression is in general an upper bound. A two-dimensional example illustrates that the bound can be strictly larger than our contact-set formula. Take $m(\theta)=(\theta,-\theta)$ on $[-1,1]$ with the constraint $m(\theta)\ge0$ and their loss $\xi(x)=x_1^-+x_2^-$, so that $\Theta_{\mathrm I}=\{0\}$, and $m_n(\theta)=m(\theta)+(Z_1,Z_2)/\sqrt n$ with $Z_1,Z_2$ independent standard normal. Moving $\theta$ by order $n^{-1/2}$ trades off the two violations, so $T_n\rightsquigarrow(Z_1+Z_2)^-$, the value of~\eqref{E:ex1}, whereas the fixed-parameter expansion yields $Z_1^-+Z_2^-$.

For the squared criterion $\sum_j(x_j^-)^2$, the first-order derivative vanishes and \citet{HongLi2018} use the second-order calculus of \citet{BonnansShapiro2000}.  Two limitations affect the squaring approach. First, squaring is criterion-specific. For certain nonsmooth $\xi$, such as $\|\cdot\|_\infty$ or the maximum-violation loss used in the misspecification index of \citet{AK2024}, kinks in $\xi$ near its zero set persist in $\xi^2$ and block the second-order route. Second, the Bonnans--Shapiro second-order formula equates the upper and lower Hadamard derivatives of $\inf_\theta\xi(m(\theta))^2$ only under active-set qualification conditions whose verification is not always straightforward. Proposition~\ref{P:ex1} avoids both issues by working at first order with $\partial\xi(m(\theta))$, which is nonempty for convex $\xi$ continuous at the points of $m(\Theta)$.

In the present notation, the limits of the sieve statistics of \citet{S2012}, \citet{hong2017} and \citet{CNS2023} take the form
\begin{equation}
\label{E:CNStype}
\inf_{\theta \in \Theta_{\mathrm{I}}} \inf_{h \in V(\theta)} \bigl\|\mathbb{X}_n(\theta) + m'(\theta)(h)\bigr\|_{\mathfrak{X}},
\end{equation}
where $V(\theta)\subseteq\overline{T_\theta\Theta}$ is the set of sieve directions entering the limit. In \citet{S2012} and \citet{hong2017}, it is the closure of the sieve directions. In \citet{CNS2023}, it is a norm-bounded set that grows with $n$, so that their approximating distribution also depends on $n$. When $\xi = \|\cdot\|_\mathfrak{X}$ and $T_0=0$, Proposition~\ref{P:ex1} yields by standard duality $\inf_{\theta\in\Theta_{\mathrm{I}}}\inf_{h\in \overline{T_\theta\Theta}}\|\mathbb{X}_n(\theta)+m'(\theta)(h)\|_\mathfrak{X}$, which is at most \eqref{E:CNStype} with equality when the sieve directions are dense in $T_\theta\Theta$.

When $\xi$ is a distance function and the population criterion can attain zero,
the subdifferential $\partial\xi(m(\theta))$ is the unit-ball-truncated normal cone to the constraint
set, and the supremum in
\eqref{E:ex1} reduces to a distance.

\begin{cor} \label{C:geometry}
    Let Assumption \ref{A:ex1} hold with $\xi(x) = d_\mathfrak{X}(x,C)$ for some convex set $C$.  Set $\ell=\xi\circ m$, and suppose that $\Theta_{\mathrm{I}}=\{\theta\in\Theta:\ell(\theta)=0\}$ is nonempty. Then,
    \begin{align*}
        r_n \inf_{\theta \in \Theta} d_\mathfrak{X}(m_n(\theta),C) & \le \inf_{\theta \in \Theta_{\mathrm{I}}} d_{\mathfrak{X}}\big(\mathbb{W}_n(\theta), T_{m(\theta)} C -m'(\theta) (T_\theta \Theta) \big) + o_p(1) \\
        & \rs \inf_{\theta \in \Theta_{\mathrm{I}}} d_{\mathfrak{X}}\big(\mathbb{X}_n(\theta),  T_{m(\theta)} C -m'(\theta) (T_\theta \Theta) \big),
    \end{align*}
    with equality if $\Theta$ is a singleton, $m$ is affine in a neighborhood of $\Theta_{\mathrm{I}}$, Assumption~\ref{A:ex22} holds, or there is an $\alpha_\ell$ satisfying $\alpha_\ell\alpha_m > 1$ and \eqref{E:localident1}.
\end{cor}

\ifmainbuild\else
\begin{figure}[htbp]
\centering
\begin{tikzpicture}[scale=1.85, >=stealth, font=\small]
  \def\aLo{0}      %
  \def\aHi{90}     %
  \def\aMink{140}  %
  \def\R{2.0}      %
  \def\Rext{2.15}  %

  \fill[blue!22]
    (0,0) -- (\R, 0)
    arc[start angle=0, end angle=90, radius=\R]
    -- cycle;

  \fill[green!22]
    (0,0) -- (0, \R)
    arc[start angle=90, end angle=140, radius=\R]
    -- cycle;

  \filldraw[fill=blue!45, draw=blue!80!black, thick]
    (0,0) -- (1.15,0) -- (0,1.15) -- cycle;
  \draw[blue!55!black, thin] (0.11,0) -- (0.11,0.11) -- (0,0.11);

  \draw[orange!80!black, thick, dashed]
    (\R, 0) arc[start angle=0, end angle=140, radius=\R];
  \draw[orange!80!black, thick, ->] (0,0) -- (\Rext, 0);
  \draw[orange!80!black, thick, ->]
    (0,0) -- ({\Rext*cos(140)}, {\Rext*sin(140)});

  \draw[blue!55!black, thick, ->] (0,0) -- (0, \Rext);

  \filldraw[black] (0,0) circle (1.5pt);
  \node[below right=1.5pt, font=\footnotesize] at (0,0) {$m(\theta)$};

  \node[blue!85!black] at (0.31, 0.31) {$C$};

  \node[blue!65!black, font=\footnotesize]
    at ({0.82*cos(45)}, {0.82*sin(45)}) {$T_{m(\theta)}C$};

  \node[green!55!black, font=\footnotesize, align=center]
    at ({0.70*cos(115)}, {0.70*sin(115)})
    {$-m'(\theta)$\\$(T_\theta\Theta)$};

  \node[orange!80!black, font=\footnotesize, align=center,
        fill=white, fill opacity=0.85, inner sep=2pt]
    (minkLab) at (1.55, 2.15)
    {$T_{m(\theta)}C$\\$-\,m'(\theta)(T_\theta\Theta)$};
  \draw[orange!70!black, thin, ->]
    (minkLab.south) -- ({1.45*cos(70)}, {1.45*sin(70)});

  \shade[shading=radial, inner color=gray!55, outer color=white]
    (-0.52,-0.71) ellipse (0.34 and 0.26);
  \node[gray!55!black, below, font=\footnotesize] at (-0.52,-0.98)
    {$\mathbb{X}_n(\theta)$};

  \draw[red!70!black, dashed, semithick] (-0.52,-0.71) -- (0,0);
  \node[red!65!black, font=\scriptsize, right] at (-0.27,-0.38)
    {$d\!\left(\mathbb{X}_n(\theta),\,T_{m(\theta)}C\!-\!m'(\theta)(T_\theta\Theta)\right)$};

  \draw[->] (-0.85,0) -- (2.35,0) node[right] {$x_1$};
  \draw[->] (0,-1.05) -- (0,2.35) node[above] {$x_2$};

\end{tikzpicture}
\caption{Geometry of Corollary~\ref{C:geometry} at a fixed $\theta\in\Theta_{\mathrm{I}}$. The bounded triangle (dark blue) represents a finite section of the convex set $C\subset\mathfrak{X}$, with $m(\theta)\in\partial C$ at the vertex. The \textcolor{blue!55!black}{blue} sector is the tangent cone $T_{m(\theta)}C$, capturing admissible first-order directions within $C$. The \textcolor{green!55!black}{green} sector is $-m'(\theta)(T_\theta\Theta)$, the image of the tangent cone to $\Theta$ under $-m'(\theta)$, capturing how $m(\theta)$ shifts as $\theta$ varies. Their Minkowski sum (bounded by the \textcolor{orange!80!black}{dashed orange arc}) combines both effects. The gray cloud marks a possible location of the Gaussian limit process $\mathbb{X}_n(\theta)$. The \textcolor{red!70!black}{dashed red} segment is the distance $d(\mathbb{X}_n(\theta), T_{m(\theta)}C-m'(\theta)(T_\theta\Theta))$ to the nearest point in the Minkowski sum cone (here, the vertex $m(\theta)$ itself). The infimum of this distance over $\theta\in\Theta_{\mathrm{I}}$ is the weak limit in Corollary~\ref{C:geometry}.}
\label{F:corollary_geometry}
\end{figure}

\fi

Figure~\ref{F:corollary_geometry}\suppref{} illustrates the three cones in Corollary~\ref{C:geometry} and how their Minkowski sum determines the effective projection target for the test statistic. By Lemma~\ref{L:tangentclosure}\suppref{}, both sides of Corollary~\ref{C:geometry} are unchanged if $C$ and this sum are replaced by their closures, so neither need be closed.

\ifmainbuild\else
If $T_0=\inf_{\theta\in\Theta}d_\mathfrak{X}(m(\theta),C)>0$, Proposition~\ref{P:ex1} gives a normal-direction formula over the minimizer set $\Theta_{\mathrm{I}}=\{\theta\in\Theta:d_\mathfrak{X}(m(\theta),C)=T_0\}$.  In a Banach space, \citet[Theorem~1]{burke1992clarke} show that the relevant subgradients are
\[
  \partial d_\mathfrak{X}(\cdot,C)(x)
  =
  \{y\in B_{\mathfrak{X}^*}(0,1): y(x)-\sup_{c\in C}y(c)=d_\mathfrak{X}(x,C)\}.
\]
Hence, the limit in \eqref{E:ex1} keeps only the supporting functionals to the nearest face of $\overline C$ that cannot be offset by local movement in $\theta$.
In a Hilbert space, the projection theorem gives a simpler form.  If $\Pi_{\overline C} x$ is the metric projection onto $\overline C$ and $x\notin \overline C$, then
\[
  \partial d_\mathfrak{X}(\cdot,C)(x)
  =
  \left\{\frac{x-\Pi_{\overline C} x}{\|x-\Pi_{\overline C} x\|_{\mathfrak{X}}}\right\}.
\]
Thus, for each minimizer $\theta\in\Theta_{\mathrm{I}}$, first-order optimality of
$\theta\mapsto d_\mathfrak{X}(m(\theta),C)$ implies that
\[
  u_\theta
  =
  \frac{m(\theta)-\Pi_{\overline C}m(\theta)}
       {\|m(\theta)-\Pi_{\overline C}m(\theta)\|_{\mathfrak{X}}}
  \in (-m'(\theta)(T_\theta\Theta))^\circ .
\]
Under the Hilbert-space identification of $\mathfrak{X}$ with its dual, the
limit in Proposition~\ref{P:ex1} reduces to
\[
  \inf_{\theta\in\Theta_{\mathrm{I}}} \langle u_\theta,\mathbb{X}_n(\theta)\rangle .
\]
Thus, the zero-distance case measures the distance from the noise to a tangent cone of feasibility, while the positive-distance case takes the component of the noise along the outward normal direction from the nearest point in $\overline C$.

\fi

\subsubsection{Examples revisited}\label{S:examples_revisited}

We now apply Proposition~\ref{P:ex1} and Corollary~\ref{C:geometry} to the examples of Section~\ref{S:examples}. Each convex-loss convergence display below assumes one of the equality conditions of Proposition~\ref{P:ex1}. 

\medskip\noindent\textit{Example~\ref{E:jtest} (continued).}
With $m(\theta)=E[g(X,\theta)]\in\R^L$ and $\xi=\norm{\cdot}_W$, the identified set is $\Theta_{\mathrm{I}}=\{\theta:m(\theta)=0\}$. When $\Theta_{\mathrm{I}}$ lies in the interior of $\Theta\subset\R^k$, the locally reachable moment directions at $\theta\in\Theta_{\mathrm{I}}$ are $m'(\theta)(\R^k)$. Let $\Pi_\theta$ denote the $W$-orthogonal projection onto the complement $\{x:m'(\theta)^\top Wx=0\}$ of this space. Corollary~\ref{C:geometry} yields
\[
  T_n \rs \inf_{\theta\in\Theta_{\mathrm{I}}} \norm{\Pi_\theta\,\mathbb{X}_n(\theta)}_W.
\]
The projection removes the locally reachable variation in $\theta$. When $\Theta_{\mathrm{I}}=\{\theta^*\}$, $m'(\theta^*)$ has full column rank $k$, and $W=\Omega^{-1}$ is the efficient weight, $T_n^2\rs\chi^2(L-k)$, recovering the $J$ test of \citet{H1982}. For a general positive definite $W$, the limit of $T_n^2$ is a quadratic form in $\mathbb{X}_n(\theta^*)$ and hence a weighted sum of chi-squares. For nontrivial $\Theta_{\mathrm{I}}$, the projection is computed locally at each identified point, and the outer infimum selects the smallest projected noise. The limit of $T_n$ may be estimated with the bootstrap of Proposition~\ref{P:ex2} (or the norm-specialized variant in \Secapp~\ref{S:stationarity}\suppref{}). The formula extends to moment functions indexed by a Donsker class, with the projection residual replaced by the distance to the reachable subspace.

\medskip\noindent\textit{Example~\ref{E:momentineq} (continued).}
With $m_j(\theta)=E[g_j(X,\theta)]$ and $\xi(x)=\max_{1\le j\le J}x_j^+$, let $J_0(\theta)=\{j:m_j(\theta)=0\}$ denote the binding index set at $\theta\in\Theta_{\mathrm{I}}$. Then
$\partial\xi(m(\theta))=\mathrm{conv}(\{0\}\cup\{e_j:j\in J_0(\theta)\})$, and Proposition~\ref{P:ex1} gives
\begin{align*}
  T_n \rs \inf_{\theta\in\Theta_{\mathrm{I}}} \sup_{y\,\in\,(-m'(\theta)(T_\theta\Theta))^\circ\,\cap\,\mathrm{conv}(\{0\}\cup\{e_j:j\in J_0(\theta)\})} y^\top\mathbb{X}_n(\theta).
\end{align*}
The convex hull keeps only binding inequalities, and its intersection with the polar keeps only weightings that cannot be relaxed by moving $\theta$ locally. When $\Theta=\{\theta_0\}$, stationarity is vacuous and the limit reduces to the maximum positive part of the Gaussian coordinates associated with binding inequalities. If $\xi$ is instead the Mahalanobis distance to the nonpositive orthant, under inverse-covariance weighting, the square of the limit is the familiar chi-bar-squared law \citep{self1987}. The same formula covers a continuum of inequalities, with the convex hull replaced by the nonnegative elements $y$ of the dual unit ball satisfying $y(m(\theta))=0$.

\ifmainbuild\else
\medskip\noindent\textit{Example~\ref{E:condmoment} (continued).}
For a finite instrument collection $\mathcal{H}=\{h_1,\ldots,h_L\}$, the moment map vanishes on $\Theta_{\mathrm{I}}$. There the subdifferential of the sup-norm loss $\xi(f)=\max_{j}|f(h_j)|$ is the full $\ell^1$ ball $\mathrm{conv}\{\pm e_1,\ldots,\pm e_L\}$, which Proposition~\ref{P:ex1} intersects with the tangent-cone polar $(-m'(\theta)(T_\theta\Theta))^\circ$. By Corollary~\ref{C:geometry} with $C=\{0\}$, the limit is the infimum over $\theta\in\Theta_{\mathrm{I}}$ of the sup-norm distance from $\mathbb{X}_n(\theta)$ to the reachable cone $-m'(\theta)(T_\theta\Theta)$. An infinite instrument class replaces $\R^L$ by $\ell^\infty(\mathcal{H})$, the maximum by a supremum over $\mathcal{H}$, and the $\ell^1$ ball by the unit ball of the dual space. Proposition~\ref{P:ex1} applies without discretization when $\{h(Z)q(X,\theta):(\theta,h)\in\Theta\times\mathcal{H}\}$ is Donsker.

\fi

\medskip\noindent\textit{Example~\ref{E:bmm} (continued).}
Let $C_-=\{f\in C(B_{\R^d}(0,1)):f(u)\leq0\text{ for every }u\in B_{\R^d}(0,1)\}$. Because $\xi(f)=\sup_{u\in B_{\R^d}(0,1)}f(u)^+=d_\infty(f,C_-)$, Corollary~\ref{C:geometry} applies. For $\theta\in\Theta_{\mathrm I}$, define the binding support directions
\[
  \mathcal U_0(\theta)=\{u\in B_{\R^d}(0,1):m(\theta)(u)=0\}.
\]
The closure of the tangent cone to $C_-$ at $m(\theta)$ consists of the continuous functions that are nonpositive on $\mathcal U_0(\theta)$. Because distance is unchanged by taking closures,
\[
  d_\infty\bigl(z,T_{m(\theta)}C_-\bigr)
  =\sup_{u\in\mathcal U_0(\theta)}z(u)^+
  \qquad\text{for every }z\in C(B_{\R^d}(0,1)).
\]
Consequently,
\begin{align*}
  T_n
  &\rs \inf_{\theta\in\Theta_{\mathrm I}}
  d_\infty\bigl(\mathbb X_n(\theta),
  T_{m(\theta)}C_- -m'(\theta)(T_\theta\Theta)\bigr)\\
  &=\inf_{\theta\in\Theta_{\mathrm I}}
  \inf_{a\in T_\theta\Theta}
  \sup_{u\in\mathcal U_0(\theta)}
  \bigl[\mathbb X_n(\theta)(u)+m'(\theta)(a)(u)\bigr]^+.
\end{align*}
Only binding support inequalities enter the limit. The local direction $a$ profiles out the component of the support-function noise that can be absorbed by moving $\theta$, after which the supremum records the largest remaining violation. If $\Theta$ is a singleton, the profiling step disappears and the limit reduces to $\sup_{u\in\mathcal U_0(\theta)}\mathbb X_n(\theta)(u)^+$. The formula is therefore a continuum counterpart of the binding-moment limit in Example~\ref{E:momentineq}.

\ifmainbuild\else
\medskip\noindent\textit{Example~\ref{E:lp} (continued).}
Let $L(\theta,x,\lambda)=c(\theta)^\top x+\lambda^\top(b(\theta)-A(\theta)x)$ and define
$\widetilde\Theta_{\mathrm I}=\{(\theta,x)\in\Theta\times K:A(\theta)x=b(\theta),\ c(\theta)^\top x=v_0\}$. As in Section~\ref{S:examples}, the envelope condition is automatic on $\widetilde\Theta_{\mathrm I}$, and stationarity in $(\theta,x)$ gives
\[
  \mathcal{K}(\theta,x) = \bigl\{\bar\lambda\in\Lambda :\; A(\theta)^\top\bar\lambda - c(\theta) \in N_xK,\quad -\nabla_\theta L(\theta,x,\bar\lambda) \in N_\theta\Theta\bigr\}.
\]
Here $N_xK=(T_xK)^\circ$ and $N_\theta\Theta=(T_\theta\Theta)^\circ$ denote normal cones.  The first condition is the usual KKT stationarity condition in $x$. The second is the additional stationarity condition generated by the composite minimization over $\theta$. Writing $\mathbb{X}_n^b(\theta)$, $\mathbb{X}_n^A(\theta)$, and $\mathbb{X}_n^c(\theta)$ for the Gaussian approximations to $\sqrt{n}(b_n-b)$, $\sqrt{n}(A_n-A)$, and $\sqrt{n}(c_n-c)$, the equality case of Theorem~\ref{T:one} gives
\[
  T_n \;\rs\; \inf_{(\theta,x)\in\widetilde\Theta_{\mathrm{I}}} \sup_{\bar\lambda\,\in\,\mathcal{K}(\theta,x)}\bigl[\mathbb{X}_n^c(\theta)^\top x + \bar\lambda^\top\bigl(\mathbb{X}_n^b(\theta) - \mathbb{X}_n^A(\theta)x\bigr)\bigr].
\]
Inside the brackets, $\mathbb{X}_n^c(\theta)^\top x$ perturbs the objective, while $\bar\lambda^\top\bigl(\mathbb{X}_n^b(\theta)-\mathbb{X}_n^A(\theta)x\bigr)$ perturbs the constraints, weighted by the dual variable. The infimum ranges over all optimal pairs $(\theta,x)$. The limit thus combines the local behavior of the LP around its optimal solutions $x$ with that of the nonlinear coefficients around the identified values of $\theta$.

\fi

\subsubsection{Studentized distance statistics}\label{S:studentized}

Studentization by an estimated covariance or scaling matrix $\Omega_n(\theta)$ is standard for scale-invariant inference in the partial identification literature \citep[e.g.,][]{andrews2013,BCS2015,CNS2023,CS2023}.  The results of this paper cover three common forms of studentized distance statistic under standard conditions.

In the first form, the weight is absorbed into the moments and $m_n(\theta)=\Omega_n(\theta)^{-1/2}g_n(\theta)$ for a raw moment vector $g_n$.  The test $r_n\inf_\theta\xi(m_n(\theta))$ then falls directly under Proposition~\ref{P:ex1}, with the sampling error of $\Omega_n$ entering the limit through $\mathbb{W}_n$.  The bootstrap results of Sections~\ref{S:boot} and~\ref{S:uniform} apply to the studentized moments as they stand.

The other two forms measure the distance from $\Sigma_n(\theta)m_n(\theta)$ to $\Sigma_n(\theta)C$, where $C$ is convex and $\Sigma_n(\theta)$ is an estimated weight operator, such as $\Omega_n(\theta)^{-1/2}$, that converges to $\Sigma(\theta)$.  In the second form, $\Sigma$ is uniformly bounded below.  This covers full weight matrices with eigenvalues bounded away from zero, as in the quasi-likelihood-ratio statistic of \citet{CS2023}, which applies an analogous bound.  The third form drops this lower bound.  In exchange, $\Sigma_n$ and $\Sigma$ preserve the constraint set $C$, and $\Sigma_n$ is consistent for $\Sigma$ in the relative sense that the norms $\norm{\Sigma_n(\theta)y}_{\mathfrak{X}}$ and $\norm{\Sigma(\theta)y}_{\mathfrak{X}}$ agree up to a factor $1+o_p(1)$, uniformly in $\theta$ and $y$.  For the diagonal weights of BCS, this requires only uniform consistency of the ratios $\hat\sigma_j/\sigma_j$, which is standard when the $\sigma_j$ are bounded away from zero.  This third form covers the studentization of equality and inequality constraints in BCS and the pure equality case $C=\{0\}$.

\Secapp~\ref{S:studentized_detail}\suppref{} develops the limit theory for the second and third forms.  In both cases, replacing $\Sigma_n$ by $\Sigma$ changes the statistic by $o_p(1)$, and the resulting known-weight statistic has a distance-to-cone limit of the form in Corollary~\ref{C:geometry} (Corollaries~\ref{C:studentized} and~\ref{C:studentized_invariant}\suppref{}).  Under the equality conditions of the two corollaries, the re-optimization bootstrap~\eqref{E:studentized_bootstrap} of Section~\ref{S:boot}, which holds the estimated weight fixed across draws, consistently estimates these limits (\Secapp~\ref{S:studentized_boot}\suppref{}).  Under uniform versions of their conditions, it also controls size uniformly (Corollary~\ref{C:unifstud}\suppref{}).

\section{Bootstrap consistency}\label{S:boot}

The limit of $T_n$ in Theorem~\ref{T:one} is non-pivotal in general, so we calibrate the test by bootstrap. Theorem~\ref{T:two} shows that re-solving the minimax problem after a scaled bootstrap perturbation, which we call the \emph{re-optimization bootstrap}, replicates this limit under the following assumptions.

\begin{asm}[Bootstrap consistency] The following hold:
\label{A:five}
\begin{enumerate}
\item There is a bootstrap analogue $\mathbb{V}_n^*$ of $\mathbb{V}_n$ which satisfies
$
\mathbb{V}_n^* \overset{\mathrm{P}}{\rs} \mathbb{G}_n
$
in $\ell^\infty(\Theta \times \mathcal{T}).$
\item For all $\theta \in \Theta$ and $n$, $t \mapsto \mathbb{V}_n^*(\theta, t)$ is almost surely upper semicontinuous and concave.
\item The following uniform Lipschitz condition holds for the functions $\theta \mapsto v(\theta, t)$:
    \begin{align*}
        C_v \equiv \sup_{\theta, \theta' \in \Theta}\sup_{t \in \mathcal{T}} \frac{|v(\theta, t) - v(\theta', t)|}{d_{\mathfrak{B}}(\theta, \theta')} < \infty.
    \end{align*}
\end{enumerate}
\end{asm}

Concavity in Assumption~\ref{A:five}(2) is automatic when the criterion is linear in $t$, as in the examples. The Lipschitz condition in part~(3) controls localization bias.

\begin{theorem} \label{T:two}
Suppose that Assumptions \ref{A:zero}, \ref{A:one}, \ref{A:two}, \ref{A:three}, and \ref{A:five} hold, and that $\Theta_{\mathrm{I}}$ is nonempty. Let $s_n\ra\infty$ with $s_n=o(r_n)$, and suppose that any one of the following holds:
\begin{itemize}
    \item[(i)] Assumption \ref{A:four}(3) holds and $\alpha_v<\infty$.
    \item[(ii)] The linearity condition $\alpha_v = \infty$ holds.
    \item[(iii)] Assumption \ref{A:four}(1) holds and $\alpha_v < \infty$.
\end{itemize}
Then,
\begin{align}
\inf_{\theta \in \Theta} \sup_{t \in \mathcal{T}}  \big(s_n v_n(\theta, t)+\mathbb{V}_n^*(\theta, t) \big) & = \inf_{\theta \in \Theta_{\mathrm{I}}} \sup_{t \in \mathcal{K}(\theta)} \mathbb{V}_n^*(\theta, t) + o_{p^*}(1) \nonumber \\
& \overset{\mathrm{P}}{\rs} \inf_{\theta \in \Theta_{\mathrm{I}}} \sup_{t \in \mathcal{K}(\theta)} \mathbb{G}_n(\theta, t). \label{E:bootstrapone}
\end{align}
\end{theorem}

All bootstrap results in this paper yield critical values in the same way.  Let $\widehat q_n^*(1-\alpha)$ be the conditional $(1-\alpha)$-quantile of a bootstrap statistic.  The proposed critical value is
\begin{equation} \label{E:critval}
\widehat c_n = \widehat q_n^*(1-\alpha)+\ve, \qquad \ve\in(0,1] \text{ fixed},
\end{equation}
and the corresponding test rejects when $T_n>\widehat c_n$.  The slack stands in for a continuity condition on the limit law, which uniform coverage would otherwise require \citep[Remark~2.7]{RomanoShaikh2012}.  Lemma~\ref{L:validity}\suppref{} shows that the test has asymptotic level $\alpha$ whenever the bootstrap law is a \emph{uniform upper-bound estimator} that asymptotically dominates the law of $T_n$ uniformly over the null configurations.  The bootstraps of this section are such estimators at a fixed distribution, and Theorem~\ref{P:unif1}\suppref{} and Proposition~\ref{P:unifbound} give conditions under which they are so uniformly.  Suppose that the limit law of $T_n$ equals the bootstrap target, as in the equality cases of Theorem~\ref{T:one} and Proposition~\ref{P:ex1}, and that its distribution function is continuous and strictly increasing at the $(1-\alpha)$-quantile.  Then no critical value that converges in probability to a constant and controls size at the given configuration is asymptotically smaller than $\widehat q_n^*(1-\alpha)$, so \eqref{E:critval} is the least conservative such choice up to the slack.  A slack $\ve_n\downarrow0$ that vanishes slowly enough keeps the test valid and, in this case, makes it asymptotically exact (\Secapp~\ref{S:validity}\suppref{}).

The \emph{step} $s_n$ is the only tuning parameter of the re-optimization bootstrap. The proof of Theorem~\ref{T:two} also shows that \eqref{E:bootstrapone} holds for the localized statistic
\[
\inf_{\theta\in\Theta}\inf_{\theta'\in B_\Theta(\theta,b_n)}\sup_{t\in\mathcal{T}}\big(\lambda_nb_n^{-1}v_n(\theta',t)+\mathbb{V}_n^*(\theta,t)\big)
\]
whenever $\lambda_n\to\infty$, $b_n\to0$, and $\lambda_nb_n^{-1}=o(r_n)$, provided also that $\lambda_nb_n^{\alpha_v-1}\to0$ when $\alpha_v<\infty$ and $\lambda_nb_n^{\alpha_\ell\alpha_v-1}\to0$ in case~(iii). Because $\theta'=\theta$ is admissible, the localized statistic is no larger than the re-optimization draw with $s_n=\lambda_nb_n^{-1}$. It therefore yields weakly smaller critical values, at the cost of a double minimization and a second tuning parameter. The two coincide when $\mathbb{V}_n^*$ does not depend on $\theta$, as in Section~\ref{S:mc}.

When neither contact-set continuity nor a polynomial minorant is available, a modified bootstrap with a penalty sequence $\eta_n\to\infty$ still estimates the limiting bound of Theorem~\ref{T:one}, provided that a known (possibly nonpolynomial) minorant of $\ell$ calibrates its tuning (Proposition~\ref{ad:one} \suppmatin{}). The same construction applies to the uniform bootstrap approximations below. \Secapp~\ref{S:stationarity}\suppref{} gives a stationarity-only variant, which replaces the contact set by the stationarity set $\mathcal{K}_0$.

\subsection{Bootstrap consistency under convex loss}

We now specialize Theorem~\ref{T:two} to the convex-loss structure $\xi\circ m$ of Proposition~\ref{P:ex1}, where the bootstrap statistic takes a more transparent form. Assumption~\ref{A:ex2} adapts parts~(1) and~(3) of Assumption~\ref{A:five} to this structure, and Proposition~\ref{P:ex2} gives bootstrap consistency under the equality conditions of Proposition~\ref{P:ex1}.

\begin{asm}
    \label{A:ex2}
    \begin{enumerate}
    \item There is a bootstrap analogue $\mathbb{W}_n^*$ of $\mathbb{W}_n$ which satisfies $\mathbb{W}_n^* \overset{\mathrm{P}}{\rs} \mathbb{X}_n$ in $\ell^\infty(\Theta, \mathfrak{X})$.
    \item The function $m$ satisfies the uniform Lipschitz condition:
    \begin{align*}
        \sup_{\theta, \theta' \in \Theta} \frac{\norm{m(\theta) - m(\theta')}_{\mathfrak{X}}}{d_{\mathfrak{B}}(\theta, \theta')} < \infty.
    \end{align*}
    \end{enumerate}
\end{asm}

\begin{prop}[Convex-loss bootstrap] \label{P:ex2}
Let Assumptions \ref{A:ex1} and \ref{A:ex2} hold with $\Theta_{\mathrm{I}}$ nonempty. Let $s_n\ra\infty$ with $s_n=o(r_n)$, and suppose that any one of the following holds:
\begin{itemize}
    \item[(i)] Assumption \ref{A:ex22} holds and $\alpha_m<\infty$.
    \item[(ii)] $m$ is affine in a neighborhood of $\Theta_{\mathrm{I}}$.
    \item[(iii)] \eqref{E:localident1} holds with $\alpha_\ell\alpha_m>1$ and $\alpha_m < \infty$.
\end{itemize}
Then,
\begin{align*}
s_n\Big(\inf_{\theta \in \Theta}\xi\big( m_n(\theta) + s_n^{-1} \mathbb{W}_n^*(\theta)\big) - T_0\Big) & =  \inf_{\theta \in \Theta_{\mathrm{I}}}
\sup_{\substack{y \in (-m'(\theta)(T_\theta \Theta))^\circ\\ \cap\, \partial \xi(m(\theta))}} y(\mathbb{W}_n^*(\theta)) + o_{p^*}(1) \\
& \overset{\mP}{\rs}
\inf_{\theta \in \Theta_{\mathrm{I}}}
\sup_{\substack{y \in (-m'(\theta)(T_\theta \Theta))^\circ\\ \cap\, \partial \xi(m(\theta))}} y(\mathbb{X}_n(\theta)).
\end{align*}
\end{prop}

The bootstrap statistic in Proposition~\ref{P:ex2} is centered at $s_nT_0$, which need not be known.  Replacing $T_0$ by the plug-in estimate $\widehat T_0=\inf_{\theta\in\Theta}\xi(m_n(\theta))$ changes the statistic by $O_p(s_n/r_n)=o_p(1)$, because $\xi$ is Lipschitz.  For nonnegative losses ($T_0\ge0$), as in all our examples, setting $T_0=0$ instead can only increase the statistic.  The resulting critical value is conservative in general and equals the oracle-centered one under the feasibility null $T_0=0$.  We recommend this centering, which \citet{HongLi2018} also allow for nonnegative functionals.  Centering at $\widehat T_0$ is asymptotically equivalent, but in the simulations of \Secapp~\ref{S:supp_centering}\suppref{} it over-rejects at moderate steps.

The display in Proposition~\ref{P:ex2} contains the geometric limits from
Section~\ref{S:convexloss} as special cases.  When
$\xi=d_{\mathfrak{X}}(\cdot,C)$, Lemma~\ref{L:normalcone}\suppref{} gives
$\partial\xi(m(\theta))=B_{\mathfrak{X}^*}(0,1)\cap(T_{m(\theta)}C)^\circ$ on
$\Theta_{\mathrm{I}}$ when $T_0=0$, and Lemma~\ref{L:positive}\suppref{} turns the supremum over
supporting functionals into the distance
$d_{\mathfrak{X}}(\mathbb{X}_n(\theta),T_{m(\theta)}C-m'(\theta)(T_\theta\Theta))$
    appearing in Corollary~\ref{C:geometry}.  The three cases parallel those of Theorem~\ref{T:two}, with $\alpha_m$ in the role of $\alpha_v$.  In each case the step need only satisfy $s_n\to\infty$ and $s_n=o(r_n)$, that is, $1\ll s_n\ll\sqrt n$ when $r_n=\sqrt n$.  For norm losses, the stationarity-only variant in \Secapp~\ref{S:stationarity}\suppref{} estimates the limiting bound~\eqref{E:ex1} of Proposition~\ref{P:ex1} assuming only collar upper hemicontinuity of the truncated derivative cone (Proposition~\ref{add:xi_norm}\suppref{}).

The re-optimization bootstrap of Proposition~\ref{P:ex2} can be further specialized to the studentized statistics of Section~\ref{S:studentized} by holding the estimated weight fixed across draws:
\begin{align} \label{E:studentized_bootstrap}
\widehat Z_{n,\Sigma_n}^*
=s_n\inf_{\theta\in\Theta}d_{\mathfrak{X}}\!\left(
\Sigma_n(\theta) m_n(\theta)+s_n^{-1}\Sigma_n(\theta)\mathbb W_n^*(\theta),
\Sigma_n(\theta) C\right).
\end{align}
Applying Lemma~\ref{L:weightreplace}\suppref{} as in \Secapp~\ref{S:studentized_limit}\suppref{}, we may replace $\Sigma_n$ by $\Sigma$ in these draws at a cost of $o_p(1)$.  The results of this section and of Section~\ref{S:uniform_convex} therefore apply to \eqref{E:studentized_bootstrap} through the corresponding known-weight problem (\Secapp~\ref{S:studentized_boot} and Corollary~\ref{C:unifstud}\suppref{}).

\begin{rem}[Relation to existing procedures]\label{R:relation}
The re-optimization bootstrap of Proposition~\ref{P:ex2} is the numerical delta method of \citet{HongLi2018} for $f\mapsto\inf_\theta\xi(f(\theta))$ with their step size $\varepsilon_n=1/s_n$, centered at the null value.  Their pointwise theory \citep[Theorems~3.1 and~3.3]{HongLi2018} presupposes a directional derivative, which Corollary~\ref{C:dirderiv} supplies.  Their uniform theory \citep[Section~3.2]{HongLi2018} treats functionals of a finite-dimensional parameter satisfying a convexity-type condition and is not immediately applicable here \citep[online Appendix~A.4]{HongLi2018}.  Theorem~\ref{P:unif1}\suppref{} and Proposition~\ref{P:unifbound} give uniform validity instead, under a uniform minorant condition.

For finitely many studentized moment equalities and inequalities and the sum-of-squares criterion, the square of the draw (see~\eqref{E:studentized_bootstrap}) is the penalize-resampling statistic of \citet{BCS2017} with $\kappa_n=\sqrt n/s_n$, as \citet[Supplement, Section~9]{HongLi2020} observe for their numerical bootstrap.  Their conditions $\kappa_n\to\infty$ and $\kappa_n/\sqrt n\to0$ are ours on $s_n$.  The recommended $\kappa_n=\sqrt{\ln n}$ \citep{as2010} gives $s_n=\sqrt{n/\ln n}$, near the upper end of this range, where the test is conservative in our simulations (Section~\ref{S:mc}).  Their Test MR takes the draw-wise minimum of this statistic and a discard-resampling statistic, and \citet{BBC2019} extend both to a number of inequalities that grows with $n$.  These papers establish uniform size control for their bootstrap tests under a minorant-type condition.  Proposition~\ref{P:ex2} and its studentized and uniform extensions (\Secapp~\ref{S:studentized_detail}\suppref{} and Section~\ref{S:uniform}) extend the analysis of this draw to general convex Lipschitz losses and Banach-valued moments, including conditional moment restrictions indexed by a Donsker class of instruments.  In particular, the penalize-resampling draw is on its own both pointwise exact, under the equality conditions of Proposition~\ref{P:ex1}, and uniformly valid.  The discard-resampling component of Test MR, and with it the estimated set of near minimizers, can therefore be dropped without losing either property asymptotically.  With an intermediate step, it also performs about as well as Test MR in our simulations (Section~\ref{S:mc}).
\end{rem}

\begin{rem}[Consistency]\label{R:consistency}
Under the alternative, the minimal violation $T_0=\inf_{\theta}\sup_{t}v(\theta,t)$ is positive.  Because $\mathbb V_n$ is tight (Assumption~\ref{A:one}), $T_n\geq r_nT_0-O_p(1)$ diverges.  Each draw of the re-optimization bootstrap is at most $s_nr_n^{-1}T_n+\sup_{\theta,t}|\mathbb V_n^*(\theta,t)|$, and the plug-in centering above only lowers it, so the critical value~\eqref{E:critval} is at most $s_nr_n^{-1}T_n+O_p(1)$.  Because $s_n=o(r_n)$, the test is consistent against every fixed alternative.
\end{rem}

\section{Uniform results} \label{S:uniform}

We now show that, under a uniform minorant condition, the re-optimization bootstrap controls size uniformly over classes of parameter sets and distributions.  Uniformity is a standard requirement \citep{RomanoShaikh2008, RomanoShaikh2012}, which is salient
here because the limit in Proposition~\ref{P:ex1} depends closely
on the geometry of the identified set.

Throughout this section, $R \subseteq \Theta$ is a closed parameter set and $P$ is a probability distribution.  The objects $v_{n,R,P}$, $\ell_{R,P}$, $m_{R,P}$, $m_{n,R,P}$, and $\mathbb{W}_{n,R,P}$ are the criterion, loss, moment map, moment estimator, and empirical process of the preceding sections, now allowed to depend on $(R,P)$ subject only to the uniform assumptions below.  Let $\mathcal{S}$ denote a fixed collection of pairs $(R,P)$ for which the identified set $\Theta_{\mathrm{I}}(R,P)=\{\theta\in R:\ell_{R,P}(\theta)=0\}$ is nonempty.  A family of processes $\{\mathbb{Z}_{n,R,P}\}$ is \textit{uniformly bounded in probability} if $\norm{\mathbb{Z}_{n,R,P}}_{\ell^\infty}$ is $O_p(1)$ uniformly in $n$ and $(R,P)$.  Theorem~\ref{P:unif1}\suppref{} establishes a uniform upper bound analogous to Theorem~\ref{T:one}, with a uniformly consistent bootstrap, in the general setting.  Proposition~\ref{P:unifbound} does the same in the convex-loss setting.

A leading application is subvector inference on a
functional $g : \Theta \ra \mathbb{R}^k$.  The hypothesis $g(\theta_0) = c$ for a
fixed $c$ is the composite null of Section~\ref{S:examples} with $\Theta$ replaced by $R_c \equiv \{\theta \in \Theta : g(\theta) = c\}$.  As $c$ ranges over the null values, the pairs $(R_c,P)$ form a subfamily of $\mathcal{S}$.  Tests may be based on the statistics
\begin{equation}\label{E:TnRP}
  T_{n,R,P} \;\equiv\; r_n\inf_{\theta\in R}\sup_{t\in\mathcal{T}} v_{n,R,P}(\theta,t).
\end{equation}
Inverting the tests based on $T_{n,R_c,P}$ over $c$ yields confidence sets for $g(\theta_0)$, and a recentered variant (Remark~\ref{R:recentered}\suppref{}) always yields nonempty ones.  Pointwise validity follows from Theorem~\ref{T:one} and the bootstrap results of Section~\ref{S:boot} \citep[cf.][]{IM2004}.  Uniform coverage requires
the tests to control size uniformly over the null values of $c$.

For $\delta>0$, define the uniform analogue of the modulus of identifiability \eqref{E:felleqn},
\begin{align} \label{E:fello}
    \overline{f}_\ell(\delta) \equiv \sup_{(R,P) \in \mathcal{S}}\sup_{\substack{\theta \in R \\ \ell_{R,P}(\theta) \le \delta}} d_{\mathfrak{B}}(\theta, \Theta_{\mathrm{I}}(R,P)).
\end{align}
The condition $\overline f_\ell(\delta)=o(1)$ means that the low-loss level sets shrink
uniformly toward the identified sets.  The bound
$\overline f_\ell(\delta)=O(\delta^{\alpha_\ell})$ is its polynomial strengthening and the
uniform counterpart of the polynomial minorant in Remark~\ref{R:polynomialminorant}.

The behavior of $\overline f_\ell$ near zero sorts the uniform results into three tiers.
In the first tier, $\overline f_\ell(\delta)=O(\delta^{\alpha_\ell})$ with
$\alpha_\ell\alpha_v>1$, or $\alpha_\ell\alpha_m>1$ under convex loss.  If the Taylor
remainder vanishes, $\overline f_\ell(\delta)=o(1)$ suffices instead.
Theorem~\ref{P:unif1}\suppref{} and Proposition~\ref{P:unifbound} then give uniform upper
bounds for the re-optimization bootstrap with any step $s_n\to\infty$, $s_n=o(r_n)$.  Uniform minorants of this kind are standard
in partial identification, especially at the linear rate $\alpha_\ell=1$.  For example, the
supplement to \citet{CHT2007} and \citet[Assumptions A.8(a) and A.3(a)]{BCS2015,BCS2017}
impose them uniformly over a class of distributions.  Each states the bound for a power
of the loss, typically its square, which supplies a uniform linear
minorant in our setting and thus places a problem in the first tier.
Calibrated projection \citep{KMS2019} avoids such a condition by testing each $\theta$
separately, whereas the uniform behavior of a statistic that minimizes over $\theta$
depends on $\overline f_\ell$.

In the second tier, only
$\overline f_\ell(\delta)=o(1)$ holds, and possibly at an arbitrarily slow rate.
Propositions~\ref{Add:2}\suppref{} and~\ref{Add:unifxi} give the same guarantee
for a localized bootstrap with an additional penalty, whose tuning reflects a known bound on
$\overline f_\ell$ near zero.  This bootstrap re-optimizes over $\theta'\in B_R(\theta,b_n)$ and penalizes the level $\ell_n(\theta)$ at the center $\theta$.  The penalty $\eta_n\ell_n(\theta)$ pins $\theta$ near $\Theta_{\mathrm{I}}$, while $\theta'$ remains free to absorb noise in the stationarity directions.  For their bootstrap critical values, \citet{BBC2019} instead allow the slope of the minorant to vanish with $n$.  Here the class $\mathcal{S}$ is fixed and the modulus may vanish at any fixed rate.

In the third tier, even $\overline f_\ell(\delta)=o(1)$
fails, and uniform validity and pointwise exactness are not generally compatible.
\Secapp~\ref{S:uniform_separation}\suppref{} constructs a Gaussian
moment-inequality family that satisfies every other condition of
Proposition~\ref{Add:unifxi}.  In that family, no estimator is both a uniform
upper-bound estimator, as the uniform results provide, and pointwise exact, as
Proposition~\ref{P:ex2} provides.  The same holds for tests whose critical values converge in probability to a constant.

\ifmainbuild\else
\subsection{Uniform validity in the general minimax setting}
Because it gives only an upper bound, Theorem~\ref{P:unif1} needs no minimax step, which
streamlines the conditions that Theorem~\ref{T:two} imposes on the same bootstrap.
Assumptions~\ref{AU:1}--\ref{AU:4} impose the remaining conditions uniformly over $\mathcal{S}$,
and the identification rate condition appears directly in Theorem~\ref{P:unif1}.

\begin{asm}[Uniform separation]
\label{AU:1}
$\Theta$ is a subset of a seminormed space $\mathfrak{B}$ satisfying:
\begin{enumerate}
\item There is some $\delta > 0$ such that, for all $(R,P) \in \mathcal{S}$ and $\theta \in \Theta_{\mathrm{I}}(R,P)$, $B_R(\theta, \delta)$ is convex.

\item For all $(R,P)\in\mathcal{S}$, the function $\ell_{R,P}$ maps $R$ to $\R_+$, and for all $\delta > 0$,
\begin{align*}
\inf_{(R,P) \in \mathcal{S}} \inf_{\substack{\theta \in R \\ d_{\mathfrak{B}}(\theta, \Theta_{\mathrm{I}}(R,P)) > \delta}} \ell_{R,P}(\theta) > 0.
\end{align*}
\end{enumerate}
\end{asm}

Assumption~\ref{AU:1}(2) is equivalent to
$\overline f_\ell(\delta)=o(1)$, with $\overline f_\ell$ as in \eqref{E:fello}.  It is
the weakest identification condition under which this section provides uniformly
valid bootstrap critical values.  Without it, the loss can be nearly zero at
parameters a fixed distance from $\Theta_{\mathrm{I}}(R,P)$, where the localization
arguments behind Theorem~\ref{P:unif1} and Proposition~\ref{P:unifbound} break down.  This is the third tier, analyzed in \Secapp~\ref{S:uniform_separation}.

Assumption~\ref{AU:1}(1) requires $R$ to be locally convex near $\Theta_{\mathrm{I}}(R,P)$.  The leading case is a linear functional $g$, for which $R_c=\{\theta:g(\theta)=c\}$ is the intersection of $\Theta$ with an affine subspace and retains its convexity properties near the identified set. \citet{S2012} provides examples of linear functionals in an NPIV setting.
When $g$ is nonlinear, the sets $R_c$ need not be convex.  They can still be accommodated when each is locally the image of a convex set under a smooth change of coordinates, through the reparameterization strategy following Assumption~\ref{A:zero}(1).  For example, if $g$ is a submersion near $\Theta_{\mathrm{I}}(R,P)$ and $\Theta_{\mathrm{I}}(R,P)$ lies in the interior of $\Theta$, each $R_c$ is locally a smooth manifold, which a chart maps to an open subset of a linear space.  Charts that are $C^2$ leave the smoothness and identification exponents of this section unchanged whenever those exponents are at most two.

\begin{asm}[Uniform Donskerness] \label{AU:2}
\begin{enumerate}
\item For a sequence $r_n \ra \infty$, the empirical processes $\mathbb{V}_{n,R,P}(\theta, t) = r_n (v_{n,R,P}(\theta, t) - v_{R,P}(\theta, t))$ converge uniformly in $(R,P) \in \mathcal{S}$ to Borel measurable elements $\mathbb{G}_{n,R,P}$ in $\ell^\infty(R \times  \mathcal{T})$ which are uniformly bounded in probability.

\item The limit processes $\mathbb{G}_{n,R,P}(\theta, t)$ are uniformly $\rho_{n,R,P}$-equicontinuous in probability with respect to a pseudometric $\rho_{n, R,P}: (\Theta \times \mathcal{T})^2 \ra \R_+$.  %

\end{enumerate}
\end{asm}

Assumption~\ref{AU:2} is a uniform Donsker-type requirement with standard primitives.  Theorems~2.8.2--2.8.4 of \citet{VW1996} characterize uniform Donskerness, under a uniform envelope condition, through uniform total boundedness of the index set and uniform asymptotic equicontinuity, and give sufficient entropy conditions.  These theorems require an index set that does not vary across distributions.  In most applications this is automatic, because $R$ enters only through the outer infimum in $T_{n,R,P}$ and not through the process $\mathbb{V}_{n,R,P}$.  The index set is then the fixed $\Theta\times\mathcal{T}$, and uniform convergence over it implies uniform convergence over every $R\times\mathcal{T}$.

\begin{asm}[Uniform differentiability] \label{AU:3}
There is a constant $\alpha_v > 1$ and, for all $(R,P)\in \mathcal{S}$, $\theta\in\Theta_{\mathrm{I}}(R,P)$, and $t\in\mathcal{T}$, a positively homogeneous functional $D_{\theta;t} v_{R,P}: \operatorname{span}(\Theta) \ra \R$ such that
\[
\sup_{(R,P) \in \mathcal{S}} \sup_{\substack{\theta \in \Theta_{\mathrm{I}}(R,P), \, \theta' \in B_R(\theta,\delta)  \\ t \in \mathcal{T}}} | v_{R,P}(\theta', t) - v_{R,P}(\theta, t) - D_{\theta; t} v_{R,P}(\theta' - \theta)| = O(\delta^{\alpha_v}).
\]
\end{asm}

\begin{asm} \label{AU:4}
For all $\ve > 0$, there exists some $\delta > 0$ such that the pseudometrics $\{\rho_{n,R,P}: n\in\N,\ (R,P) \in \mathcal{S}\}$ satisfy the uniform equicontinuity condition
\[
\sup_{n\in\N}\sup_{(R,P) \in \mathcal{S}} \sup_{\substack{ \theta \in R,\, \theta' \in B_R(\theta,\delta) \\ t \in \mathcal{T}}} \rho_{n,R,P}((\theta, t), (\theta', t)) < \ve.
\]
\end{asm}

The next result covers the first tier, in which either $v_{R,P}(\cdot,t)$ is linear in
$\theta$, so that $\alpha_v=\infty$, or $\overline f_\ell(\delta)=O(\delta^{\alpha_\ell})$
with $\alpha_\ell\alpha_v>1$.

\begin{theorem}[Uniform version of Theorems~\ref{T:one} and~\ref{T:two}] \label{P:unif1}
    Suppose that Assumptions \ref{AU:1}, \ref{AU:2}, \ref{AU:3}, and \ref{AU:4} hold, and let $\overline f_\ell$ be defined by \eqref{E:fello}. Let $s_n\ra\infty$ with $s_n=o(r_n)$, and let either of the following also hold:
        \begin{itemize}
            \item The linearity condition $\alpha_v = \infty$ holds.
            \item $\alpha_v < \infty$, and there is an $\alpha_\ell$ such that $\overline{f}_\ell(\delta) = O(\delta^{\alpha_\ell})$ and $\alpha_\ell\alpha_v > 1$.
        \end{itemize}
        Then, uniformly in $(R,P) \in \mathcal{S}$, the statistic~\eqref{E:TnRP} satisfies
    \begin{align}
        T_{n,R,P} \le \inf_{\theta \in R} \sup_{t \in \mathcal{T}} \big(s_n v_{n,R,P}(\theta, t) + \mathbb{V}_{n,R,P}(\theta, t)\big) + o_p(1). \label{E:upperbdprop}
    \end{align}
    Given a bootstrap process $\mathbb{V}_{n,R,P}^*$ with $\mathbb{V}_{n,R,P}^* \overset{\mathrm{P}}{\rs} \mathbb{G}_{n,R,P}$ uniformly in $(R,P) \in \mathcal{S}$, replacing $\mathbb{V}_{n,R,P}$ by $\mathbb{V}_{n,R,P}^*$ in \eqref{E:upperbdprop} yields an estimator of the law of its right side that is uniformly consistent in the bounded-Lipschitz metric.
\end{theorem}

Theorems~\ref{T:one}, \ref{T:two}, and~\ref{P:unif1} play complementary roles for
the same bootstrap.  For each fixed $(R,P)$, Theorem~\ref{T:one} gives the limit of $T_{n,R,P}$,
exact under Assumption~\ref{A:four}, and Theorem~\ref{T:two} shows that the re-optimization
bootstrap estimates it.  Theorem~\ref{P:unif1} shows that the same bootstrap is a uniform
upper-bound estimator over $\mathcal{S}$.  With the critical value~\eqref{E:critval},
Lemma~\ref{L:validity} then gives uniform size control.
Propositions~\ref{ad:one} and~\ref{Add:2} play the same two roles for the localized, penalized
bootstrap when the conditions of Theorem~\ref{P:unif1} fail but $\overline f_\ell(\delta)=o(1)$.
Proposition~\ref{ad:one} shows that it is pointwise consistent for the limiting bound of
Theorem~\ref{T:one}, and Proposition~\ref{Add:2}, whose tuning uses a known upper bound on
$\overline f_\ell$ near zero, shows that it is a uniform upper-bound estimator over $\mathcal{S}$. The pointwise and uniform results of \citet{HongLi2018} play the same two roles for
certain directionally differentiable functionals (Remark~\ref{R:relation}).

\fi

\ifmainbuild\else
\subsection{Uniform validity under convex loss}\fi\label{S:uniform_convex}

This subsection specializes Theorem~\ref{P:unif1}\suppref{} to the convex-loss
setting of Propositions~\ref{P:ex1} and~\ref{P:ex2}.  Under the conditions of
Proposition~\ref{P:unifbound} below, the re-optimization bootstrap of Proposition~\ref{P:ex2} is a uniform upper-bound
estimator over $\mathcal{S}$ for every step $s_n$ with $s_n\to\infty$ and $s_n=o(r_n)$.  Throughout, $T_0(R,P) = \inf_{\theta \in R}\xi(m_{R,P}(\theta))$ is the $(R,P)$ analogue of $T_0$, and the loss is $\ell_{R,P}(\theta) = \xi(m_{R,P}(\theta)) - T_0(R,P)$.  As in Section~\ref{S:convexloss}, $\Theta_{\mathrm{I}}(R,P)$ is then the set of minimizers, so $\mathcal{S}$ may contain configurations with $T_0(R,P)>0$.

A family of processes $\{\mathbb{X}_{n,R,P}\}$ is \textit{quantitatively pre-tight} if, for every $\ve > 0$, there are measurable sets $S_{n,R,P} \subset \ell^\infty(\Theta, \mathfrak{X})$ such that $P(\mathbb{X}_{n,R,P} \in S_{n,R,P}) \ge 1 - \ve$ for all $n$ and $(R,P)$, and $\sup_{n,(R,P)}N(S_{n,R,P}, \norm{\cdot}_\infty, \delta) < \infty$ for every $\delta > 0$.  Here $N(S,\norm{\cdot}_\infty,\delta)$ is the smallest number of $\norm{\cdot}_\infty$-balls of radius $\delta$ that cover $S$.  Part~(4) of Assumption~\ref{A:uniformbound} below imposes this on the class of limit processes.  The condition is weaker than the uniform pre-tightness of Assumption~\ref{A:ex1}(5), which requires a single fixed set $S$ for all $n$. It is also weaker than the uniform tightness of \citet[\S\S1.3 and~1.12]{VW1996}, which requires that fixed set to be compact.\ifmainbuild\else\footnote{The condition is a natural byproduct of uniform convergence of the empirical process. %
Theorem~2.8.2 of \citet{VW1996} characterizes uniform Donskerness and pre-Gaussianity through uniform total boundedness of the index set and uniform asymptotic equicontinuity. The limits are therefore uniformly equicontinuous with respect to the $L_2(P)$ seminorms, under which the index set is totally bounded uniformly in $P$. An Arzel\`{a}--Ascoli argument then yields uniformly bounded covering numbers for the sets $S_{n,R,P}$ when, with probability at least $1-\ve$, the processes take values in a fixed totally bounded subset of $\mathfrak{X}$, as they do in finite dimensions under Assumption~\ref{A:uniformbound}(4).
}\fi

\begin{asm}\label{A:uniformbound}
\begin{enumerate}
    \item $\Theta$ is a subset of a seminormed space $\mathfrak{B}$, and there is some $\delta > 0$ such that, for all $(R,P) \in \mathcal{S}$ and $\theta \in \Theta_{\mathrm{I}}(R,P)$, $B_R(\theta, \delta)$ is convex.

    \item There is some $\delta > 0$ and some bounded, convex, open region on which $\xi: \mathfrak{X} \ra \R \cup \{\infty\}$ is convex, finite, and Lipschitz, and which contains $B_{\mathfrak{X}}(0, \delta) + \bigcup_{(R,P) \in \mathcal{S}} m_{R,P}(\Theta)$.
    \item The functions $\{m_{R,P}: (R,P) \in \mathcal{S}\}$ are Fr\'{e}chet differentiable over the sets $R$ with Taylor remainder of order $\alpha_m > 1$ in that
    \begin{align*}
        \sup_{(R,P) \in \mathcal{S}} \sup_{\substack{\theta \in R,\, \theta' \in B_R(\theta,\delta)}} \norm{m_{R,P}(\theta') - m_{R,P}(\theta) - m'_{R,P}(\theta) (\theta' - \theta)}_{\mathfrak{X}} = O(\delta^{\alpha_m}).
    \end{align*}

    \item The processes $\mathbb{W}_{n,R,P}(\theta) = r_n (m_{n,R,P}(\theta) - m_{R,P}(\theta))$ converge uniformly in $\ell^\infty(\Theta, \mathfrak{X})$ to limit processes $\{\mathbb{X}_{n,R,P}(\theta): n \in \N, (R,P) \in \mathcal{S}\}$ that are \textit{quantitatively pre-tight}, uniformly bounded in probability, and uniformly equicontinuous in probability such that:
    \begin{align*}
        \lim_{\delta \downarrow 0}\sup_{n \in \N}\sup_{(R,P) \in \mathcal{S}}  \PP{P}{\sup_{\substack{\theta \in \Theta,\, \theta' \in B_\Theta(\theta,\delta)}} \norm{\mathbb{X}_{n,R,P}(\theta) - \mathbb{X}_{n,R,P}(\theta')}_{\mathfrak{X}} > \ve} = 0 \text{ for all } \ve > 0.
    \end{align*}

    \item There exist bootstrap analogues $\mathbb{W}_{n,R,P}^*$ of $\mathbb{W}_{n,R,P}$ satisfying $\mathbb{W}_{n,R,P}^* \overset{\mathrm{P}}{\rs} \mathbb{X}_{n,R,P}$ in $\ell^\infty(\Theta, \mathfrak{X})$ uniformly in $(R,P) \in \mathcal{S}$.
    \item The functions $\{m_{R,P}: (R,P) \in \mathcal{S}\}$ are uniformly Lipschitz with constant $C_m$.
\end{enumerate}
\end{asm}

Assumption~\ref{A:uniformbound} is the uniform counterpart of Assumptions~\ref{A:ex1} and~\ref{A:ex2}. Its part~(6), the Lipschitz condition of Assumption~\ref{A:ex2}(2), is used only by Proposition~\ref{Add:unifxi} below and by case~(1) of Corollary~\ref{C:unifstud}\suppref{}.  Assumption~\ref{A:ex22} needs no counterpart, because Proposition~\ref{P:unifbound} establishes an upper bound only.

\begin{prop}[Uniform version of Propositions~\ref{P:ex1} and~\ref{P:ex2}] \label{P:unifbound}
    Let Assumption \ref{A:uniformbound}(1)--(5) hold, and define $\overline{f}_\ell$ by \eqref{E:fello} for $\ell_{R,P}(\theta) = \xi(m_{R,P}(\theta)) - T_0(R,P)$.
    Let $s_n\ra\infty$ with $s_n=o(r_n)$, and suppose that either of the following also holds:
    \begin{itemize}
        \item The linearity condition $\alpha_m = \infty$ holds and $\overline{f}_\ell(\delta) = o(1)$.
        \item $\alpha_m < \infty$, and there is an $\alpha_\ell$ such that $\overline{f}_\ell(\delta) = O(\delta^{\alpha_\ell})$ and $\alpha_\ell\alpha_m > 1$.
    \end{itemize}
    Then,
    \small
    \begin{align}
        &r_n( \inf_{\theta \in R} \xi(m_{n,R,P}(\theta)) - T_0 (R,P) ) & \nonumber \\
        & \quad\le
        s_n\Big(\inf_{\theta \in R} \xi \big( m_{n,R,P}(\theta) + s_n^{-1} \mathbb{W}_{n,R,P}(\theta) \big) - T_0 (R,P)\Big)  + o_p(1)\label{E:unifub}
    \end{align}
    \normalsize
    uniformly in $(R,P) \in \mathcal{S}$. Moreover, the bootstrap estimator of the law of the right-hand side of \eqref{E:unifub} which replaces $\mathbb{W}_{n,R,P}$ with $\mathbb{W}_{n,R,P}^*$ is uniformly consistent in the bounded-Lipschitz metric.
\end{prop}

Lemma~\ref{L:validity}\suppref{} then gives size control over the configurations in $\mathcal{S}$ with $T_0(R,P)=0$.  The centering $T_0(R,P)$ may again be replaced by the plug-in $\inf_{\theta\in R}\xi(m_{n,R,P}(\theta))$, which changes the draw by at most $O_p(s_n/r_n)=o_p(1)$ uniformly over $\mathcal{S}$.  Under the feasibility null, the centering may also be set to $0$.  The studentized bootstrap \eqref{E:studentized_bootstrap}, computed over $R$, inherits these guarantees under uniform versions of the conditions of Section~\ref{S:studentized}, for both general and invariant weights (Corollary~\ref{C:unifstud}\suppref{}).

Proposition~\ref{Add:unifxi} covers the second tier case in which the polynomial minorant is weakened to $\overline f_\ell(\delta)=o(1)$.  Its penalty sequence $\eta_n\to\infty$ is tuned with a known bound on $\overline f_\ell$ near zero and with $\overline{C}_v$, the uniform Lipschitz constant of the maps $\theta\mapsto t(m_{R,P}(\theta))$ over the supporting affine functions $t$ of $\xi$.  The latter is at most $L_\xi C_m$, where $L_\xi$ is the Lipschitz constant of $\xi$ on the region of Assumption~\ref{A:uniformbound}(2) and $C_m$ is that of Assumption~\ref{A:uniformbound}(6).

\begin{prop}[Convex-loss uniform bound under weaker identification] \label{Add:unifxi}
    Let Assumption \ref{A:uniformbound} hold, and suppose that $\alpha_m<\infty$ and $\overline{f}_\ell(\delta) = o(1)$. Let $(\lambda_n)$, $(\eta_n)$, and $(b_n)$ be sequences of tuning parameters chosen such that $\lambda_n,\eta_n \ra \infty$, $b_n \ra 0$, $\lambda_n = o (b_n^{1-\alpha_m})$, $\lambda_n b_n^{-1} = o (\eta_n)$, $\eta_n = o(r_n)$, and $\lambda_n b_n^{-1} \overline{f}_\ell( c \eta_n^{-1} \lambda_n)^{\alpha_m} = o(1)$ for some $c > \overline{C}_v$. Then, uniformly in $(R,P) \in \mathcal{S}$,
    \begin{align}
        &r_n( \inf_{\theta \in R} \xi(m_{n,R,P}(\theta)) - T_0 (R,P) ) & \nonumber \\
        & \quad\le
        \inf_{\theta \in R}\inf_{\theta' \in B_R(\theta,b_n)} \Big( \eta_n \big(\xi( m_{n,R,P}(\theta)) - T_0(R,P)\big) \nonumber\\
        & \hspace{4em}+ \lambda_n b_n^{-1}\big(\xi ( m_{n,R,P}(\theta') + \lambda_n^{-1} b_n \mathbb{W}_{n,R,P}(\theta)) - T_0(R,P)\big)\Big)  + o_p(1) \label{E:unfbd2}
    \end{align}
    Moreover, the bootstrap estimator of the law of the right-hand side of \eqref{E:unfbd2} which replaces $\mathbb{W}_{n,R,P}$ with $\mathbb{W}_{n,R,P}^*$ is uniformly consistent in the bounded-Lipschitz metric.
\end{prop}

\ifmainbuild\else
\begin{rem}[Recentered statistic]\label{R:recentered}
For subvector inference on $g(\theta_0)$, a practical variant is the \emph{recentered} statistic $T_{n,R,P} - T_{n,\Theta,P}$.  Inverting the recentered tests always yields a nonempty confidence set, as noted by \citet{CHT2007} and \citet[online Appendix~A.4]{HongLi2018}.
Under the equality case of Theorem~\ref{T:one} applied separately to $(R,P)$ and to $(\Theta,P)$, the recentered statistic equals
\begin{equation*}
  \inf_{\theta\in\Theta_{\mathrm{I}}(R,P)}\sup_{t\in\mathcal{K}(\theta)}\mathbb{V}_{n,R,P}(\theta,t)
  - \inf_{\theta\in\Theta_{\mathrm{I}}(\Theta,P)}\sup_{t\in\mathcal{K}(\theta)}\mathbb{V}_{n,\Theta,P}(\theta,t)
  + o_p(1).
\end{equation*}
Applied in the same way, the first line of~\eqref{E:bootstrapone} in Theorem~\ref{T:two} shows that the difference of the corresponding bootstrap statistics, computed from the same bootstrap draws, consistently estimates the law of this difference.

When $P$ is fixed and the equality case of Theorem~\ref{T:one} holds at $(\Theta,P)$, the recentered statistic also has an upper bound that holds uniformly over the tested values $c$.  Theorem~\ref{P:unif1} bounds $T_{n,R_c,P}$ uniformly in $c$, and the subtracted term $T_{n,\Theta,P}$ does not depend on $c$.  If the bootstrap statistics of Theorems~\ref{P:unif1} and~\ref{T:two} use the same bootstrap draws, their difference is therefore a uniform upper-bound estimator over $\{(R_c,P)\}$.  Uniformity over a class of distributions would instead require the lower bound on $T_{n,\Theta,P}$ to hold uniformly in $P$.
\end{rem}

\fi

\section{Monte Carlo study}\label{S:mc}

Our simulation evidence substantiates two of the main claims of the paper.  First, it shows that
Proposition~\ref{P:ex1} accurately characterizes the limiting distribution of a minimax test statistic in a
model that is nonlinear in $\theta$, where it yields limits of the type studied by \citet{CS2023} and
\citet{CSS2025} for models linear in a nuisance parameter.  Second, it shows that the
penalize-resampling component of Test~MR of \citet{BCS2017}, which is the re-optimization bootstrap,
performs on par with the full test when its step is intermediate.  Sections~\ref{S:boot}
and~\ref{S:uniform} also establish its pointwise exactness and uniform size control in the general settings
of the paper.  A multiscale hybrid that reduces its dependence on the step inherits both properties, and the
simulations bear them out in our design.

\subsection{Design and procedures}

We use a Gaussian design that is low dimensional but nonlinear in $\theta$.  Let
$\theta=(\theta_1,\theta_2)\in[0,3]^2$ and let $X_i\sim N(\mu,I_6)$ independently, with moment map
\[
 m(\theta)=\bigl(\mu_1-\theta_1,\ \theta_1-\mu_2,\ \mu_3-\theta_2,\ \mu_4-\theta_2,\
 \theta_1^2-\mu_5,\ \theta_2^2-\mu_6\bigr)',
\]
whose first, second and fifth coordinates are inequalities and whose remaining coordinates are equalities.  The
statistic measures the Euclidean distance of the sample moments to the constraint set
$C=\{x\in\R^6:x_1,x_2,x_5\le0,\ x_3=x_4=x_6=0\}$, so that
\begin{equation}\label{E:mc_stat}
 \xi(x)=\bigl\|(x_1^+,x_2^+,x_3,x_4,x_5^+,x_6)\bigr\|_2,
 \qquad
 T_n=\sqrt n\inf_{\theta\in[0,3]^2}\xi(m_n(\theta)).
\end{equation}

We consider four null designs, which differ in the local geometry of the identified set.  In the
\emph{baseline null}, $\mu=\mu^0=(0.4,1.2,0.8,0.8,0.5,0.8^2)$, so the identified set is the interval
$\theta_1\in[0.4,\sqrt{0.5}]$ at $\theta_2=0.8$, and the inequalities are slack at its interior points.  A
\emph{nonlinear-boundary null} replaces $(\mu_2,\mu_5)$ by $(0.4,0.4^2)$, so that all three inequalities bind at
$\theta_1=0.4$.  A \emph{quadratic-boundary null} replaces only $\mu_5$ by $0.4^2$, so that the lower linear
bound and the quadratic upper bound meet at $\theta_1=0.4$.  Finally, a \emph{shrinking-slack null} sets
$(\mu_1,\mu_2)=(0.4-1/\sqrt n,0.4+1/\sqrt n)$, so that the linear bounds close in on each other as $n$ grows.
We also consider two families of local alternatives, indexed by $h\in\{0.5,1,1.5,2,3\}$, that move along the
nonlinear directions of the model.  \emph{Curved-equality} alternatives set $\mu_{3,n}=\mu_{4,n}=0.8+h/\sqrt n$,
so that the linear equalities contradict the quadratic one, and \emph{curved-inequality} alternatives set
$\mu_{5,n}=0.4^2-h/\sqrt n$, so that the lower linear bound exceeds the quadratic upper bound.

In every null design the criterion has the linear minorant
$\xi(m(\theta))\geq0.8\,d(\theta,\Theta_{\mathrm I})$ (\Secapp~\ref{S:supp_mc_subvector}\suppref{}), so the
identification rate is $\alpha_\ell=1$, while the moments are smooth of order $\alpha_m=2$.  We also use a subvector
design, which adds a parameter $\theta_3\geq\theta_1$ that bounds $\theta_1$ from above and tests $H_0:\theta_3=c$ for
values of $c$ near $0.4$, where the restricted identified set shrinks to a point.  The same minorant holds uniformly over its null
values $c\geq0.4$, so Proposition~\ref{P:unifbound} continues to apply.  We report its main results below, with additional
details deferred to \Secapp~\ref{S:supp_mc_subvector}\suppref{}.

Because our design is finite dimensional, the re-optimization bootstrap can be compared directly with the tests of
\citet{BCS2017}.  In this design $\xi$ is the Euclidean distance to the constraint set, so Corollary~\ref{C:geometry} gives the limit of $T_n$ in closed form.  Let
$Z\sim N(0,I_6)$, so that $\sqrt n(m_n-m)$ converges to $\mathbb X=(Z_1,-Z_2,Z_3,Z_4,-Z_5,-Z_6)$ at every $\theta$.  At the
baseline null only the three equalities bind, and a local movement of $\theta_2$ absorbs one of their directions,
so that
\begin{equation}\label{E:mc_limit}
 T_n\rs\inf_{\theta\in\Theta_{\mathrm I}}d\bigl(\mathbb X,\,T_{m(\theta)}C-m'(\theta)(T_\theta\Theta)\bigr)
 =\min_{g\in\R}\bigl\|(Z_3-g,\ Z_4-g,\ 1.6g-Z_6)\bigr\|,
\end{equation}
whose square is $\chi^2_2$.  At the nonlinear-boundary null the inequalities bind as well, and the square of the
limit acquires the independent term $\min_{u\in\R}\|((Z_1-u)^+,(u-Z_2)^+,(0.8u-Z_5)^+)\|^2$.  Neglecting the
local movement of $\theta$ amounts to setting $u=g=0$ and gives a conservative bound, which we call the
\emph{envelope bound}.  Its square is $\chi^2_3$ at the baseline null.  \Secapp~\ref{S:supp_losses}\suppref{} repeats the fixed-null experiments under the $\ell_1$ and $\ell_\infty$
distances to $C$, whose limits the corollary also gives.

We implement the re-optimization bootstrap of Proposition~\ref{P:ex2} with $T_0=0$ (the centering
recommended after that proposition, which is exact under the feasibility null).  Given $W^*\sim N(0,I_6)$ and $\widetilde W^*=(W_1^*,-W_2^*,W_3^*,W_4^*,-W_5^*,-W_6^*)'$, its draw and
critical value are
\[
 B_{n}^*
 =s_n\inf_{\theta\in[0,3]^2}\xi\bigl(m_n(\theta)+s_n^{-1}\widetilde W^*\bigr),
 \qquad
 c_n^{\mathrm{RO}}=q_{.95}\bigl(B_n^*\mid X_1,\ldots,X_n\bigr),
\]
and the step takes the form
\[
 s_n=\lambda_0\sqrt{\ln n}\,n^{1/4},
 \qquad
 \lambda_0\in\Lambda\equiv\{.01,.02,.05,.1,.2,.5,1\}.
\]
Every fixed $\lambda_0$ yields $1\ll s_n\ll\sqrt n$, as Propositions~\ref{P:ex2} and~\ref{P:unifbound}
require, and the rate $n^{1/4}$ lies halfway between these limits on a logarithmic scale.  To reduce dependence on the choice of $\lambda_0$, we also propose a
multiscale hybrid, which takes the smallest critical value over the grid, $c_n^{\mathrm{MS}}=\min_{\lambda_0\in\Lambda}
c_{n,\lambda_0}^{\mathrm{RO}}$.  Because each critical value in the grid is asymptotically at least the
$(1-\alpha)$-quantile of $T_n$, up to the slack of~\eqref{E:critval}, so is their minimum, and the hybrid is
valid without any size adjustment (Corollary~\ref{C:minvalidity}\suppref{}).  When the limit distribution is
continuous and strictly increasing at that quantile, every critical value in the grid converges to it by
Proposition~\ref{P:ex2}, so the hybrid is also exact up to the slack.  Throughout, we set
the slack $\varepsilon$ of~\eqref{E:critval} to $10^{-4}$, which is below $0.01\%$ of the mean critical value
in every design and changes about $0.001\%$ of the simulated decisions, all toward acceptance.

We compare these procedures with the three tests of \citet{BCS2017}, using their recommended
$\kappa_n=\sqrt{\ln n}$.  Their discard-resampling test (DR) applies hard generalized moment selection to the
inequalities over an estimated set of near minimizers and keeps every equality, without re-solving for $\theta$, and it
coincides with Test~RS of \citet{BCS2015}.  Their penalize-resampling test (PR) uses the re-optimization draw with
$s_n=\sqrt n/\kappa_n$ (Remark~\ref{R:relation}), which corresponds to $\lambda_0\approx0.76$--$1.09$ here, a step
large enough to make the bootstrap conservative in our designs (Figure~\ref{F:mc_main}(d)).  Their Test~MR takes
the draw-wise minimum of the two.  By Proposition~\ref{P:ex2}, PR is pointwise
exact. Because MR produces critical values that are at most those of PR and is also valid \citep{BCS2017}, it inherits this exactness.

Each design uses $2{,}500$ replications at $n\in\{500,1{,}000,5{,}000,10{,}000\}$, with $599$ bootstrap draws,
and all procedures share the samples and the Gaussian perturbations within a replication.  At the baseline and nonlinear-boundary nulls we also compare each conditional bootstrap distribution with $15{,}000$ draws of $T_n$ through their Wasserstein-1
($W_1$) distance.  One run of the simulation code takes about twelve minutes on $32$ cores, and
\Secapp~\ref{S:supp_tables}\suppref{} reports power over the full grid of alternatives, an estimated-covariance
design, further subvector results, alternative centerings, and two other losses.

\subsection{Results}

\begin{table}[!htbp]
\centering
\caption{Null rejection probabilities at nominal level $0.05$.  DR, PR and MR are the
discard-resampling, penalize-resampling and minimum-resampling tests of \citet{BCS2017},
with $\kappa_n=\sqrt{\ln n}$.  DR coincides with Test~RS of \citet{BCS2015}.  The re-optimization
bootstrap uses $\lambda_0=.2$, and the hybrid takes the smallest re-optimization critical value
over $\Lambda$.  The subvector rows report the largest rejection probability over the $161$ fixed and $21$ drifting
null values of \Secapp~\ref{S:supp_mc_subvector}\suppref{}.  Entries use $2{,}500$ replications
($10{,}000$ in the subvector rows), with Monte Carlo standard errors $0.0044$
($0.0022$) at a rejection probability of $0.05$.}
\label{T:mc_main}
\begin{tabular}{lrrrrrr}
\hline
Design & $n$ & DR & PR & MR & Re-opt. & Hybrid \\
\hline
Baseline & 500 & 0.020 & 0.009 & 0.044 & 0.036 & 0.036 \\
 & 1{,}000 & 0.022 & 0.012 & 0.051 & 0.043 & 0.043 \\
 & 5{,}000 & 0.021 & 0.014 & 0.049 & 0.046 & 0.046 \\
 & 10{,}000 & 0.016 & 0.013 & 0.046 & 0.045 & 0.046 \\
Nonlinear boundary & 500 & 0.017 & 0.009 & 0.038 & 0.046 & 0.050 \\
 & 1{,}000 & 0.014 & 0.005 & 0.032 & 0.039 & 0.041 \\
 & 5{,}000 & 0.021 & 0.016 & 0.042 & 0.048 & 0.051 \\
 & 10{,}000 & 0.013 & 0.014 & 0.034 & 0.044 & 0.046 \\
Quadratic boundary & 500 & 0.017 & 0.006 & 0.035 & 0.038 & 0.040 \\
 & 1{,}000 & 0.013 & 0.008 & 0.031 & 0.037 & 0.038 \\
 & 5{,}000 & 0.020 & 0.018 & 0.048 & 0.056 & 0.057 \\
 & 10{,}000 & 0.017 & 0.017 & 0.038 & 0.044 & 0.046 \\
Shrinking slack & 500 & 0.013 & 0.003 & 0.032 & 0.028 & 0.028 \\
 & 1{,}000 & 0.011 & 0.007 & 0.035 & 0.028 & 0.029 \\
 & 5{,}000 & 0.006 & 0.006 & 0.029 & 0.025 & 0.026 \\
 & 10{,}000 & 0.011 & 0.014 & 0.036 & 0.038 & 0.038 \\
Subvector, largest & 500 & 0.023 & 0.007 & 0.049 & 0.041 & 0.042 \\
 & 1{,}000 & 0.022 & 0.007 & 0.048 & 0.042 & 0.043 \\
 & 5{,}000 & 0.023 & 0.014 & 0.051 & 0.048 & 0.049 \\
 & 10{,}000 & 0.022 & 0.015 & 0.052 & 0.051 & 0.051 \\
\hline
\end{tabular}
\end{table}

\begin{figure}[!tb]
\centering
\includegraphics[width=\textwidth]{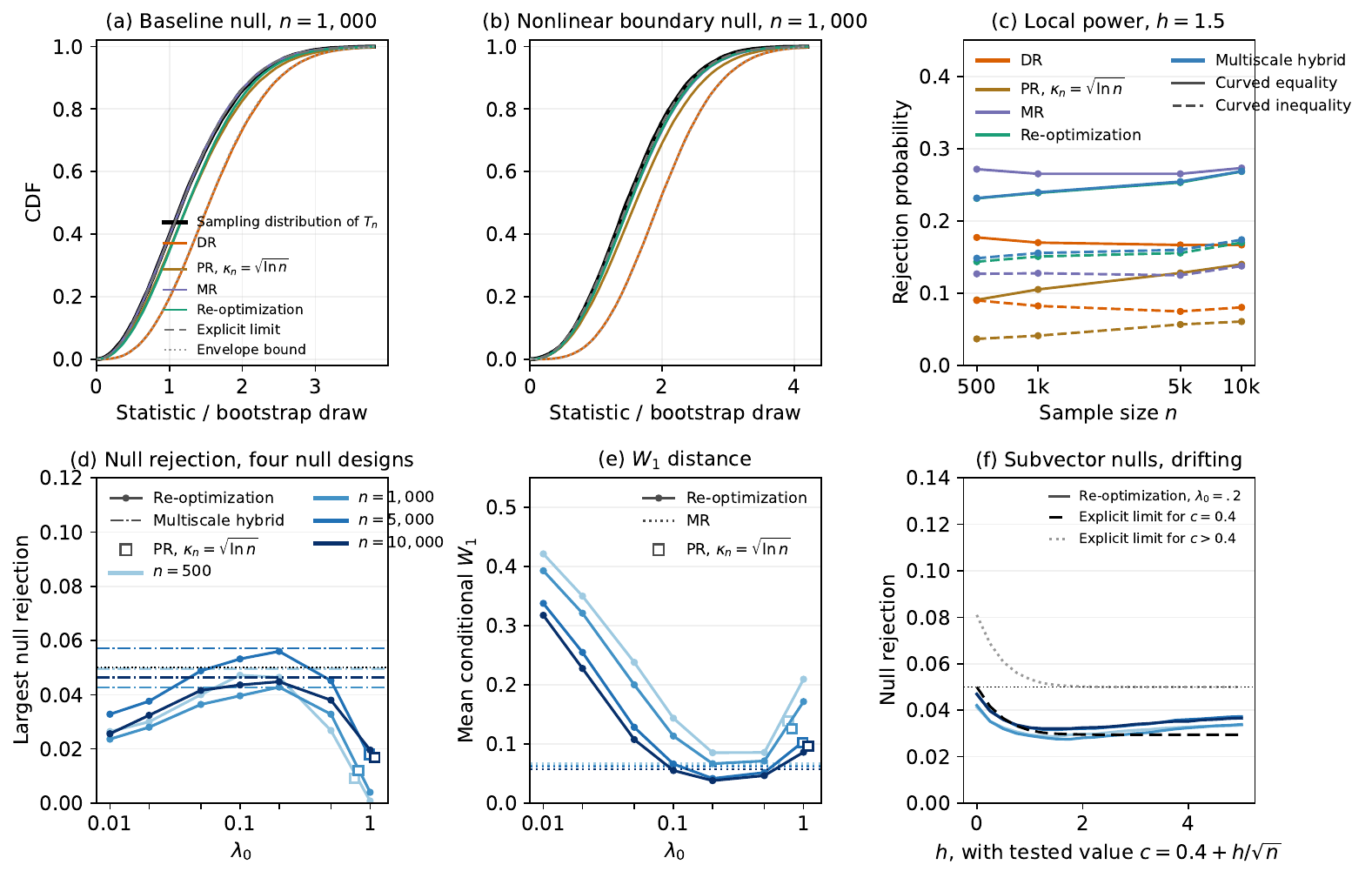}
\caption{Monte Carlo summary.  (a), (b) Sampling distribution of $T_n$, its explicit limit, the envelope bound, and
conditional bootstrap CDFs averaged over $80$ samples, at $n=1{,}000$.  (c) Rejection probabilities at $h=1.5$
against the curved-equality (solid) and curved-inequality (dashed) alternatives.  (d) Largest null rejection
probability over the four null designs against $\lambda_0$, with the hybrid (horizontal lines) and PR
(squares), by sample size.  (e) Mean conditional $W_1$ distance to the sampling distribution, averaged over the baseline and nonlinear-boundary nulls,
with MR (horizontal lines) and PR (squares).  (f) Null rejection of the re-optimization bootstrap at
$\lambda_0=.2$ along $c=0.4+h/\sqrt n$ in the subvector design, by sample size, with the rejection probabilities predicted by critical values from the explicit limits for $c=0.4$ and for $c>0.4$.  Dotted lines mark the nominal level.}
\label{F:mc_main}
\end{figure}

As Corollary~\ref{C:geometry} predicts, the sampling distribution of $T_n$ lies within $W_1$ distance $0.01$ of the explicit limit at both nulls and
every $n$ (Figure~\ref{F:mc_main}(a)--(b)).  At $n=1{,}000$ its $0.95$-quantiles are $2.48$ and $2.78$, and those
of the limit are $2.45$ and $2.78$.  \Secapp~\ref{S:supp_losses}\suppref{} shows that the same agreement holds
under the $\ell_1$ and $\ell_\infty$ losses.  The DR bootstrap distribution, by contrast, is shifted to the right.  This is because DR retains every equality direction, including one that the limit~\eqref{E:mc_limit} absorbs by moving $\theta_2$.  The re-optimization bootstrap distribution approaches the sampling distribution as $n$ grows.  At $\lambda_0=.2$ its $W_1$
distance falls from $0.106$ to $0.034$ at the baseline null and from $0.065$ to $0.042$ at the nonlinear boundary.  Across values of $\lambda_0$ the distance is U-shaped (Figure~\ref{F:mc_main}(e)), because a small step does not properly localize the bootstrap toward the identified set, whereas a large step amplifies its sampling noise.

The re-optimization bootstrap and the hybrid have null rejection probabilities close to those of Test~MR (Table~\ref{T:mc_main}).  Across
the four null designs they reject with probability $0.025$--$0.056$ and $0.026$--$0.057$, and MR with
$0.029$--$0.051$, while DR and PR are conservative on their own.  With an intermediate step, a single re-optimization draw thus
performs on par with MR, without the discard-resampling draw or its estimated set of near
minimizers.  Null rejection is flat over intermediate values of $\lambda_0$ and falls toward both ends of the grid, and the upper end is close to PR's step $s_n=\sqrt{n/\ln n}$ (Figure~\ref{F:mc_main}(d)).  The best value of $\lambda_0$ depends on the loss (\Secapp~\ref{S:supp_losses}\suppref{}), and the hybrid insures against a poor choice while
raising null rejection above that at $\lambda_0=.2$ by at most $0.004$ in every null design.

Size control persists as the restricted identified set shrinks.  Over the $161$ fixed and $21$ drifting null values of the subvector design, the
largest rejection probability is $0.041$--$0.051$ for the re-optimization bootstrap and $0.042$--$0.051$ for the
hybrid (Table~\ref{T:mc_main}).  Along the drifting null values $c=0.4+h/\sqrt n$, rejection is close to the
nominal level at $c=0.4$ and slightly below it at nearby values (Figure~\ref{F:mc_main}(f)).  At $c=0.4$ the
restricted identified set collapses to the single point $\theta_1=0.4$, where the lower bound on $\theta_1$ binds,
and the limit of $T_n$ jumps upward.  The bootstrap critical values at nearby null values are drawn up toward the
critical value at $c=0.4$, which makes the test conservative in a neighborhood of $c=0.4$ that shrinks as $n$
grows.  This conservatism is the price of uniform validity, because critical values that did not rise near
$c=0.4$ would over-reject there, as the curve for $c>0.4$ in Figure~\ref{F:mc_main}(f) shows.
\Secapp~\ref{S:supp_mc_subvector}\suppref{} reports the fixed null values and power against $c<0.4$.

All procedures share the test statistic $T_n$, so that their differences in power are differences in critical values
(Figure~\ref{F:mc_main}(c)).  MR is more powerful against the curved-equality alternatives and the hybrid against
the curved-inequality ones, and both reject more often than DR and PR.

\FloatBarrier

\section{Conclusion}\label{S:conclusion}

This paper characterizes the limiting distribution of the minimax statistic
$r_n\inf_{\theta\in\Theta}\sup_{t\in\mathcal{T}}v_n(\theta,t)$ under a composite
feasibility null. The limit is described by a contact set at each identified
parameter value, the loss directions that are locally binding and cannot be
absorbed by a local movement of $\theta$. Theorem~\ref{T:one} gives this
description as an upper bound in general, and as the exact limit under convexity in $\theta$, a local
identification rate, or continuity of the contact set.  For convex losses, Proposition~\ref{P:ex1} gives the closed
form~\eqref{E:limit_intro}, in which the subdifferential of the loss is restricted to functionals that are
nonnegative on the locally reachable moment directions.

The re-optimization bootstrap estimates the limit consistently by re-solving the minimum-violation
problem after a scaled perturbation, without estimating the contact set.  For the minimum-violation
functional it is the numerical delta method of \citet{HongLi2018}, and for finitely many moment
equalities and inequalities it corresponds to the penalize-resampling statistic of \citet{BCS2017}.
The results thus extend pointwise exactness to penalize resampling and uniform validity to the
numerical delta method, for general convex losses on Banach-valued moments.  We establish uniform validity under a uniform minorant condition.  Without it, pointwise exactness and uniform validity can be
incompatible (\Secapp~\ref{S:uniform_separation}\suppref{}).  In the simulations of
Section~\ref{S:mc}, the bootstrap tracks the sampling distribution closely at intermediate steps and performs about as well as Test~MR without its discard-resampling component.  A multiscale
hybrid over a grid of steps reduces sensitivity to tuning.

Two extensions are natural. The first concerns the index class $\mathcal{T}$, which the analysis
holds fixed. Allowing it to grow, as in the sieve approximations of fully infinite-dimensional
problems or the many inequalities of \citet{BBC2019}, would tie these literatures together more
closely than Remark~\ref{R:sieve} does. In such settings, the slack of Lemma~\ref{L:validity}\suppref{} could
replace the usual anti-concentration step. The convex-loss
results also cover a positive minimal violation (\ifmainbuild\else Section~\ref{S:convexloss}\fi),
complementing the inference on the misspecification index of \citet{AK2024}. The second
extension is to apply the re-optimization bootstrap to the least-violating parameters themselves.

\pagebreak

\appendix

\ifshowproofs
\section{Proofs}\label{S:proofs}

\subsection*{Additional notation}
The proofs use the following notation, which is not needed outside the proofs.

An \textit{affine map} over $\mathfrak{X}$ is any function $t: x \mapsto y(x) + a$, where $y \in \mathfrak{X}^*$ and $a \in \R$. The set of such functions is readily identified with the affine dual $(\mathfrak{X} \times \R)^* = \mathfrak{X}^* \times \R$ by the agreement $(x \mapsto y(x) + a) \mapsto (y,a)$. We assign the set of affine maps $\mathcal{A}$ over $\mathfrak{X}$ the norm and weak-* topologies that pertain to the dual of $\mathfrak{X} \times \R$. A subset $\mathcal{T}$ of $\mathcal{A}$ is defined to be bounded if it is bounded in norm. For an affine map $t \in \mathcal{A}$, we use $y[t] = t - t(0)$ to denote the linear part of $t$.

\subsection{Proof of Theorem \ref{T:one}}

All functionals to which we apply expectations and probabilities---in particular the $\inf$--$\sup$ statistic $T_n$ and the localized-bootstrap map---are $1$-Lipschitz on $\ell^\infty$ and therefore inherit asymptotic measurability from the underlying processes. They also preserve bounded-Lipschitz closeness, conditionally or not. If $f\in\mathrm{BL}_1$ and $g$ is $1$-Lipschitz, then $f\circ g\in\mathrm{BL}_1$, so $d_{\mathrm{BL}}(g(\mathbb{V}),g(\mathbb{G}))\le d_{\mathrm{BL}}(\mathbb{V},\mathbb{G})$ \citep[with $d_{\mathrm{BL}}$ as in][\S1.12]{VW1996}.

    Let $\Theta_{\mathrm{I}}$ be nonempty. By Assumption \ref{A:one}(1), BL convergence of $\mathbb{V}_n$ to limits with $\|\mathbb{G}_n\|_{\ell^\infty}=O_p(1)$ uniformly in $n$ implies $\|\mathbb{V}_n\|_{\ell^\infty}=O_p(1)$, and in particular
    \begin{align} \label{E:vbd}
        \sup_{\theta \in \Theta} \sup_{t \in \mathcal{T}} |v_n(\theta, t) - v(\theta, t) | = O_p(r_n^{-1}).
    \end{align}
    Therefore,
    \begin{align} \label{E:ellbd}
        \sup_{\theta \in \Theta} | \ell_n(\theta) - \ell(\theta)| = O_p(r_n^{-1}).
    \end{align}

    By Assumption \ref{A:two}(1), we may pick a sequence $(b_n)$ satisfying $b_n=o(1)$, $r_nb_n\to\infty$, and $r_nb_n^{\alpha_v}\to0$.
    Because $\Theta_{\mathrm{I}}$ is nonempty, $\inf_{\theta \in \Theta} \ell_n(\theta) = O_p(r_n^{-1})$. Assumption \ref{A:zero}(2) and \eqref{E:ellbd} imply that, for all $\delta > 0$, there is some $\eta(\delta) > 0 $ such that
    \begin{align*}
        \inf_{\theta: d_{\mathfrak{B}}(\theta, \Theta_{\mathrm{I}}) > \delta} \ell_n(\theta) & \ge \eta(\delta) - O_p(r_n^{-1}).
    \end{align*}
Hence, for all $\delta > 0$, $\inf_{\theta \in \Theta} \ell_n(\theta)$ is equal to $\inf_{\theta: d_{\mathfrak{B}}(\theta, \Theta_{\mathrm{I}}) \le \delta} \ell_n(\theta)$ with probability approaching $1$. By a diagonal argument, we may choose a sequence $(a_n)$ bounded below by $(b_n)$ and satisfying $a_n\to0$ such that
    \begin{align}
        r_n \inf_{\theta \in \Theta} \sup_{t \in \mathcal{T}} v_n(\theta, t) & = \inf_{\theta \in \Theta_{\mathrm{I}}} \inf_{\theta' \in B_\Theta(\theta, a_n)} \sup_{t \in \mathcal{T}} r_n v_n(\theta', t) + o_p(1) \nonumber \\
        & = \inf_{\theta \in \Theta_{\mathrm{I}}}\inf_{\theta' \in B_\Theta(\theta, a_n)} \sup_{t \in \mathcal{T}} r_n(v_n(\theta, t) + v(\theta', t) - v(\theta, t)) + o_p(1), \label{E:ll}
    \end{align}
    where the second line follows from Assumption \ref{A:one} and the resulting asymptotic $\rho$-equicontinuity of $\mathbb{V}_n$, together with Assumption \ref{A:three} (uniform continuity of $\rho$ with respect to $d_{\mathfrak{B}}$ for fixed $t$). Then, because $B_\Theta(\theta, b_n) \subseteq B_\Theta(\theta, a_n)$, \eqref{E:ll} can be bounded above, up to the $o_p(1)$ term, by
\begin{align}
&\inf_{\theta \in \Theta_{\mathrm{I}}} \inf_{\theta' \in B_\Theta(\theta, b_n)} \sup_{t \in \mathcal{T}} ( r_n v_n(\theta, t) + r_n D_{\theta; t}v (\theta' - \theta)) + r_n O(b_n^{\alpha_v}) \nonumber\\
& \quad = \inf_{\theta \in \Theta_{\mathrm{I}}}  \sup_{t \in \mathcal{T}}( r_n v_n(\theta, t) + \inf_{\theta' \in B_\Theta(\theta, b_n)} r_n D_{\theta; t}v (\theta' - \theta)) + o(1) \nonumber  \\
& \quad = \inf_{\theta \in \Theta_{\mathrm{I}}}  \sup_{t \in \mathcal{T}}( r_n v_n(\theta, t) + r_n b_n \psi_{b_n} (\theta, t)) + o(1).  \label{E:lastline}
\end{align}
Because $\Theta_{\mathrm{I}}$ is compact, $b_n$ is eventually small enough that $B_\Theta(\theta,b_n)$ is convex for every $\theta \in \Theta_{\mathrm{I}}$, by Assumption~\ref{A:zero}(1). The second line therefore follows by applying Sion's minimax theorem \citep[Corollary~3.3]{Sion1958} pathwise to $(\theta',t)\mapsto r_n v_n(\theta,t)+r_nD_{\theta;t}v(\theta'-\theta)$ on $B_\Theta(\theta,b_n)\times\mathcal T$: Assumption~\ref{A:two}(1) gives lower semicontinuity and convexity in $\theta'$, Assumption~\ref{A:two}(2) gives concavity in $t$, and Assumption~\ref{A:three} gives compactness of $\mathcal T$ and upper semicontinuity in $t$.
The error term which arises is uniform in $\theta, \theta'$, and $t$. Lemma \ref{L:gn}\suppref{} implies that \eqref{E:lastline} may be rewritten as $\inf_{\theta \in \Theta_{\mathrm{I}}} \sup_{t \in \mathcal{K}(\theta)} \mathbb{V}_n(\theta, t)$ up to a uniformly negligible error, which establishes the first line of \eqref{E:thm}. The second line of \eqref{E:thm} then follows from Lipschitz continuity of the map $g\mapsto\inf_{\theta\in\Theta_{\mathrm{I}}}\sup_{t\in\mathcal{K}(\theta)}g(\theta,t)$ and the bounded-Lipschitz convergence of $\mathbb{V}_n$ to $\mathbb{G}_n$.

We now take Assumption \ref{A:four} into account and prove the last claim of the theorem. Note that, in the process of establishing the first line of \eqref{E:thm}, an inequality (instead of an equality, up to $o_p(1)$ term) appears only in bounding \eqref{E:ll} by \eqref{E:lastline}.

Assume that the first case of Assumption \ref{A:four} holds. Then, Lemma \ref{L:gn2}\suppref{} applies, and we may choose the sequence $b_n$ in \eqref{E:lastline} such that it is equal to \eqref{E:ll} up to an $o_p(1)$ term. Suppose next that the second or third case of Assumption \ref{A:four} holds. We first show that the second case implies the third. By Assumptions \ref{A:zero}(1), \ref{A:zero}(3), and \ref{A:four}(2), there is some $\delta_0 > 0$ such that, for all $\theta \in \Theta_{\mathrm{I}}$, $B_\Theta(\theta, \delta_0)$ is convex and $v(\cdot, t)$ is convex on $B_\Theta(\theta, \delta_0)$ for all $t \in \mathcal{K}(\theta)$. Let $\alpha \in (0,1)$ be arbitrary. Then, for such $\theta, t$, and for $\theta' \in B_\Theta(\theta, \delta_0)$, Assumption \ref{A:two} implies
\begin{align} \label{E:ps30}
v(\theta', t) - v(\theta, t) &\ge \alpha^{-1} (v(\alpha \theta' + (1-\alpha) \theta, t) - v(\theta, t))  =  D_{\theta; t}v(\theta' - \theta) + O(\alpha^{\alpha_v - 1}),
\end{align}
where the remainder $O(\alpha^{\alpha_v-1})$ is uniform in $\theta, t$, and $\theta' \in B_\Theta(\theta, \delta_0)$, and vanishes as $\alpha\downarrow0$ because $\alpha_v>1$. Because $t\in\mathcal K(\theta)$ implies $\gamma(\theta,t)=0$, local convexity and positive homogeneity imply that $D_{\theta;t}v(\theta'-\theta)\ge0$ in \eqref{E:ps30} (Lemma~\ref{L:gsc}). Letting $\alpha$ tend to $0$ therefore implies that $v(\theta', t) \ge v(\theta, t) \ge 0$. Hence $\mathcal{K}(\theta)\subseteq\widetilde{\mathcal{K}}(\theta')$ whenever $\theta\in\Theta_{\mathrm{I}}$ and $\theta'\in B_\Theta(\theta,\delta_0)\setminus\Theta_{\mathrm{I}}$, so $\widetilde{\mathcal{K}}$ is lower hemicontinuous at $\Theta_{\mathrm{I}}$.

Assume now that the third case holds. Given the result of Lemma \ref{L:gn}, it is sufficient to show that
\begin{align} \label{E:conv1}
T_n \ge \inf_{\theta \in \Theta_{\mathrm{I}}} \sup_{t \in \mathcal{K}(\theta)} \mathbb{V}_n(\theta, t) + o_p(1).
\end{align}
Let $\ve>0$. By Assumption \ref{A:one} and the resulting asymptotic $\rho$-equicontinuity of $\mathbb{V}_n$ (cf.\ \eqref{E:ll}), there is $\delta>0$ such that the event $\sup_{\rho((\theta,t),(\theta',t'))<\delta}|\mathbb{V}_n(\theta,t)-\mathbb{V}_n(\theta',t')|>\ve$ has outer probability at most $\ve+o(1)$. Let $S_\delta$ be as in Lemma \ref{L:contactlb}, applied with $\mathcal{K}'=\widetilde{\mathcal{K}}$, which equals $\mathcal{K}$ on $\Theta_{\mathrm{I}}$ and is therefore nonempty there. Because $\Theta_{\mathrm{I}}$ is compact, the $a_n$-neighborhood of $\Theta_{\mathrm{I}}$ in \eqref{E:ll} lies in $S_\delta$ for $n$ large. For $\theta'\in S_\delta$ and $t'\in\widetilde{\mathcal{K}}(\theta')$, $v(\theta',t')\ge0$, so $\sup_{t\in\mathcal{T}}r_nv_n(\theta',t)\ge\sup_{t'\in\widetilde{\mathcal{K}}(\theta')}\mathbb{V}_n(\theta',t')$. By the first line of \eqref{E:ll} and Lemma \ref{L:contactlb}, $T_n\ge\inf_{\theta\in\Theta_{\mathrm{I}}}\sup_{t\in\mathcal{K}(\theta)}\mathbb{V}_n(\theta,t)-\ve-o_p(1)$ with probability at least $1-\ve-o(1)$. As $\ve$ is arbitrary, \eqref{E:conv1} follows.

Finally, \eqref{E:vbd} implies $r_n v_n(\theta,t) \ge r_n v(\theta,t)-O_p(1)$ uniformly in $(\theta,t)$. Taking the supremum over $t$ and then the infimum over $\theta$ yields $T_n \ge r_n\inf_{\theta\in\Theta}\ell(\theta)-O_p(1)$.

\begin{lemma}  \label{L:gsc}
Under Assumptions \ref{A:zero}, \ref{A:two}, and \ref{A:three}, there is some $\underline{\delta} > 0$ such that, for every $\theta$ with $d_{\mathfrak{B}}(\theta,\Theta_{\mathrm{I}})\le\underline{\delta}$, the maps $\psi_\delta(\theta,\cdot)$ (for $\delta > 0$) and $\gamma(\theta, \cdot)$ are $\mathcal{U}$-upper semicontinuous on $\mathcal{T}$, and $\psi_\delta(\theta,t)$ decreases to $\gamma(\theta,t)$ as $\delta\downarrow0$ with $\delta\le\underline{\delta}$ for every $t\in\mathcal{T}$.
\end{lemma}
\begin{proof}
By the local convexity condition of Assumption \ref{A:zero}(1) and the compactness of $\Theta_{\mathrm{I}}$ in Assumption \ref{A:zero}(3), there is some $\underline{\delta} > 0$ such that $B_\Theta(\theta, \underline{\delta})$ is convex for every $\theta$ with $d_{\mathfrak{B}}(\theta, \Theta_{\mathrm{I}}) \le \underline{\delta}$. If $0<\delta_1\le\delta_2\le\underline{\delta}$, convexity implies
\[
    \frac{B_\Theta(\theta,\delta_2)-\theta}{\delta_2}
    \subset
    \frac{B_\Theta(\theta,\delta_1)-\theta}{\delta_1}.
\]
Positive homogeneity of $D_{\theta;t}v$ therefore gives $\psi_{\delta_1}(\theta,t)\le\psi_{\delta_2}(\theta,t)$, so $\psi_\delta(\theta,t)$ decreases as $\delta\downarrow0$ and $\gamma$ is the pointwise limit of the $\psi_\delta$.
As the maps
\begin{align*}
t \mapsto \frac{D_{\theta; t} v(\theta' - \theta)}{\delta}
\end{align*}
are upper semicontinuous, $\psi_\delta(\theta, \cdot)$ is upper semicontinuous over $\mathcal{T}$ for all such $\theta$. Furthermore, as $\gamma$ is the pointwise decreasing limit of the functions $\psi_\delta(\theta, \cdot)$, it is also upper semicontinuous.
\end{proof}

\begin{lemma} \label{L:conv0}
Let $W\times Z$ be compact with respect to the product topology $\mathcal{V} = \mathcal{V}_W \times \mathcal{V}_Z$, and let $\rho: W \times Z \ra \R$ be a pseudometric continuous with respect to $\mathcal{V}$ in the sense that $x_\alpha \ra x$ (for some net $x_\alpha$) implies $\rho(x_\alpha, x) \ra 0$.
Then, $W$ and $Z$ are totally bounded with respect to the pseudometrics
\begin{align*}
\rho_W: (w, w') \mapsto \sup_{z \in Z} \rho((w,z), (w', z)) \text{ and }
\rho_Z: (z, z') \mapsto \sup_{w \in W} \rho((w,z), (w, z')),
\end{align*}
which are continuous with respect to $\mathcal{V}_W$ and $\mathcal{V}_Z$, respectively.
\end{lemma}
\begin{proof}
By $\mathcal{V}$-continuity of $\rho$, for any $(w,z) \in W \times Z$ and $\ve > 0$, there is an open neighborhood $V(w,z)$ of $(w,z)$ such that $\sup_{(w', z') \in V} \rho((w,z),(w', z')) < \ve$. The cylinder sets form a basis for the product topology, so we may assume that $V(w,z) = V_W \times V_Z$ for open sets $V_W, V_Z$. By compactness, $W \times Z$ is covered by a finite collection of such sets $V_W(1) \times V_Z(1), \ldots, V_W(k) \times V_Z(k)$.

Let $w \in W$, and by possibly relabeling, let $V_W(1), \ldots , V_W(m)$ be the subcollection of open sets $V_W(j)$ containing $w$. Then $V_w = \bigcap_{j = 1}^m V_W(j)$ is an open set containing $w$. Moreover, $\bigcup_{j = 1}^m V_W(j) \times V_Z(j)$ contains $\{w\} \times Z$.

Let $w' \in V_w$, and denote by $(w_j, z_j)$ the center of the cover element $V_W(j) \times V_Z(j)$. For each $z \in Z$ there is a $j \le m$ with $z \in V_Z(j)$, and because $w, w' \in V_W(j)$, both $(w,z)$ and $(w', z)$ lie in $V_W(j) \times V_Z(j)$. By the triangle inequality and the defining property of the cover element,
\begin{align*}
\rho((w,z), (w', z)) & \le \rho((w_j, z_j), (w, z)) + \rho((w_j, z_j), (w', z)) < 2\ve.
\end{align*}
Taking the supremum over $z \in Z$ gives $\rho_W(w, w') \le 2\ve$. This implies that $\rho_W$ is $\mathcal{V}_W$-continuous. By compactness, $W$ can be covered with finitely many sets $V_w$, so it is $\rho_W$-totally bounded. The proof for $Z$ is the same, and establishing that $\rho_W$ and $\rho_Z$ are pseudometrics is routine.
\end{proof}

\begin{lemma} \label{L:conv1}
Make the assumptions of Lemma \ref{L:conv0}, and let $S \subset W \times Z$ with fibers $\overline{S}(w) = \{(w', z')\in S: w' = w\}$.
Then, for arbitrary $\ve > 0$, there exists a finite set $\{w_1, \ldots , w_J\} \subset W$ such that the fibers $\{\overline{S}(w_j): 1 \le j \le J \}$ satisfy
\begin{align*}
\sup_{w \in W} \min_{1 \le j \le J} \rho(\overline{S}(w), \overline{S}(w_j)) < \ve.
\end{align*}
\end{lemma}
\begin{proof}
For $w \in W$, let $S(w) = \{z \in Z: (w,z) \in S\}$ denote the fiber of $S$ at $w$.
By Lemma \ref{L:conv0}, $Z$ must be totally bounded with pseudometric $\rho_Z$, so partition $Z$ into a finite number of $\rho_Z$ $\ve/4$-diameter balls $B_1, \ldots, B_L$. Consider the map $\psi: W \ra \{0,1\}^L$ given by $\psi(w) = (\one_{S(w) \cap B_1 \neq \emptyset}, \ldots , \one_{S(w) \cap B_L \neq \emptyset})$, which returns as its $\ell^\text{th}$ coordinate an indicator of whether or not $S(w)$ and $B_\ell$ have a nonempty intersection. Furthermore, partition $W$ into $\rho_W$ $\ve/4$-diameter balls $C_1, \ldots , C_M$. For each ball $C_m$, the image $\psi(C_m) \subset \{0,1\}^L$ is a finite set, and there exist
$
w_{m,1}, \ldots , w_{m, i_m} \in C_m
$
which satisfy
\begin{align*}
\{\psi(w_{m,i}): 1 \le i \le i_m\} = \psi(C_m).
\end{align*}

The collection $\{w_{m, i}: 1 \le m \le M, 1\le i \le i_m\}$ satisfies our requirements. To see that this is the case, let $w \in W$ be arbitrary, and let $m$ be such that $w \in C_m$. Then there exists an index $i$ with $\psi(w) = \psi(w_{m,i})$.
Let $z$ be any element of $Z$ such that $(w, z) \in S$. For some $\ell$, we have $z \in B_\ell$, and so $S(w) \cap B_\ell \neq \emptyset$. It follows that $S(w_{m,i}) \cap B_\ell \neq \emptyset$ as well, and so there is some point $z'$ in the intersection. The triangle inequality then implies that
\begin{align*}
\rho((w,z), (w_{m,i}, z')) \le \rho_W(w, w_{m,i}) + \rho_Z(z, z') < \ve.
\end{align*}
As $z$ was arbitrary, every point of $\overline{S}(w)$ lies within $\ve$ of $\overline{S}(w_{m,i})$. Because the matching $\psi(w)=\psi(w_{m,i})$ is symmetric in $w$ and $w_{m,i}$---each coordinate records only whether $S(\cdot)\cap B_\ell\neq\emptyset$---the identical argument with the roles of $w$ and $w_{m,i}$ interchanged shows that every point of $\overline{S}(w_{m,i})$ lies within $\ve$ of $\overline{S}(w)$. Both directed distances, and hence the symmetric Hausdorff pseudodistance $\rho(\overline{S}(w),\overline{S}(w_{m,i}))$, are therefore below $\ve$, concluding the case where $S(w)$ is nonempty. If $S(w)$ is empty, then $\psi(w)$ is the zero vector, so the representative $w_{m,i}$ satisfying $\psi(w)=\psi(w_{m,i})$ also has $S(w_{m,i})=\emptyset$ and $\rho(\overline{S}(w),\overline{S}(w_{m,i}))=0$.
\end{proof}

\begin{lemma} \label{L:monotonedecreasing}
Let $(\mathcal{T}, \mathcal{U})$ be a compact topological space and $g_n: \mathcal{T} \ra \R \cup \{-\infty\}$ a collection of nonpositive and upper semicontinuous functions decreasing pointwise to a function $g: \mathcal{T} \ra \R \cup \{-\infty\}$. Let $d: \mathcal{T} \times \mathcal{T} \ra \R$ be a $\mathcal{U}$-continuous pseudometric, $K = g^{-1}(0)$ be nonempty, and let $(c_n)$ be any positive sequence with $c_n=o(1)$. Then, $\lim_{n \ra \infty} d(\{t: |g_n(t)| \le c_n \} , K) = 0$.
\end{lemma}
\begin{proof}
The limit function $g$ is upper semicontinuous, because $\{g\ge c\}=\bigcap_n\{g_n\ge c\}$ is closed for every $c\in\R$. For $\delta > 0$, let $K^{(\delta)} = \{t \in \mathcal{T}: d(t, K) < \delta\}$. By continuity of $d$, $K^{(\delta)}$ contains a $\mathcal{U}$-open set which in turn contains $K$. By upper semicontinuity and compactness, there is some $\ve > 0$ such that $\sup_{t \in (K^{(\delta)})^c} g(t)  \le - \ve$.

By the triangle inequality, it can readily be verified that $t \mapsto d(t,K^{(\delta)})$ is $\mathcal{U}$-continuous. Let $\tilde{g}$ be a $\mathcal{U}$-continuous function which is $0$ on $K^{(\delta)}$, bounded below by $-\ve$, and $- \ve$ on $((K^{(\delta)})^{(\delta)})^c$ (for instance, let $\tilde{g} = -\ve \chi(d(t,K^{(\delta)}))$ where $\chi$ is continuous and satisfies $\chi(0) = 0$, $\chi(x) = 1$ for $x \ge \delta$). Then, $\tilde{g}$ upper bounds $g$ and the sequence of pointwise maxima $\tilde{g}_n \equiv \max\{\tilde{g}, g_n\}$ consists of upper semicontinuous functions majorizing the functions $g_n$ and decreasing pointwise to $\tilde{g}$. Because $\tilde{g}$ is continuous, the differences $\tilde{g}_n - \tilde{g}$ are upper semicontinuous, nonnegative, and decrease pointwise to zero. For $\eta > 0$ the sets $E_n = \{t \in \mathcal{T}: \tilde{g}_n(t) - \tilde{g}(t) \ge \eta\}$ are therefore closed and nested, with empty intersection, so compactness gives $E_n = \emptyset$ for some $n$. Hence $\tilde{g}_n$ converges uniformly to $\tilde{g}$. It follows that there exists some $n'$ such that
\begin{equation} \label{E:monodec}
\sup_{\substack{n \ge n' \\ t \in (K^{(2\delta)})^c}} g_n(t) \le \sup_{\substack{n \ge n' \\ t \in ((K^{(\delta)})^{(\delta)})^c}} g_n(t) \le -\ve/2.
\end{equation}
Let $n'$ be large enough so that $\sup_{n \ge n'} c_n < \ve / 2$. Then, for $n \ge n'$, one has $\{t: |g_n(t)| \le c_n \} \subset K^{(2\delta)}$, whence \[
\limsup_{n \ra \infty} \sup_{t: |g_n(t)| \le c_n } d(t, K) \le 2\delta.
\]
The lemma follows because $\delta$ was arbitrary, and one always has $K \subset \{t: |g_n(t)| \le c_n \}$.

\end{proof}

\begin{lemma}\label{L:gn}
Under Assumptions \ref{A:zero}, \ref{A:one}, \ref{A:two}, and \ref{A:three}, for any sequence $b_n=o(1)$ with $r_nb_n\to\infty$ and $r_nb_n^{\alpha_v}\to0$, $\mathcal{K}(\theta)$ is nonempty for all $\theta \in \Theta_{\mathrm{I}}$, and
\[
\inf_{\theta \in \Theta_{\mathrm{I}}} \sup_{t \in \mathcal{T}} (r_n v_n(\theta, t) +  r_n b_n \psi_{b_n} (\theta, t)) = \inf_{\theta \in \Theta_{\mathrm{I}}} \sup_{t \in \mathcal{K}(\theta)} \mathbb{V}_n(\theta, t) + o_p(1).
\]
\end{lemma}
\begin{proof}
We claim that
\begin{align} \label{E:gn1}
\inf_{\theta \in \Theta_{\mathrm{I}}} \sup_{t \in \mathcal{T}} (r_n v_n(\theta, t) +  r_n b_n \psi_{b_n} (\theta, t)) \le \inf_{\theta \in \Theta_{\mathrm{I}}} \sup_{t \in \mathcal{K}(\theta)} \mathbb{V}_n(\theta, t) + o_p(1).
\end{align}
This is sufficient to prove the lemma, because then
\begin{align*}
\inf_{\theta \in \Theta_{\mathrm{I}}} \sup_{t \in \mathcal{K}(\theta)} \mathbb{V}_n(\theta, t)  &= \inf_{\theta \in \Theta_{\mathrm{I}}} \sup_{t \in \mathcal{K}(\theta)} (r_n v_n(\theta, t) + r_n b_n \psi_{b_n} (\theta, t)) \le \inf_{\theta \in \Theta_{\mathrm{I}}} \sup_{t \in \mathcal{T}} (r_n v_n(\theta,t) +  r_n b_n \psi_{b_n} (\theta, t))  \\
& \le \inf_{\theta \in \Theta_{\mathrm{I}}} \sup_{t \in \mathcal{K}(\theta)} \mathbb{V}_n(\theta, t)  + o_p(1)
\end{align*}
(by Lemma~\ref{L:gsc}, when $\gamma(\theta,t)=0$, monotonicity and nonpositivity imply $\psi_\delta(\theta,t)=0$ for all sufficiently small $\delta$, and hence $r_n b_n\psi_{b_n}(\theta,t)=0$ for all large $n$). To establish \eqref{E:gn1}, rewrite its left hand side as
\begin{align} \label{E:ps108}
\inf_{\theta \in \Theta_{\mathrm{I}}} \sup_{t \in \mathcal{T}} (\mathbb{V}_n(\theta, t) + r_n v(\theta, t) +  r_n b_n \psi_{b_n} (\theta, t)).
\end{align}

For the sequence $(b_n)$, let $g_n$ be the function $(\theta, t) \mapsto v(\theta, t) + \psi_{b_n}(\theta, t)$. Let $(c_n)$ be any sequence decreasing to $0$ slowly enough such that $c_n (r_n b_n)$ diverges. Then,
\begin{align}
    &\sup_{\theta \in \Theta_{\mathrm{I}}} \sup_{t: g_n(\theta, t) < -c_n } (\mathbb{V}_n(\theta, t) + r_n v(\theta, t) +  r_n b_n \psi_{b_n} (\theta, t)) \nonumber \\
    & \quad \le \sup_{\theta \in \Theta_{\mathrm{I}}} \sup_{t \in \mathcal{T}} \mathbb{V}_n(\theta, t)  - c_n r_n b_n  = O_p(1) - c_n r_n b_n.
        \label{E:bnan1}
\end{align}
By \eqref{E:ellbd}, $T_n$ is bounded below by $-O_p(1)$. Because, for such $b_n$, \eqref{E:lastline} is an upper bound for $T_n$ up to $o_p(1)$ by the proof of Theorem~\ref{T:one}, the same lower bound applies to \eqref{E:lastline} and to the left hand side of \eqref{E:gn1}. As \eqref{E:bnan1} is less than this lower bound with probability approaching $1$, and $v(\theta, \cdot)$, $\psi_{b_n}(\theta, \cdot)$ are nonpositive for $\theta \in \Theta_{\mathrm{I}}$, we conclude that \eqref{E:ps108} can be rewritten as
 \begin{align} \label{E:gnan2}
    \inf_{\theta \in \Theta_{\mathrm{I}}} \sup_{\substack{ t :g_n(\theta, t) \ge -c_n}} (\mathbb{V}_n(\theta,t) + r_n v(\theta, t) + r_n b_n \psi_{b_n} (\theta, t)) \le \inf_{\theta \in \Theta_{\mathrm{I}}} \sup_{\substack{ t :g_n(\theta, t) \ge -c_n}} \mathbb{V}_n(\theta, t)
 \end{align}
up to an asymptotically negligible term, and that the sets $A_n(\theta)\equiv\{t: g_n(\theta, t)\ge -c_n\}$, which are deterministic, are nonempty for all large $n$. Let $\tilde b_n=\sup_{m\ge n}b_m$ and $\tilde g_n(\theta,t)=v(\theta,t)+\psi_{\tilde b_n}(\theta,t)$. Then $\tilde b_n\downarrow0$, so the sets $\tilde A_n(\theta)\equiv\{t:\tilde g_n(\theta,t)\ge -c_n\}$ are nested and closed in $\mathcal{T}$ by Lemma \ref{L:gsc}. As $\tilde b_n\ge b_n$, $0\ge\tilde g_n\ge g_n$, so $A_n(\theta)\subseteq\tilde A_n(\theta)$ and the sets $\tilde A_n(\theta)$ are nonempty for all large $n$. By Assumption \ref{A:three}(2), they have nonempty intersection over $n$. By \eqref{E:gamma}, their intersection must be precisely $\mathcal{K}(\theta)$, which must be nonempty for all $\theta \in \Theta_{\mathrm{I}}$.

Let $\ve> 0$ be arbitrary. Both parts of Assumption \ref{A:one} jointly imply that there is some $\delta > 0$ such that
\begin{align} \label{E:equicont}
\limsup_{n \ra \infty} \P{ \sup_{\rho((\theta, t ), (\theta', t')) < \delta} |\mathbb{V}_n(\theta, t) - \mathbb{V}_n(\theta', t') |> \ve} < \ve.
\end{align}
By Lemma \ref{L:conv1}, we may pick $\{\theta_1, \ldots , \theta_J \}$ that satisfy
\begin{align*}
\sup_{\theta \in \Theta_{\mathrm{I}}} \min_{1 \le j \le J} \rho(\overline{S}(\theta), \overline{S}(\theta_j)) < \delta,
\end{align*}
where $\overline{S}(\theta) = \{(\theta, t): t \in \mathcal{K}(\theta)\}$ is defined as in Lemma \ref{L:conv1}. Let $\theta \in \Theta_{\mathrm{I}}$ be arbitrary. There exists a $\theta_j$ such that every $t' \in \mathcal{K}(\theta_j)$ has a corresponding $t \in \mathcal{K}(\theta)$ for which $\rho((\theta, t), (\theta_j, t')) < \delta$. Hence,
\begin{align*}
    \inf_{\theta \in \Theta_{\mathrm{I}}} \sup_{t \in \mathcal{K}(\theta)} \mathbb{V}_n(\theta, t) \ge \inf_{1 \le j \le J} \sup_{t \in \mathcal{K}(\theta_j)} \mathbb{V}_n(\theta_j, t) - \sup_{\rho((\theta, t),(\theta', t')) < \delta} | \mathbb{V}_n(\theta, t) - \mathbb{V}_n(\theta', t')|,
\end{align*}
which, in view of \eqref{E:equicont}, implies that
\begin{align}
\limsup_{n \ra \infty} \P{\inf_{\theta \in \Theta_{\mathrm{I}}}\sup_{t \in \mathcal{K}(\theta)} \mathbb{V}_n (\theta, t) < \inf_{1 \le j \le J} \sup_{t \in \mathcal{K}(\theta_j)} \mathbb{V}_n(\theta_j, t) - \ve } \le \ve. \label{E:infG}
\end{align}

For any given $\theta \in \Theta_{\mathrm{I}}$, Lemmas \ref{L:conv0} and \ref{L:monotonedecreasing} imply that
\begin{align*}
 \rho_\mathcal{T} (\tilde A_n(\theta),  \mathcal{K}(\theta)) \ra 0.
\end{align*}
where the $\rho_\mathcal{T}$ is defined in Lemma \ref{L:conv0}. Because $\mathcal{K}(\theta)\subseteq A_n(\theta)\subseteq\tilde A_n(\theta)$ for large $n$, the same convergence holds with $A_n(\theta)$. The convergence above is uniform over the finite set $\{\theta_1, \ldots, \theta_J\}$. By asymptotic equicontinuity,
\begin{align}
& \limsup_{n \ra \infty} \P{\inf_{1 \le j \le J} \sup_{t \in \mathcal{K}(\theta_j)}  \mathbb{V}_n(\theta_j, t) \le \inf_{\theta \in \Theta_{\mathrm{I}}} \sup_{t: g_n(\theta, t) \ge -c_n} \mathbb{V}_n(\theta, t) - \ve } \nonumber \\
& \qquad \le \limsup_{n \ra \infty}\P{\inf_{1 \le j \le J} \sup_{t \in \mathcal{K}(\theta_j)}  \mathbb{V}_n(\theta_j, t) \le \inf_{1 \le j \le J} \sup_{t : g_n(\theta_j, t) \ge - c_n} \mathbb{V}_n(\theta_j, t) - \ve }  = 0. \label{E:last2}
\end{align}
As $\ve > 0$ was arbitrary, the union bound, \eqref{E:infG}, and \eqref{E:last2} imply that
\begin{align} \label{E:mm}
 \inf_{\theta \in \Theta_{\mathrm{I}}} \sup_{t: g_n(\theta, t) \ge -c_n} \mathbb{V}_n(\theta, t)  \le \inf_{\theta \in \Theta_{\mathrm{I}}}\sup_{t \in \mathcal{K}(\theta)} \mathbb{V}_n(\theta, t) + o_p(1).
\end{align}
In conjunction with \eqref{E:gnan2}, this establishes \eqref{E:gn1}.
\end{proof}

\begin{lemma} \label{L:gn2}
Under Assumptions \ref{A:zero}, \ref{A:one}, \ref{A:two}, \ref{A:three}, and \ref{A:four}(1), there is a sequence $(b_n)$ satisfying $b_n=o(1)$, $r_nb_n\to\infty$, $r_nb_n^{\alpha_v}\to0$, and $T_n = \inf_{\theta \in \Theta_{\mathrm{I}}} \inf_{\theta' \in B_\Theta(\theta, b_n)} \sup_{t \in \mathcal{T}} r_n v_n(\theta' , t) + o_p(1)$.
\end{lemma}
\begin{proof} Let $f_\ell$ be defined by \eqref{E:felleqn}.
By Assumption \ref{A:four}(1), $\lim_{n \ra \infty} r_n f_\ell(r_n^{-1} c)^{\alpha_v} = 0$ for all constants $c > 0$. Hence, we may choose a diverging sequence $(d_n)$ such that $r_n^{-1}d_n\to0$ and $r_n f_\ell(r_n^{-1} d_n)^{\alpha_v}\to0$.

Let $q_n=f_\ell(r_n^{-1} d_n)$ for all $n$. Because Assumption \ref{A:zero}(2) guarantees that $f_\ell(\delta) = o(1)$, $q_n$ is $o(1)$, and $q_n^{\alpha_v} = o(r_n^{-1})$. Moreover, the definition of $f_\ell$ in \eqref{E:felleqn} implies that any $\theta$ such that $d_{\mathfrak{B}}(\theta, \Theta_{\mathrm{I}}) > q_n$ is such that $\ell(\theta) > r_n^{-1} d_n$. Hence,
\begin{align*}
    \inf_{\theta: d_{\mathfrak{B}}(\theta, \Theta_{\mathrm{I}}) > q_n} r_n \ell_n(\theta) \ge d_n - O_p(1),
\end{align*}
and the argument establishing \eqref{E:ll} implies that, with probability approaching $1$, $T_n$ is equal to $\inf_{\theta \in \Theta_{\mathrm{I}}} \inf_{\theta' \in B_\Theta(\theta, q_n)} \sup_{t \in \mathcal{T}} r_n v_n(\theta' , t)$.

Finally, let $(\tilde{b}_n)$ satisfy $\tilde{b}_n=o(1)$, $r_n\tilde{b}_n\to\infty$, and $r_n\tilde{b}_n^{\alpha_v}\to0$, which is possible because $\alpha_v>1$. Let $b_n = \max\{q_n, \tilde{b}_n\}$. Then, it is straightforward to verify that $b_n$ satisfies the requirements of the lemma.
\end{proof}

\begin{lemma}
\label{L:hemicont} Let $Z$ be totally bounded with respect to pseudometric $d_Z$ and let $(W, \mathcal{W})$ be a topological space with respect to which pseudometric $d_W$ is continuous. Let the correspondence $K: (W, \mathcal{W}) \ra (Z, d_Z)$ be lower hemicontinuous with nonempty values at some set $S\subset W$ (in the sense of Definition \ref{D:hemiat}). Then, for every $\delta > 0$, there is a $\mathcal{W}$-open neighborhood $S_\delta$ of $S$ on which $K$ has nonempty values and
\begin{align*}
\sup_{w \in S_\delta} \inf_{w' \in S} (d_W(w,w') + \sup_{z \in K(w')} d_Z(z,K(w))) < \delta.
\end{align*}
\end{lemma}
\begin{proof}
Let $s \in S$ be arbitrary. The set $K(s)$ is a subset of the totally bounded space $Z$ and has a finite cover of $\delta /8$-open balls $\{B(z_1(s), \delta/8), \ldots B(z_{m(s)}(s), \delta /8)\}$, where $\{z_1(s), \ldots, z_{m(s)}(s)\} \subset K(s)$. By lower hemicontinuity of $K$ on $S$, there exists a neighborhood $U(s)$ of $s$ such that $K(w)$ intersects all of these $\delta/8$-balls whenever $w \in U(s)$ and $w \not\in S$. For such $w$,
\begin{align*}
\sup_{z \in K(s)} d_Z(z, K(w)) \le \sup_{z \in K(s)} \min_{1 \le m \le m(s)} d_Z(z, z_m(s)) + d_Z(z_m(s), K(w)) \le \delta/4.
\end{align*}
Assume without loss of generality that $U(s)$ has $d_W$-diameter less than $\delta/2$; if not, intersect it with a $\mathcal{W}$-open set contained in a $\delta/4$ $d_W$-open ball around $s$ (which exists, by $\mathcal{W}$-continuity of $d_W$). It suffices to take $S_\delta = \bigcup_{s \in S} U(s)$, because for each $w\in S_\delta$ the quantity inside the supremum is then at most $3\delta/4$.
\end{proof}

\begin{lemma}[Contact-set lower bound] \label{L:contactlb}
Let Assumptions \ref{A:zero}(3) and \ref{A:three} hold, and let $\mathcal{K}':\Theta\ra\mathcal{T}$ be a correspondence that has nonempty values on $\Theta_{\mathrm{I}}$ and is lower hemicontinuous at $\Theta_{\mathrm{I}}$. Then, for every $\delta>0$, there is an open neighborhood $S_\delta$ of $\Theta_{\mathrm{I}}$ on which $\mathcal{K}'$ has nonempty values and such that, for every $\mathbb{V}\in\ell^\infty(\Theta\times\mathcal{T})$,
\begin{align*}
\inf_{\theta'\in S_\delta}\sup_{t'\in\mathcal{K}'(\theta')}\mathbb{V}(\theta',t')\ \ge\ \inf_{\theta\in\Theta_{\mathrm{I}}}\sup_{t\in\mathcal{K}'(\theta)}\mathbb{V}(\theta,t)-\sup_{\rho((\theta,t),(\theta',t'))<\delta}|\mathbb{V}(\theta,t)-\mathbb{V}(\theta',t')|.
\end{align*}
The same holds with $S_\delta$ replaced by any smaller open neighborhood of $\Theta_{\mathrm{I}}$.
\end{lemma}
\begin{proof}
Let $\rho_{\mathcal{T}}(t,t') = \sup_{\theta\in\Theta_{\mathrm{I}}}\rho((\theta,t),(\theta,t'))$ and $\rho_\Theta(\theta,\theta') = \sup_{t\in\mathcal{T}}\rho((\theta,t),(\theta',t))$. Lemma \ref{L:conv0}, applied with $W=\Theta_{\mathrm{I}}$, implies that $\mathcal{T}$ is totally bounded with respect to $\rho_{\mathcal{T}}$, which is continuous with respect to $\mathcal{U}$. Thus, $\mathcal{K}'$ is lower hemicontinuous at $\Theta_{\mathrm{I}}$ with respect to the topology on $\mathcal{T}$ induced by $\rho_\mathcal{T}$. Therefore, Lemma \ref{L:hemicont} implies that there is some open neighborhood $S_\delta$ of $\Theta_{\mathrm{I}}$ on which $\mathcal{K}'$ has nonempty values and which satisfies
\begin{align*}
\sup_{\theta' \in S_\delta} \inf_{\theta \in \Theta_{\mathrm{I}}} (d_{\mathfrak{B}} (\theta, \theta') + \sup_{t \in \mathcal{K}'(\theta)} \rho_\mathcal{T} (t, \mathcal{K}'(\theta')) ) < \delta.
\end{align*}
By \eqref{E:pseudometricdelta}, $S_\delta$ may also be chosen such that the previous display holds with $\rho_\Theta$ replacing $d_{\mathfrak{B}}$. For such a $S_\delta$, every $\theta' \in S_\delta$ may be matched with a $\theta \in \Theta_{\mathrm{I}}$ with the property that every $t \in \mathcal{K}'(\theta)$ may be matched with a $t' \in \mathcal{K}'(\theta')$ satisfying $\rho((\theta, t), (\theta',t')) \le \rho_{\mathcal{T}} (t, t') + \rho_{\Theta} (\theta, \theta') < \delta$. Hence $\sup_{t\in\mathcal{K}'(\theta)}\mathbb{V}(\theta,t)\le\sup_{t'\in\mathcal{K}'(\theta')}\mathbb{V}(\theta',t')$ plus the modulus in the display, and taking the infimum over $\theta'\in S_\delta$ gives the claim. Shrinking $S_\delta$ only raises the left side.
\end{proof}

\subsection{Proof of Proposition \ref{P:ex1}}
    We prove the proposition for the case $T_0 = 0$. The transformed function $\xi - T_0$ may always be used in place of $\xi$, so this suffices.
    After this replacement, we relabel the centered criterion $\xi-T_0$ as $\xi$, so that $\Theta_{\mathrm{I}}=\{\theta:\xi(m(\theta))=0\}$ and all contact sets below are formed for the centered loss.
    Assumptions \ref{A:ex1}(1) and (3) directly imply Assumptions \ref{A:zero}(1), (2), and (3).

    For $\delta > 0$, let $\Theta_{\mathrm{I}}^{(\delta)}$ denote the $\delta$-enlargement of $\Theta_{\mathrm{I}}$. By Assumption \ref{A:ex1}(3), $\ell(\theta)$ is strictly bounded away from $0$ outside of every $\Theta_{\mathrm{I}}^{(\delta)}$. Therefore, as in \eqref{E:ll}, the $r_n^{-1}$-consistency of $m_n$ for $m$ in Assumption \ref{A:ex1}(5) implies that $r_n \inf_{\theta \in \Theta} \ell_n(\theta)$ agrees with $r_n \inf_{\theta \in \Theta_{\mathrm{I}}^{(\delta)}} \ell_n(\theta)$ up to an asymptotically negligible term.

    Lemma \ref{L:duality}\suppref{}, applied on the bounded convex set $D$ supplied by Assumption \ref{A:ex1}(2), implies that there is a convex, norm-bounded, and weak-* compact set of affine functions $\mathcal{T}$ over $\mathfrak{X}$ such that $\xi(x) = \max_{t \in \mathcal{T}} t(x) $ for all $x \in D$. Assumptions \ref{A:ex1}(2) and (5) imply that $m_n(\Theta)\subset D$ with probability approaching one. On this event, $\inf_{\theta\in\Theta}\sup_{t\in\mathcal T}v_n(\theta,t)=\inf_{\theta\in\Theta}\xi(m_n(\theta))$, so the statistic in Theorem~\ref{T:one} is the left-hand side of \eqref{E:ex1}. We now verify the conditions of Theorem \ref{T:one} under Assumption \ref{A:ex1} with this choice of $\mathcal{T}$, $\mathcal{U}$ the weak-* topology, and the arrangements
    \begin{align} \label{E:vfns}
        v(\theta, t) & = t(m(\theta))  \text{ and } v_n(\theta, t) = t(m_n(\theta)) .
    \end{align}

    The functions $t$ in $\mathcal{T}$ are affine, and thus $\mathbb{V}_n(\theta, t) = r_n (t(m_n(\theta)) - t(m(\theta))) = y[t](r_n (m_n(\theta) - m(\theta) ) ) $.
    Because $\mathcal{T}$ is norm bounded, $L \equiv \sup_{t \in \mathcal{T}} \norm{y[t]}_{\mathfrak{X}^*}$ is finite. Hence, the map $\mathbb{W}_n(\cdot) \mapsto ((\theta, t) \mapsto y[t](\mathbb{W}_n(\theta)))$ from $\ell^\infty(\Theta, \mathfrak{X})$ to $\ell^\infty(\Theta \times \mathcal{T})$ is Lipschitz continuous, and Assumption \ref{A:one}(1) holds with $\mathbb{G}_n(\theta, t) = y[t](\mathbb{X}_n(\theta))$ for all $\theta$ and $t$.

    Let $\mathcal{T}$ be equipped with the weak-* topology $\mathcal{U}$. By Lemma \ref{L:pretight}\suppref{}, the pseudometric defined by \eqref{E:rhomaineqn}\suppref{} is continuous on $\Theta \times \mathcal{T}$, satisfies the equicontinuity requirement of Assumption \ref{A:one}(2), and also fulfills \eqref{E:pseudometricdelta} and thus Assumption \ref{A:three}(1).

    Assumption \ref{A:ex1}(4) and linearity of the functionals $y[t]$ imply Assumption \ref{A:two}(1) with $D_{\theta; t}v$ the continuous linear operator $\theta' - \theta \mapsto y[t](m'(\theta)(\theta' - \theta))$ and $\alpha_v = \alpha_m$. Because
    \begin{align}
        v_n(\theta, t) + D_{\theta; t}v(\theta' - \theta) & = t(m_n(\theta)) + y[t](m'(\theta)(\theta' - \theta)) \nonumber \\
        &= t(m_n(\theta) + m'(\theta)(\theta' - \theta)), \nonumber
    \end{align}
    the mapping $t \mapsto v_n(\theta, t) + D_{\theta; t}v(\theta' - \theta) $ is itself affine and hence concave, satisfying Assumption \ref{A:two}(2).

    Assumption \ref{A:three}(2) is guaranteed by Lemma \ref{L:duality}. The pointwise continuity requirements of Assumption \ref{A:three}(3) hold by definition of the weak-* topology. Hence, Theorem \ref{T:one} implies \eqref{E:ex1} with $\mathcal{K}(\theta)$ as specified in Lemma \ref{L:duality}. If either $\Theta$ is a singleton or $m$ is affine in a neighborhood of $\Theta_{\mathrm{I}}$, the map $\theta \mapsto v(\theta, t)$ is convex on that neighborhood, so that Assumption \ref{A:four}(2) holds and Theorem \ref{T:one} implies that \eqref{E:ex1} holds with equality. On the other hand, if \eqref{E:localident1} holds with $\alpha_\ell\alpha_m  > 1$, then Assumption \ref{A:four}(1) holds, and again \eqref{E:ex1} is an equality. Under Assumption \ref{A:ex22}, Assumption \ref{A:four}(3) holds by Lemma~\ref{L:ex22K}\suppref{}, and again \eqref{E:ex1} is an equality.

\begin{lemma} \label{L:pretight}
    Let Assumptions \ref{A:ex1}(5) and (6) hold and let $\mathcal{T}$ be equipped with the weak-* topology. Then, there exists a continuous pseudometric $\rho$ on $\Theta \times \mathcal{T}$ satisfying \eqref{E:pseudometricdelta} with respect to which the limit processes $\mathbb{G}_n$ are equicontinuous.
\end{lemma}
\begin{proof}
    For $w\in\ell^\infty(\Theta,\mathfrak{X})$, let
    \[
        \omega_w(\delta)=\sup_{\substack{\theta\in\Theta,\,\theta'\in B_\Theta(\theta,\delta)}}
        \norm{w(\theta)-w(\theta')}_{\mathfrak{X}} .
    \]
    Fix $k\in\mathbb{N}$.  By Assumption~\ref{A:ex1}(6), for each $m\in\mathbb{N}$ we may choose $\delta_{k,m}>0$ so that
    \[
        \sup_n\P{\omega_{\mathbb{X}_n}(\delta_{k,m})>1/m}<2^{-(m+1)}/k .
    \]
    Hence, by the union bound, the set
    \[
        E_k=\{w:\omega_w(\delta_{k,m})\le 1/m\text{ for all }m\in\mathbb{N}\}
    \]
    consists of uniformly $d_{\mathfrak{B}}$-continuous functions and satisfies
    $\inf_n\P{\mathbb{X}_n\in E_k}\ge 1-1/(2k)$.  Assumption~\ref{A:ex1}(5) gives a totally bounded set $K_k\subset\ell^\infty(\Theta,\mathfrak{X})$ with $\inf_n\P{\mathbb{X}_n\in K_k}\ge 1-1/(2k)$.  Therefore $S_k\equiv E_k\cap K_k$ is totally bounded, consists of uniformly continuous functions, and satisfies $\inf_n\P{\mathbb{X}_n\in S_k}\ge 1-1/k$.

    The union $\bigcup_{k \in \N} S_{k}$ is separable, contains $\mathbb{X}_n$ almost surely for each $n$, and hence contains a dense sequence $X_j \in \ell^\infty(\Theta, \mathfrak{X})$. Define $\rho$ with respect to this sequence as
    \begin{align} \label{E:rhomaineqn}
        \rho((\theta, t), (\theta', t')) & = \sum_{j = 1}^\infty 2^{-j} \frac{|y[t](X_j(\theta)) - y[t'](X_j(\theta'))|}{1 + |y[t](X_j(\theta)) - y[t'](X_j(\theta'))|}.
    \end{align}
    Clearly, $\rho$ is symmetric and nonnegative. Each of its summands measures the distance between $y[t](X_j(\theta))$ and $y[t'](X_j(\theta'))$ under a standard bounded metric transformation, and $\rho$ thus satisfies the triangle inequality.

    Let $\ve > 0$, $\theta$, and $t$ be arbitrary and let $J$ be such that $2^{-J} < \ve/2$, so that the tail terms having $j > J$ in \eqref{E:rhomaineqn} contribute at most $\ve/2$. Let $C=1\vee\sup_{t' \in \mathcal{T}} \norm{y[t']}_{\mathfrak{X}^*}$. Let $U \subset \mathcal{T}$ be the open neighborhood of $t$ containing points $t'$ such that $|y[t](X_j(\theta)) - y[t'](X_j(\theta))| < \ve/4$ for all $j \le J$. Furthermore, let $\delta > 0$ be small enough such that $\max_{1 \le j \le J}\sup_{\theta' \in B_\Theta(\theta,\delta)} \norm{X_j(\theta) - X_j(\theta')}_{\mathfrak{X}} < \ve /(4 C)$. Then for $(\theta', t') \in B_\Theta(\theta, \delta) \times U$ and $j \le J$,
    \begin{align}
        |y[t](X_j(\theta)) - y[t'](X_j(\theta'))| & \le |y[t](X_j(\theta)) - y[t'](X_j(\theta))| + |y[t'](X_j(\theta)) - y[t'](X_j(\theta'))| \nonumber \\
        & \le \ve/4 + C \norm{ X_j (\theta) - X_j(\theta')}_\mathfrak{X} < \ve/2. \label{E:rhoeps}
    \end{align}
    As the map $d \mapsto d/(1 + d)$ is Lipschitz continuous for $d \ge 0$, \eqref{E:rhoeps} implies that $\rho((\theta, t), (\theta', t')) < \ve$, which implies the continuity of $\rho$. Similarly, we have
    \begin{align*}
        \sup_{\substack{\theta\in\Theta,\,\theta'\in B_\Theta(\theta,\delta)}} \sup_{t \in \mathcal{T}} |y[t](X_j(\theta)) - y[t](X_j(\theta')) | \le C \sup_{\substack{\theta\in\Theta,\,\theta'\in B_\Theta(\theta,\delta)}} \norm{X_j(\theta) - X_j(\theta')}_\mathfrak{X} < \ve/4
    \end{align*}
    for $j \le J$, so that $\sup_{\substack{\theta\in\Theta,\,\theta'\in B_\Theta(\theta,\delta)}} \sup_{t \in \mathcal{T}} \rho((\theta, t), (\theta', t))$ is at most $\ve$, and \eqref{E:pseudometricdelta} is fulfilled.

    The last step is to show that the limit processes $\mathbb{G}_n$ are equicontinuous with respect to $\rho$. Again, let $\ve > 0$ be arbitrary. Choose $k$ so large that $1/k<\ve$, and then choose $J$ such that the $\ell^\infty(\Theta,\mathfrak{X})$ balls of radius $\ve/(4C)$ around $\{X_1,\ldots,X_J\}$ cover $S_k$. With probability at least $1-\ve$, $\mathbb{X}_n$ lies in this union of balls uniformly in $n$. Let $X$ be any point in this union of balls, and let $\delta$ be small enough such that $|y[t](X_j(\theta)) - y[t'](X_j(\theta'))| < \ve/2$ for all $j \le J$ whenever $\rho((\theta, t), (\theta', t')) < \delta$. By definition of $X$ and the triangle inequality,
    \begin{align*}
        \sup_{\rho((\theta, t), (\theta', t')) < \delta} |y[t](X(\theta)) - y[t'](X(\theta'))| & \le  \sup_{\substack{j \le J \\ \rho((\theta, t), (\theta', t')) < \delta}} |y[t](X_j(\theta)) - y[t'](X_j(\theta'))| + \ve /2 < \ve.
    \end{align*}
    Hence, whenever $\mathbb{X}_n$ is in this union of balls, which happens with probability uniformly at least $1-\ve$, $\sup_{\rho((\theta, t), (\theta', t')) < \delta} |\mathbb{G}_n(\theta, t) - \mathbb{G}_n(\theta', t') |$ is at most $\ve$, as desired.

\end{proof}

\begin{lemma} \label{L:ex22K}
Let Assumption \ref{A:ex1} hold, and let $T_0=0$, $\mathcal{T}$, and $v$ be as in the proof of Proposition~\ref{P:ex1}, with $\mathcal{U}$ the weak-* topology. Under Assumption \ref{A:ex22}, $\mathcal{K}$ is lower hemicontinuous at $\Theta_{\mathrm{I}}$. Hence Assumption \ref{A:four}(3) holds.
\end{lemma}
\begin{proof}
Suppose that Assumption \ref{A:ex22}(1) holds, and let $\theta_0 \in \Theta_{\mathrm{I}}$. By definition of the subgradient, any $y \in \partial \xi(m(\theta_0))$ is such that $\xi(x) \ge y(x) - y(m(\theta_0)) $ for all $x \in \mathfrak{X}$. In particular, if $x \in \xi^{-1}((-\infty,0])$, $y(x) - y(m(\theta_0)) \le 0$, which implies that $y \in (T_{m(\theta_0)}(\xi^{-1}((-\infty,0])))^\circ$. Lemma \ref{L:duality} thus implies that, at any point $\theta_0 \in \Theta_{\mathrm{I}}$,
    \begin{align*}
        \mathcal{K}(\theta_0) &= \{ y - y(m(\theta_0)): y \in (-m'(\theta_0)(T_{\theta_0} \Theta))^\circ \cap \partial \xi(m(\theta_0))\}\\
        & = \{ y - y(m(\theta_0)): y \in (-m'(\theta_0)(T_{\theta_0} \Theta))^\circ \cap (T_{m(\theta_0)} (\xi^{-1}((-\infty,0])))^\circ \cap \partial \xi(m(\theta_0))\}
    \end{align*}

    We show that this correspondence is weak-* lower hemicontinuous at $\theta_0$. Let $D$ be the bounded neighborhood of $m(\Theta)$ relative to which $\mathcal{T}$ is defined by Lemma \ref{L:duality}, and let $L$ be the Lipschitz constant which is used to define $\tilde{\xi}$. Let $U \subset \mathcal{T}$ be weak-* open and contain an $\ve$-norm ball around a point $t = y - (\xi + \delta(\cdot | \overline{D}))^* (y) \in \mathcal{K}(\theta_0)$, where $y \in (-m'(\theta_0)(T_{\theta_0} \Theta))^\circ \cap \partial \xi(m(\theta_0))$. By the Fenchel-Young inequality \eqref{E:fenchelyoung}, $y(m(\theta_0)) = (\xi + \delta(\cdot | \overline{D}))^* (y)$.

    By \eqref{E:xistarcont}, the conjugate $(\xi + \delta(\cdot | \overline{D}))^*$ is the pointwise supremum of a collection of uniformly norm-continuous functions, and is itself norm-continuous. Therefore, there is some $\delta_1 > 0$ sufficiently small that the difference between $(\xi + \delta(\cdot | \overline{D}))^* (y)$ and $(\xi + \delta(\cdot | \overline{D}))^*(y')$ is less than $\ve/2$ whenever $\norm{y - y'}_{\mathfrak{X}^*} < \delta_1$.

   Assumption \ref{A:ex22}(1) implies that there is some $\delta_2 > 0$ such that $\partial \xi(m(\theta_0)) \cap (-m'(\theta)(T_\theta \Theta))^\circ \cap (T_{m(\theta)} (\xi^{-1}((-\infty,0])))^\circ$ contains a point $y'$ satisfying $\norm{y - y'}_{\mathfrak{X}^*} < \delta_1 \wedge \ve /2$ whenever $\theta$ is not in $\Theta_{\mathrm{I}}$ and $d_{\mathfrak{B}}(\theta, \theta_0) < \delta_2$.
   Because $y' \in \partial \xi(m(\theta_0))$, $y'$ is also an element of $B_{\mathfrak{X}^*}(0,L)$. Hence, the affine function $t' \equiv y' - (\xi + \delta(\cdot | \overline{D}))^* (y')$ is in $\mathcal{T}$ by its definition in \eqref{E:affineTdefn}. By the triangle inequality, the norm distance between $t$ and $t'$ is at most $\ve$. Because $y'$ is in $(-m'(\theta)(T_\theta \Theta))^\circ$, $\gamma(\theta, t') = 0$ by Lemma~\ref{L:tangentconecond}. Because $y'$ is also in $\partial \xi(m(\theta_0))$, the Fenchel-Young inequality implies that $(\xi + \delta(\cdot | \overline{D}))^*(y') = y'(m(\theta_0))$. Finally, because $y'$ is in $(T_{m(\theta)}(\xi^{-1}((-\infty,0])))^\circ$, and $m(\theta_0)$ is in $\xi^{-1}((-\infty,0])$,
    \begin{align*}
        v(\theta, t') & = t'(m(\theta))  = y'(m(\theta)) - (\xi + \delta(\cdot | \overline{D}))^* (y') = y'(m(\theta)) - y'(m(\theta_0)) \ge 0.
    \end{align*}
    Hence, $t'$ is in $\mathcal{K}(\theta) \cap U$, and $\mathcal{K}$ is lower hemicontinuous at $\Theta_{\mathrm{I}}$.

    Similarly, Assumption \ref{A:ex22}(2) implies the existence of points $y' \in (-m'(\theta)(T_\theta \Theta))^\circ \cap \partial \xi(m(\theta))$ which become arbitrarily close to $y$ as $\theta$ becomes close to $\theta_0 \in \Theta_{\mathrm{I}}$. As $\partial \xi (m(\theta)) = \partial \tilde{\xi}(m(\theta))$ is contained in the effective domain of $\tilde{\xi}^*$, $y'$ is in $B_{\mathfrak{X}^*}(0,L).$ Moreover,
    \begin{align*}
        y'(m(\theta)) - (\xi + \delta(\cdot | \overline{D} ))^*(y') = \xi(m(\theta)) \ge 0,
    \end{align*}
    so that $t' \equiv y' - (\xi + \delta(\cdot |\overline{D} ))^*(y')$ is in $\mathcal{K}(\theta)$, and again $\mathcal{K}$ is lower hemicontinuous at $\Theta_{\mathrm{I}}$.

Because $\mathcal{K}(\theta')\subseteq\widetilde{\mathcal{K}}(\theta')$ for $\theta'\notin\Theta_{\mathrm{I}}$, Assumption \ref{A:four}(3) holds.
\end{proof}

\subsection{Proofs of the corollaries}
\subsubsection{Proof of Corollary \ref{C:dirderiv}}
    Apply Theorem~\ref{T:one} deterministically, with $r_n=h_n^{-1}$, $v_n=v+h_n\varphi_n$, $\mathbb{V}_n=\varphi_n$, and $\mathbb{G}_n=\varphi$. Assumptions~\ref{A:zero}, \ref{A:two}, and~\ref{A:three} hold by supposition, with Assumption~\ref{A:two}(2) and the requirement on $v_n$ in Assumption~\ref{A:three}(3) following from $v_n\in\mathbb{D}_\Phi$. Assumption~\ref{A:one} holds because $\varphi_n\to\varphi$ uniformly and $\varphi$ is bounded and uniformly $\rho$-continuous. Because $\Theta_{\mathrm{I}}$ is nonempty, $\Phi(v)=0$ by its definition, so the statistic of Theorem~\ref{T:one} is $T_n=(\Phi(v+h_n\varphi_n)-\Phi(v))/h_n$. The theorem gives $T_n\le\inf_{\theta\in\Theta_{\mathrm{I}}}\sup_{t\in\mathcal{K}(\theta)}\varphi_n(\theta,t)+o(1)$, and the right side converges to $\inf_{\theta\in\Theta_{\mathrm{I}}}\sup_{t\in\mathcal{K}(\theta)}\varphi(\theta,t)$. Under Assumption~\ref{A:four}, the first line of \eqref{E:thm} is an equality, which gives the limit.

\subsubsection{Proof of Corollary \ref{C:geometry}}
    Apply Proposition~\ref{P:ex1} with $\xi=d_\mathfrak{X}(\cdot,C)$; under the corollary's normalization, its left side is the displayed statistic. It remains only to identify the dual expression. By Lemma \ref{L:normalcone} below, $\partial \xi (m(\theta)) = B_{\mathfrak{X}^*}(0,1) \cap (T_{m(\theta)} C)^\circ$ for all $\theta \in \Theta_{\mathrm{I}}$. For such $\theta$,
    \begin{align*}
        (-m'(\theta)(T_\theta \Theta))^\circ \cap \partial \xi(m(\theta)) & = B_{\mathfrak{X}^*}(0,1) \cap ( T_{m(\theta)} C - m'(\theta) (T_\theta \Theta))^\circ
    \end{align*}
    By Lemma \ref{L:positive}, taking suprema over this intersection yields the distance function to the cone $ T_{m(\theta)} C - m'(\theta) (T_\theta \Theta) $.

\subsection{Proof of Theorem \ref{T:two}}

Under Assumption~\ref{A:four}, and when $\mathbb{G}_n=\mathbb{G}$ is fixed and tight, \eqref{E:bootstrapone} also follows from Corollary~\ref{C:dirderiv} and \citet[Section~3]{HongLi2018}. That route does not extend to Propositions~\ref{ad:one} and~\ref{Add:2} or Addendum~\ref{add:bootstrap_K0}, whose bootstrap statistics are not numerical derivatives of $\Phi$. The direct argument below covers them as well and makes the role of the step explicit.

Choose $\tau\in(1/\alpha_v,1)$ in case (i), $\tau\in(0,1)$ in case (ii), and $\tau\in(\max\{1/\alpha_v,1/(\alpha_\ell\alpha_v)\},1)$ in case (iii). Set $b_n=s_n^{-\tau}$ and $\lambda_n=s_nb_n$. Then $\lambda_n\to\infty$, $b_n\to0$, and $\lambda_nb_n^{-1}=s_n=o(r_n)$. In cases (i) and (iii), $\lambda_nb_n^{\alpha_v-1}=s_n^{1-\tau\alpha_v}\to0$, and in case (iii), $s_nb_n^{\alpha_\ell\alpha_v}=s_n^{1-\tau\alpha_\ell\alpha_v}\to0$. Let
\[
Z^*_{n,\mathrm{loc}}=\inf_{\theta\in\Theta}\inf_{\theta'\in B_\Theta(\theta,b_n)}\sup_{t\in\mathcal{T}}\big(s_nv_n(\theta',t)+\mathbb{V}_n^*(\theta,t)\big)
\]
and $\omega_n^*(b)=\sup\{|\mathbb{V}_n^*(\theta,t)-\mathbb{V}_n^*(\theta',t)|:\theta\in\Theta,\ \theta'\in B_\Theta(\theta,b),\ t\in\mathcal{T}\}$. We first show that \eqref{E:bootstrapone} holds with $Z^*_{n,\mathrm{loc}}$ in place of its left side. Define $\psi_n(\theta, t) = b_n^{-1} \inf_{\theta' \in B_\Theta(\theta, b_n)} D_{\theta; t} v(\theta' - \theta)$ for all $\theta$ and $t$. By \eqref{E:vbd},
    \begin{align}
        Z^*_{n,\mathrm{loc}} & = \inf_{\theta \in \Theta} \inf_{\theta' \in B_\Theta(\theta,b_n)} \sup_{t\in\mathcal T}\big(s_n v(\theta',t)+\mathbb V_n^*(\theta,t)\big)+o_{p^*}(1). \label{E:auxl1}
    \end{align}
    The proof establishes the first line of \eqref{E:bootstrapone} by bounding $Z^*_{n,\mathrm{loc}}$ above and below. The second line then follows from the bootstrap continuous mapping theorem.

    \textbf{An initial lower bound for $\mathbf{Z^*_{n,\mathrm{loc}}}$:} we begin from a crude upper bound for $Z^*_{n,\mathrm{loc}}$, bounded in probability:
    \begin{align} \label{E:crudeupper}
        \inf_{\theta \in \Theta_{\mathrm{I}}} \sup_{t \in \mathcal{T}} (s_n  v(\theta, t)+\mathbb{V}_n^*(\theta, t) )   &\le \inf_{\theta \in \Theta_{\mathrm{I}}} s_n \ell(\theta) +\sup_{\substack{\theta \in \Theta \\ t \in \mathcal{T}}} |\mathbb{V}_n^*(\theta, t)|.
    \end{align}
    We now use this upper bound to localize $Z^*_{n,\mathrm{loc}}.$ Let $\delta > 0$ be arbitrary, and let $S_\delta \supset \Theta_{\mathrm{I}}$ be an open neighborhood of $\Theta_{\mathrm{I}}$ in $\mathfrak{B}$, chosen sufficiently small that $B_\Theta(\theta,b_n)$ is convex for every $\theta\in S_\delta$ and all sufficiently large $n$ (which is possible by Assumption~\ref{A:zero}(1)). By compactness of $\Theta_{\mathrm{I}}$, there is some $\delta' > 0$ such that $B_\Theta(\theta, \delta')$ is contained in $S_\delta$ whenever $\theta \in \Theta_{\mathrm{I}}$. By Assumption \ref{A:five}(3),
    \begin{align}
        &\inf_{\theta \in S_\delta^c} \inf_{\theta' \in B_\Theta(\theta, b_n)} \sup_{t \in \mathcal{T}} \big(s_n v(\theta, t) + \mathbb{V}_n^*(\theta, t) + s_n(v(\theta' , t) - v(\theta, t))\big) \nonumber \\
        &\quad \ge \inf_{\theta \in S_\delta^c} s_n \ell(\theta) - \sup_{\substack{\theta \in \Theta \\ t \in \mathcal{T}}} |\mathbb{V}_n^*(\theta, t)| - s_n b_n C_v. \label{E:auxl2}
    \end{align}
    Assumption \ref{A:zero}(2) implies that the previous line is at least $s_n c - \sup_{\substack{\theta \in \Theta \\ t \in \mathcal{T}}} |\mathbb{V}_n^*(\theta, t)| - s_n b_n C_v$ for some $c > 0$ depending upon $\delta$.
    On the other hand, Assumptions \ref{A:one}(1) and \ref{A:five}(1) imply the conditional tightness statement
    \begin{align}
        \label{E:binf1}
        \PP{n}{\sup_{\substack{\theta \in \Theta , t \in \mathcal{T} }} |\mathbb{V}_n^*(\theta, t)| > a_n } = o_p(1)
    \end{align}
    for every sequence $a_n\to\infty$.
    Let $c' < c$ be positive. By the union bound,
    \small
    \begin{align}
        &\PP{n}{\inf_{\theta \in \Theta} \inf_{\theta' \in B_\Theta(\theta, b_n)} \sup_{t \in \mathcal{T}} (s_n v(\theta', t) + \mathbb{V}_n^*(\theta, t)) \neq \inf_{\theta \in S_\delta} \inf_{\theta' \in B_\Theta(\theta, b_n)} \sup_{t \in \mathcal{T}} (s_n v(\theta', t) + \mathbb{V}_n^*(\theta, t)) } \nonumber \\
        & \quad \le \PP{n}{\inf_{\theta \in \Theta} \inf_{\theta' \in B_\Theta(\theta, b_n)} \sup_{t \in \mathcal{T}} (s_n v(\theta', t) + \mathbb{V}_n^*(\theta, t) ) > s_n c'} \nonumber \\
        & \quad \quad + \PP{n}{\inf_{\theta \in S_\delta^c} \inf_{\theta' \in B_\Theta(\theta, b_n)}\sup_{t \in \mathcal{T}} (s_n v(\theta', t) + \mathbb{V}_n^*(\theta, t) ) \le s_n c'}\label{E:binf2}
    \end{align}
    \normalsize
    By \eqref{E:crudeupper}, \eqref{E:auxl2}, and \eqref{E:binf1}, \eqref{E:binf2} converges in probability to $0$ and allows us to replace the infimum over $\Theta$ in \eqref{E:auxl1} with an infimum over $S_\delta$. Uniformly over $\theta\in S_\delta$, Assumptions \ref{A:two}, \ref{A:three}, and \ref{A:five}(2) then give
    \begin{align}
        \label{E:auxl0}
        &\inf_{\theta' \in B_\Theta(\theta, b_n)} \sup_{t \in \mathcal{T}} \big(s_n v(\theta',t)+\mathbb{V}_n^*(\theta,t)\big) \nonumber \\
        &\quad = \inf_{\theta' \in B_\Theta(\theta, b_n)} \sup_{t \in \mathcal{T}} \big(s_n v(\theta,t)+\mathbb{V}_n^*(\theta,t)+s_n D_{\theta;t}v(\theta'-\theta)\big)+O(s_n b_n^{\alpha_v}) \nonumber \\
        &\quad = \sup_{t \in \mathcal{T}} \big(s_n v(\theta,t)+\mathbb{V}_n^*(\theta,t)+\lambda_n\psi_n(\theta,t)\big)+O(\lambda_n b_n^{\alpha_v-1}),
    \end{align}
    where the error term is uniform over $\theta\in S_\delta$ and the second equality is justified below. In cases (i) and (iii), $\lambda_nb_n^{\alpha_v-1}=o(1)$, while the Taylor remainder is absent in case (ii). Replacing $v$ by $v_n$ in \eqref{E:auxl0} perturbs it by at most $s_n\sup_{\theta,t}|v_n(\theta,t)-v(\theta,t)|=o_p(1)$, uniformly in $\theta$, because $s_n=o(r_n)$. Sion's minimax theorem applies to the $v_n$ version because $B_\Theta(\theta,b_n)$ is convex on $S_\delta$ and Assumptions~\ref{A:two}(2) and~\ref{A:five}(2) make the integrand concave in $t$. Replacing $v_n$ by $v$ then yields
    \begin{align}
        Z^*_{n,\mathrm{loc}} = \inf_{\theta \in S_\delta} \sup_{t \in \mathcal{T}} ( s_n v(\theta, t) + \mathbb{V}_n^*(\theta, t) + \lambda_n \psi_n(\theta, t)) + o_{p^*}(1). \label{E:auxl11}
    \end{align}

    Fix $\ve > 0$. By Assumptions \ref{A:one}(2) (equicontinuity) and \ref{A:five}(1), there is a $\delta_\ve > 0$ which has the property that
    \begin{align} \label{E:binf3}
        \PP{n}{\sup_{\rho((\theta, t),(\theta', t')) < \delta_\ve }|\mathbb{V}_n^*(\theta, t) - \mathbb{V}_n^*(\theta', t')| > \ve} \le \ve + o_p(1).
    \end{align}

    \textbf{A lower bound under Assumption \ref{A:four}(3):} suppose that Assumption \ref{A:four}(3) holds. This includes case (ii). There, near each $\theta\in\Theta_{\mathrm{I}}$, $v(\cdot,t)=v(\theta,t)+D_{\theta;t}v(\cdot-\theta)$ is convex, and the argument following \eqref{E:ps30} in the proof of Theorem~\ref{T:one} shows that Assumption~\ref{A:four}(3) holds. Let $S'_{\delta_\ve}$ be the neighborhood of Lemma~\ref{L:contactlb}, applied with $\mathcal{K}'=\widetilde{\mathcal{K}}$, which equals $\mathcal{K}$ and is therefore nonempty on $\Theta_{\mathrm{I}}$. As $\Theta_{\mathrm{I}}$ is compact, some $\eta$-enlargement of $\Theta_{\mathrm{I}}$ lies in $S'_{\delta_\ve}$. By \eqref{E:binf2}, the infimum over $\theta$ in \eqref{E:auxl1} may be restricted to the $\eta/2$-enlargement, so that every $\theta'\in B_\Theta(\theta,b_n)$ lies in $S'_{\delta_\ve}$ once $b_n<\eta/2$. For such $\theta'$ and $t'\in\widetilde{\mathcal{K}}(\theta')$, $s_nv(\theta',t')\ge0$. The supremum over $t$ in \eqref{E:auxl1} is therefore at least $\sup_{t'\in\widetilde{\mathcal{K}}(\theta')}\mathbb{V}_n^*(\theta,t')\ge\sup_{t'\in\widetilde{\mathcal{K}}(\theta')}\mathbb{V}_n^*(\theta',t')-\omega_n^*(b_n)$. Lemma~\ref{L:contactlb}, \eqref{E:binf3}, and \eqref{E:pseudometricdelta} then imply
    \begin{align} \label{E:binf5}
        \PP{n}{Z^*_{n,\mathrm{loc}} < \inf_{\theta \in \Theta_{\mathrm{I}}} \sup_{t \in \mathcal{K}(\theta)} \mathbb{V}_n^*(\theta, t) - 2 \ve} \le \ve + o_p(1).
    \end{align}

    \textbf{A lower bound without Assumption \ref{A:four}(3):} we next consider case (iii), in which Assumption \ref{A:four}(1) holds. Let $c_v > C_v$. Arguing as in \eqref{E:auxl2} with Assumption \ref{A:five}(3) supplies the bound
    \begin{align} \label{E:binf22}
        \inf_{\theta: \ell(\theta) > c_v s_n^{-1} \lambda_n} \inf_{\theta' \in B_\Theta(\theta, b_n) } \sup_{t \in \mathcal{T}} \big( s_n v(\theta', t) + \mathbb{V}_n^*(\theta, t)  \big) \ge c_v \lambda_n - s_n b_n C_v - \sup_{\substack{\theta \in \Theta,  t \in \mathcal{T}}}| \mathbb{V}_n^*(\theta, t)| ,
    \end{align}
    which diverges in the sense of \eqref{E:auxl2} because $s_nb_n=\lambda_n\to\infty$ and $c_v>C_v$. Arguing as in \eqref{E:binf2} allows the outer infimum in \eqref{E:auxl1} to be replaced with an infimum over those $\theta$ having $\ell(\theta) \le c_v b_n$. From the definition of $f_\ell$, every such $\theta$ is such that $d_{\mathfrak{B}}(\theta, \Theta_{\mathrm{I}}) \le f_\ell(c_v b_n)$, and thus
    \begin{align} \label{E:zneq11}
        \PP{n}{Z^*_{n,\mathrm{loc}} \ge \inf_{\theta \in \Theta_{\mathrm{I}}} \inf_{\theta' \in B_\Theta(\theta,  f_\ell(c_v b_n))} \inf_{\theta'' \in B_\Theta(\theta', b_n)} \sup_{t \in \mathcal{T}} \big( s_n v(\theta'', t) + \mathbb{V}_n^*(\theta', t) \big)- \ve} = 1 - o_p(1).
    \end{align}
    The asymptotic equicontinuity condition \eqref{E:binf3} allows for $\mathbb{V}_n^*(\theta', t)$ to be replaced with $\mathbb{V}_n^*(\theta, t)$ in \eqref{E:zneq11}. Clearly, $\mathcal{T}$ can also be replaced with $\mathcal{K}(\theta)$ while preserving the inequality. Finally, by the triangle inequality, \eqref{E:zneq11} implies that
    \begin{align} \label{E:zneq111}
        \PP{n}{Z^*_{n,\mathrm{loc}} \ge \inf_{\theta \in \Theta_{\mathrm{I}}} \inf_{\theta'' \in B_\Theta(\theta,  f_\ell(c_v b_n) + b_n)} \sup_{t \in \mathcal{K}(\theta) } \big( s_n v(\theta'', t) + \mathbb{V}_n^*(\theta, t) \big)- \ve} \ge 1 - \ve - o_p(1).
    \end{align}

    For $t \in \mathcal{K}(\theta)$, $v(\theta,t)=0$ and $D_{\theta;t}v(\theta''-\theta)\ge0$ by Lemma~\ref{L:gsc}. Because $f_\ell(\delta)=O(\delta^{\alpha_\ell})$, $s_nf_\ell(c_vb_n)^{\alpha_v}=O(s_nb_n^{\alpha_\ell\alpha_v})=o(1)$ by the choice of $b_n$. Moreover, $s_nb_n^{\alpha_v}=\lambda_nb_n^{\alpha_v-1}=o(1)$. Hence, Assumption \ref{A:two}(1) implies that the term $v(\theta'', t)$ in \eqref{E:zneq111} is at least $v(\theta, t) + D_{\theta; t} v(\theta'' - \theta) - O( f_\ell(c_v b_n)^{\alpha_v} + b_n^{\alpha_v})$, where the final error term is $o(s_n^{-1})$. Therefore, $v(\theta'', t)$ may be suppressed from \eqref{E:zneq111}, which yields the desired lower bound:
    \begin{align} \label{E:zneqn12}
        \PP{n}{Z^*_{n,\mathrm{loc}} \ge \inf_{\theta \in \Theta_{\mathrm{I}}} \sup_{t \in \mathcal{K}(\theta)} \mathbb{V}_n^*(\theta, t) - \ve} \ge 1 - \ve - o_p(1).
    \end{align}

    \textbf{An upper bound for $\mathbf{Z^*_{n,\mathrm{loc}}}$:} we now establish an upper bound for $Z^*_{n,\mathrm{loc}}$ independently of Assumption \ref{A:four}(3). By \eqref{E:auxl11}, $Z^*_{n,\mathrm{loc}}$ is bounded above up to an $o_{p^*}(1)$ term by
    \begin{align}
        \label{E:upperbb1}
        \inf_{\theta \in \Theta_{\mathrm{I}}} \sup_{t \in \mathcal{T}} (s_n v(\theta,t) + \mathbb{V}_n^*(\theta, t) + \lambda_n \psi_n(\theta, t)).
    \end{align}
    Note that \eqref{E:upperbb1} is at least $\inf_{\theta \in \Theta_{\mathrm{I}}} \sup_{t \in \mathcal{K}(\theta)} \mathbb{V}_n^*(\theta, t) $, which is $O_{p^*}(1)$. Let $c > 0$ be any constant, and let $g_n$ be the function $(\theta, t) \mapsto v(\theta, t) + \psi_n(\theta, t)$, which is upper semicontinuous (Assumption \ref{A:three}(3) and Lemma \ref{L:gsc}). For all $\theta \in \Theta_{\mathrm{I}}$, $v(\theta, t)$ is nonpositive, and thus $s_n v(\theta,t) \le \lambda_n v(\theta, t)$ and $g_n \le 0$. Hence, $\sup_{\theta \in \Theta_{\mathrm{I}}} \sup_{t: g_n(\theta, t) < -c} (s_n v(\theta, t) + \mathbb{V}_n^*(\theta, t) + \lambda_n \psi_n(\theta, t))$ is at most $\sup_{\theta \in \Theta_{\mathrm{I}}} \sup_{t \in \mathcal{T}} \mathbb{V}_n^*(\theta, t) - \lambda_n c - o_{p^*}(1)$. Let $c' \in (0,c)$. By the union bound, this allows \eqref{E:upperbb1} to be reformulated in terms of $g_n$:
    \small
    \begin{align}
         &\PP{n}{\inf_{\theta \in \Theta_{\mathrm{I}}} \sup_{t \in \mathcal{T}} (s_n v(\theta,t) + \mathbb{V}_n^*(\theta, t) + \lambda_n \psi_n(\theta, t)) \neq  \inf_{\theta \in \Theta_{\mathrm{I}}} \sup_{t: g_n(\theta, t) \ge -c} (s_n v(\theta, t) + \mathbb{V}_n^*(\theta, t) + \lambda_n \psi_n(\theta, t)) } \nonumber \\
         &\quad \le \PP{n}{\inf_{\theta \in \Theta_{\mathrm{I}}} \sup_{t \in \mathcal{T}} (s_n v(\theta,t) + \mathbb{V}_n^*(\theta, t) + \lambda_n \psi_n(\theta, t)) < \sup_{\theta \in \Theta_{\mathrm{I}}} \sup_{t \in \mathcal{T}} \mathbb{V}_n^*(\theta, t) - \lambda_n c'} \nonumber\\
         & \quad \quad + \PP{n}{\sup_{\theta \in \Theta_{\mathrm{I}}} \sup_{t: g_n(\theta, t) < -c} (s_n v(\theta, t) + \mathbb{V}_n^*(\theta, t) + \lambda_n \psi_n(\theta, t)) \ge \sup_{\theta \in \Theta_{\mathrm{I}}} \sup_{t \in \mathcal{T}} \mathbb{V}_n^*(\theta, t) - \lambda_n c'}.\label{E:binf6}
    \end{align}
    \normalsize
    By \eqref{E:binf1} and the bound noted after \eqref{E:upperbb1}, \eqref{E:binf6} converges in probability to $0$.
    By \eqref{E:binf3} and Lemma \ref{L:conv1}, there is a finite subset $\{\theta_1, \ldots, \theta_J\} \subset \Theta_{\mathrm{I}}$ which has the property that
    \begin{align} \label{E:binf7}
        \PP{n}{\inf_{\theta \in \Theta_{\mathrm{I}}} \sup_{t \in \mathcal{K}(\theta)} \mathbb{V}_n^*(\theta, t) < \inf_{1 \le j \le J} \sup_{t \in \mathcal{K}(\theta_j)} \mathbb{V}_n^*(\theta_j, t) - \ve} \le \ve + o_p(1).
    \end{align}
    Using the upper-envelope construction for $b_n$ from the proof of Lemma \ref{L:gn}, \eqref{E:binf3} and \eqref{E:monodec} in the proof of Lemma \ref{L:monotonedecreasing} imply that a small enough choice of $c$ guarantees that
    \begin{align} \label{E:binf8}
         \PP{n}{ \inf_{1 \le j \le J} \sup_{t \in \mathcal{K}(\theta_j)} \mathbb{V}_n^*(\theta_j, t) < \inf_{1 \le j \le J} \sup_{t: g_n(\theta_j, t) \ge -c} \mathbb{V}_n^*(\theta_j, t)- \ve} \le \ve + o_p(1).
    \end{align}
    Recall that $g_n$ is nonpositive on $\Theta_{\mathrm{I}} \times \mathcal{T}.$
    Hence, \eqref{E:binf6}, \eqref{E:binf7}, and \eqref{E:binf8} imply the following upper bound for $Z^*_{n,\mathrm{loc}}$:
    \begin{align} \label{E:binf9}
         \PP{n}{ \inf_{\theta \in \Theta_{\mathrm{I}}} \sup_{t \in \mathcal{K}(\theta)} \mathbb{V}_n^*(\theta, t) < Z^*_{n,\mathrm{loc}} - 3 \ve} \le 2\ve + o_p(1).
    \end{align}
    Because $\ve$ was arbitrarily chosen, the applicable lower bound, \eqref{E:binf5} in cases (i) and (ii) and \eqref{E:zneqn12} in case (iii), together with \eqref{E:binf9} and Assumption \ref{A:five}(1) imply \eqref{E:bootstrapone} for $Z^*_{n,\mathrm{loc}}$.

    \textbf{Equivalence of $\mathbf{Z^*_{n,\mathrm{loc}}}$ and $\mathbf{Z_n^*}$:} let $Z_n^*$ denote the left side of \eqref{E:bootstrapone}. Replacing $\mathbb{V}_n^*(\theta,t)$ by $\mathbb{V}_n^*(\theta',t)$ in $Z^*_{n,\mathrm{loc}}$ changes it by at most $\omega_n^*(b_n)$ and, because $\theta\in B_\Theta(\theta,b_n)$, turns it into $Z_n^*$. Because $\omega_n^*(b_n)=o_{p^*}(1)$ by \eqref{E:binf3} and \eqref{E:pseudometricdelta}, \eqref{E:bootstrapone} follows.

\subsection{Localized bootstrap under weaker identification (Proposition~\ref{ad:one})}

\begin{prop}[Localized bootstrap under weaker identification] \label{ad:one}
    Suppose that Assumptions \ref{A:zero}, \ref{A:one}, \ref{A:two}, \ref{A:three}, and \ref{A:five} hold, that $\Theta_{\mathrm{I}}$ is nonempty, and that $\alpha_v<\infty$. Let $(\lambda_n)$, $(\eta_n)$, and $(b_n)$ be sequences of tuning parameters chosen such that $\lambda_n,\eta_n \ra \infty$, $b_n \ra 0$, $\lambda_n = o(b_n^{1-\alpha_v})$, $\lambda_n b_n^{-1} = o(\eta_n)$, $\eta_n = o(r_n)$, and $\lambda_n b_n^{-1} f_\ell(c_v \eta_n^{-1} \lambda_n)^{\alpha_v} = o(1)$ for some $c_v > C_v$. Then,
    \begin{align} \label{E:boostraptwo}
        \inf_{\theta \in \Theta} \inf_{\theta' \in B_\Theta(\theta, b_n)}\bigl(\eta_n \ell_n(\theta) + \sup_{t \in \mathcal{T}} (\lambda_n b_n^{-1} v_n(\theta', t) + \mathbb{V}_n^*(\theta, t))\bigr) & = \inf_{\theta \in \Theta_{\mathrm{I}}} \sup_{t \in \mathcal{K}(\theta)} \mathbb{V}_n^*(\theta, t) + o_{p^*}(1)\nonumber \\
        & \overset{\mathrm{P}}{\rs} \inf_{\theta \in \Theta_{\mathrm{I}}} \sup_{t \in \mathcal{K}(\theta)} \mathbb{G}_n(\theta, t).
    \end{align}
\end{prop}

\begin{proof}
    Let $Z_n^*$ denote the left hand side of \eqref{E:boostraptwo}. The proof largely follows that of Theorem \ref{T:two}. Let $s_n=\lambda_nb_n^{-1}$, and let $\psi_n$ be defined as therein. Note that $\lambda_n b_n^{\alpha_v - 1} = o(1)$ and $\lambda_n b_n^{-1} r_n^{-1} = o(1)$ by the imposed rate conditions. Hence, replicating \eqref{E:binf2}--\eqref{E:auxl11} leads to the asymptotic equivalence
    \begin{align} \label{E:auxm1}
        Z_n^* = \inf_{\theta \in \Theta} \big(\eta_n \ell(\theta) + \sup_{t \in \mathcal{T}}  (s_n v(\theta, t) + \mathbb{V}_n^*(\theta, t) + \lambda_n \psi_n(\theta, t))\big) + o_{p^*}(1).
    \end{align}
    The upper bound \eqref{E:crudeupper} still applies for $Z_n^*$ by definition of $\Theta_{\mathrm{I}}.$

    \textbf{A lower bound for $\mathbf{Z_n^*}$:} by supposition, $c_v > C_v$. As in \eqref{E:binf22},
    \small
    \begin{align}
        &\inf_{\theta: \ell(\theta) > c_v \eta_n^{-1} \lambda_n } \inf_{\theta' \in B_\Theta(\theta, b_n)} \big(\eta_n \ell(\theta) + \sup_{t \in \mathcal{T}}  (s_n v(\theta', t) + \mathbb{V}_n^*(\theta, t)) \big) \nonumber \\
        &\qquad \ge c_v (\lambda_n + s_n \eta_n^{-1}\lambda_n) - \sup_{\theta \in \Theta, t \in \mathcal{T}} | \mathbb{V}_n^*(\theta, t)| - \lambda_n C_v, \label{E:add1bd}
    \end{align}
    \normalsize
    which diverges in probability. A comparison with the upper bound \eqref{E:crudeupper} and nonnegativity of $\ell$ lead to the following analogue of \eqref{E:zneq111}:
    \begin{align} \label{E:upperbd33}
        \PP{n}{Z_n^* \ge \inf_{\theta \in \Theta_{\mathrm{I}}} \inf_{\theta' \in B_\Theta(\theta,  f_\ell(c_v \eta_n^{-1} \lambda_n) + b_n)} \sup_{t \in \mathcal{K}(\theta)} \big( s_n v(\theta', t) + \mathbb{V}_n^*(\theta, t) \big) - \ve} \ge 1 - \ve - o_p(1).
    \end{align}
     By Assumption \ref{A:two}(1), $v(\theta',t)$ is at least $- O(f_\ell(c_v \eta_n^{-1} \lambda_n)^{\alpha_v} + b_n^{\alpha_v})$ for $\theta \in \Theta_{\mathrm{I}}$ and $t \in \mathcal{K}(\theta)$. The assumed rate conditions guarantee that the final approximation error is $o(s_n^{-1})$. Hence, the term $v(\theta', t)$  may be omitted from \eqref{E:upperbd33}, leading to the desired lower bound for $Z_n^*$.

    \textbf{An upper bound for $\mathbf{Z_n^*}$:} an initial upper bound for $Z_n^*$ may be produced by replacing the outer infimum over $\Theta$ in \eqref{E:auxm1} with an infimum over $\Theta_{\mathrm{I}}$. By definition, $\ell$ vanishes over $\Theta_{\mathrm{I}}$, so that this upper bound is exactly \eqref{E:upperbb1}. The arguments producing \eqref{E:binf9} apply and supply the complementary upper bound. Assumption~\ref{A:five}(1) gives the weak convergence in the second line of \eqref{E:boostraptwo}.

\end{proof}

\subsection{Proof of Proposition \ref{P:ex2}}
    We verify the conditions for Theorem \ref{T:two}. As in the proof of Proposition \ref{P:ex1}, we assume that $T_0 = 0$ in the proof. We also let $\mathcal{T}$ be the set furnished by Lemma \ref{L:duality} in \eqref{E:affineTdefn}.

    Assumptions \ref{A:zero}--\ref{A:three} are satisfied under \eqref{E:vfns} by Assumption \ref{A:ex1}. In the notation adopted by the proof of Proposition \ref{P:ex1}, let $\mathbb{V}_n^*(\theta, t) = y[t](\mathbb{W}_n^*(\theta))$. By Assumption \ref{A:ex2}(1), $\mathbb{V}_n^*(\theta, t) \overset{\mP}{\rs} y[t] (\mathbb{X}_n(\theta)) =  \mathbb{G}_n(\theta, t)$, and Assumption \ref{A:five}(1) is satisfied. Moreover, $\mathbb{V}_n^*$ is linear and weak-* continuous in $t$ for any $\theta$ and $n$, so Assumption \ref{A:five}(2) holds.
    By norm-boundedness of $\mathcal{T}$, Assumption \ref{A:ex2}(2) implies Assumption \ref{A:five}(3).

    As in the proof of Proposition~\ref{P:ex1}, $\alpha_v=\alpha_m$. Under Assumption \ref{A:ex22}, Assumption \ref{A:four}(3) holds by Lemma~\ref{L:ex22K}, and case (i) of Theorem \ref{T:two} applies. If $m$ is affine in a neighborhood of $\Theta_{\mathrm{I}}$, then $v(\cdot,t)=t(m(\cdot))$ is affine, whence convex, there, so Assumption~\ref{A:four}(2) holds and Assumption~\ref{A:four}(3) follows by the argument after \eqref{E:ps30} in the proof of Theorem~\ref{T:one}.  Theorem~\ref{T:two} then applies. Similarly, if \eqref{E:localident1} holds with $\alpha_\ell\alpha_m>1$ and $\alpha_m < \infty$, then Assumption \ref{A:four}(1) holds with $\ell(\theta) = \xi(m(\theta))$ and $\alpha_\ell$ as specified by \eqref{E:localident1}, and case (iii) of Theorem \ref{T:two} applies.

    In all cases of the proposition, Theorem \ref{T:two} implies
    \begin{align}
        \inf_{\theta \in \Theta} \sup_{t \in \mathcal{T}} \big( s_n t(m_n(\theta)) + y[t](\mathbb{W}_n^*(\theta)) \big) &= \inf_{\theta \in \Theta_{\mathrm{I}}} \sup_{y \in (-m'(\theta)(T_\theta \Theta))^\circ \cap \partial \xi(m(\theta))} y(\mathbb{W}_n^*(\theta)) + o_{p^*}(1) \nonumber \\
        &\overset{\mP}{\rs} \inf_{\theta \in \Theta_{\mathrm{I}}} \sup_{y \in (-m'(\theta)(T_\theta \Theta))^\circ \cap \partial \xi(m(\theta))} y(\mathbb{X}_n(\theta)), \label{E:ex210}
    \end{align}
    where the contact set is identified as in the proof of Proposition~\ref{P:ex1}. By the identity $t = y[t] + t(0)$, the supremum on the left of \eqref{E:ex210} may be rewritten as
    \begin{align} \label{E:cld9}
    \sup_{t \in \mathcal{T}} s_n \big(t(m_n(\theta)) + y[t](s_n^{-1} \mathbb{W}_n^*(\theta))\big) =
    \sup_{t \in \mathcal{T}} s_n t\big(m_n(\theta) + s_n^{-1} \mathbb{W}_n^*(\theta)\big).
    \end{align}
    As $n$ diverges, $m_n(\theta)$ concentrates uniformly on $m(\Theta)$ over $\theta\in\Theta$, and the distribution of $s_n^{-1} \mathbb{W}_n^*(\theta)$ concentrates at $0$ with probability approaching $1$.  Because the function $x \mapsto \sup_{t \in \mathcal{T}} t(x)$ agrees with $\xi$ on an open neighborhood of $m(\Theta)$, \eqref{E:cld9} agrees with $s_n\xi\big(m_n(\theta) + s_n^{-1} \mathbb{W}_n^*(\theta)\big)$ with probability approaching $1$, as desired.

\subsection{Proofs of uniform results}

\subsubsection{Proof of Theorem \ref{P:unif1}}
    Let
    \[
        Z_{n,R,P} = \inf_{\theta \in R} \sup_{t \in \mathcal{T}}\big(s_n v_{R,P}(\theta,t) + \mathbb{V}_{n,R,P}(\theta, t)\big).
    \]
     By Assumption \ref{AU:2}(1), $\mathbb{V}_{n,R,P} = O_p(1)$ uniformly in $(R,P)$, and thus $s_n v_{n,R,P} (\theta, t)$ is equal to $s_n v_{R,P}(\theta, t)$ up to some uniformly $o_p(1)$ term. Hence, $Z_{n,R,P}$ is the right side of \eqref{E:upperbdprop} up to such a term.
     Assumption~\ref{AU:1}(2) gives $\overline{f}_\ell(\delta)=o(1)$, as noted following the assumption.

    Let $c_n\ra\infty$ with $c_n=o(s_n)$ and, if $\alpha_v<\infty$, $s_n\overline f_\ell(c_n/s_n)^{\alpha_v}\ra0$. Such a sequence exists because $\overline f_\ell(\delta)=O(\delta^{\alpha_\ell})$ and $\alpha_\ell\alpha_v>1$. For example, $c_n=s_n^{\kappa}$ with $\kappa>0$ sufficiently small suffices. Because $\sup_t v_{R,P}(\theta,t)=\ell_{R,P}(\theta)$, the infimum in $Z_{n,R,P}$ over $\theta$ with $\ell_{R,P}(\theta)>c_n/s_n$ is at least $c_n-\sup_{\theta,t}|\mathbb{V}_{n,R,P}(\theta,t)|$, which diverges uniformly, while restricting the infimum in $Z_{n,R,P}$ to $\theta\in\Theta_{\mathrm{I}}(R,P)$ shows that $Z_{n,R,P}\le O_p(1)$ uniformly. With probability approaching one uniformly, the infimum may therefore be restricted to $\theta$ with $\ell_{R,P}(\theta)\le c_n/s_n$. Each such $\theta$ lies within $d_n\equiv\overline f_\ell(c_n/s_n)+s_n^{-1}$ of some $\theta_0\in\Theta_{\mathrm{I}}(R,P)$, where $s_n^{-1}$ is a slack term, and replacing $\mathbb{V}_{n,R,P}(\theta,t)$ by $\mathbb{V}_{n,R,P}(\theta_0,t)$ perturbs $Z_{n,R,P}$ by at most the modulus of $\mathbb{V}_{n,R,P}$ over $d_n$-balls, which is uniformly $o_p(1)$ by Assumptions \ref{AU:2} and \ref{AU:4}. Therefore, uniformly in $(R,P)\in\mathcal{S}$,
    \begin{align} \label{E:binf34}
        Z_{n,R,P} = \inf_{\theta_0 \in \Theta_{\mathrm{I}}(R,P)}\inf_{\theta \in B_R(\theta_0, d_n)} \sup_{t \in \mathcal{T}} \big(s_n v_{R,P}(\theta,t) + \mathbb{V}_{n,R,P}(\theta_0, t) \big) + o_p(1).
    \end{align}
    Here $d_n\ra0$, and $s_nd_n^{\alpha_v}=o(1)$ in the finite-smoothness case, while the Taylor remainder is identically zero in the linearity case. Lemma \ref{L:unifL} shows that the right side of \eqref{E:binf34} is a uniform upper bound for $T_{n,R,P}$, which gives \eqref{E:upperbdprop}. The bootstrap statement follows because the right side of \eqref{E:upperbdprop} differs by a uniformly $o_p(1)$ term from a functional of $\mathbb{V}_{n,R,P}$ that is $1$-Lipschitz on $\ell^\infty(R\times\mathcal{T})$, and the same holds with $\mathbb{V}_{n,R,P}^*$ in place of $\mathbb{V}_{n,R,P}$.

\begin{lemma} \label{L:unifL}
    Let Assumptions \ref{AU:1}(1), \ref{AU:2}, \ref{AU:3}, and \ref{AU:4} hold. Let $(s_n)$ and $(d_n)$ be any sequences such that $s_n \ra \infty$, $d_n \ra 0$, $s_n = o(r_n)$, and $s_n d_n^{\alpha_v} = o(1).$ Then, uniformly in $(R,P) \in \mathcal{S}$,
    \begin{align*}
        T_{n,R,P} \le \inf_{\theta \in \Theta_{\mathrm{I}}(R,P)} \inf_{\theta' \in B_R(\theta, d_n)}  \sup_{t \in \mathcal{T}} ( s_n v_{R,P}(\theta' , t) + \mathbb{V}_{n,R,P}(\theta, t))+ o_p(1)
    \end{align*}
\end{lemma}
\begin{proof}
    For all $n$, let $k_n = s_n d_n / r_n$. If $\alpha_v<\infty$, then $r_n k_n^{\alpha_v} =  s_n d_n^{\alpha_v} (s_n / r_n)^{\alpha_v - 1} = o(1)$. If $\alpha_v=\infty$, the corresponding Taylor remainder is identically zero. Hence,
    \begin{align}
        T_{n,R,P} & \le r_n\inf_{\theta \in \Theta_{\mathrm{I}}(R,P)} \inf_{\theta' \in B_R(\theta, k_n)} \sup_{t \in \mathcal{T}} v_{n,R,P}(\theta', t) \nonumber \\
        & = \inf_{\theta \in \Theta_{\mathrm{I}}(R,P)} \inf_{\theta' \in B_R(\theta, k_n)}  \sup_{t \in \mathcal{T}} (r_n v_{R,P}(\theta', t) + \mathbb{V}_{n,R,P}(\theta', t)) \nonumber \\
        & = \inf_{\theta \in \Theta_{\mathrm{I}}(R,P)} \inf_{\theta' \in B_R(\theta, k_n)}   \sup_{t \in \mathcal{T}}( r_n v_{R,P}(\theta, t) + r_n D_{\theta; t} v_{R,P}(\theta' - \theta) + \mathbb{V}_{n,R,P}(\theta, t)) \nonumber \\
        & \qquad + O_p(r_n k_n^{\alpha_v} ) + o_p(1), \label{E:unif11}
    \end{align}
    where the $o_p(1)$ term is uniform in $(R,P) \in \mathcal{S}$, and the third line is a consequence of Assumptions \ref{AU:2}(2), \ref{AU:4} (asymptotic equicontinuity with respect to $d_{\mathfrak{B}}$, with equicontinuity of the limit $\mathbb{G}_{n,R,P}$ transferring to $\mathbb{V}_{n,R,P}$ through the uniform convergence in Assumption~\ref{AU:2}(1)), and \ref{AU:3}. Because $B_R(\theta, d_n)$ is convex for all $n$ sufficiently large (Assumption \ref{AU:1}(1), as $d_n\ra0$) and $k_n<d_n$, $\theta + (k_n / d_n)(\theta'' - \theta)$ lies in $B_R(\theta, k_n)$ whenever $\theta'' \in B_R(\theta, d_n)$. Restricting the infimum over $\theta'$ in \eqref{E:unif11} to such points, positive homogeneity gives $r_nD_{\theta;t}v_{R,P}(\theta'-\theta)=s_nD_{\theta;t}v_{R,P}(\theta''-\theta)$. For $\theta \in \Theta_{\mathrm{I}}(R,P)$, $v_{R,P}(\theta, t)\le0$, so $r_nv_{R,P}(\theta,t)\le s_nv_{R,P}(\theta,t)$ once $s_n\le r_n$. An upper bound for \eqref{E:unif11}, up to a uniform $o_p(1)$ term, is therefore
    \begin{align*}
        & \inf_{\theta \in \Theta_{\mathrm{I}}(R,P)} \inf_{\theta'' \in B_R(\theta, d_n)}  \sup_{t \in \mathcal{T}}  ( s_n v_{R,P}(\theta, t)  + \mathbb{V}_{n,R,P}(\theta, t) +  s_n  D_{\theta; t} v_{R,P}(\theta'' - \theta) ) \\
        & \quad = \inf_{\theta \in \Theta_{\mathrm{I}}(R,P)} \inf_{\theta'' \in B_R(\theta, d_n)}  \sup_{t \in \mathcal{T}} ( s_n v_{R,P}(\theta'' , t) + \mathbb{V}_{n,R,P}(\theta, t))+ O(s_n d_n^{\alpha_v}).
    \end{align*}

\end{proof}

\subsubsection{Uniform bound under weaker identification (Proposition~\ref{Add:2})}

\begin{prop}[Uniform bound under weaker identification] \label{Add:2}
    Let Assumptions \ref{AU:1}, \ref{AU:2}, \ref{AU:3}, and \ref{AU:4} hold, suppose that $\alpha_v<\infty$, and suppose that
    \[
    \overline{C}_v \equiv \sup_{(R,P) \in \mathcal{S}} \sup_{\theta, \theta' \in R} \sup_{t \in \mathcal{T}} \frac{| v_{R,P}(\theta, t) - v_{R,P}(\theta', t)|}{d_{\mathfrak{B}}(\theta, \theta')} < \infty.
    \]
    Let $(\lambda_n)$, $(\eta_n)$, and $(b_n)$ be sequences of tuning parameters chosen such that $\lambda_n,\eta_n \ra \infty$, $b_n \ra 0$, $\lambda_n = o(b_n^{1-\alpha_v})$, $\lambda_n b_n^{-1} = o(\eta_n)$, $\eta_n = o(r_n)$, and $\lambda_n b_n^{-1} \overline{f}_\ell(\overline{c}_v \eta_n^{-1} \lambda_n)^{\alpha_v} = o(1)$ for some $\overline{c}_v> \overline{C}_v$. Then, uniformly in $(R,P) \in \mathcal{S},$
    \begin{align}
        \label{E:add2bd}
        T_{n,R,P} \le \inf_{\theta \in R} \inf_{\theta' \in B_R(\theta, b_n)} \big( \eta_n \ell_{n,R,P}(\theta) + \sup_{t \in \mathcal{T}} ( \lambda_n b_n^{-1} v_{n,R,P}(\theta', t) + \mathbb{V}_{n,R,P} (\theta, t))\big) + o_p(1).
    \end{align}
    The bootstrap statement of Theorem~\ref{P:unif1} holds with \eqref{E:add2bd} in place of \eqref{E:upperbdprop}.
\end{prop}

\begin{proof}[Proof of Proposition \ref{Add:2}]
    Let $s_n = \lambda_n b_n^{-1}$ and $Z_{n,R,P}$ denote the right side of \eqref{E:add2bd}. By the assumed rate conditions, $\eta_n$ and $\lambda_n b_n^{-1}$ are $o(r_n)$, and the terms $\ell_{n,R,P}$ and $v_{n,R,P}$ in \eqref{E:add2bd} may be replaced with $\ell_{R,P}$ and $v_{R,P}$, respectively, with the imposition of a uniformly negligible error. As in \eqref{E:add1bd},
    \begin{align*}
        &\inf_{\theta: \ell_{R,P}(\theta) > \overline{c}_v \eta_n^{-1}\lambda_n} \inf_{\theta' \in B_R(\theta, b_n)} \big( \eta_n \ell_{R,P}(\theta) + \sup_{t \in \mathcal{T}} ( s_n v_{R,P}(\theta', t) + \mathbb{V}_{n,R,P} (\theta, t)) \big)
    \end{align*}
    is at least $\overline{c}_v(\lambda_n + s_n \eta_n^{-1}\lambda_n ) - \lambda_n \overline{C}_v - O_p(1)$, where the $O_p(1)$ term is uniform. Arguing as for \eqref{E:binf34}, reducing the infima in \eqref{E:add2bd} to $\theta \in \Theta_{\mathrm{I}}(R,P)$ and $\theta' \in B_R(\theta,  \overline{f}_\ell(\overline{c}_v \eta_n^{-1} \lambda_n) + 2b_n)$ leads to an asymptotically uniform lower bound for $Z_{n,R,P}$. By nonnegativity, the term $\eta_n \ell_{R,P}$ may be omitted from this lower bound. By the assumed  rate conditions, $s_n \overline{f}_\ell(\overline{c}_v \eta_n^{-1} \lambda_n)^{\alpha_v}$ and $s_n b_n^{\alpha_v}$ are $o(1)$. Hence, Lemma \ref{L:unifL} establishes that this lower bound for $Z_{n,R,P}$ is in fact an upper bound for $T_{n,R,P}$.  The bootstrap statement follows as in the proof of Theorem~\ref{P:unif1}.
\end{proof}

\subsubsection{Proof of Proposition \ref{P:unifbound}}
    Apply Lemma \ref{L:duality} on the common bounded convex neighborhood supplied by Assumption \ref{A:uniformbound}(2), and let $\mathcal{T}$ be the resulting weak-* compact set of supporting affine functionals. We verify the conditions of Theorem \ref{P:unif1} with $v_{R,P}(\theta, t) = t (m_{R,P}(\theta)) - T_0(R,P)$ and empirical analogues $v_{n,R,P}(\theta, t) = t(m_{n,R,P}(\theta)) - T_0(R,P)$. Assumption \ref{A:uniformbound}(4) gives $\mathbb W_{n,R,P}=O_p(1)$ uniformly, so with probability tending to one uniformly in $(R,P)$, the arguments $m_{n,R,P}(\theta)$ and $m_{n,R,P}(\theta)+s_n^{-1}\mathbb W_{n,R,P}(\theta)$ lie in this common neighborhood. Hence $\sup_{t\in\mathcal T}v_{n,R,P}(\theta,t)=\xi(m_{n,R,P}(\theta))-T_0(R,P)$, and the analogous perturbed identity gives \eqref{E:unifub}. Under this arrangement, Assumption \ref{AU:1} is a consequence of Assumption \ref{A:uniformbound}(1) and the assumed conditions on $\overline{f}_\ell$. Assumption \ref{AU:2}(1) follows from Assumption \ref{A:uniformbound}(4) and the norm boundedness of $\mathcal{T}$.

    Lemma \ref{L:quantitative} below defines a sequence $(X^{n,R,P}_j)_{j \in \N} \subset \ell^\infty(\Theta, \mathfrak{X})$ for every $n,R,P$ such that $\{X^{n,R,P}_j: n \in \N, (R,P) \in \mathcal{S}\}$ is uniformly equicontinuous with respect to $d_{\mathfrak{B}}$ for every $j$ and the union of the $\delta$-balls around the points $X^{n,R,P}_j$ for $j = 1, \ldots, J$ contains the processes $\mathbb{X}_{n,R,P}$ with probability converging uniformly to $1$ in $J$. Define, as in \eqref{E:rhomaineqn}, the family of pseudometrics
    \begin{align*}
        \rho_{n,R,P}((\theta, t), (\theta', t')) & =  \sum_{j = 1}^\infty 2^{-j} \frac{|y[t](X_j^{n,R,P}(\theta)) - y[t'](X_j^{n,R,P}(\theta'))|}{1 + |y[t](X_j^{n,R,P}(\theta)) - y[t'](X_j^{n,R,P}(\theta'))|}
    \end{align*}
    The proof of Lemma~\ref{L:pretight} uses only two properties of its approximating points: uniform equicontinuity over $\Theta$, and approximation of $\mathbb{X}_n$ in probability. By Lemma~\ref{L:quantitative}, the points $X_j^{n,R,P}$ satisfy both properties uniformly in $n$ and $(R,P)$. The identical argument therefore establishes that $\rho_{n,R,P}$ satisfies Assumption~\ref{AU:4} uniformly, and that $\mathbb{G}_{n,R,P}(\theta,t) = y[t](\mathbb{X}_{n,R,P}(\theta))$ is equicontinuous in probability with respect to $\rho_{n,R,P}$ uniformly in $(R,P) \in \mathcal{S}$, giving Assumption~\ref{AU:2}(2).

    Assumption \ref{A:uniformbound}(3) implies Assumption \ref{AU:3} with $\alpha_v = \alpha_m$. The two cases considered by Proposition \ref{P:unifbound} correspond directly to the two cases of Theorem \ref{P:unif1}.

    Finally, the bootstrap statement follows as in the proof of Theorem~\ref{P:unif1}. By Assumption \ref{A:uniformbound}(5), the same local-domain event holds for the bootstrap analogue with probability approaching one uniformly; restricting to this event changes the bounded-Lipschitz distance by $o_p(1)$. Hence, replacing $\mathbb{W}_{n,R,P}$ with its bootstrap analogue results in a uniformly consistent bootstrap estimator.

\begin{lemma} \label{L:quantitative}
    Let $\{\mathbb{X}_{n,R,P}\}$ satisfy Assumption \ref{A:uniformbound}(4). Then, for all $n \in \N$ and $(R,P) \in \mathcal{S}$, there exists a sequence $(X^{n,R,P}_j)_{j \in \N}$ such that $\{X_j^{n,R,P}\}$ is uniformly $d_{\mathfrak{B}}$-equicontinuous for all $j$ and, for all $\delta > 0$, the probability that $\mathbb{X}_{n,R,P}$ is contained in the union of the $\delta$-balls around $X^{n,R,P}_j$ for $j = 1, \ldots, J$ converges uniformly in $n,R,P$ to $1$ as $J$ diverges.
\end{lemma}
\begin{proof}[Proof of Lemma \ref{L:quantitative}]
Let $\omega_w(\delta)$ be the modulus of continuity defined in the proof of Lemma~\ref{L:pretight}. The diagonal equicontinuity construction used there, applied with the suprema over $n$ and $(R,P)\in\mathcal{S}$ supplied by Assumption~\ref{A:uniformbound}(4), yields for each $k\in\N$ a set $E_k\subset\ell^\infty(\Theta,\mathfrak{X})$ of uniformly $d_{\mathfrak{B}}$-continuous functions such that
\[
    \inf_{n\in\N}\inf_{(R,P)\in\mathcal{S}}
    \PP{P}{\mathbb{X}_{n,R,P}\in E_k}\ge 1-1/(2k).
\]
Quantitative pre-tightness gives, for every $n\in\N$ and $(R,P)\in\mathcal{S}$, a measurable set $K_k^{n,R,P}\subset \ell^\infty(\Theta,\mathfrak{X})$ with the same probability bound $1-1/(2k)$ and with covering numbers bounded uniformly in $n$ and $(R,P)$. Set $S_k^{n,R,P}=E_k\cap K_k^{n,R,P}$. Then $S_k^{n,R,P}$ consists of uniformly continuous functions, has uniformly bounded covering numbers, and satisfies
\[
    \inf_{n\in\N}\inf_{(R,P)\in\mathcal{S}}
    \PP{P}{\mathbb{X}_{n,R,P}\in S_k^{n,R,P}}\ge 1-1/k .
\]

For every $k,n,R,P$, let
\[
    N_k=\sup_{n\in\N}\sup_{(R,P)\in\mathcal{S}}
    N(S_k^{n,R,P},\norm{\cdot}_\infty,1/k)<\infty,
\]
and choose a $2/k$-net $\{X_{1,k}^{n,R,P},\ldots,X_{N_k,k}^{n,R,P}\}$ for $S_k^{n,R,P}$ with points in $S_k^{n,R,P}$, by picking one point of $S_k^{n,R,P}$ in each ball of a minimal $1/k$-cover. Pad with repeated elements if the actual covering number is smaller than $N_k$. Let $(X_j^{n,R,P})_{j\in\N}$ enumerate these nets block by block as $k$ ranges over $\N$. For each fixed $j$, the points $X_j^{n,R,P}$ belong to a common set $E_k$, and are therefore uniformly $d_{\mathfrak{B}}$-equicontinuous in $n,R,P$.

Finally, fix $\delta>0$ and $\ve>0$. Choose $k$ so large that $2/k<\delta$ and $1/k<\ve$. For all $J$ large enough to include the entire $k$th block of the enumeration, the union of the $\delta$-balls around $X_j^{n,R,P}$, $j=1,\ldots,J$, contains $S_k^{n,R,P}$ for every $n,R,P$. Therefore
\[
    \inf_{n\in\N}\inf_{(R,P)\in\mathcal{S}}
    \PP{P}{\mathbb{X}_{n,R,P}\in \bigcup_{j=1}^J
    \{w:\norm{w-X_j^{n,R,P}}_\infty<\delta\}}\ge 1-1/k>1-\ve ,
\]
which proves the desired uniform approximation property.
\end{proof}

\subsubsection{Proof of Proposition \ref{Add:unifxi}}
    Assumption \ref{AU:1}(2) is a consequence of the assumption placed on $\overline{f}_\ell$. The remainder of Assumption \ref{AU:1} and Assumptions \ref{AU:2}, \ref{AU:3}, and \ref{AU:4} are implied by Assumption \ref{A:uniformbound} as in the proof of Proposition~\ref{P:unifbound}, and the Lipschitz condition of Proposition~\ref{Add:2} holds by Assumption~\ref{A:uniformbound}(6). The dual identification used in the proof of Proposition~\ref{P:unifbound}, applied on the common local domain from Assumption~\ref{A:uniformbound}(2), collapses the mixed $\lambda_n b_n^{-1} v_{n,R,P}(\theta', t) + \mathbb{V}_{n,R,P}(\theta, t)$ term inside \eqref{E:add2bd} into the convex-loss expression $\lambda_n b_n^{-1}(\xi(m_{n,R,P}(\theta') + \lambda_n^{-1} b_n \mathbb{W}_{n,R,P}(\theta)) - T_0(R,P))$ appearing in \eqref{E:unfbd2}. Hence, Proposition \ref{Add:2} establishes \eqref{E:unfbd2}, and the bootstrap statement follows as in the proof of Theorem~\ref{P:unif1}.

\subsection{Auxiliary lemmas}

\begin{lemma} \label{L:tangentconecond}
    Let $\Theta$ be a subset of a seminormed space $\mathfrak{B}$, convex in a neighborhood of $\theta \in \Theta$. Let $D: \mathrm{cone}(\Theta - \theta) \ra \R$ be a positively homogeneous function. Then, with the conventions $0/0=0$ and $c/0=c\cdot\infty$ for $c\neq0$,
    \begin{align}
        \liminf_{\delta \downarrow 0} \inf_{\theta' \in  B_\Theta(\theta, \delta)} \frac{D (\theta' - \theta)}{\delta} = \inf_{b \in T_\theta \Theta} \frac{D(b)}{\norm{b}_{\mathfrak{B}}}. \label{E:tangentconecond}
    \end{align}
\end{lemma}

\begin{proof}
    Let $\overline{\delta}$ be small enough that $B_\Theta(\theta, \delta)$ is convex for all positive $\delta \le \overline{\delta}$. For such $\delta$, it is straightforward to verify that $\mathrm{cone}(B_\Theta(\theta, \delta) - \theta) = T_\theta \Theta$. Moreover, because the infimum on the left hand side of \eqref{E:tangentconecond} is at most $0$,
    \begin{align*}
        \inf_{\theta' \in  B_\Theta(\theta, \delta)} \frac{D(\theta' - \theta)}{\delta} & \ge \inf_{\theta' \in B_\Theta(\theta ,\delta)} \frac{D(\theta' - \theta)}{\norm{\theta' - \theta}_{\mathfrak{B}}} = \inf_{b \in B_\Theta(\theta, \delta) - \theta} \frac{D(b)}{\norm{b}_\mathfrak{B}} = \inf_{b \in T_\theta \Theta} \frac{D(b)}{\norm{b}_\mathfrak{B}}.
    \end{align*}
    On the other hand, let $b \in T_\theta \Theta$ be nonzero and arbitrary. There exists an element $b' \in B_\Theta(\theta, \overline{\delta}) - \theta$ such that $b = c \cdot b'$ for some $c > 0$. If $\norm{b'}_{\mathfrak{B}}=0$, then $\theta+b'\in B_\Theta(\theta,\delta)$ for every $\delta>0$, so the left hand side of \eqref{E:tangentconecond} is $-\infty$ when $D(b)<0$, and it is at most $0 \le D(b)/\norm{b}_{\mathfrak{B}}$ otherwise.  Assume therefore that $\delta_b = \norm{b'}_{\mathfrak{B}}>0$. Then, for all $\delta \le \delta_b,$
    \begin{align*}
         \inf_{\theta' \in  B_\Theta(\theta, \delta)} \frac{D(\theta' - \theta)}{\delta} \le \frac{D(b' \cdot \delta/\delta_b)}{\delta} = \frac{D(b')}{\delta_b} = \frac{D(b)}{\norm{b}_{\mathfrak{B}}}.
    \end{align*}
    Hence, the left hand side of \eqref{E:tangentconecond} is at most its right hand side, which implies the equality.
\end{proof}

\begin{lemma}\label{L:duality}
Suppose that Assumption \ref{A:ex1}(2) holds with a bounded, convex, open neighborhood $D$ of $m(\Theta)$.
Then, there is a convex, norm bounded, and weak-* compact set of affine functions $\mathcal{T} \subset \mathcal{A}$ such that $\xi(x)  = \max_{t \in \mathcal{T}} t(x)$ for all $x \in D$ . If, in addition, Assumption~\ref{A:ex1}(4) holds, $T_0=0$, $v(\theta, t) = t(m(\theta))$, and $\Theta$ is convex in a neighborhood of $\Theta_{\mathrm{I}}$, then for $\theta \in \Theta_{\mathrm{I}}$
\begin{align}\label{E:subgradient}
\mathcal{K}(\theta)
=
\{y-y(m(\theta)):y\in (-m'(\theta)(T_\theta\Theta))^\circ
\cap \partial\xi(m(\theta))\}.
\end{align}
\end{lemma}

\begin{proof}[Proof of Lemma \ref{L:duality}]
For convex functions $f_1$ and $f_2$ on a Banach space $\mathfrak{X}$, let the infimal convolution be defined as the mapping $x \mapsto f_1 \square f_2 (x) = \inf_{x' \in \mathfrak{X}} (f_1(x - x') + f_2(x'))$, and let the convex conjugate of $f_1$ be the function mapping $y \mapsto f_1^*(y) = \sup_{x \in \mathfrak{X}} (\inner{y,x} - f_1(x))$, for all $y \in \mathfrak{X}^*$ \citep{Rock1970}.

Let $\overline{D}$ be the norm closure of $D$, and suppose that $\xi$ is $L$-Lipschitz continuous on $D$. Let $\delta(\cdot | \overline{D}) = \left\{ \begin{array}{l l}
0 \text{ if }x \in \overline{D} \\
\infty \text{ if }x \not\in \overline{D}
\end{array}
\right. $ be the indicator function of $\overline{D}$. Make the definition $\tilde{\xi} = (\xi + \delta(\cdot|\overline{D})) \square \, L \norm{\cdot}_{\mathfrak{X}}$. For $x \in D$ and $x' \in x - \overline{D}$, the point $x - c x' = (1-c)x + c(x-x')$ is in $D$ for all $c \in (0,1)$, because $D$ is open and hence $x$ is interior \citep[Lemma 5.28]{A1999}. A consequence of Lipschitz continuity of $\xi$ over $D$ is thus that $\xi(x - x') - \xi(x) \ge c^{-1}(\xi(x - c x') - \xi(x)) \ge -L \norm{x'}_{\mathfrak{X}}$ for all $c \in (0,1)$. Hence, $\xi$ is uniformly bounded below on $\overline{D}$, and
\begin{align*}
\tilde{\xi} & = ((\xi + \delta(\cdot | \overline{D})) \square L \norm{\cdot}_{\mathfrak{X}} ) (x)  = \inf_{x' \in \mathfrak{X}} (\xi(x - x') + \delta(x - x' |\overline{D}) + L \norm{x'}_{\mathfrak{X}} ) \\
& = \inf_{x' \in x - \overline{D}} ( \xi(x - x') + L\norm{x'}_{\mathfrak{X}})  \ge \xi(x).
\end{align*}
Clearly $\tilde{\xi} \le \xi$ on $D$, so that $\tilde{\xi}(x) = \xi(x)$ for all $x \in D$. Moreover, Lemma 3.8.3 of \citet{cheridito2013} implies that $\tilde{\xi}$ is norm-continuous on $\mathfrak{X}$ and is finite everywhere.

The definition of the convex conjugate yields
\begin{align*}
\tilde{\xi}^* & = (\xi + \delta(\cdot |\overline{D} ))^* + (L \norm{\cdot}_{\mathfrak{X}})^* = (\xi + \delta(\cdot |\overline{D}))^* + \delta(\cdot|B_{\mathfrak{X}^*}(0,L)),
\end{align*}
where $B_{\mathfrak{X}^*}(0,L) \subset \mathfrak{X}^*$ is the closed ball of radius $L$. Note that $\tilde{\xi}^*$ is the pointwise supremum of a collection of functions which are individually weak-* continuous over $\mathfrak{X}^*$, so it is weak-* lower semi-continuous and has closed epigraph in the product space $\mathfrak{X}^* \times \R$ when $\mathfrak{X}^*$ is endowed with its weak-* topology.

By boundedness of $\overline{D}$,
\begin{align}
(\xi + \delta(\cdot|\overline{D}))^*(y) & = \sup_{x \in \mathfrak{X}} \inner{y, x} - \xi(x) - \delta(\cdot |\overline{D}) \nonumber \\
& = \sup_{x \in \overline{D}} \inner{y, x} - \xi(x) \le \norm{y}_{\mathfrak{X}^*}\sup_{x \in \overline{D}} \norm{x}_{\mathfrak{X}}  - \inf_{x \in \overline{D}} \xi(x) < \infty.  \label{E:xistarcont}
\end{align}
By its definition, $(\xi + \delta(\cdot |\overline{D} ))^*$ is also bounded below. Thus, $(\xi + \delta(\cdot|\overline{D}))^*$ is uniformly bounded above and below on every norm-bounded subset of $\mathfrak{X}^*$. Let $C\ge \sup_{y \in B_{\mathfrak{X}^*}(0,L)} (\xi + \delta(\cdot|\overline{D}))^*(y)$, and define
\begin{align} \label{E:affineTdefn}
\mathcal{A} \supset \mathcal{T} &\equiv \{ x \mapsto \inner{y, x} - c \, | \, y \in B_{\mathfrak{X}^*}(0,L), c \in [(\xi + \delta( \cdot |\overline{D}))^* (y), C ]\}.
\end{align}
Then $\mathcal{T}$ is the image of the set
\begin{align} \label{E:weakstarcont}
\{(y,c)\, | \, y \in B_{\mathfrak{X}^*}(0,L), c \in [ (\xi + \delta(\cdot | \overline{D}))^*(y), C]\} \subset \mathfrak{X}^* \times \R \cong (\mathfrak{X} \times \R)^*
\end{align}
under the map $\zeta: (y,c) \mapsto ( x \mapsto \inner{y,x} - c)$. When $(\mathfrak{X} \times \R)^*$ is endowed with the weak-* topology and $\mathcal{A}$ is endowed with the topology of pointwise convergence, it is clear that $\zeta$ is a continuous map. Moreover, we observe that the left hand side of \eqref{E:weakstarcont} is norm bounded in $(\mathfrak{X} \times \R)^*$, and is also the intersection of the epigraph of $\tilde{\xi}^*$ with the weak-* closed set $\mathfrak{X}^* \times (-\infty, C]$. Therefore it is both norm bounded and weak-* closed, whence compact in the weak-* topology. As $\zeta$ is linear and \eqref{E:weakstarcont} is the intersection of two convex sets, $\mathcal{T}$ is also convex.

Because $\tilde{\xi}$ is finite and norm-continuous, the Fenchel-Moreau theorem \citep[Theorem~3.7.6]{cheridito2013} implies that $\tilde{\xi} = \tilde{\xi}^{**}$, and thus
\begin{align*}
\tilde{\xi}(x) & = \sup_{y \in \mathfrak{X}^*} (\inner{y, x} - (\xi + \delta(\cdot|\overline{D}))^*(y) - \delta(y|B_{\mathfrak{X}^*}(0,L))) = \sup_{y \in B_{\mathfrak{X}^*}(0,L)} (\inner{y, x} - (\xi + \delta(\cdot | \overline{D}))^*(y))  \\
& = \sup_{t \in \mathcal{T}} t(x) = \max_{t \in \mathcal{T}} t(x),
\end{align*}
where the maximum is attained by weak compactness.
In particular, $\xi(x) = \max_{t \in \mathcal{T}} t(x)$ for all $x \in D$.

We turn now to \eqref{E:subgradient}. Suppose that $m(\Theta_{\mathrm{I}})$ is contained in $D$ and fix a point $\theta \in \Theta_{\mathrm{I}}$. By the preceding arguments, $\tilde{\xi}$ agrees with $\xi$ on a neighborhood of $m(\theta)$. By Fr\'{e}chet differentiability of $m$ and norm boundedness of $\mathcal{T}$, we may take $D_{\theta; t} v: \theta' - \theta \mapsto y[t](m'(\theta) (\theta' - \theta))$ in the statement of Assumption \ref{A:two}. By Lemma \ref{L:tangentconecond}, the condition that $\gamma(\theta, t)$ is at least $0$ is equivalent to the requirement that $y[t]$ is in $(-m'(\theta)(T_\theta \Theta))^\circ$. Hence,
\[
\mathcal{K}(\theta) = \{t \in \mathcal{T}: t(m(\theta)) = 0 \}\cap \{t \in \mathcal{T}: y[t] \in (-m'(\theta)(T_\theta \Theta))^\circ\},
\]
and the definition of $\Theta_{\mathrm{I}}$ implies that $\tilde{\xi}(m(\theta)) = 0$. If $t$ is in $\mathcal{K}(\theta)$, we conclude that $t(m(\theta)) = \tilde{\xi}(m(\theta)) = \xi(m(\theta))$ and that $t$ is majorized by $\tilde{\xi}$, whence $\xi$, in a neighborhood of $m(\theta)$. As $\xi$ is convex, an affine function that lies below $\xi$ near $m(\theta)$ and equals it at $m(\theta)$ lies below $\xi$ on all of $D$. Hence this happens if and only if $y[t]$ is in $\partial \xi(m(\theta))$ and $t(m(\theta)) = \xi(m(\theta))$, so that
\begin{align}
\mathcal{K}(\theta) &= \{t \in \mathcal{T}: t(m(\theta)) = 0 \text{ and } y[t]  \in  (-m'(\theta)(T_\theta \Theta))^\circ \cap \partial \xi (m(\theta))\} \nonumber \\
& = \{y[t] + t(0) : t  \in \mathcal{T}, \,  t(m(\theta)) = 0 \text{ and } y[t]  \in  (-m'(\theta)(T_\theta \Theta))^\circ \cap \partial \xi (m(\theta))\} \label{E:ktheta2}
\end{align}
Let $t$ be such that $t(m(\theta)) = 0$ and $y[t] \in (-m'(\theta)(T_\theta \Theta))^\circ \cap \partial \xi (m(\theta))$, so that $y[t](m(\theta)) = t(m(\theta)) - t(0) = -t(0)$. Hence, \eqref{E:ktheta2} implies that $\mathcal{K}(\theta)$ is contained in the set in \eqref{E:subgradient}.

On the other hand, let $y \in (-m'(\theta)(T_\theta \Theta))^\circ \cap \partial \xi (m(\theta))$.
By the definitions of the subdifferential and the conjugate, $y\in\partial\tilde\xi(x)$ if and only if $\tilde\xi^*(y)=y(x)-\tilde\xi(x)$, the equality case of the Fenchel--Young inequality. Hence $\partial \xi(m(\theta)) = \partial \tilde{\xi}(m(\theta))$ is contained in the effective domain of $\tilde{\xi}^*$, which is contained in $B_{\mathfrak{X}^*}(0,L)$. Thus, $y \in B_{\mathfrak{X}^*}(0,L)$. By the same equality,
\begin{align} \label{E:fenchelyoung}
(\xi + \delta(\cdot|\overline{D}))^*(y) = \tilde{\xi}^*(y) = y(m(\theta))
\end{align}
By \eqref{E:affineTdefn}, the affine function $y - y(m(\theta))$ is in $\mathcal{T}$, and \eqref{E:ktheta2} contains, and is equal to, \eqref{E:subgradient}.
\end{proof}

The next two lemmas collect standard facts from Banach-space convex analysis in the form used above. The distance and polar identities are consequences of Hahn-Banach separation and convex duality; see, for example, \citet{fabian2011} and \citet{burke1992clarke}. The hemicontinuity argument also uses the usual weak and weak-* compactness and separation tools for locally convex spaces, as in \citet{C1994} and \citet{A1999}. We include the proofs for completeness, because the statements are specialized to the cone-valued correspondences used in the paper.

Condition~(1) of the next lemma is collar upper hemicontinuity (Definition~\ref{D:collar}).

\begin{lemma} \label{L:duality2}
Let $\mathcal{A}$ be a topological space and $\mathfrak{X}$ a Banach space. Let $C: \mathcal{A} \ra \mathfrak{X}$ be a correspondence sending points $a \in \mathcal{A}$ to convex cones $C(a) \subset \mathfrak{X}$. Let $C^\circ(a): \mathcal{A} \ra \mathfrak{X}^*$ be the correspondence mapping $a \mapsto \{y \in \mathfrak{X}^*: y(x) \le 0, \text{ for all } x \in C(a)\}$.
Then, the following are equivalent:
\begin{enumerate}
\item The correspondence $a \mapsto C(a) \cap B_{\mathfrak{X}}(0,1)$ is collar upper hemicontinuous at $a$
\item $a \mapsto C^\circ(a)$ is lower hemicontinuous with respect to the norm topology on $\mathfrak{X}^*$ at $a$
\item $a \mapsto C^\circ(a) \cap B_{\mathfrak{X}^*}(0,1)$ is lower hemicontinuous with respect to the norm topology at $a$.
\end{enumerate}
\end{lemma}

\begin{proof}[Proof of Lemma \ref{L:duality2}] For most functional analytic results, we refer to \citet{C1994}.
\begin{enumerate}[itemindent=56pt, leftmargin=0mm]
\item[(1)$\Rightarrow$(2):]  Suppose that $a \mapsto C(a) \cap B_{\mathfrak{X}}(0,1)$ is collar upper hemicontinuous at $a \in \mathcal{A}$. Let $V \subset \mathfrak{X}^*$ be open in the norm topology and intersect $C(a)^\circ$. Then, for some $y \in C(a)^\circ$ and $\ve > 0$, the open ball $B_{\mathfrak{X}^*}(y, \ve)$ is contained in $V$.

Note that the weakly open set $O = y^{-1}(-\infty, \ve/2)$ contains $C(a) \cap B_{\mathfrak{X}}(0,1)$. By supposition, applied with the weakly open neighborhood $D = y^{-1}(-\ve/2, \ve/2)$ of $0$, there is a neighborhood $U$ of $a$ such that $C(a') \cap B_{\mathfrak{X}}(0,1) \subseteq (C(a) \cap B_{\mathfrak{X}}(0,1)) + D$ whenever $a' \in U$. Because $y(c) \le 0$ for $c \in C(a)$ and $y(d) < \ve/2$ for $d \in D$, this collar is contained in $O$, so $C(a') \cap B_{\mathfrak{X}}(0,1) \subset O$ whenever $a' \in U$. If $C(a') \cap B_{\mathfrak{X}}(0,1) \subset O$, then we have $y(c) \le \ve \norm{c}/2 $ for all $c \in C(a')$. It follows from Lemma \ref{L:positive} below that $d(y, C(a')^\circ) \le \ve/2$. Therefore, for all $a' \in U$, $C(a')^\circ \cap V \supset C(a')^\circ \cap B_{\mathfrak{X}^*}(y, \ve) \neq \emptyset$.

\item[(2)$\Leftrightarrow$(3):] Each $C(a')^\circ$ is a cone, and dilation $y \mapsto \lambda y$ ($\lambda > 0$) is a homeomorphism of $\mathfrak{X}^*$ preserving it, so the two lower-hemicontinuity conditions are equivalent. Let $V \subseteq \mathfrak{X}^*$ be open in the norm topology, and let $\mathrm{Int}$ denote the norm interior. For (2)$\Rightarrow$(3), suppose $V$ meets $C(a)^\circ \cap B_{\mathfrak{X}^*}(0,1)$; because $C(a)^\circ$ is a cone, scaling the witness toward $0$ shows that $V$ also meets $C(a)^\circ \cap \mathrm{Int}\, B_{\mathfrak{X}^*}(0,1)$. By (2) there is a neighborhood $U$ of $a$ such that $C(a')^\circ$ meets the open set $V \cap \mathrm{Int}\, B_{\mathfrak{X}^*}(0,1)$ for every $a' \in U$, and any such point lies in $C(a')^\circ \cap B_{\mathfrak{X}^*}(0,1) \cap V$. For (3)$\Rightarrow$(2), suppose $V$ meets $C(a)^\circ$ at a point $y$; as $0 \in C(a')^\circ$ for every $a'$, we may take $y \neq 0$. With $\lambda = 1 \wedge \norm{y}_{\mathfrak{X}^*}^{-1}$ we have $\lambda y \in C(a)^\circ \cap B_{\mathfrak{X}^*}(0,1) \cap \lambda V$, so by (3) there is a neighborhood $U$ of $a$ such that $C(a')^\circ \cap B_{\mathfrak{X}^*}(0,1)$ meets the open set $\lambda V$ for every $a' \in U$; any witness $w$ then gives $\lambda^{-1} w \in C(a')^\circ \cap V$.

\item[(3)$\Rightarrow$(1):] We prove the contrapositive. Suppose that $a \mapsto C(a)\cap B_{\mathfrak{X}}(0,1)$ fails to be collar upper hemicontinuous at $a$, so that there exists some net $a_i$ converging to $a$ (indexed by neighborhoods of $a$ and ordered by reverse inclusion) and some weakly open neighborhood $D$ of $0$ for which $C(a_i)\cap B_{\mathfrak{X}}(0,1) \not \subset \bigcup_{c \in C(a) \cap B_{\mathfrak{X}}(0,1)} (c + D)$, for all $i$. Without loss of generality, we may assume that $D = \bigcap_{j = 1}^J y_j^{-1}(-\ve, \ve)$ for some choice of functionals $y_1, \ldots, y_J \in \mathfrak{X}^*$.

Let $\iota: \mathfrak{X} \hookrightarrow  \mathfrak{X}^{**}$ be the canonical embedding. By the Banach-Alaoglu theorem, there is a subnet $a_{i'}$ such that, for all $i'$, there is some $x_{i'} \in \iota( (C(a_{i'}) \cap B_{\mathfrak{X}}(0,1))  \setminus \bigcup_{c \in C(a) \cap B_{\mathfrak{X}}(0,1)} (c + D))$ converging (weak-*) to some $x$ in the closed unit ball of $\mathfrak{X}^{**}$. If we define the weak-* open set $\tilde{D} = \bigcap_{j = 1}^J \{z \in \mathfrak{X}^{**}: z(y_j) \in (-\ve, \ve)\}$, the triangle inequality implies that $(x + \tilde{D}/2) \cap (\iota(C(a) \cap B_{\mathfrak{X}}(0,1))) = \emptyset$. Moreover, $(\mathfrak{X}^*, \sigma(\mathfrak{X}^*, \mathfrak{X}^{**})),  (\mathfrak{X}^{**}, \sigma(\mathfrak{X}^{**}, \mathfrak{X}^{*}))$ form a dual pair by \citet{C1994}, V.1.3. Hence, by the geometric Hahn-Banach theorem for locally convex spaces, there is some $y'$ in the unit ball of $\mathfrak{X}^{*}$ and $\alpha, \alpha' \in \R$ such that $\inner{y', \iota(c)} \le \alpha < \alpha' \le \inner{y', x}$ for all $c \in C(a) \cap B_{\mathfrak{X}}(0,1)$.

As $C(a)$ is a cone, we must have $\alpha \ge 0$. We have $y'(c) = \iota(c)(y') \le \alpha$ for all $c \in C(a) \cap B_{\mathfrak{X}}(0,1)$, so by Lemma \ref{L:positive}, there must be some $y \in C(a)^\circ$ satisfying $\norm{ y - y'}_{\mathfrak{X}^*} \le \frac{\alpha + \alpha'}{2}$ (say). Moreover, by weak convergence,
\begin{align*}
y(\iota^{-1}(x_{i'})) &\ge y'(\iota^{-1}(x_{i'})) - \norm{y - y'}_{\mathfrak{X}^*} = x_{i'} (y') - \norm{y - y'}_{\mathfrak{X}^*} \nonumber \\
& \ra x(y') - \norm{y - y'}_{\mathfrak{X}^*} \ge \frac{\alpha' - \alpha}{2} > 0.
\end{align*}
Hence,
\begin{align}
\liminf_{i'} \inf_{\norm{\tilde{y} - y}_{\mathfrak{X}^*} < (\alpha' - \alpha)/4} \tilde{y}(\iota^{-1}(x_{i'})) \ge \frac{\alpha' - \alpha}{4} > 0. \label{E:rest}
\end{align}
Rescaling $y$ if necessary, we may assume without loss of generality that $y \in C(a)^\circ \cap B_{\mathfrak{X}^*}(0,1)$: if $M=1\vee\norm{y}_{\mathfrak{X}^*}$, then \eqref{E:rest} continues to hold with $y/M$ after reducing the displayed radius and lower bound. From \eqref{E:rest}, we deduce that $C(a_{i'})^\circ$ eventually has empty intersection with a norm-open ball around $y$. Hence, $a \mapsto C(a)^\circ \cap B_{\mathfrak{X}^*}(0,1)$ cannot be lower hemicontinuous.
\end{enumerate}

\end{proof}

\begin{lemma} \label{L:positive}
Let $C$ be a convex cone in a real Banach space $\mathfrak{X}$, and let $y \in \mathfrak{X}^*$. Then $d(x, C) = \sup_{y' \in C^\circ \cap B_{\mathfrak{X}^*}(0,1)} \inner{x, y' }$ and $d(y, C^\circ) = \sup_{x \in C \cap B_{\mathfrak{X}}(0,1)} \inner{x,y}$.
\end{lemma}
\begin{proof}
We begin with the result for the distance to a convex cone in $\mathfrak{X}$, which is e.g.\ Theorem 7.3 in \citet{Deutsch2012} when $\mathfrak{X}$ is a Hilbert space. By the Hahn-Banach theorem,
\begin{align*}
d(x, C) & = \inf_{c \in C} \norm{x - c}_\mathfrak{X} = \inf_{c \in C} \sup_{\substack{ y' \in \mathfrak{X}^* \\ \norm{y'}_{\mathfrak{X}^*} \le 1}} \inner{x-c, y'}
\end{align*}
The minimax theorem implies
\begin{align}
\sup_{ \norm{y'}_{\mathfrak{X}^*} \le 1} (\inner{x, y'} - \sup_{c \in C} \inner{c , y'}) &=  \sup_{ \norm{y'}_{\mathfrak{X}^*} \le 1} ( \inner{x,y'} - \infty \cdot \one_{y' \not \in C^\circ}) \nonumber\\
& = \sup_{\substack{   y' \in C^\circ  \\ \norm{y'}_{\mathfrak{X}^*} \le 1 }} \inner{x , y'}. \label{E:thorns}
\end{align}
Note that $C^\circ$ is weak-* closed \citep[\S3.4, with weak-* compact intersection with $B_{\mathfrak{X}^*}(0,1)$]{fabian2011}. Dually, we may write
\begin{align}
d(y, C^\circ) &= \lim_{L \ra \infty} d(y, C^\circ \cap B_{\mathfrak{X}^*}(0,L)) = \lim_{L \ra \infty} \inf_{\substack{y' \in C^\circ \\ \norm{y'}_{ \mathfrak{X}^*} \le L }}\sup_{\norm{x}_{\mathfrak{X}} \le 1} \inner{x, y - y'} \nonumber \\
& = \lim_{L \ra \infty}\sup_{\norm{x}_{\mathfrak{X}} \le 1} ( \inner{x, y } -  \sup_{\substack{y' \in C^\circ \\ \norm{y'}_{ \mathfrak{X}^*} \le L }} \inner{x, y'} ) \nonumber \\
& = \lim_{L \ra \infty} \sup_{\norm{x}_{\mathfrak{X}} \le 1} ( \inner{x, y } - L \cdot d(x, C) ), \label{E:pilot}
\end{align}
where the last line uses \eqref{E:thorns}. Let $x_L$ maximize \eqref{E:pilot} up to some error that is $o(1)$ as $L \ra \infty$. As \eqref{E:pilot} is bounded below by $0$, and $|\inner{x,y}| \le \norm{y}_{\mathfrak{X}^*}$ whenever $\norm{x} \le 1$, one must have $d(x_L, C) = O(L^{-1})$. Letting $\alpha_L = o(1)$ be such that $L^{-1} = o(\alpha_L)$, one has
\begin{align*}
d(y, C^\circ) & \le \lim_{L \ra \infty} \sup_{\substack{\norm{ x}_{\mathfrak{X}} \le 1 \\ d(x, C) \le \alpha_L} } \inner{x,y} \le \lim_{L \ra  \infty}  \sup_{\substack{c \in C \cap B_{\mathfrak{X}}(0,1+2\alpha_L) \\ \norm{v}_{\mathfrak{X}} \le 2 \alpha_L } } \inner{ c + v , y} \\
& \le \lim_{L \ra \infty} \Big( (1+2\alpha_L) \sup_{c \in C \cap B_{\mathfrak{X}}(0,1)} \inner{c , y} + \sup_{\norm{v}_{\mathfrak{X}} \le 2 \alpha_L } \inner{v, y} \Big) = \sup_{ c \in C \cap B_{\mathfrak{X}}(0,1)} \inner{c,y}.
\end{align*}
Here the second inequality decomposes a near-feasible $x$ as $c + v$ with $c \in C$, $\norm{v}_{\mathfrak{X}} \le 2\alpha_L$, and hence $\norm{c}_{\mathfrak{X}} \le 1 + 2\alpha_L$. The third inequality rescales $c$ by $(1+2\alpha_L)^{-1}$, which keeps it in $C$ because $C$ is a cone, and the final equality uses $\sup_{\norm{v}_{\mathfrak{X}} \le 2\alpha_L} \inner{v,y} = 2\alpha_L \norm{y}_{\mathfrak{X}^*} \to 0$.
On the other hand,
\begin{align*}
d(y, C^\circ) & = \inf_{\substack{y' \in C^\circ } }\sup_{\norm{x}_{\mathfrak{X}} \le 1} \inner{x, y - y'} \ge \inf_{y' \in C^\circ} \sup_{ x \in C \cap B_{\mathfrak{X}}(0,1)} \inner{x, y - y'} = \sup_{ x \in C \cap B_{\mathfrak{X}}(0,1)} \inner{x, y},
\end{align*}
which concludes.
\end{proof}

\begin{lemma} \label{L:tangentclosure}
    Let $C$ be a convex subset of a Banach space $\mathfrak{X}$ and let
    $x\in\overline C$.  Then
    \[
        \overline{T_x C}=\overline{T_x\overline C}.
    \]
    Consequently, $(T_xC)^\circ=(T_x\overline C)^\circ$.
\end{lemma}
\begin{proof}
    For a convex set $S$, the definition gives
    $T_xS=\mathrm{cone}(S-x)$ whether or not $x\in S$.
    Hence $T_xC=\mathrm{cone}(C-x)$ and $T_x\overline C=\mathrm{cone}(\overline C-x)$,
    and these cones have the same closure because $C$ is dense in $\overline C$.
    The equality of polars follows because a
    polar depends only on the norm closure of the underlying cone.
\end{proof}

\begin{lemma} \label{L:normalcone}
    Let $\mathfrak{X}$ be a Banach space and $C \subset \mathfrak{X}$ convex. For all $x \in \overline{C}$,
    \[
        \partial d_{\mathfrak{X}}(\cdot, C)(x) = B_{\mathfrak{X}^*}(0, 1) \cap (T_x C)^\circ.
    \]
\end{lemma}
\begin{proof}[Proof of Lemma \ref{L:normalcone}]
    Because $d_{\mathfrak{X}}(\cdot, C) = d_{\mathfrak{X}}(\cdot, \overline{C})$ and Lemma~\ref{L:tangentclosure} gives $(T_x C)^\circ = (T_x \overline{C})^\circ$, we may assume $C$ is closed. The conclusion is then Theorem~1 of \citet{burke1992clarke}.
\end{proof}

The next lemma is stated for a domain $\mathcal{A}$ containing $\theta_0$.  Taking $\mathcal{A}=\Theta$ for every $\theta_0\in\Theta_{\mathrm{I}}$ gives Assumption~\ref{A:ex22}(1) as stated.  Taking $\mathcal{A}=\{\theta_0\}\cup(\Theta\setminus\Theta_{\mathrm{I}})$ gives the weaker form in the footnote to that assumption, which is all Lemma~\ref{L:ex22K} uses.

\begin{lemma} \label{L:geometry}
    Suppose that $\xi(x) = d_{\mathfrak{X}}(x, C)$ for some convex set $C$ and that $T_0=0$. Let $\theta_0 \in \Theta_{\mathrm{I}}$ and $\theta_0 \in \mathcal{A} \subseteq \Theta$. Suppose that the correspondence $\theta \mapsto B_{\mathfrak{X}}(0,1) \cap \big(\mathrm{cone}(m(\theta_0) - m(\theta)) + T_{m(\theta_0)} C -m'(\theta) (T_\theta \Theta)\big)$, restricted to $\mathcal{A}$, satisfies condition~(1) of Lemma~\ref{L:duality2} at $\theta_0$, a condition weaker than upper hemicontinuity with respect to the weak topology. Then, the correspondence in Assumption~\ref{A:ex22}(1), restricted to $\mathcal{A}$, is lower hemicontinuous at $\theta_0$ with respect to the norm topology on $\mathfrak{X}^*$.
\end{lemma}

\begin{proof}[Proof of Lemma \ref{L:geometry}]
    By Lemma \ref{L:normalcone}, $\partial \xi(m(\theta_0)) = B_{\mathfrak{X}^*}(0,1) \cap (T_{m(\theta_0)} C)^\circ$ for $\theta_0 \in \Theta_{\mathrm{I}}.$ Hence, the correspondence in Assumption \ref{A:ex22}(1) maps $\theta$ to
    \begin{align*}
    &(T_{m(\theta)} (\xi^{-1}((-\infty,0])))^\circ \cap (T_{m(\theta_0)} C)^\circ \cap (-m'(\theta) (T_\theta \Theta))^\circ \cap B_{\mathfrak{X}^*}(0,1) \\
    &\qquad = \big(T_{m(\theta)} C + T_{m(\theta_0)} C - m'(\theta) (T_\theta \Theta) \big)^\circ \cap B_{\mathfrak{X}^*}(0,1).
    \end{align*}
    A point $y$ is in the polar of $T_{m(\theta)} C + T_{m(\theta_0)} C$ if and only if $y(c - m(\theta))$ and $y(c - m(\theta_0))$ are nonpositive for all $c$ in $C$. Because $m(\theta_0)$ is in the closure of $C$, these inequalities may be equivalently expressed as $y(m(\theta_0) - m(\theta)) \le 0 $ and $y(c - m(\theta_0)) \le 0$ for all $c$ in $C$. Hence, the polar of $T_{m(\theta)} C + T_{m(\theta_0)} C$ is the polar of $\mathrm{cone}(m(\theta_0) - m(\theta)) + T_{m(\theta_0)} C$. The conclusion of the lemma then follows from Lemma \ref{L:duality2}, applied with the subspace topology on $\mathcal{A}$.
\end{proof}

\else
\section{Proofs of the main results}

\subsection*{Additional notation}
The proofs use the following notation, which is not needed outside the proofs.

An \textit{affine map} over $\mathfrak{X}$ is any function $t: x \mapsto y(x) + a$, where $y \in \mathfrak{X}^*$ and $a \in \R$. The set of such functions is readily identified with the affine dual $(\mathfrak{X} \times \R)^* = \mathfrak{X}^* \times \R$ by the agreement $(x \mapsto y(x) + a) \mapsto (y,a)$. We assign the set of affine maps $\mathcal{A}$ over $\mathfrak{X}$ the norm and weak-* topologies that pertain to the dual of $\mathfrak{X} \times \R$. A subset $\mathcal{T}$ of $\mathcal{A}$ is defined to be bounded if it is bounded in norm. For an affine map $t \in \mathcal{A}$, we use $y[t] = t - t(0)$ to denote the linear part of $t$.

\fi

\ifmainbuild\else
\section{Stationarity-only bootstraps}\label{S:stationarity}

The bootstrap of Theorem~\ref{T:two} targets a limit over the full contact set $\mathcal{K}$. This section gives a variant whose target replaces $\mathcal{K}$ by the stationarity set $\mathcal{K}_0$ of \eqref{E:k0setdef}. A level penalty $\ell_n(\theta)$ localizes $\theta$ near $\Theta_{\mathrm{I}}$, while the local difference $v_n(\theta',t)-v_n(\theta,t)$ estimates only the derivative condition defining $\mathcal{K}_0(\theta)$. Because $\mathcal{K}(\theta)\subseteq\mathcal{K}_0(\theta)$ on $\Theta_{\mathrm{I}}$, the resulting critical values are generally conservative relative to those of Theorem~\ref{T:two}. If $v(\theta,t)=0$ for all $(\theta,t)\in\Theta_{\mathrm{I}}\times\mathcal{T}$, as in the norm-loss Examples~\ref{E:jtest} and~\ref{E:condmoment}, then $\mathcal{K}=\mathcal{K}_0$ on $\Theta_{\mathrm{I}}$ and the two targets coincide. The variant's continuity condition, Assumption~\ref{A:five'}, concerns $\mathcal{K}_0$ alone.

\begin{asm}[Stationarity-set continuity] \label{A:five'}
The correspondence $\theta \mapsto \mathcal{K}_0(\theta)$ is lower hemicontinuous at $\Theta_{\mathrm{I}}$.
\end{asm}

\begin{addendum}[Stationarity-only bootstrap]\label{add:bootstrap_K0}
Suppose that Assumptions~\ref{A:zero}--\ref{A:three}, \ref{A:five}, and \ref{A:five'} hold, $t\mapsto D_{\theta;t} v(\theta' - \theta)$ is concave for all $\theta, \theta' \in \Theta$, $\alpha_v<\infty$, and that $\Theta_{\mathrm{I}}$ is nonempty. Let $(\lambda_n)$ satisfy $\lambda_n\ra\infty$ and $\lambda_n = o(r_n^{1-1/\alpha_v})$, and let $b_n = r_n^{-1/\alpha_v}$. Then,
\begin{align}
&\inf_{\theta \in \Theta} \Big(\lambda_n b_n^{-1}  \ell_n(\theta) + \inf_{\theta' \in B_\Theta(\theta, b_n)} \sup_{t \in \mathcal{T}} \big(\lambda_n b_n^{-1} (v_n(\theta', t) - v_n(\theta, t)) +\mathbb{V}_n^*(\theta, t) \big) \Big) \nonumber \\
&\qquad =  \inf_{\theta \in \Theta_{\mathrm{I}}} \sup_{t \in \mathcal{K}_0(\theta)} \mathbb{V}_n^*(\theta, t) + o_{p^*}(1)\nonumber \\
&\qquad \overset{\mathrm{P}}{\rs} \inf_{\theta \in \Theta_{\mathrm{I}}} \sup_{t \in \mathcal{K}_0(\theta)} \mathbb{G}_n(\theta, t). \label{E:bootstrap_K0}
\end{align}
\end{addendum}

\begin{proof}
    Let $Z_{n,0}^*$ denote the left hand side of \eqref{E:bootstrap_K0}, and let $s_n=\lambda_n b_n^{-1}$. The rate restrictions imply $s_n=o(r_n)$. By Assumptions \ref{A:one} and \ref{A:five}(1), the terms obtained by replacing $\ell_n$ and $v_n(\theta',t)-v_n(\theta,t)$ with their population counterparts are uniformly $o_{p^*}(1)$ after multiplication by $s_n$. Hence,
    \begin{align} \label{E:bootstrap_K00}
    Z_{n,0}^* = \inf_{\theta \in \Theta}\inf_{\theta'\in B_\Theta(\theta,b_n)}\sup_{t\in\mathcal T}\big(s_n\ell(\theta)+s_n(v(\theta',t)-v(\theta,t))+\mathbb V_n^*(\theta,t)\big)+o_{p^*}(1).
    \end{align}
    Taking $\theta\in\Theta_{\mathrm I}$ and $\theta'=\theta$ gives $Z_{n,0}^*=O_{p^*}(1)$. Let $S$ be an arbitrarily small open neighborhood of $\Theta_{\mathrm I}$ contained in the local-convexity neighborhood of Assumption~\ref{A:zero}(1). By Assumptions~\ref{A:zero}(2) and~\ref{A:five}(3), there is a $c>0$ such that, uniformly over $\theta\in S^c$, $\theta'\in B_\Theta(\theta,b_n)$, and $t\in\mathcal T$,
    \[
    s_n\ell(\theta)+s_n(v(\theta',t)-v(\theta,t))+\mathbb V_n^*(\theta,t)
    \ge s_n c-\lambda_n C_v-\|\mathbb V_n^*\|_\infty.
    \]
    The right-hand side equals $\lambda_n(c/b_n-C_v)-O_{p^*}(1)$ and diverges. Thus, by \eqref{E:binf1}, \eqref{E:bootstrap_K00} may be restricted to $\theta\in S$ with probability approaching one. On this neighborhood, the same minimax rearrangement as in \eqref{E:auxl0} applies, using the addendum's concavity assumption on $t\mapsto D_{\theta;t}v(\theta'-\theta)$, and yields
    \begin{align} \label{E:bootstrap_K01}
    Z_{n,0}^* = \inf_{\theta \in S}  \sup_{t \in \mathcal{T}} \big(s_n \ell(\theta) + \mathbb{V}_n^*(\theta, t) + \lambda_n \psi_n(\theta, t)\big) + o_{p^*}(1).
    \end{align}

    For the lower bound, apply Lemma~\ref{L:contactlb} with $\mathcal{K}'=\mathcal{K}_0$, which is lower hemicontinuous at $\Theta_{\mathrm{I}}$ by Assumption~\ref{A:five'} and nonempty on $\Theta_{\mathrm{I}}$ because it contains $\mathcal{K}$. After shrinking if necessary, take $S$ in \eqref{E:bootstrap_K01} to lie in the neighborhood $S'_\delta$ supplied by the lemma and within $\underline{\delta}$ of $\Theta_{\mathrm{I}}$ ($\underline{\delta}$ as in Lemma~\ref{L:gsc}). For $t'\in\mathcal{K}_0(\theta')$ we have $\gamma(\theta',t')=0$, so Lemma~\ref{L:gsc} and the monotonicity of $\delta\mapsto\psi_\delta(\theta',t')$ give $\psi_n(\theta',t')=\psi_{b_n}(\theta',t')\ge\gamma(\theta',t')=0$; because $\theta''=\theta'$ is admissible in the infimum defining $\psi_{b_n}(\theta',t')$ (and $D_{\theta';t'}v(0)=0$), that infimum is also nonpositive, so the directional penalty vanishes, while $\ell(\theta')\ge 0$. Lemma~\ref{L:contactlb} and \eqref{E:binf3} therefore imply
    \[
    Z_{n,0}^* \ge \inf_{\theta\in\Theta_{\mathrm{I}}}\sup_{t\in\mathcal{K}_0(\theta)}\mathbb{V}_n^*(\theta,t)+o_{p^*}(1).
    \]

    For the upper bound, restrict \eqref{E:bootstrap_K01} to $\Theta_{\mathrm{I}}$. As $\ell$ vanishes on $\Theta_{\mathrm{I}}$, the resulting upper bound is
    \[
    \inf_{\theta\in\Theta_{\mathrm{I}}}\sup_{t\in\mathcal{T}}\big(\mathbb{V}_n^*(\theta,t)+\lambda_n\psi_n(\theta,t)\big)+o_{p^*}(1).
    \]
    The functions $\psi_n(\theta,\cdot)$ are upper semicontinuous by Lemma~\ref{L:gsc} and nonpositive on $\Theta_{\mathrm{I}}\times\mathcal{T}$. Applying the upper-envelope construction from the proof of Lemma~\ref{L:gn} with $\psi_n$ in place of $v+\psi_n$, the sets $\{t:\psi_{\tilde b_n}(\theta,t)\ge -c_n\}$ are nested and closed in $\mathcal{T}$, and by \eqref{E:gamma} together with the nonpositivity of $\gamma$ on $\Theta_{\mathrm{I}}$ their intersection is $\mathcal{K}_0(\theta)$. Display \eqref{E:monodec} in the proof of Lemma~\ref{L:monotonedecreasing} then supplies the analogue of \eqref{E:binf8} with $\mathcal{K}_0$ in place of $\mathcal{K}$, and the remaining steps \eqref{E:binf6}--\eqref{E:binf9} apply verbatim. This yields the matching upper bound. The conditional weak convergence in Assumption~\ref{A:five}(1) completes the proof.
\end{proof}

\subsection{Norm loss}

For the norm loss $\xi(x)=\norm{x}_{\mathfrak{X}}$ with $T_0=0$, the moments vanish on $\Theta_{\mathrm{I}}$, so $\mathcal{K}(\theta)=\mathcal{K}_0(\theta)$ thereon and only the stationarity condition is relevant.  Proposition~\ref{add:xi_norm} specializes Addendum~\ref{add:bootstrap_K0} to this case.  It estimates the target of Proposition~\ref{P:ex2} under collar upper hemicontinuity (Definition~\ref{D:collar}) of the truncated derivative cone, which by Lemma~\ref{L:duality2}\suppref{} is equivalent to lower hemicontinuity of its polar.

\begin{prop}[Norm-loss bootstrap] \label{add:xi_norm}
Let Assumptions \ref{A:ex1} and \ref{A:ex2} hold, and suppose that $\Theta_{\mathrm{I}}$ is nonempty, $\alpha_m<\infty$, and $\xi(x)=\norm{x}_{\mathfrak{X}}$ with $T_0=0$. Let $(\lambda_n)$ satisfy $\lambda_n\ra\infty$ and $\lambda_n = o(r_n^{1-1/\alpha_m})$, and let $b_n = r_n^{-1/\alpha_m}$. If the correspondence $\theta \mapsto -m'(\theta)(T_\theta \Theta) \cap B_{\mathfrak{X}}(0,1)$ is collar upper hemicontinuous at every point $\theta_0$ of $\Theta_{\mathrm{I}}$, then
\begin{align*}
\inf_{\theta \in \Theta} \left[
\begin{aligned}
&\lambda_n b_n^{-1} \norm{m_n(\theta)}_{\mathfrak{X}} \\
&\quad + \inf_{\theta' \in B_\Theta(\theta, b_n)}
\norm{\lambda_n b_n^{-1}\big(m_n(\theta')-m_n(\theta)\big)+\mathbb{W}_n^*(\theta)}_{\mathfrak{X}}
\end{aligned}
\right] \\
& = \inf_{\theta \in \Theta_{\mathrm{I}}}
d_{\mathfrak{X}}\bigl(\mathbb{W}_n^*(\theta),-m'(\theta)(T_\theta\Theta)\bigr) + o_{p^*}(1) \\
& \overset{\mP}{\rs}
\inf_{\theta \in \Theta_{\mathrm{I}}}
d_{\mathfrak{X}}\bigl(\mathbb{X}_n(\theta),-m'(\theta)(T_\theta\Theta)\bigr).
\end{align*}
\end{prop}

\begin{proof}
    When $\xi$ is the norm on $\mathfrak{X}$, take $\mathcal{T}=B_{\mathfrak{X}^*}(0,1)$ with its weak-* topology and set $v(\theta,t)=t(m(\theta))$. As $T_0=0$, $\ell(\theta)=\norm{m(\theta)}_{\mathfrak{X}}$, so $m(\theta)=0$ for every $\theta\in\Theta_{\mathrm{I}}$. Hence $v(\theta,t)=0$ for all $(\theta,t)\in\Theta_{\mathrm{I}}\times\mathcal{T}$, and $\partial\xi(m(\theta))=B_{\mathfrak{X}^*}(0,1)$ on $\Theta_{\mathrm{I}}$.

    As in the proof of Proposition~\ref{P:ex1}, Lemma~\ref{L:tangentconecond} shows that $t\in\mathcal{K}_0(\theta)$ if and only if
    \[
    t \in (-m'(\theta)(T_\theta\Theta))^\circ \cap B_{\mathfrak{X}^*}(0,1)
    \]
    for every $\theta$ in the neighborhood of $\Theta_{\mathrm{I}}$ on which $\Theta$ is convex.
    The assumed condition on $\theta\mapsto -m'(\theta)(T_\theta\Theta)\cap B_{\mathfrak{X}}(0,1)$, which is weaker than upper hemicontinuity with respect to the weak topology, therefore implies norm lower hemicontinuity of $\theta\mapsto\mathcal{K}_0(\theta)$ by Lemma~\ref{L:duality2}, and hence lower hemicontinuity in the weak-* topology, which is the relevant topology on $\mathcal{T}$ (see the proof of Proposition~\ref{P:ex1}). Assumptions~\ref{A:zero}--\ref{A:three} and \ref{A:five} are verified as in the proof of Proposition~\ref{P:ex2}, and Addendum~\ref{add:bootstrap_K0} applies.
    The limiting expression produced by Addendum~\ref{add:bootstrap_K0} is the supremum of $\mathbb{W}_n^*(\theta)$ over the polar of the cone $-m'(\theta)(T_\theta\Theta)$ intersected with the unit ball of $\mathfrak{X}^*$; Lemma~\ref{L:positive} identifies this support function with the distance from $\mathbb{W}_n^*(\theta)$ to $-m'(\theta)(T_\theta\Theta)$.

    Finally, with $s_n=\lambda_n b_n^{-1}$,
    \[
\sup_{t\in B_{\mathfrak{X}^*}(0,1)}\left\{s_n t\big(m_n(\theta')-m_n(\theta)\big)+t(\mathbb{W}_n^*(\theta))\right\}
    =
    \norm{s_n\big(m_n(\theta')-m_n(\theta)\big)+\mathbb{W}_n^*(\theta)}_{\mathfrak{X}},
    \]
    which gives the displayed statistic and the stated limiting distribution.
\end{proof}

Collar upper hemicontinuity allows the unit-ball-truncated derivative cone to rotate or shrink as $\theta$ approaches an identified value $\theta_0$, but not to acquire new limiting directions outside the derivative cone at $\theta_0$ (left panel of Figure~\ref{F:geometry_norm}\suppref{}).

\begin{rem}[Constant rank]
\label{R:constantrank_norm}
Let $\mathfrak{X}=\R^d$ and let $\theta_0\in\Theta_{\mathrm{I}}$ be interior to $\Theta\subseteq\R^k$, so that, for $\theta$ near $\theta_0$, $-m'(\theta)(T_\theta \Theta)$ is the column space of $m'(\theta)$.  For norm loss, it suffices for Proposition~\ref{add:xi_norm} that $\theta\mapsto -m'(\theta)(T_\theta \Theta)\cap B_{\mathfrak{X}}(0,1)$ be upper hemicontinuous, without the membership condition~\eqref{E:conemembership} of Remark~\ref{R:constantrank}, because only the stationarity set enters.  This holds when $m'$ is continuous and has constant rank near $\theta_0$.  Constant rank near every point of $\Theta_{\mathrm{I}}$, each interior to $\Theta$, also gives local linear growth of $\norm{m(\theta)}_{\mathfrak{X}}$ in $d_{\mathfrak{B}}(\theta,\Theta_{\mathrm{I}})$, so that, when $\Theta$ is bounded, \eqref{E:localident1} holds with $\alpha_\ell=1$ and Proposition~\ref{P:ex1} holds with equality.  Constant rank generalizes the full-column-rank condition behind the GMM $J$ statistic.  Full column rank makes $\Theta_{\mathrm{I}}$ locally a point, while constant rank $r<k$ leaves it a $(k-r)$-dimensional manifold along which the derivative cone may vary continuously.
\end{rem}

\subsection{Uniform validity}

The bootstrap of Addendum~\ref{add:bootstrap_K0} is, up to a uniformly $o_{p^*}(1)$ term, at least the bootstrap of Theorem~\ref{P:unif1} with $s_n=\lambda_nb_n^{-1}$. This comparison uses the inequality $\ell_n(\theta)\ge v_n(\theta,t)$ and the uniform equicontinuity of $\mathbb{V}^*_{n,R,P}$, which follows from Assumptions~\ref{AU:2} and~\ref{AU:4} and the uniform convergence of $\mathbb{V}^*_{n,R,P}$ to $\mathbb{G}_{n,R,P}$. This bootstrap is therefore also a uniform upper-bound estimator under the conditions of Theorem~\ref{P:unif1}, and, under the hypotheses of Addendum~\ref{add:bootstrap_K0}, it is pointwise exact when
$\mathcal{K}=\mathcal{K}_0$ on $\Theta_{\mathrm{I}}$ and Theorem~\ref{T:one} holds with equality.  The triangle inequality shows that the bootstrap of Proposition~\ref{add:xi_norm} is also, up to a uniformly $o_{p^*}(1)$ term, at least that of Proposition~\ref{P:unifbound} with $s_n=\lambda_nb_n^{-1}$.  It is therefore a uniform upper-bound estimator under the conditions of Proposition~\ref{P:unifbound}, and, under the hypotheses of Proposition~\ref{add:xi_norm}, it is pointwise exact when Proposition~\ref{P:ex1} holds with equality.

\section{Studentization with estimated weights}\label{S:studentized_detail}

This section treats the second and third forms of studentization in Section~\ref{S:studentized}.  Both measure the distance from $\Sigma_n(\theta) m_n(\theta)$ to $\Sigma_n(\theta) C$ for a convex constraint set $C$.  Here, $\Sigma_n(\theta)$ is the linear operator induced by the studentization, such as $\Omega_n(\theta)^{-1/2}$, taken in the space $\mathcal{L}(\mathfrak{X})$ of bounded linear operators on $\mathfrak{X}$ with the operator norm.  The two corollaries below address the sampling error in $\Sigma_n$ in complementary ways.  Corollary~\ref{C:studentized} uses a uniform lower bound on $\Sigma$, under which replacing $\Sigma_n$ by $\Sigma$ changes the statistic by at most a factor $1+O_p(\sup_\theta\norm{\Sigma_n(\theta)-\Sigma(\theta)}_{\mathrm{op}})$, so that the error $\Sigma_n-\Sigma$ has no first-order effect.  It covers full weight matrices, as in the quasi-likelihood-ratio statistic of \citet{CS2023}.  Corollary~\ref{C:studentized_invariant} drops the lower bound and instead assumes that $\Sigma_n$ and $\Sigma$ leave $C$ invariant.

\subsection{Limit theory}\label{S:studentized_limit}

The following lemma allows us to remove the effect of estimating $\Sigma$ in the test statistic and the bootstrap draws.

\begin{lemma}[Weight replacement]\label{L:weightreplace}
    Let $C\subset\mathfrak{X}$, $x\in\ell^\infty(\Theta,\mathfrak{X})$, and $A,B:\Theta\to\mathcal{L}(\mathfrak{X})$, and set
    $F_A(x)=\inf_{\theta\in\Theta,\,c\in C}\norm{A(\theta)(x(\theta)-c)}_{\mathfrak{X}}=\inf_{\theta\in\Theta}d_{\mathfrak{X}}(A(\theta)x(\theta),A(\theta)C)$.  Suppose that, for some
    $\varepsilon\ge0$,
    \begin{align}\label{E:wtreplace}
        (1+\varepsilon)^{-1}\norm{B(\theta)y}_{\mathfrak{X}}\le\norm{A(\theta)y}_{\mathfrak{X}}\le(1+\varepsilon)\norm{B(\theta)y}_{\mathfrak{X}}
        \qquad\text{for all }\theta\in\Theta\text{ and }y\in\mathfrak{X}.
    \end{align}
    Then $(1+\varepsilon)^{-1}F_B(x)\le F_A(x)\le(1+\varepsilon)F_B(x)$.  The hypothesis holds with
    $\varepsilon=\varepsilon'/(1-\varepsilon')$ if $\norm{(A(\theta)-B(\theta))y}_{\mathfrak{X}}\le\varepsilon'\norm{B(\theta)y}_{\mathfrak{X}}$
    for all $\theta$ and $y$, for some $\varepsilon'<1$.  In particular, it holds with $\varepsilon=\delta/(\kappa-\delta)$ if
    $\norm{B(\theta)y}_{\mathfrak{X}}\ge\kappa\norm{y}_{\mathfrak{X}}$ for all $\theta$ and $y$ and
    $\delta=\sup_\theta\norm{A(\theta)-B(\theta)}_{\mathrm{op}}<\kappa$.
\end{lemma}
\begin{proof}
    The first conclusion is immediate from \eqref{E:wtreplace}.
    If, in addition, $\norm{(A(\theta)-B(\theta))y}_{\mathfrak{X}}\le\varepsilon'\norm{B(\theta)y}_{\mathfrak{X}}$, then
    $\norm{A(\theta)y}_{\mathfrak{X}}\le(1+\varepsilon')\norm{B(\theta)y}_{\mathfrak{X}}\le(1+\varepsilon)\norm{B(\theta)y}_{\mathfrak{X}}$ and
    $\norm{B(\theta)y}_{\mathfrak{X}}\le\norm{A(\theta)y}_{\mathfrak{X}}+\varepsilon'\norm{B(\theta)y}_{\mathfrak{X}}$, so that
    $\norm{B(\theta)y}_{\mathfrak{X}}\le(1-\varepsilon')^{-1}\norm{A(\theta)y}_{\mathfrak{X}}=(1+\varepsilon)\norm{A(\theta)y}_{\mathfrak{X}}$.
    Under the lower bound, $\norm{(A(\theta)-B(\theta))y}_{\mathfrak{X}}\le\delta\norm{y}_{\mathfrak{X}}\le(\delta/\kappa)\norm{B(\theta)y}_{\mathfrak{X}}$,
    and setting $\varepsilon'=\delta/\kappa$ gives $\varepsilon=\delta/(\kappa-\delta)$.
\end{proof}

\begin{cor}[Studentization] \label{C:studentized}
    Let Assumption~\ref{A:ex1} hold with
    $\xi(x)=d_\mathfrak{X}(x,C)$ for some convex set $C$.  Set
    $\ell=\xi\circ m$, and suppose that
    $\Theta_{\mathrm{I}}=\{\theta\in\Theta:\ell(\theta)=0\}$ is nonempty.
    Let $\Sigma:\Theta\to\mathcal{L}(\mathfrak{X})$ be bounded and uniformly continuous, with
    $\inf_{\theta\in\Theta}\inf_{\norm{x}_{\mathfrak{X}}=1}\norm{\Sigma(\theta)x}_{\mathfrak{X}}>0$, and suppose that,
    for every bounded $B\subset\mathfrak{X}$, $(\theta,c)\mapsto\Sigma(\theta)(m(\theta)-c)$ is uniformly
    Fréchet differentiable over $\Theta\times B$ with some order $\widetilde\alpha>1$.  Let
    $\Sigma_n:\Theta\to\mathcal{L}(\mathfrak{X})$ satisfy
    $\sup_{\theta\in\Theta}\norm{\Sigma_n(\theta)-\Sigma(\theta)}_{\mathrm{op}}=o_p(1)$.  Then,
    \begin{align}
        r_n \inf_{\theta \in \Theta} d_{\mathfrak{X}}(\Sigma_n(\theta) m_n(\theta), \Sigma_n(\theta) C) & \le \inf_{\theta \in \Theta_{\mathrm{I}}} d_{\mathfrak{X}}\big(\Sigma(\theta) \mathbb{W}_n(\theta), \Sigma(\theta)  T_{m(\theta)} C -\Sigma(\theta) m'(\theta) (T_\theta \Theta) \big) + o_p(1)\nonumber \\
        & \rs \inf_{\theta \in \Theta_{\mathrm{I}}} d_{\mathfrak{X}}\big(\Sigma(\theta) \mathbb{X}_n(\theta), \Sigma(\theta) T_{m(\theta)}  C - \Sigma(\theta) m'(\theta) (T_\theta \Theta) \big), \label{E:studentized}
    \end{align}
    with equality if $\Theta$ is a singleton, $(\theta,c)\mapsto\Sigma(\theta)(m(\theta)-c)$ is affine, or $m$ is Lipschitz near $\Theta_{\mathrm{I}}$ and there is an
    $\alpha_\ell$ satisfying $\alpha_\ell\widetilde\alpha>1$ and
    \eqref{E:localident1}.
\end{cor}

By the product rule, the differentiability condition of Corollary~\ref{C:studentized} holds with $\widetilde\alpha=\min\{\alpha_m,\alpha_\Sigma,2\}$ when $m$ and $\Sigma$ are Lipschitz and $\Sigma$ is uniformly Fréchet differentiable with order $\alpha_\Sigma$.

\begin{proof}
    Let $T_n$ denote the statistic on the left side of \eqref{E:studentized}, and let
    $\widetilde T_n=r_n\inf_{\theta\in\Theta}\inf_{c\in C}\norm{\Sigma(\theta)(m_n(\theta)-c)}_{\mathfrak{X}}$ be its
    known-weight counterpart.  Replacing $C$ by its closure leaves $T_n$, $\widetilde T_n$, and $\ell$ unchanged, because
    $AC\subseteq A\overline C\subseteq\overline{AC}$ for $A=\Sigma_n(\theta)$ and $A=\Sigma(\theta)$.  By
    Lemma~\ref{L:tangentclosure}, it also leaves the closure of the cone in \eqref{E:studentized} unchanged.  We therefore
    assume that $C$ is closed.  Let $\delta_n=\sup_\theta\norm{\Sigma_n(\theta)-\Sigma(\theta)}_{\mathrm{op}}=o_p(1)$.
    By the lower bound on $\Sigma$, there is $\kappa>0$ with $\norm{\Sigma(\theta)x}_{\mathfrak{X}}\ge2\kappa\norm{x}_{\mathfrak{X}}$
    for all $\theta\in\Theta$ and $x\in\mathfrak{X}$.  On the event $\{\delta_n<2\kappa\}$, whose
    probability approaches one, Lemma~\ref{L:weightreplace} with $x=m_n$, $A=\Sigma_n$, and $B=\Sigma$ gives
    \begin{align}
        |T_n-\widetilde T_n|\le\frac{\delta_n}{2\kappa-\delta_n}\,\widetilde T_n.
        \label{E:studreplace}
    \end{align}
    Taking $\theta=\theta_0\in\Theta_{\mathrm{I}}$ and $c=m(\theta_0)\in C$ gives
    $\widetilde T_n\le r_n\norm{\Sigma(\theta_0)}_{\mathrm{op}}\norm{m_n(\theta_0)-m(\theta_0)}_{\mathfrak{X}}=O_p(1)$, so the
    right side of \eqref{E:studreplace} is $O_p(\delta_n)=o_p(1)$.

    Fix $\eta>0$, and let $U$ be a bounded closed ball whose interior contains every point within distance $\eta$ of
    $m(\Theta)$, which exists because $m(\Theta)$ is bounded by Assumption~\ref{A:ex1}(2).  Because
    $\norm{\Sigma(\theta)y}_{\mathfrak{X}}\ge2\kappa\norm{y}_{\mathfrak{X}}$, a pair $(\theta,c)$ with
    $\norm{m(\theta)-c}_{\mathfrak{X}}>\eta$ has objective
    \begin{align}\label{E:studlocalize}
        r_n\norm{\Sigma(\theta)(m_n(\theta)-c)}_{\mathfrak{X}}\ge2\kappa\big(r_n\eta-\sup_{\theta'}\norm{\mathbb{W}_n(\theta')}_{\mathfrak{X}}\big).
    \end{align}
    On the event $\{\widetilde T_n<2\kappa(r_n\eta-\sup_\theta\norm{\mathbb{W}_n(\theta)}_{\mathfrak{X}})\}$, such pairs
    therefore do not contribute to the infimum defining $\widetilde T_n$ and may be discarded, so that $\widetilde T_n$
    is unchanged if $C$ is replaced by $\widetilde C=C\cap U$.  This event has probability approaching one, because
    $\widetilde T_n$ and $\sup_\theta\norm{\mathbb{W}_n(\theta)}_{\mathfrak{X}}$ are $O_p(1)$.

    When rewritten in this way, $\widetilde T_n$ is the statistic of Corollary~\ref{C:geometry} with constraint set $\{0\}$, parameter space
    $\widetilde\Theta=\Theta\times\widetilde C$ in $\mathfrak{B}\times\mathfrak{X}$ with the seminorm
    $\norm{\theta}_{\mathfrak{B}}+\norm{c}_{\mathfrak{X}}$, maps
    $\widetilde m(\theta,c)=\Sigma(\theta)(m(\theta)-c)$ and
    $\widetilde m_n(\theta,c)=\Sigma(\theta)(m_n(\theta)-c)$, and lifted loss $\widetilde\ell=\norm{\widetilde m}_{\mathfrak{X}}$.  We verify Assumption~\ref{A:ex1}
    for this lifted problem.  Because $\Sigma(\theta)$ is injective and $C$ is closed, $\widetilde\ell(\theta,c)=0$ exactly when
    $\theta\in\Theta_{\mathrm{I}}$ and $c=m(\theta)$, so the lifted identified set is
    $\widetilde\Theta_{\mathrm{I}}=\{(\theta,m(\theta)):\theta\in\Theta_{\mathrm{I}}\}$.  This set is compact by
    continuity of $m$, and $\widetilde\Theta$ is convex near it.  Moreover,
    $\widetilde\ell(\theta,c)\ge2\kappa\norm{m(\theta)-c}_{\mathfrak{X}}\ge2\kappa\,\ell(\theta)$, so small values
    of $\widetilde\ell$ force $\theta$ near the compact set $\Theta_{\mathrm{I}}$ and $c$ near $m(\theta)$.  A standard
    compactness argument, using continuity of $m$, then forces $(\theta,c)$ to lie near
    $\widetilde\Theta_{\mathrm{I}}$, which gives \eqref{E:elldeltaeqn}.  Parts~(1) and~(3) therefore hold.  Part~(2) holds because $\Sigma$, $m$, and
    $\widetilde C$ are bounded, and part~(4) is the differentiability hypothesis.  Parts~(5) and~(6) hold because
    $\Sigma(\theta)\mathbb{W}_n(\theta)$ converges to $\Sigma(\theta)\mathbb{X}_n(\theta)$ by boundedness of $\Sigma$, and this limit is equicontinuous in probability because $\Sigma$ is also uniformly continuous.

    It remains to characterize the limit.  For $\theta\in\Theta_{\mathrm{I}}$, continuity of $\Sigma$ gives
    $\widetilde m(\theta+\tau h,m(\theta)+\tau k)=\tau\Sigma(\theta)(m'(\theta)h-k)+o(\tau)$.  As $m(\theta)$
    is interior to $U$, $T_{(\theta,m(\theta))}\widetilde\Theta=T_\theta\Theta\times T_{m(\theta)}C$, so
    the cone in Corollary~\ref{C:geometry} is
    \[
        -\widetilde m'(\theta,m(\theta))\big(T_{(\theta,m(\theta))}\widetilde\Theta\big)
        =\Sigma(\theta)T_{m(\theta)}C-\Sigma(\theta)m'(\theta)(T_\theta\Theta).
    \]
    This yields \eqref{E:studentized}.

    The equality clause follows from that of Corollary~\ref{C:geometry}.  It is immediate when
    $\widetilde m$ is affine, in particular when $\Theta$ is a singleton.  Otherwise, where $\widetilde\ell$ is small,
    choosing $\theta_0\in\Theta_{\mathrm{I}}$ nearest to $\theta$ gives, by \eqref{E:localident1},
    \[
        d((\theta,c),\widetilde\Theta_{\mathrm{I}})
        \le(1+L)d_{\mathfrak{B}}(\theta,\Theta_{\mathrm{I}})+\norm{m(\theta)-c}_{\mathfrak{X}}
        \le K\big(\widetilde\ell(\theta,c)^{\alpha_\ell}+\widetilde\ell(\theta,c)\big),
    \]
    where $L$ is a Lipschitz constant for $m$ near $\Theta_{\mathrm{I}}$ and $K$ depends only on $L$, $\kappa$,
    and the bound in \eqref{E:localident1}.  Elsewhere $d((\theta,c),\widetilde\Theta_{\mathrm{I}})$ is bounded and $\widetilde\ell$ is bounded away
    from zero.  Hence \eqref{E:localident1} holds for the lifted problem with exponent
    $\min\{\alpha_\ell,1\}$, and $\min\{\alpha_\ell,1\}\widetilde\alpha>1$.
\end{proof}

Corollary~\ref{C:studentized_invariant} instead measures the distance from the studentized moments to $C$ itself and requires the weight to leave $C$ unchanged, in the sense that $\Sigma_n(\theta)C=\Sigma(\theta)C=C$.  This holds for diagonal studentization of equality and inequality constraints, which positive rescalings leave invariant.  The results of \citet{BCS2015} cover the square of such a statistic, with $C$ encoding which moments enter as equalities and which as inequalities.  The condition also holds trivially for the pure equality case $C=\{0\}$.

\begin{cor}[Invariant studentization] \label{C:studentized_invariant}
    Let $C\subset\mathfrak{X}$ be closed and convex, and let $\Sigma:\Theta\to\mathcal{L}(\mathfrak{X})$ be bounded
    and satisfy $\Sigma(\theta)C=C$ for all $\theta$.  Let Assumption~\ref{A:ex1} hold with
    $\xi(x)=d_{\mathfrak{X}}(x,C)$, population map $\theta\mapsto\Sigma(\theta)m(\theta)$, and
    empirical map $\theta\mapsto\Sigma(\theta)m_n(\theta)$, and let $\mathbb{W}_n:=r_n(m_n-m)\rs\mathbb{X}_n$ in
    $\ell^\infty(\Theta,\mathfrak{X})$.  Set
    \[
        \ell(\theta)=d_{\mathfrak{X}}(\Sigma(\theta)m(\theta),C),
        \qquad
        \Theta_{\mathrm{I}}=\{\theta\in\Theta:\ell(\theta)=0\},
    \]
    and suppose that $\Theta_{\mathrm{I}}$ is nonempty.  Suppose further that, with probability approaching one,
    $\Sigma_n(\theta)C=C$ and
    \begin{align}\label{E:normequiv}
        (1+\varepsilon_n)^{-1}\norm{\Sigma(\theta)y}_{\mathfrak{X}}\le\norm{\Sigma_n(\theta)y}_{\mathfrak{X}}
        \le(1+\varepsilon_n)\norm{\Sigma(\theta)y}_{\mathfrak{X}}
    \end{align}
    for all $\theta\in\Theta$ and $y\in\mathfrak{X}$, for some $\varepsilon_n=o_p(1)$.  Then,
    \begin{align}
        r_n \inf_{\theta \in \Theta}
        d_{\mathfrak{X}}(\Sigma_n(\theta) m_n(\theta), C)
        & \le \inf_{\theta \in \Theta_{\mathrm{I}}}
        d_{\mathfrak{X}}\big(\Sigma(\theta) \mathbb{W}_n(\theta),
        T_{\Sigma(\theta)m(\theta)}C-(\Sigma m)'(\theta)(T_\theta\Theta)\big)+o_p(1)\nonumber\\
        & \rs \inf_{\theta \in \Theta_{\mathrm{I}}}
        d_{\mathfrak{X}}\big(\Sigma(\theta) \mathbb{X}_n(\theta),
        T_{\Sigma(\theta)m(\theta)}C-(\Sigma m)'(\theta)(T_\theta\Theta)\big),
        \label{E:studentized_invariant}
    \end{align}
    with equality whenever the equality conditions of Corollary~\ref{C:geometry} hold for the map
    $\theta\mapsto\Sigma(\theta)m(\theta)$.
\end{cor}
For diagonal weights $\Sigma(\theta)=\mathrm{diag}(\sigma_j(\theta)^{-1})$ and $\Sigma_n(\theta)=\mathrm{diag}(\hat\sigma_j(\theta)^{-1})$ under an $\ell_p$ norm, the condition \eqref{E:normequiv} holds with $\varepsilon_n=\sup_\theta\max_j\max\{\hat\sigma_j(\theta)/\sigma_j(\theta),\sigma_j(\theta)/\hat\sigma_j(\theta)\}-1$, so it requires only uniform consistency of the ratios $\hat\sigma_j/\sigma_j$.

Because $\Sigma(\theta)C=C$ does not vary with $\theta$, the proof of Corollary~\ref{C:studentized_invariant} replaces $\Sigma_n$ by $\Sigma$ using \eqref{E:normequiv} and applies Corollary~\ref{C:geometry} directly over $\Theta$.

\begin{proof}
    Let $T_n=r_n\inf_{\theta\in\Theta}d_{\mathfrak{X}}(\Sigma_n(\theta)m_n(\theta),C)$ and
    $\widetilde T_n=r_n\inf_{\theta\in\Theta}d_{\mathfrak{X}}(\Sigma(\theta)m_n(\theta),C)$.  On the event that
    $\Sigma_n(\theta)C=C$ and \eqref{E:normequiv} hold for all $\theta$ and $y$, whose probability approaches one,
    $d_{\mathfrak{X}}(\Sigma_n(\theta)z,C)=d_{\mathfrak{X}}(\Sigma_n(\theta)z,\Sigma_n(\theta)C)$ and
    $d_{\mathfrak{X}}(\Sigma(\theta)z,C)=d_{\mathfrak{X}}(\Sigma(\theta)z,\Sigma(\theta)C)$.
    Lemma~\ref{L:weightreplace} with $x=m_n$, $A=\Sigma_n$, and $B=\Sigma$ therefore gives
    $|T_n-\widetilde T_n|\le\varepsilon_n\widetilde T_n$, as in \eqref{E:studreplace}.  Because
    $\Sigma(\theta_0)m(\theta_0)\in C$ for $\theta_0\in\Theta_{\mathrm{I}}$,
    $\widetilde T_n\le\norm{\Sigma(\theta_0)}_{\mathrm{op}}\norm{\mathbb{W}_n(\theta_0)}_{\mathfrak{X}}=O_p(1)$, and
    $T_n=\widetilde T_n+o_p(1)$.  The statistic $\widetilde T_n$ is that of Corollary~\ref{C:geometry} for the
    population map $\Sigma m$ and empirical map $\Sigma m_n$, whose process
    $r_n(\Sigma m_n-\Sigma m)=\Sigma\mathbb{W}_n$ converges to $\Sigma\mathbb{X}_n$ by the continuous mapping
    theorem, because $\Sigma$ is bounded.  Corollary~\ref{C:geometry} therefore yields
    \eqref{E:studentized_invariant}, with equality under its conditions.
\end{proof}

\subsection{Bootstrap}\label{S:studentized_boot}

It is straightforward to state conditions under which \eqref{E:studentized_bootstrap} is consistent for the limits of both corollaries.  Applied as in their proofs, with $x=m_n+s_n^{-1}\mathbb W_n^*$, Lemma~\ref{L:weightreplace} shows that replacing $\Sigma_n$ by $\Sigma$ in \eqref{E:studentized_bootstrap} changes each draw by a vanishing multiple of the known-weight draw.  That draw is bounded in probability, so \eqref{E:studentized_bootstrap} is asymptotically equivalent to the bootstrap of Proposition~\ref{P:ex2} for the corresponding known-weight problem, in which $\Sigma$ replaces $\Sigma_n$.  For Corollary~\ref{C:studentized_invariant}, \eqref{E:studentized_bootstrap} is therefore consistent for its limit when the hypotheses of Proposition~\ref{P:ex2} hold with loss $d_{\mathfrak{X}}(\cdot,C)$, population map $\Sigma m$, empirical map $\Sigma m_n$, and bootstrap process $\Sigma\mathbb W_n^*$.

For Corollary~\ref{C:studentized}, the known-weight draw is
$s_n\inf_{\theta\in\Theta,\,c\in C}\norm{\Sigma(\theta)(m_n(\theta)+s_n^{-1}\mathbb W_n^*(\theta)-c)}_{\mathfrak{X}}$.
As in the proof of that corollary (see \eqref{E:studlocalize}), pairs $(\theta,c)$ with $c$ far from $m(\theta)$
are, with probability approaching one, never near the infimum in the draws, so $C$ may be replaced by the bounded
set $\widetilde C=C\cap U$.  The draw is then the bootstrap of Proposition~\ref{P:ex2} for the lifted norm problem
of that proof, with loss $\norm{\cdot}_{\mathfrak{X}}$, parameter space $\Theta\times\widetilde C$, maps
$(\theta,c)\mapsto\Sigma(\theta)(m(\theta)-c)$ and $(\theta,c)\mapsto\Sigma(\theta)(m_n(\theta)-c)$, and bootstrap
process $\Sigma(\theta)\mathbb W_n^*(\theta)$, which does not depend on $c$.  Under the hypotheses of
Corollary~\ref{C:studentized}, Assumption~\ref{A:ex1} holds for this problem, and Assumption~\ref{A:ex2} holds when
it holds for the original problem and $\Sigma$ is Lipschitz.  Under the conditions of the equality clause of
Corollary~\ref{C:studentized}, Proposition~\ref{P:ex2} therefore gives consistency of
\eqref{E:studentized_bootstrap} for the limiting distribution on the right side of \eqref{E:studentized}.  This follows from case~(ii) when the lifted map is affine,
and otherwise from case~(iii) with exponent $\min\{\alpha_\ell,1\}$ and $\widetilde\alpha$ in place of $\alpha_m$.

\subsection{Uniform inference}

Corollary~\ref{C:unifstud} extends Proposition~\ref{P:unifbound} uniformly to the studentized statistics, with the bootstrap \eqref{E:studentized_bootstrap} computed over $R$ and the estimated weight $\Sigma_{n,R,P}$ held fixed across draws.

\begin{cor}[Uniform studentized inference]\label{C:unifstud}
    Let $C\subset\mathfrak{X}$ be closed and convex, and for each $(R,P)\in\mathcal{S}$ let
    $\Sigma_{R,P}:\Theta\to\mathcal{L}(\mathfrak{X})$ be a weight with
    $\sup_{(R,P)\in\mathcal{S}}\sup_{\theta}\norm{\Sigma_{R,P}(\theta)}_{\mathrm{op}}<\infty$.  Let $s_n\ra\infty$ with
    $s_n=o(r_n)$, and suppose that one of the following holds.
    \begin{enumerate}
        \item (General weights.)  Assumption~\ref{A:uniformbound} holds for the loss $d_{\mathfrak{X}}(\cdot,C)$, population
        maps $m_{R,P}$, and empirical maps $m_{n,R,P}$.  The
        weights satisfy
        $\inf_{(R,P)\in\mathcal{S}}\inf_{\theta}\inf_{\norm{x}_{\mathfrak{X}}=1}\norm{\Sigma_{R,P}(\theta)x}_{\mathfrak{X}}>0$
        and are uniformly Lipschitz.  For
        every bounded $B\subset\mathfrak{X}$, $(\theta,c)\mapsto\Sigma_{R,P}(\theta)(m_{R,P}(\theta)-c)$ is uniformly
        Fr\'{e}chet differentiable over $R\times B$ with order $\widetilde\alpha>1$, uniformly in $(R,P)$, and one of the
        conditions of Proposition~\ref{P:unifbound} holds with $\widetilde\alpha$ in place of $\alpha_m$.  Finally,
        $\sup_\theta\norm{\Sigma_{n,R,P}(\theta)-\Sigma_{R,P}(\theta)}_{\mathrm{op}}=o_p(1)$ uniformly in $(R,P)$.
        \item (Invariant weights.)  $\Sigma_{R,P}(\theta)C=C$ for all $\theta$ and $(R,P)$.  Assumption~\ref{A:uniformbound}(1)--(5)
        and one of the conditions of Proposition~\ref{P:unifbound} hold for the loss $d_{\mathfrak{X}}(\cdot,C)$, population
        maps $\Sigma_{R,P}m_{R,P}$, empirical maps $\Sigma_{R,P}m_{n,R,P}$, and bootstrap processes
        $\Sigma_{R,P}\mathbb{W}^*_{n,R,P}$.  On events whose probability approaches one uniformly in $(R,P)$,
        $\Sigma_{n,R,P}(\theta)C=C$ and \eqref{E:normequiv} holds with $\Sigma_{n,R,P}$, $\Sigma_{R,P}$, and $\varepsilon_{n,R,P}$ in place of $\Sigma_n$,
        $\Sigma$, and $\varepsilon_n$, where $\varepsilon_{n,R,P}=o_p(1)$ uniformly in $(R,P)$.
    \end{enumerate}
    Then, uniformly in $(R,P)\in\mathcal{S}$ with $T_0(R,P)=0$,
    \begin{align}
        &r_n\inf_{\theta\in R}d_{\mathfrak{X}}\big(\Sigma_{n,R,P}(\theta)m_{n,R,P}(\theta),\Sigma_{n,R,P}(\theta)C\big)\nonumber\\
        &\qquad\le s_n\inf_{\theta\in R}d_{\mathfrak{X}}\Big(\Sigma_{R,P}(\theta)\big(m_{n,R,P}(\theta)+s_n^{-1}\mathbb{W}_{n,R,P}(\theta)\big),
        \Sigma_{R,P}(\theta)C\Big)+o_p(1),\label{E:unifstud}
    \end{align}
    and the bootstrap \eqref{E:studentized_bootstrap}, computed over $R$ with weight $\Sigma_{n,R,P}$ and bootstrap processes
    $\mathbb{W}^*_{n,R,P}$, is uniformly consistent in the bounded-Lipschitz metric for the law of the right side of
    \eqref{E:unifstud}.
\end{cor}

Case~(1) covers weight matrices with eigenvalues bounded away from zero, as in the quasi-likelihood-ratio statistic of \citet{CS2023}, and case~(2) covers diagonal studentization when the ratios $\hat\sigma_j/\sigma_j$ converge to one uniformly.

\begin{proof}
    Fix $(R,P)\in\mathcal{S}$ with $T_0(R,P)=0$. We suppress the dependence on $(R,P)$ in subscripts, and all bounds below hold uniformly in
    $(R,P)$.  Let $T_n$ denote the left side of \eqref{E:unifstud}, $\widetilde T_n$ its value with $\Sigma$ in place of
    $\Sigma_n$, and $Z_n$ the right side of \eqref{E:unifstud} without the $o_p(1)$ term.  Let $Z_n^*$ be the draw \eqref{E:studentized_bootstrap} computed over $R$, and let
    $\widetilde Z_n^*$ be $Z_n$ with $\mathbb{W}_n$ replaced by $\mathbb{W}_n^*$.  As in the proofs of Corollaries~\ref{C:studentized} and~\ref{C:studentized_invariant},
    $\Sigma(\theta_0)m(\theta_0)\in\Sigma(\theta_0)C$ for $\theta_0\in\Theta_{\mathrm{I}}(R,P)$, so evaluating the infima
    defining $\widetilde T_n$, $Z_n$, and $\widetilde Z_n^*$ at $\theta_0$ bounds
    $\widetilde T_n$ by $\norm{\Sigma\mathbb{W}_n}_\infty$ and $Z_n$ and $\widetilde Z_n^*$ by
    $(s_n/r_n)\norm{\Sigma\mathbb{W}_n}_\infty+\norm{\Sigma\mathbb{W}_n^{(*)}}_\infty$, where $\mathbb{W}_n^{(*)}$ is
    $\mathbb{W}_n$ or $\mathbb{W}_n^*$.  These bounds are $O_p(1)$ uniformly by Assumption~\ref{A:uniformbound}(4)--(5),
    in case~(1) because
    $\norm{\Sigma}_\infty=\sup_{(R,P),\,\theta}\norm{\Sigma_{R,P}(\theta)}_{\mathrm{op}}<\infty$.

    \textbf{Replacing the estimated weight:} Lemma~\ref{L:weightreplace}, applied as in the proofs of
    Corollaries~\ref{C:studentized} and~\ref{C:studentized_invariant} with $\Theta$ replaced by $R$, gives
    $|T_n-\widetilde T_n|\le\varepsilon_n\widetilde T_n$ and $|Z_n^*-\widetilde Z_n^*|\le\varepsilon_n\widetilde Z_n^*$ on
    events whose probability approaches one uniformly, with $\varepsilon_n=\delta_n/(2\kappa-\delta_n)$ in case~(1) for
    $\delta_n=\sup_\theta\norm{\Sigma_n(\theta)-\Sigma(\theta)}_{\mathrm{op}}$ and $2\kappa$ the uniform lower bound on the
    weights.  Hence $T_n=\widetilde T_n+o_p(1)$ uniformly, and the bootstrap laws of
    $Z_n^*$ and $\widetilde Z_n^*$ are uniformly close in the bounded-Lipschitz metric.  It therefore suffices to
    show that $\widetilde T_n\le Z_n+o_p(1)$ uniformly and that the bootstrap law of $\widetilde Z_n^*$ is uniformly
    consistent for $Z_n$.

    \textbf{Case~(1):} fix $\eta>0$, and let $U$ be a bounded closed ball whose interior contains every point within distance $\eta$ of
    $\bigcup_{(R,P)}m_{R,P}(\Theta)$, which is bounded by Assumption~\ref{A:uniformbound}(2).  By \eqref{E:studlocalize}
    and its analogues for $Z_n$ and $\widetilde Z_n^*$, with $r_n$ replaced by $s_n$, pairs $(\theta,c)$ with
    $\norm{m(\theta)-c}_{\mathfrak{X}}>\eta$ do not contribute to the infima in $\widetilde T_n$, $Z_n$, and
    $\widetilde Z_n^*$ on events whose probability approaches one uniformly, because these quantities and
    $\norm{\mathbb{W}_n}_\infty$, $\norm{\mathbb{W}_n^*}_\infty$ are $O_p(1)$ uniformly.  On these events, $C$ may be
    replaced by $\widetilde C=C\cap U$, and $\widetilde T_n$, $Z_n$, and $\widetilde Z_n^*$ become the statistic, the
    bound, and the bootstrap of Proposition~\ref{P:unifbound} for the lifted family, indexed by the pairs $(R,P)\in\mathcal{S}$ with $T_0(R,P)=0$, with loss
    $\norm{\cdot}_{\mathfrak{X}}$, parameter sets $R\times\widetilde C$, maps
    $\widetilde m_{R,P}(\theta,c)=\Sigma(\theta)(m(\theta)-c)$ and $\widetilde m_{n,R,P}(\theta,c)=\Sigma(\theta)(m_n(\theta)-c)$,
    and processes $\Sigma\mathbb{W}_n$ and $\Sigma\mathbb{W}_n^*$.  It remains to verify Assumption~\ref{A:uniformbound}(1)--(5)
    and the conditions of Proposition~\ref{P:unifbound} for this family.  Parts~(1), (2), (4), and~(5) hold as in the
    proof of Corollary~\ref{C:studentized}, uniformly in $(R,P)$ because the bounds used there depend only on
    $\norm{\Sigma}_\infty$, the Lipschitz constants, and $U$.  Part~(3) follows from the uniform Fr\'{e}chet
    differentiability hypothesis.

    It remains to bound the modulus $\overline f_{\widetilde\ell}$ of the lifted loss
    $\widetilde\ell_{R,P}=\norm{\widetilde m_{R,P}}_{\mathfrak{X}}$, as in the equality part of the proof of
    Corollary~\ref{C:studentized}.  For $(\theta,c)\in R\times\widetilde C$, choosing $\theta_0\in\Theta_{\mathrm{I}}(R,P)$
    with $d_{\mathfrak{B}}(\theta,\theta_0)\le2d_{\mathfrak{B}}(\theta,\Theta_{\mathrm{I}}(R,P))$ gives, by \eqref{E:fello},
    \begin{align*}
        d\big((\theta,c),\widetilde\Theta_{\mathrm{I}}(R,P)\big)
        &\le2(1+C_m)d_{\mathfrak{B}}(\theta,\Theta_{\mathrm{I}}(R,P))+\norm{m(\theta)-c}_{\mathfrak{X}}\\
        &\le2(1+C_m)\,\overline f_\ell\big(\widetilde\ell_{R,P}(\theta,c)/(2\kappa)\big)+\widetilde\ell_{R,P}(\theta,c)/(2\kappa),
    \end{align*}
    where $C_m$ is the uniform Lipschitz constant of Assumption~\ref{A:uniformbound}(6), and the second inequality uses
    $\widetilde\ell_{R,P}(\theta,c)\ge2\kappa\norm{m(\theta)-c}_{\mathfrak{X}}\ge2\kappa\,\ell_{R,P}(\theta)$.  Thus
    $\overline f_{\widetilde\ell}(\delta)\le2(1+C_m)\,\overline f_\ell(\delta/(2\kappa))+\delta/(2\kappa)$, so
    $\overline f_{\widetilde\ell}(\delta)=o(1)$ when $\overline f_\ell(\delta)=o(1)$, and
    $\overline f_{\widetilde\ell}(\delta)=O(\delta^{\min\{\alpha_\ell,1\}})$ with $\min\{\alpha_\ell,1\}\widetilde\alpha>1$ when
    $\overline f_\ell(\delta)=O(\delta^{\alpha_\ell})$.  The conditions of Proposition~\ref{P:unifbound} therefore hold for
    the lifted family, which gives both claims.

    \textbf{Case~(2):} because $\Sigma(\theta)C=C$,
    \[
        \widetilde T_n=r_n\inf_{\theta\in R}d_{\mathfrak{X}}\big(\Sigma(\theta)m_n(\theta),C\big),\qquad
        Z_n=s_n\inf_{\theta\in R}d_{\mathfrak{X}}\big(\Sigma(\theta)m_n(\theta)+s_n^{-1}\Sigma(\theta)\mathbb{W}_n(\theta),C\big).
    \]
    These are the statistic and the bound of Proposition~\ref{P:unifbound} for the
    family in case~(2), and $\widetilde Z_n^*$ is its bootstrap. Thus, Proposition~\ref{P:unifbound} yields both claims.
\end{proof}

\section{Uniform upper-bound estimators and critical values}\label{S:validity}

The validity lemma below turns our bootstrap results into critical values with
asymptotic level $\alpha$ through~\eqref{E:critval}.  Let $\mathcal{L}_{n,R,P}$ denote the law of $T_{n,R,P}$ under $P$ and let
$\mathcal{H}$ be the nondecreasing members of $\mathrm{BL}_1$ on $\R$.  A sequence
of data-dependent laws $\widehat{\mathcal{L}}_n$ is a \emph{uniform upper-bound
estimator} for $T_{n,R,P}$ if
\begin{align*}
\sup_{(R,P)\in\mathcal{S}}\mathrm{P}_P\left\{
\sup_{h\in\mathcal{H}}\left(\int h\,d\mathcal{L}_{n,R,P}
-\int h\,d\widehat{\mathcal{L}}_n\right)>\eta\right\}\longrightarrow0
\qquad\text{for every }\eta>0,
\end{align*}
that is, $\widehat{\mathcal{L}}_n$ asymptotically dominates $\mathcal{L}_{n,R,P}$ in
increasing bounded-Lipschitz order, uniformly over $\mathcal{S}$.  Although
$\widehat{\mathcal{L}}_n$ depends on $(R,P)$, through the sample drawn from $P$ and the
restricted set $R$, we suppress this dependence in the notation. A bootstrap that uniformly and consistently estimates the law of $T_{n,R,P}$, or that of any $Y_{n,R,P}$ with $T_{n,R,P}\le Y_{n,R,P}+o_p(1)$ uniformly over $\mathcal{S}$, is of this kind, because every $h\in\mathcal{H}$ is nondecreasing and $1$-Lipschitz.
Theorem~\ref{T:two} supplies such estimators pointwise, and Theorem~\ref{P:unif1}
and Proposition~\ref{P:unifbound} supply them uniformly over $\mathcal{S}$.

\begin{lemma}[Size control with a value slack]\label{L:validity}
Let $\widehat{\mathcal{L}}_n$ be a uniform upper-bound estimator for $T_{n,R,P}$, let
$\alpha\in(0,1)$ and $\varepsilon\in(0,1]$, and let
$\widehat c_{n,R,P}=\widehat q_{n,R,P}(1-\alpha)+\varepsilon$, where
$\widehat q_{n,R,P}(1-\alpha)$ is the $(1-\alpha)$-quantile of
$\widehat{\mathcal{L}}_n$.  Then
\begin{align*}
\liminf_{n}\ \inf_{(R,P)\in\mathcal{S}}\ \mathrm{P}_P\bigl(T_{n,R,P}\le \widehat c_{n,R,P}\bigr)\ \ge\ 1-\alpha .
\end{align*}
\end{lemma}

\begin{proof}
Fix $\varepsilon\in(0,1]$ and $\eta\in(0,1-\alpha)$. We may suppress the subscripts
$(R,P)$, as all bounds are uniform over $\mathcal{S}$.  For $t\in\R$ let $h_t$ be the function equal to
$0$ on $(-\infty,t]$, $1$ on $[t+\varepsilon,\infty)$, and linear in between, so
that $\mathbf 1\{x>t+\varepsilon\}\le h_t(x)\le\mathbf 1\{x>t\}$ and
$\varepsilon h_t\in\mathcal{H}$.  Let $F$ be the distribution function of
$\mathcal{L}_{n,R,P}$ and, for $\beta\in(0,1)$, let $Q_\beta=\inf\{x: F(x)\ge\beta\}$
be the $\beta$-quantile of $T_{n,R,P}$.

Write $\widehat q=\widehat q_{n,R,P}(1-\alpha)$ and
$\delta_n=\sup_{h\in\mathcal{H}}\bigl(\int h\,d\mathcal{L}_{n,R,P}-\int h\,d\widehat{\mathcal{L}}_n\bigr)$.
Because $h_{\widehat q}\le\mathbf 1\{x>\widehat q\}$ and $\widehat q$ is the
$(1-\alpha)$-quantile of $\widehat{\mathcal{L}}_n$,
$\int h_{\widehat q}\,d\widehat{\mathcal{L}}_n\le\widehat{\mathcal{L}}_n(\widehat q,\infty)\le\alpha$.
Applying the upper-bound property to $\varepsilon h_{\widehat q}\in\mathcal{H}$ gives
$\varepsilon\int h_{\widehat q}\,d\mathcal{L}_{n,R,P}\le\varepsilon\int h_{\widehat q}\,d\widehat{\mathcal{L}}_n+\delta_n\le\varepsilon\alpha+\delta_n$,
so $\int h_{\widehat q}\,d\mathcal{L}_{n,R,P}\le\alpha+\delta_n/\varepsilon$.

The lower bound $\mathbf 1\{x>t+\varepsilon\}\le h_t(x)$ holds for every $t$, so
$1-F(t+\varepsilon)\le\int h_t\,d\mathcal{L}_{n,R,P}$ holds for all $t$, and in particular at the random point
$\widehat q$.  On the event $\{\delta_n\le\varepsilon\eta\}$ the previous display gives
$1-F(\widehat q+\varepsilon)\le\alpha+\eta$, that is
$F(\widehat q+\varepsilon)\ge1-\alpha-\eta$, whence
$\widehat q+\varepsilon\ge Q_{1-\alpha-\eta}$.  On that event
$\{T_n>\widehat q+\varepsilon\}\subseteq\{T_n>Q_{1-\alpha-\eta}\}$, and
$\mathrm{P}_P(T_n>Q_{1-\alpha-\eta})=1-F(Q_{1-\alpha-\eta})\le\alpha+\eta$.  Therefore
$\mathrm{P}_P(T_n>\widehat q+\varepsilon)\le\alpha+\eta+\mathrm{P}_P(\delta_n>\varepsilon\eta)$.

Taking $\sup_{(R,P)\in\mathcal{S}}$ and then $\limsup_n$, the last term vanishes, so
$\limsup_n\sup_{(R,P)\in\mathcal{S}}\mathrm{P}_P(T_{n,R,P}>\widehat c_{n,R,P})\le\alpha+\eta$.
The left-hand side does not depend on $\eta$, and $\eta\in(0,1-\alpha)$ was arbitrary,
so the bound holds with $\eta=0$.
\end{proof}

The minimum of finitely many such critical values inherits the same guarantee.  Candidate critical values can therefore be aggregated without losing asymptotic validity.

\begin{cor}[Minimum of critical values]\label{C:minvalidity}
Let $\widehat{\mathcal{L}}_n^{(1)},\ldots,\widehat{\mathcal{L}}_n^{(J)}$ be uniform
upper-bound estimators for $T_{n,R,P}$ with $J$ fixed, and let
$\widehat q^{(j)}_{n,R,P}(1-\alpha)$ denote the $(1-\alpha)$-quantile of
$\widehat{\mathcal{L}}_n^{(j)}$.  For $\alpha\in(0,1)$ and $\varepsilon\in(0,1]$, set
$\widehat c_{n,R,P}=\min_{1\le j\le J}\widehat q^{(j)}_{n,R,P}(1-\alpha)+\varepsilon$.
Then
\begin{align*}
\liminf_{n}\ \inf_{(R,P)\in\mathcal{S}}\ \mathrm{P}_P\bigl(T_{n,R,P}\le \widehat c_{n,R,P}\bigr)\ \ge\ 1-\alpha .
\end{align*}
\end{cor}

\begin{proof}
Fix $\eta\in(0,1-\alpha)$ and let $Q_{1-\alpha-\eta}$ be the
$(1-\alpha-\eta)$-quantile of $T_{n,R,P}$, which depends on the law of $T_{n,R,P}$
and on $\eta$ alone and so serves every $j$.  Write $\delta_n^{(j)}$ for the
quantity $\delta_n$ in the proof of Lemma~\ref{L:validity} formed from
$\widehat{\mathcal{L}}_n^{(j)}$.  That proof shows
$\widehat q^{(j)}_{n,R,P}(1-\alpha)+\varepsilon\ge Q_{1-\alpha-\eta}$ on the event
$\{\delta_n^{(j)}\le\varepsilon\eta\}$.  Because $J$ is fixed, the intersection of
these $J$ events has probability tending to one uniformly over $\mathcal{S}$, and on
it $\min_{j}\widehat q^{(j)}_{n,R,P}(1-\alpha)+\varepsilon\ge Q_{1-\alpha-\eta}$.
The remainder of the proof of Lemma~\ref{L:validity} applies unchanged.
\end{proof}

A minimum taken draw by draw is also covered.  The proof of Theorem~\ref{P:unif1} shows that the bootstrap of
Theorem~\ref{P:unif1} and Proposition~\ref{P:unifbound} is, up to a uniformly $o_p(1)$ term, a $1$-Lipschitz
functional of $\mathbb V^*_{n,R,P}$ whose counterpart in $\mathbb V_{n,R,P}$ bounds $T_{n,R,P}$.  Over a fixed
grid of steps, the minimum of these functionals is again $1$-Lipschitz and still bounds $T_{n,R,P}$, so the law of
the draw-wise minimum is a uniform upper-bound estimator.  Its critical value never exceeds that of the hybrid
with the same grid.  In the designs of Section~\ref{S:mc}, it tends to over-reject, with null rejection reaching
$0.077$.  For this reason our hybrid takes the minimum of critical values.

The lemma requires no continuity of the law of $T_{n,R,P}$, because the slack absorbs atoms.
The slack may be dropped by imposing continuity of the limiting distribution at
its critical value, such as Assumption~A.2 of \citet{BCS2017}, whose Remark~B.2
notes that a small constant added to the critical value avoids that requirement.
The fixed slack may also be replaced by any sequence $\varepsilon_n\downarrow0$ with
$\delta_n=o_p(\varepsilon_n)$ uniformly over $\mathcal{S}$.  The uniform guarantees of
Lemma~\ref{L:validity} and Corollary~\ref{C:minvalidity} still hold in that case, because
the proof of Lemma~\ref{L:validity} uses $\varepsilon$ only through the term
$\mathrm{P}_P(\delta_n>\varepsilon\eta)$ of its final bound.  If, in addition, $\widehat{\mathcal{L}}_n$ is consistent for the
law of $T_{n,R,P}$ at some $(R,P)$ whose limiting distribution function is continuous and strictly increasing at its
$(1-\alpha)$-quantile, the critical value $\widehat q_{n,R,P}(1-\alpha)+\varepsilon_n$
is also asymptotically exact there \citep[cf.][Remark~2.6]{RomanoShaikh2012}.

\section{Supplementary Monte Carlo results}\label{S:supp_tables}

This section supplements Section~\ref{S:mc}.  It extends its power analysis, adds an estimated-covariance
design, gives more detail on the subvector design, and examines alternative centerings and
other losses.

\subsection{Local power}\label{S:supp_mc_figures}

The figures in this subsection extend Figure~\ref{F:mc_main}(c) to the full grid of $h$.  At every $h$, DR and PR
are less powerful than the re-optimization bootstrap, the hybrid and MR, as at $h=1.5$
(Figure~\ref{F:supp_mc_power}).  Figure~\ref{F:mc_size_adjusted} shows that this gain over DR comes from better size
calibration, as size adjustment closes and in some cases reverses the power gap between DR and the
re-optimization-based procedures.

\begin{figure}[!htbp]
\centering
\includegraphics[width=0.95\textwidth]{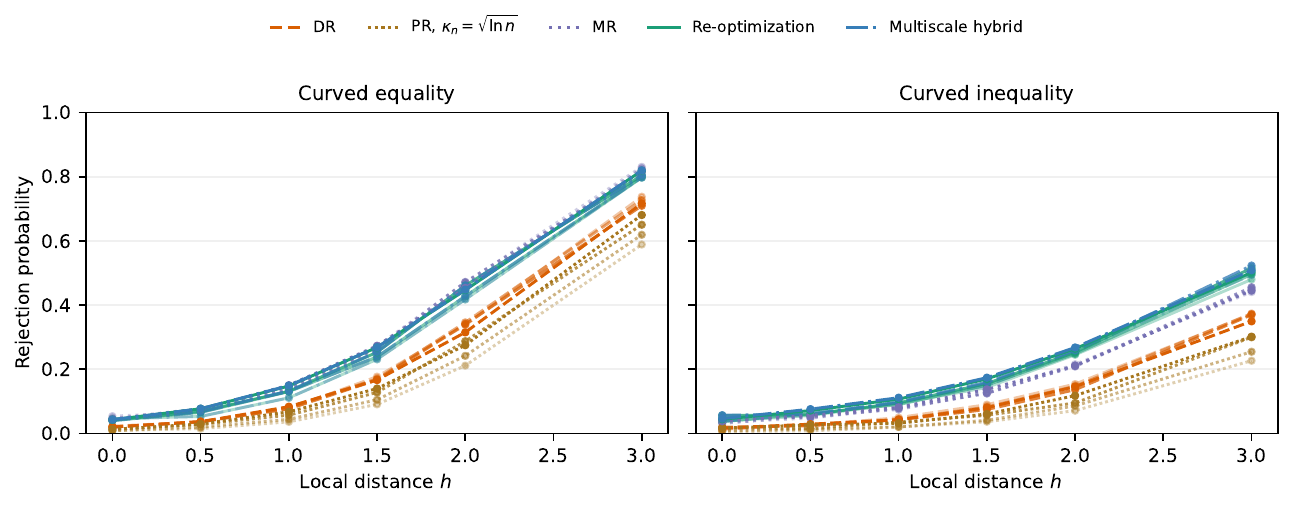}
\caption{Rejection probabilities against the local alternatives at nominal level $0.05$, with the
re-optimization bootstrap at $\lambda_0=.2$.  Opacity increases with $n$.}
\label{F:supp_mc_power}
\end{figure}

\begin{figure}[!htbp]
\centering
\includegraphics[width=0.95\textwidth]{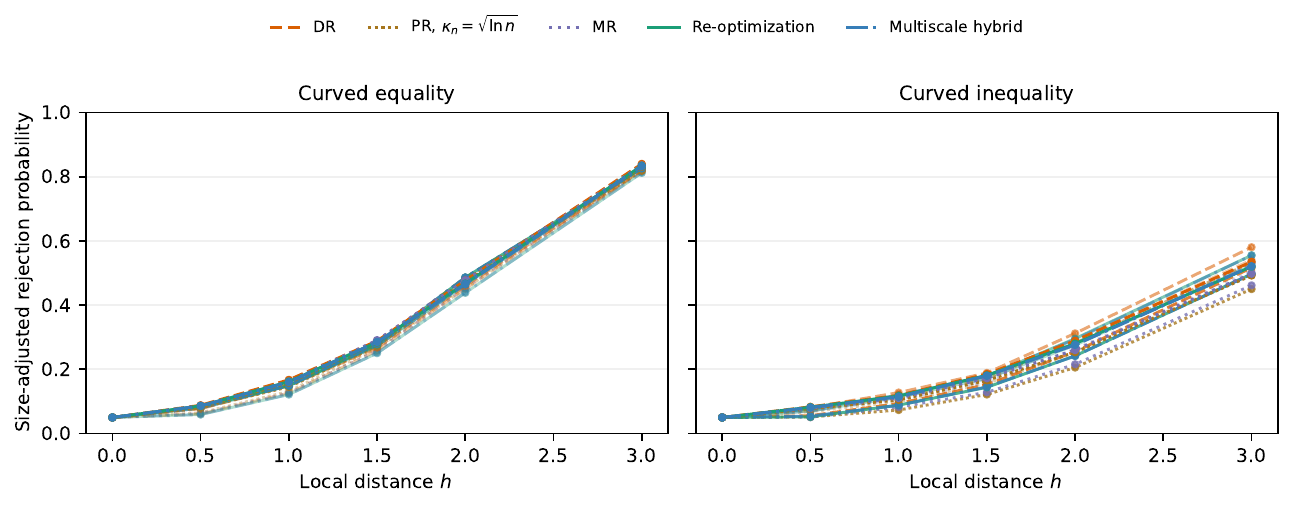}
\caption{Size-adjusted rejection probabilities against the local alternatives.  Each procedure is calibrated
to its simulated null distribution at $h=0$.  Opacity increases with $n$.}
\label{F:mc_size_adjusted}
\end{figure}

\clearpage

\subsection{Estimated covariance and an observation-level multiplier bootstrap}
\label{S:supp_mc_studentized}

The main experiment uses the known covariance and draws the Gaussian perturbation directly.  Here the
observations have Student-$t_5$ rather than Gaussian innovations, scaled to unit variance. The
re-optimization bootstrap at $\lambda_0=.2$ weights each moment by its inverse sample standard deviation, held fixed across draws
as in~\eqref{E:studentized_bootstrap}, and forms its perturbation from an observation-level Gaussian multiplier
process.  Because diagonal weighting leaves the constraint cone invariant, Corollary~\ref{C:studentized_invariant}
and \Secapp~\ref{S:studentized_boot} imply that the estimated weight has no first-order effect.
Table~\ref{T:mc_studentized} supports this finding, as rejection probabilities under estimated and oracle identity
weighting differ by at most $0.006$ (less than one Monte Carlo standard error).

\begingroup
\renewcommand*{\theHtable}{\thesection.\arabic{table}}
\begin{table}[!htbp]
\centering
\caption{Rejection probabilities at nominal level $0.05$ with estimated covariance
and an observation-level multiplier bootstrap.
``Estimated'' weights each moment by its inverse sample standard deviation, while
``oracle'' uses unit weights.  Each entry uses $500$ Monte Carlo replications and
$399$ multiplier draws.}
\label{T:mc_studentized}
\begin{tabular}{llrrr}
\hline
Design & $n$ & Estimated & Oracle & MC s.e. \\
\hline
Baseline null & 500 & 0.028 & 0.030 & 0.007 \\
              & 2,000 & 0.034 & 0.036 & 0.008 \\
Nonlinear boundary null & 500 & 0.028 & 0.030 & 0.007 \\
                        & 2,000 & 0.044 & 0.046 & 0.009 \\
Curved equality, $h=1.5$ & 500 & 0.246 & 0.242 & 0.019 \\
                         & 2,000 & 0.240 & 0.234 & 0.019 \\
Curved inequality, $h=1.5$ & 500 & 0.154 & 0.148 & 0.016 \\
                           & 2,000 & 0.160 & 0.156 & 0.016 \\
\hline
\end{tabular}
\end{table}
\endgroup

\subsection{Subvector design}
\label{S:supp_mc_subvector}

This subsection gives further details on the subvector design of Section~\ref{S:mc}, which augments the baseline with an additional parameter
$\theta_3\in[0,3]$ and the restriction $\theta_3\geq\theta_1$.
Testing $H_0:\theta_3=c$ restricts $\theta$ to the slice $R_c=\{\theta\in\Theta:\theta_3=c\}$.  As the moments do
not involve $\theta_3$, the restriction $\theta_3\geq\theta_1$ only caps $\theta_1$ at $c$, and for $c\geq0.4$
the restricted identified set is
\[
 \Theta_{\mathrm I}(R_c,P)
 = [0.4,\min\{c,\sqrt{0.5}\}]\times\{0.8\}\times\{c\}.
\]
It collapses in $\theta_1$ at $c=0.4$, and for $c<0.4$ the restricted model is infeasible, with
$T_0(R_c,P)>0$.

We first verify the linear minorant of Section~\ref{S:mc}.  Because the identified set is the product of an interval in $\theta_1$ and the point
$\theta_2=0.8$, the squared distance to it is the sum of a $\theta_1$ part and a $\theta_2$ part.  Likewise
$\xi(m(\theta))^2$ is the sum of the squared violations of the $\theta_1$ moments (coordinates 1, 2 and 5) and of
the $\theta_2$ moments (coordinates 3, 4 and 6), so it suffices to bound the norm of each block of violations by a fixed
multiple of the corresponding distance.  The equality block has norm at least $\sqrt2\,|\theta_2-0.8|$.  Below the identified interval of $\theta_1$, the violated
lower bound $\mu_1-\theta_1$ equals the distance to it.  Above the interval, the violated upper bound is either
linear with slope one or quadratic, with
$\theta_1^2-\mu_5=(\theta_1-\sqrt{\mu_5})(\theta_1+\sqrt{\mu_5})\geq2\sqrt{\mu_5}\,(\theta_1-\sqrt{\mu_5})$.  Because
$2\sqrt{\mu_5}\geq0.8$, this gives $\xi(m(\theta))\geq0.8\,d(\theta,\Theta_{\mathrm I})$ in each null design of
Section~\ref{S:mc}.  On $R_c$ the same argument, with $\mu_5=0.5$ and the upper bound active only when
$c>\sqrt{0.5}$, gives $\xi(m(\theta))\geq d(\theta,\Theta_{\mathrm I}(R_c,P))$ for $\theta\in R_c$, uniformly over the null
family.  Thus $\overline f_\ell(\delta)\leq\delta$, and $\alpha_\ell=1$ holds uniformly even though the restricted
identified set collapses to a point at $c=0.4$.

Figure~\ref{F:supp_mc_subvector} shows that the re-optimization bootstrap and the hybrid control size over the
whole null grid and that their power rises quickly below $c=0.4$.  It covers the fixed values
$c\in\{0.20,\allowbreak0.2025,\allowbreak\ldots,\allowbreak0.80\}$ and the drifting values $c=0.4+h/\sqrt n$, $h\in\{-3,-2.75,\ldots,5\}$, with
$10{,}000$ replications.  Over the $161$ fixed and $21$ drifting null values, the largest rejection probability is
$0.041$--$0.051$ for the re-optimization bootstrap, $0.042$--$0.051$ for the hybrid and $0.048$--$0.052$ for MR,
while DR and PR stay below $0.024$.  Panel~(c) adds the other procedures to Figure~\ref{F:mc_main}(f).  At $c=0.38$ and $n=10{,}000$ the hybrid rejects with probability $0.408$, the re-optimization bootstrap with
$0.398$, MR with $0.374$ and DR with $0.296$.

\begin{figure}[!htbp]
\centering
\includegraphics[width=\textwidth]{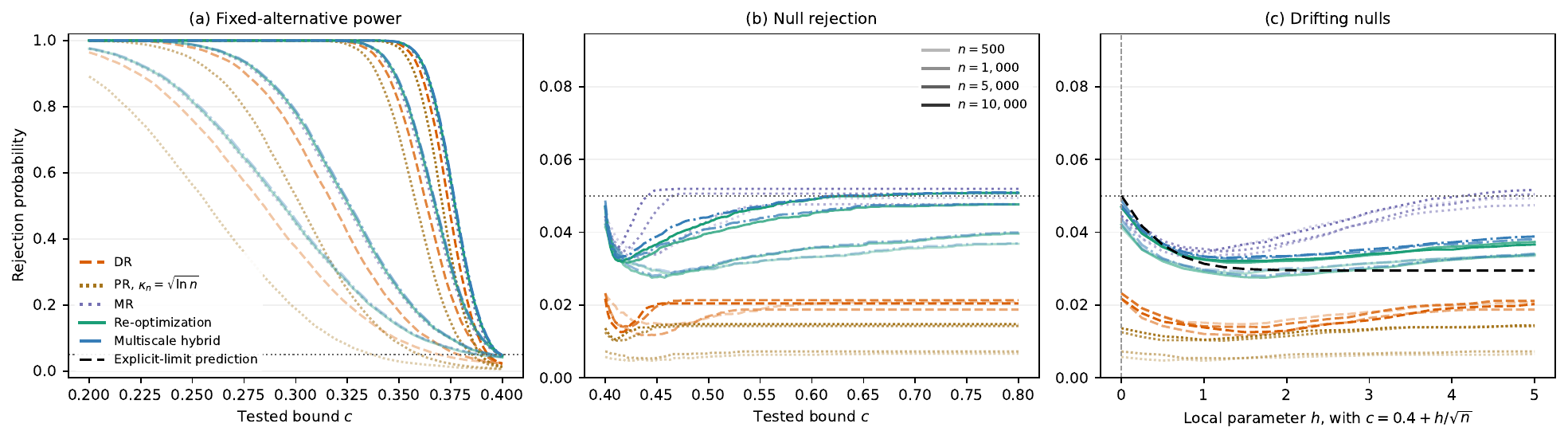}
\caption{Subvector rejection probabilities: (a) power at fixed $c\leq0.4$, (b) null rejection at fixed
$c\geq0.4$, and (c) null rejection along $c=0.4+h/\sqrt n$, $h\geq0$, with the prediction of the
explicit limit for $c=0.4$ (dashed).  The dotted line is the nominal level.  Opacity increases with $n$.}
\label{F:supp_mc_subvector}
\end{figure}

\subsection{Centering}\label{S:supp_centering}

The draw $B_n^*$ of Section~\ref{S:mc} is centered at $T_0=0$, which is exact under the feasibility null.
We compare it with three other centerings.  The infeasible oracle uses $B_n^*-s_nT_0$, which equals $B_n^*$ under every null design because $T_0=0$.  Plug-in centering
uses $B_n^*-s_n\widehat T_0$ with $\widehat T_0=T_n/\sqrt n$, which, because $\widetilde W^*$ does not depend on $\theta$, is exactly the
numerical delta method of \citet{HongLi2018} with their step size equal to $s_n^{-1}$.  The uncentered draw is their variant for
nonnegative functionals, which they note is conservative \citep[Section~3.2]{HongLi2018}.  Plug-in
centering can also be bias corrected, because Proposition~\ref{P:ex2} estimates the law of
$\sqrt n(\widehat T_0-T_0)$ and thus estimates the bias of $\widehat T_0$.  The self-corrected (SC) variant
subtracts this estimate from $\widehat T_0$.  The correction is $O_p(n^{-1/2})$, so it leaves the limit of the
bootstrap unchanged.

\begin{figure}[p]
\centering
\includegraphics[width=\textwidth]{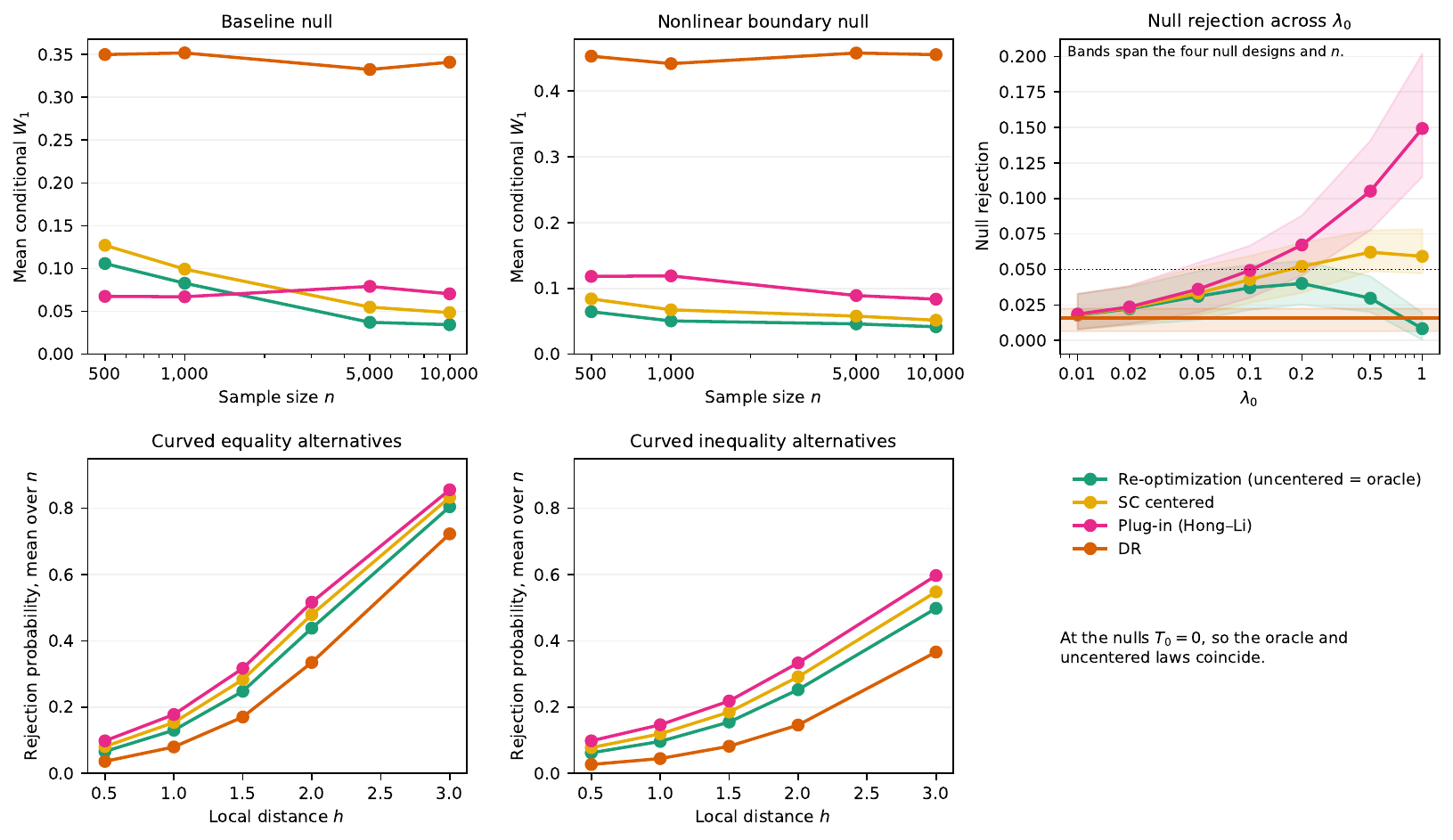}
\caption{Centering diagnostics.  Upper left and middle: mean conditional $W_1$ distances at the baseline and nonlinear-boundary nulls,
$\lambda_0=.2$.  Upper right: mean null rejection across the step grid, with bands spanning the four null
designs and $n$.  Lower panels: local power at $\lambda_0=.2$, averaged over $n$.}
\label{F:supp_mc_centering}
\end{figure}

Centering trades size for power (Figure~\ref{F:supp_mc_centering}).  SC and plug-in centering are asymptotically
valid, but at these sample sizes they tend to over-reject.  At $\lambda_0=.2$ and across the null designs, SC rejects with
probability $0.034$--$0.069$ and plug-in centering with probability $0.046$--$0.088$, against $0.025$--$0.056$
for the uncentered bootstrap, and both generally rise with $\lambda_0$.  In exchange, SC recovers most of the oracle's power under the
alternatives.  Plug-in centering, by contrast, rejects exactly when $T_n(1+s_n/\sqrt n)$ exceeds the uncentered
critical value.  It only rescales the statistic, so it leaves size-adjusted power unchanged to within $0.004$, and
its size distortion vanishes only at the slow rate $s_n/\sqrt n$.

SC over-rejects less with a smaller step.  With $s_n=\lambda_0\sqrt{\ln n}$ and intermediate values of
$\lambda_0$, its null rejection stays close to $0.05$ at every $n$, and it remains more powerful than the
uncentered bootstrap with the same step.

\subsection{Other losses}\label{S:supp_losses}

Neither the explicit limit nor the re-optimization bootstrap relies on the Euclidean loss, so we repeat the
fixed-null experiments under two others.  The $\ell_1$ loss sums the positive parts of the inequalities and the
absolute values of the equalities, and the $\ell_\infty$ loss takes their maximum.  A linear minorant still holds for these losses,
because each is bounded below by a multiple of the Euclidean one, and Corollary~\ref{C:geometry} still applies
with the $\ell_1$ or $\ell_\infty$ distance in place of the Euclidean one.

Under both losses, the Corollary~\ref{C:geometry} limit lies within $W_1$ distance $0.016$ of the sampling
distribution at both nulls and every $n$ (Figure~\ref{F:supp_mc_losses}).  The corresponding envelope bound, which holds $\theta$
fixed as a first-order chain rule does \citep[online Appendix~A.4]{HongLi2018}, is similarly conservative and lies
$0.47$--$1.02$ away.  The
re-optimization bootstrap again tracks the sampling distribution and shows the same U-shape in $\lambda_0$, but its best
step depends on the loss.  Under $\ell_1$ with $n\leq1{,}000$ the $W_1$ distance is smallest at $\lambda_0=.5$ rather
than $.2$, and at $\lambda_0=.2$ the test is conservative at the baseline null in small samples, rejecting with probability
$0.016$ at $n=500$ and $0.041$ at $n=10{,}000$.  A step tuned to one loss need not suit another, which is a reason to aggregate over
steps.

\begin{figure}[!htbp]
\centering
\includegraphics[width=\textwidth]{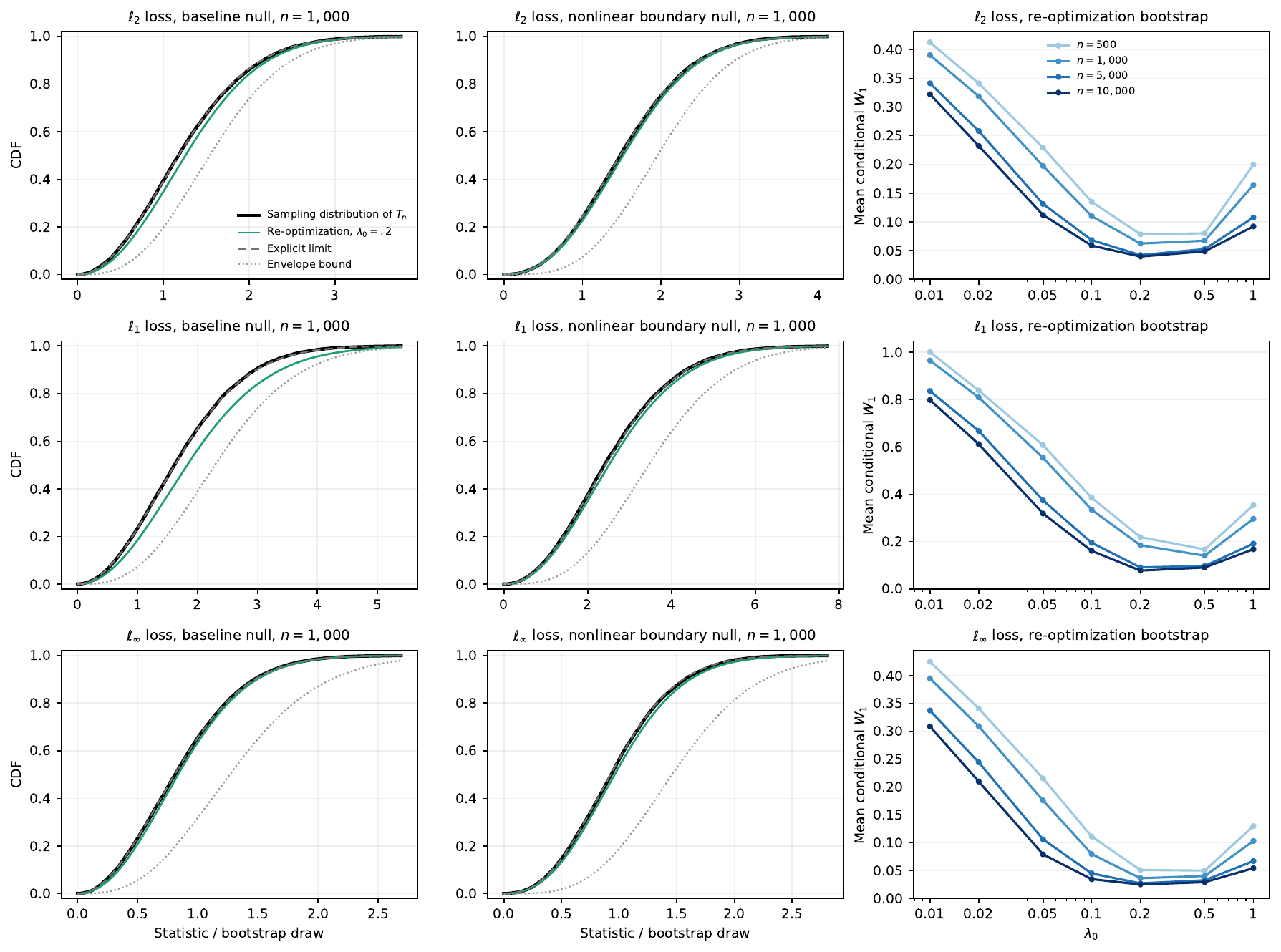}
\caption{Explicit limits and the re-optimization bootstrap under three losses.  Left and middle: sampling distribution of
$T_n$, explicit limit, envelope bound, and mean conditional bootstrap CDF at $\lambda_0=.2$, at $n=1{,}000$.
Right: mean conditional $W_1$ distance to the sampling distribution, averaged over the baseline and nonlinear-boundary nulls, against $\lambda_0$.}
\label{F:supp_mc_losses}
\end{figure}

\clearpage

\section{Failure without uniform separation}
\label{S:uniform_separation}

Our uniform separation condition (Assumption~\ref{AU:1}(2), equivalently
$\overline f_\ell(\delta)=o(1)$) requires the loss to stay bounded away from zero at
parameters a fixed distance from the identified sets, uniformly over $\mathcal{S}$.
When it fails (the third tier of Section~\ref{S:uniform}), the pointwise
exactness of Proposition~\ref{P:ex2} and the uniform upper bounds of
Propositions~\ref{P:unifbound} and~\ref{Add:unifxi} can be incompatible.  We
construct a family of DGPs in which Proposition~\ref{P:ex2} applies at every
member and every structural condition of Proposition~\ref{Add:unifxi} holds except
uniform separation.  In this family no
uniform upper-bound estimator can be pointwise exact, so
uniform consistency for the finite-sample laws is impossible.

\Secapp~\ref{S:validity} defines the uniform upper-bound estimators and the class $\mathcal{H}$.  Here
$P^{(n)}$ is the joint law of the size-$n$ sample under $P$,
$\norm{P-Q}_{\mathrm{TV}}=\sup_A|P(A)-Q(A)|$ is the total variation distance, and an \emph{estimator}
$\widehat{\mathcal{L}}_n$ is any measurable, law-valued function of the data.

\subsection{A two-point bound}

\begin{lemma}\label{L:twopoint}
Fix $(R,P_0)\in\mathcal{S}$ and, for all large $n$, $(R,P_n)\in\mathcal{S}$.  Suppose
there are $\tau\in[0,1)$, $\varepsilon>0$, $h\in\mathcal{H}$, and laws
$\Lambda_n$ such that, for all large $n$,
\[
\text{(i)}\quad \norm{P_0^{(n)}-P_n^{(n)}}_{\mathrm{TV}}\le\tau,
\qquad\qquad
\text{(ii)}\quad \int h\,d\Lambda_n - \int h\,d\mathcal{L}_{n,R,P_0}\ge\varepsilon.
\]
Then every estimator $\widehat{\mathcal{L}}_n$ satisfies
\[
\max\!\left\{\,
\EE{P_0}{d_{\mathrm{BL}}\bigl(\widehat{\mathcal{L}}_n,\mathcal{L}_{n,R,P_0}\bigr)},\ \
\EE{P_n}{\Bigl(\textstyle\int h\,d\Lambda_n-\int h\,d\widehat{\mathcal{L}}_n\Bigr)^{+}}
\right\}
\ \ge\ \frac{\varepsilon(1-\tau)}{4}.
\]
\end{lemma}

\begin{proof}
Let $s=\int h\,d\widehat{\mathcal{L}}_n\in[-1,1]$, which is a bounded measurable function
of the data, and set $A=\{s<\int h\,d\mathcal{L}_{n,R,P_0}+\varepsilon/2\}$.  On
$A^c$, $\widehat{\mathcal{L}}_n$ differs from $\mathcal{L}_{n,R,P_0}$ by at least
$\varepsilon/2$ in the $h$-direction, and $\norm{h}_{\mathrm{BL}}\le1$, so
$d_{\mathrm{BL}}(\widehat{\mathcal{L}}_n,\mathcal{L}_{n,R,P_0})\ge\varepsilon/2$. In particular, the first term in the maximum above is at least $\tfrac{\varepsilon}{2}P_0(A^c)$.  On $A$,
condition (ii) gives $\int h\,d\Lambda_n-s>\varepsilon/2$, so the second term is
at least $\tfrac{\varepsilon}{2}P_n(A)\ge\tfrac{\varepsilon}{2}(P_0(A)-\tau)$ by
(i).  Letting $p=P_0(A)$, the larger of $\tfrac{\varepsilon}{2}(1-p)$ and
$\tfrac{\varepsilon}{2}(p-\tau)$ is smallest at $p=(1+\tau)/2$, where both equal
$\varepsilon(1-\tau)/4$.
\end{proof}

Lemma~\ref{L:twopoint} is a two-point construction in the sense of Le Cam
\citep[Section~2.4.2, Theorem~2.2]{Tsybakov2009}.  Nothing is assumed about
$\widehat{\mathcal{L}}_n$ beyond measurability, so the bound binds plug-in,
bootstrap, and moment-selection or thresholding rules alike.  Below, $\Lambda_n$ is the finite-sample law of $T_{n,R,P_n}$ under $P_n$.

\subsection{A family without uniform separation}

Take $\Theta=[0,1]$ with the restricted set $R=\Theta$ fixed. For each $n$, let the data consist of iid vectors $X_1, \ldots, X_n$ with
$X_i\sim N\!\bigl((0,c),I_2\bigr)$, where $c\in[0,1]$ indexes
$\mathcal{S}$. Set
\begin{equation}\label{E:noseparation}
m_n(\theta)=(1-\theta)\bar X_1+\theta\bar X_2,\qquad \xi(x)=x^{+},\qquad r_n=\sqrt n.
\end{equation}
This is the moment-inequality null $H_0:\exists\,\theta\in[0,1]$ with
$E[(1-\theta)X_1+\theta X_2]\le0$.   The population map is $m_c(\theta)=c\theta$,
so $\ell_c(\theta)=c\theta$ and
\[
\Theta_{\mathrm{I}}(c)=\{0\} \text{ for }c > 0,\qquad \Theta_{\mathrm{I}}(0)=[0,1].
\]

\paragraph{Every structural condition of Proposition~\ref{Add:unifxi} holds except uniform separation.}
We verify Assumption~\ref{A:uniformbound}, which supplies those conditions in the
convex-loss setting.  \emph{(1)} $\Theta=[0,1]$ is compact and convex, and
$R=\Theta$, so every $B_R(\theta,\delta)$ is convex.  \emph{(2)} $\xi(x)=x^{+}$ is
convex, finite, and $1$-Lipschitz on $\R$.  \emph{(3)} Each $m_c(\theta)=c\theta$
is affine, hence uniformly Fr\'echet differentiable with $m'_c(\theta)=c\le1$ and
identically zero Taylor remainder, so the remainder bound of part~(3) holds for
every $\alpha_m>1$.  \emph{(4)} The centered
process
\[
\mathbb{W}_n(\theta)=(1-\theta)\sqrt n\,\bar X_1+\theta\sqrt n(\bar X_2-c)
\]
has a law that does not depend on $c$, as each coordinate is standard Gaussian
under every $P_c$, and the process converges to the tight, equicontinuous limit
$\mathbb{X}(\theta)=(1-\theta)Z_1+\theta Z_2$ with $Z_1,Z_2\sim N(0,1)$
independent, so the weak convergence and the pre-tightness bounds are uniform in
$c$.  \emph{(5)} The standard multiplier bootstrap for the sample mean supplies
the uniform conditional weak convergence.  \emph{(6)} Each $m_c$ is $c$-Lipschitz with
$c\le1$, so $C_m=1$.  These conditions imply
the pointwise hypotheses of Proposition~\ref{P:ex2} at each $c$, including
pointwise separation, because $\ell_c(\theta)=c\theta\ge c\delta$ off a
$\delta$-neighborhood of $\Theta_{\mathrm{I}}(c)=\{0\}$ when $c>0$.  As $m$ is
affine, Proposition~\ref{P:ex2} applies in its affine case~(ii), and
Proposition~\ref{P:ex1} holds with equality.

Only the uniform separation condition $\overline{f}_\ell(\delta) = o(1)$ of
Proposition~\ref{Add:unifxi} fails in this family.  At $\theta=1$, the loss
$\ell_c(1)=c$ tends to zero as $c\downarrow0$ while
$d_{\mathfrak{B}}(1,\Theta_{\mathrm{I}}(c))=1$, so the uniform modulus
\eqref{E:fello} obeys $\overline{f}_\ell(\delta)\ge1$ for every $\delta>0$.

\paragraph{The finite-sample law of $T_{n,R,P}$.}
The design is exactly Gaussian at every sample size.  Under $P_c$, let
$Z_1=\sqrt n\,\bar X_1$ and $Z_2=\sqrt n(\bar X_2-c)$ be independent standard
normal variables. Then $\sqrt n\,m_n(\theta)=(1-\theta)Z_1+\theta(Z_2+\sqrt n\,c)$ is affine
in $\theta$, and
\[
T_{n,R,P_c}=\min\{Z_1,\,Z_2+\sqrt n\,c\}^{+}\qquad\text{ for every }n.
\]
At $c=0$ the law $\mathcal{L}_{n,R,P_0}=\min\{Z_1,Z_2\}^{+}$ does not
depend on $n$, and we write $\mathcal{L}_{R,P_0}$.

\paragraph{Total variation.}
Each distribution $P_c^{(n)}$ is multivariate Gaussian and differs from $P_0^{(n)}$ only through a
mean shift by a vector of Euclidean length $\sqrt{n} c$. Theorem~1 of \citet{BarsovUlyanov1987} therefore gives $\norm{P_0^{(n)}-P_c^{(n)}}_{\mathrm{TV}}=2\Phi(c\sqrt n/2)-1$, where $\Phi$ is the standard Gaussian distribution function.

\paragraph{The impossibility.}
Fix $c>0$ and define the drift $c_n=c/\sqrt n$, writing $P_n\equiv P_{c_n}$. Then, for all $n$, the law of $T_{n,R,P_n}$ is
$\Lambda_n \equiv \min\{Z_1,\,Z_2+c\}^{+}$, and the total-variation identity
gives $\norm{P_0^{(n)}-P_n^{(n)}}_{\mathrm{TV}}=2\Phi(c/2)-1\equiv \tau(c)<1$, which fulfills
condition~(i) of Lemma~\ref{L:twopoint}.  For $t\ge0$, $P(\Lambda_n>t)=(1-\Phi(t))(1-\Phi(t-c))$ strictly exceeds
$P(\mathcal{L}_{R,P_0}>t)=(1-\Phi(t))^2$, so $\Lambda_n$ strictly stochastically
dominates $\mathcal{L}_{R,P_0}$.  The
strictly increasing function $h(x)=\tfrac12\tanh(x)\in\mathcal H$ therefore strictly separates
the two laws, with
$\int h\,d\Lambda_n-\int h\,d\mathcal{L}_{R,P_0}=\varepsilon$ for some
$\varepsilon>0$ that does not depend on $n$.

Let $\widehat{\mathcal{L}}_n$ be any uniform upper-bound estimator for the law of
$T_{n,R,P}$.  Its defining
condition (\Secapp~\ref{S:validity}) at $(R,P_n)$, applied to the nondecreasing $h$ above and using
$\mathcal{L}_{n,R,P_n}=\Lambda_n$, gives $\int h\,d\widehat{\mathcal{L}}_n\ge\int
h\,d\Lambda_n-o_{P_n}(1)$.  The integrand of the second term of the maximum in
Lemma~\ref{L:twopoint} is therefore bounded and converges to zero in
$P_n$-probability, so
\[
\EE{P_n}{\Bigl(\textstyle\int h\,d\Lambda_n-\int h\,d\widehat{\mathcal{L}}_n\Bigr)^{+}}
\longrightarrow0 .
\]
The lemma then bounds the first term away from zero:
\[
\liminf_{n}\ \EE{P_0}{d_{\mathrm{BL}}\bigl(\widehat{\mathcal{L}}_n,\mathcal{L}_{R,P_0}\bigr)}
\ \ge\ \frac{\varepsilon\,(1-\tau(c))}{4}\ >\ 0 .
\]
A uniform upper-bound estimator is therefore not pointwise consistent at $P_0$. The same
inequality excludes any uniformly consistent estimator for the law of $T_{n,R,P}$ under $P$, whose consistency at both
$P_0$ and $P_n$ would drive both terms of the maximum to zero (see
\citealt{Voronin2025} for a linear-programming analogue).

\paragraph{The impossibility for tests.}
The impossibility result developed above also extends to pointwise exact and uniformly valid tests built upon critical values.  Let $\alpha<1/4$, and suppose a critical value $\widehat c_n$ converges in
$P_0$-probability to a constant $\bar c$ and gives a pointwise exact test at $P_0$, so that
$P_0(T_{n,R,P_0}>\bar c)=\alpha$.  Because the distribution of $T_{n,R,P_0}$ has an atom of mass $3/4$ at zero and is
continuous on $(0,\infty)$, this forces $\bar c>0$.  We next show that $\widehat c_n\to\bar c$ in
$P_n$-probability as well.  The likelihood ratio of $P_n^{(n)}$ to $P_0^{(n)}$ is the function
\[
 L_n(x)=\prod_{i=1}^n\frac{\phi(x_{i2}-c/\sqrt n)}{\phi(x_{i2})}
 =\exp\Bigl(\frac{c}{\sqrt n}\sum_{i=1}^n x_{i2}-\frac{c^2}{2}\Bigr),\qquad x\in\R^{2n},
\]
where $\phi$ is the standard normal density.  Because $n^{-1/2}\sum_i x_{i2}$ is standard normal under
$P_0^{(n)}$, the tail integral
\[
 \int_{\{L_n>M\}}L_n\,dP_0^{(n)}=\int_{\{e^{cz-c^2/2}>M\}}e^{cz-c^2/2}\,\phi(z)\,dz
\]
is the same for every $n$ and vanishes as $M\to\infty$.  Therefore, for any events $A_n$ and any $M>0$,
\[
 P_n(A_n)=\int_{A_n}L_n\,dP_0^{(n)}\leq M\,P_0(A_n)+\int_{\{L_n>M\}}L_n\,dP_0^{(n)}.
\]
Let $\eta>0$ and define the events $A_n=\{|\widehat c_n-\bar c|>\eta\}$, which satisfy $P_0(A_n)\to0$.  Letting $n\to\infty$ and then $M\to\infty$ in the preceding display gives
$P_n(A_n)\to0$.  As the law of $T_{n,R,P_n}$ does not depend on $n$ and is continuous at $\bar c$, it follows
that the probability $P_n(T_{n,R,P_n}>\widehat c_n)$ converges to the constant $P_n(T_{n,R,P_n}>\bar c)$, which
exceeds $P_0(T_{n,R,P_0}>\bar c)=\alpha$ by the strict dominance above.  The test is therefore not uniformly valid.

This obstruction parallels the impossibility result of \citet{HiranoPorter2012}
for nondifferentiable functionals, although Lemma~\ref{L:twopoint} does not restrict
attention to regular or unbiased estimators.  \citet{RomanoShaikh2012} develop the
underlying distinction between pointwise and uniform validity.

\fi

\bibliographystyle{chicago}
\begin{spacing}{1}
\bibliography{Thesisbib}
\end{spacing}

\end{document}